\documentclass[11pt]{article}

\usepackage{amsmath}
\usepackage{epsfig, amssymb, amsfonts, dsfont}
\usepackage{amsthm}
\usepackage{amstext}
\usepackage{amsopn}
\usepackage{mathrsfs}
\usepackage{subfigure}
\usepackage{graphicx}
\usepackage[left=2.3cm,top=2cm,right=2.3cm,bottom=2cm, nohead,foot=1cm]{geometry}
\usepackage[makeroom]{cancel}
\usepackage{color}
\usepackage{enumerate}

\usepackage{comment}
\usepackage{stackrel}
\usepackage{mathtools}
\usepackage{amsbsy}

\makeatletter
\newcommand*\rel@kern[1]{\kern#1\dimexpr\macc@kerna}
\newcommand*\widebar[1]{%
  \begingroup
  \def\mathaccent##1##2{%
    \rel@kern{0.8}%
    \overline{\rel@kern{-0.8}\macc@nucleus\rel@kern{0.2}}%
    \rel@kern{-0.2}%
  }%
  \macc@depth\@ne
  \let\math@bgroup\@empty \let\math@egroup\macc@set@skewchar
  \mathsurround\z@ \frozen@everymath{\mathgroup\macc@group\relax}%
  \macc@set@skewchar\relax
  \let\mathaccentV\macc@nested@a
  \macc@nested@a\relax111{#1}%
  \endgroup
}
\makeatother

\numberwithin{equation}{section}

\newtheorem{theorem}{Theorem}[section]

\newtheorem{lemma}[theorem]{Lemma}

\newtheorem{corollary}[theorem]{Corollary}
\newtheorem{proposition}[theorem]{Proposition}

\newtheorem{notation}[theorem]{Notation}
\theoremstyle{definition}
\newtheorem{definition}[theorem]{Definition}
\newtheorem{remark}[theorem]{Remark}
\newtheorem{example}[theorem]{Example}

\newcommand{\R}{{\mathbb R}}

\DeclareFontFamily{U}{mathx}{\hyphenchar\font45}
\DeclareFontShape{U}{mathx}{m}{n}{
      <5> <6> <7> <8> <9> <10>
      <10.95> <12> <14.4> <17.28> <20.74> <24.88>
      mathx10
      }{}
\DeclareSymbolFont{mathx}{U}{mathx}{m}{n}
\DeclareMathSymbol{\bigtimes}{1}{mathx}{"91}

\usepackage{relsize}

\newcommand{\Z}{{\mathbb Z}}

\newcommand{\sbf}{{s}}
\newcommand{\bbf}{{b}}

\date{\vspace{-9ex}}

\title{Weak-disorder limit at criticality for  directed polymers\\ on  hierarchical  graphs}

\date{  }

  \author{ \textbf{Jeremy Thane Clark}\footnote{ {\tt
jeremy@olemiss.edu}} \vspace{.1cm}  \\  University of Mississippi, Department of Mathematics   }

\begin{document}
\maketitle

\begin{abstract}
We prove a distributional limit theorem conjectured in [Journal of Statistical Physics \textbf{174}, No.\ 6, 1372-1403    (2019)] for  partition functions defining models of  directed  polymers on diamond  hierarchical graphs with disorder variables placed at the graphical edges. The limiting regime involves a joint scaling in which  the number of hierarchical layers, $n\in \mathbb{N}$, of the graphs grows as the inverse temperature, $\beta\equiv \beta(n)$, vanishes with a fine-tuned dependence on $n$.  The conjecture pertains to the marginally relevant disorder  case of the  model wherein the branching parameter $b \in \{2,3,\ldots\}$ and the segmenting parameter $s \in  \{2,3,\ldots\}$ determining the  hierarchical graphs are equal, which  coincides with  the diamond fractal embedding  the  graphs having Hausdorff dimension two.     Unlike the analogous weak-disorder scaling limit for random  polymer models on  hierarchical graphs in the disorder relevant  $b<s$ case (or for the (1+1)-dimensional polymer  on the rectangular lattice),  the distributional convergence of the partition function  when $b=s$ cannot be approached through a term-by-term convergence to a Wiener chaos expansion, which does not exist for the continuum model emerging in the limit. The analysis proceeds by controlling  the distributional convergence of the partition functions  in terms of the Wasserstein distance through a perturbative generalization of Stein's method at a critical  step. In addition, we prove that a similar limit theorem holds for the analogous model with disorder variables placed at the vertices of the graphs.

\end{abstract}

\section{Introduction}\label{SectionIntro}

In probabilistic frameworks, a \textit{disordered system} usually refers to a relatively simple and familiar random object whose ``pure" probabilistic law  is distorted through its coupling to a random ``environment" formed by an array of random variables (local impurities) or a random field. If the size of the model depends on a parameter $L\in \mathbb{N}$, a central question for these disordered systems is whether typical realizations of the random environment create either a qualitative or only a quantitative change in the law of the random object as $L\nearrow \infty$.   For a given coupling strength $\beta \in [0,\infty)$ of the system to the environment, these large-scale behaviors are respectively referred to as \textit{strongly disordered} or \textit{weakly disordered}. A disordered system  is further classified as \textit{disorder relevant}  if it exhibits strong disorder for any fixed  $\beta$ as the system size grows or as \textit{disorder irrelevant} otherwise. Finally, models at the border between the disorder relevant and disorder irrelevant regimes are referred to as  \textit{marginally relevant} or \textit{marginally irrelevant}, and these boundary models manifest anomalous finer scaling behavior  as the coupling strength vanishes. 

One of the most closely studied  disorder models is the \textit{directed polymer in a random environment}, which usually refers to a $d$-dimensional simple symmetric random walk (SSRW) whose trajectories are reweighed within a Gibbsian formalism that depends on  an inverse temperature parameter, $\beta$, and an array of centered i.i.d.\ random variables labeled by  the time-space lattice $\{1,\ldots, L\}\times \Z^d$ for a polymer length $L\in \mathbb{N}$.  The parameter  $\beta$ effectively controls the strength of the polymer's coupling to the environment, and $\beta=0$ corresponds to a pure SSRW.  Established results in this field imply that the $(d+1)$-dimensional polymer model is disorder relevant when $d=1$, marginally relevant when $d=2$, and disorder irrelevant in all higher dimensions; see Comets's recent book~\cite{Comets}. 

In this article, we prove a distributional  limit theorem   for partition functions defined from a hierarchical model  for  directed polymers in a random environment for which the disorder is marginally relevant.  Our limiting regime, which involves a joint scaling wherein the number of hierarchical layers of the model grows while the disorder strength decays to zero,  is similar to the critical weak-disorder scaling regime for $(2+1)$-dimensional polymers proposed by  Caravenna, Sun, and Zygouras in~\cite{CSZ2,CSZ4}.  While~\cite{CSZ4} proves the existence of a subsequential distributional limit of the partition functions within this critical scaling regime and fully characterizes the correlation structure of any such  limit, the uniqueness of the subsequential distributional limit currently remains open.  Although the hierarchical symmetry of the model considered in this article  makes a detailed limit analysis within the critical weak-disorder regime less difficult than for the  rectangular lattice polymer model with marginally relevant disorder, the hierarchical setting provides some  insights that are likely general for  weak-disorder scaling limits at criticality for marginally relevant  systems. 

The continuum polymer model corresponding to the scaling limit of this article is studied in~\cite{Clark3,Clark4}.  We will return to a broader discussion of related work in Section~\ref{SecDiscussion} after  defining  our hierarchical model and presenting a first version of our main result.


\section{The setup and a statement of the main result} \label{SectionMainResult}

This section begins by defining a family of random measures on directed paths crossing diamond hierarchical graphs and concludes with the statement of a limit theorem for the total masses  of the measures (Theorem~\ref{ThmMain}), which was conjectured in~\cite{Clark1}.  The models  in this section have \textit{bond-disorder}, i.e., disorder variables placed at the edges of the graphs, while the models discussed in the next section have disorder at the vertices.

\subsection{Construction of the diamond hierarchical graphs}\label{SecDHG}
Hierarchical diamond graphs $D_n^{b,s}$, $n\in \mathbb{N}_0$ are recursively  defined  through a construction determined by a  branching number $b\in \{2,3,\ldots\}$ and a segmenting number  $\sbf \in \{2,3,\ldots\}$. The zeroth  graph, $D_0^{b,s}$, is simply  two root vertices, $A$ and $B$, with an edge between them. The first-generation graph, $D_1^{b,s}$, is formed by $b$ parallel branches connecting $A$ and $B$, wherein each branch has $s$ edges running in sequence.  For $n\geq 2$ the graph $D_n^{b,s}$ is defined recursively from $D_{n-1}^{b,s}$ by embedding a copy of $D_1^{b,s}$ in place of each edge on  $D_{n-1}^{b,s}$.  The set of edges, $E_n^{b,s}$, on $D_n^{b,s}$ thus contains $(bs)^n$ elements.  
\begin{center}
\includegraphics[scale=.6]{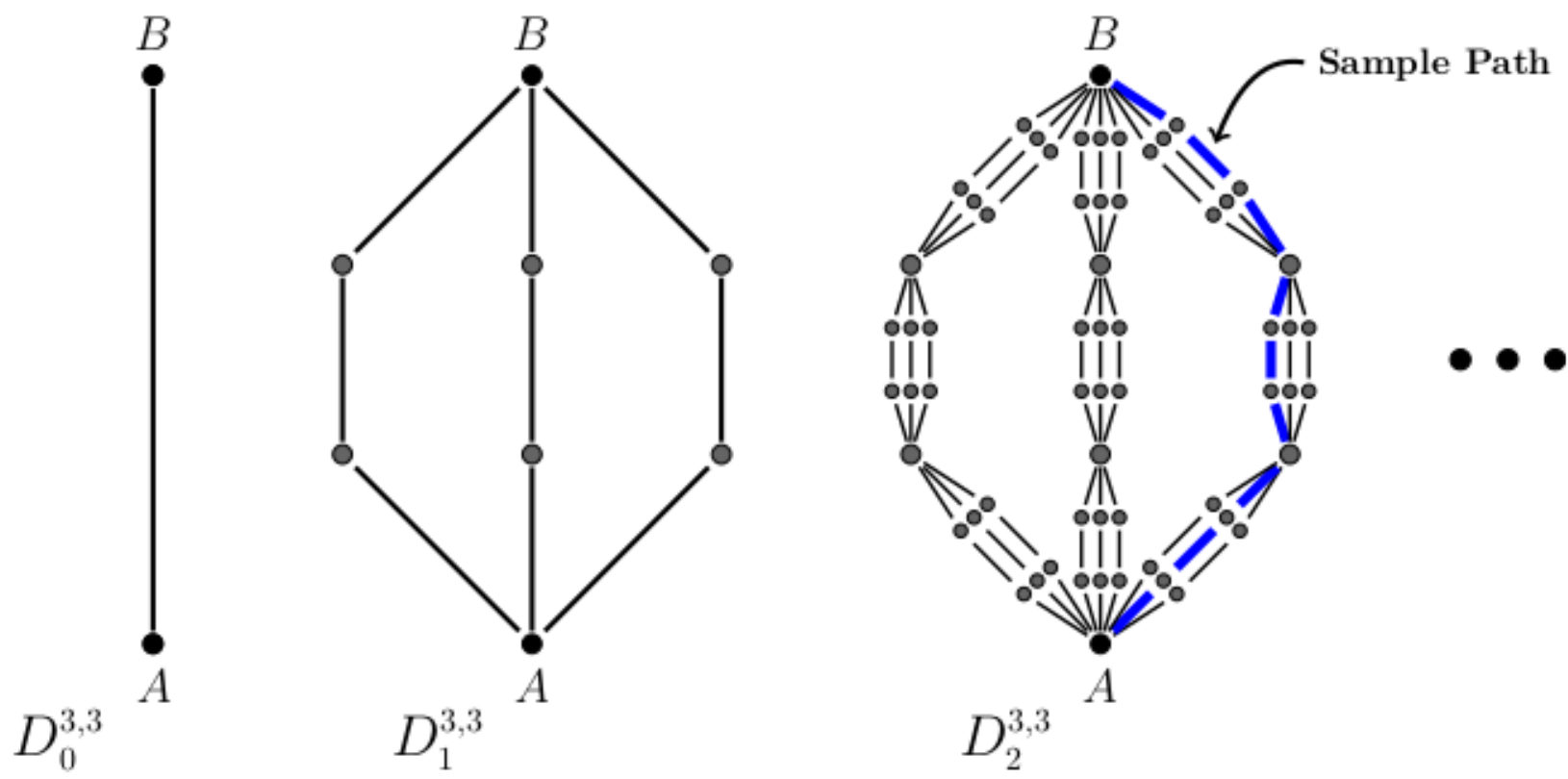}\\
\small The first three recursively-defined diamond graphs with $\bbf=3$  and $\sbf=3$.
\end{center}
A \textit{directed path} on $D_{n}^{b,s}$ is a function $p:\{1,\ldots, s^n\}\rightarrow E^{b,s}_n$ for which $p(1)$ is incident to $A$, $p(s^n)$ is incident to $B$, and successive edges $p(k)$, $p(k+1)$ share a common vertex for $1\leq k<s^n$. In other terms, the path moves monotonically upwards from $A$ up to $B$, as seen in the figure.  We denote the set of directed paths on $D_{n}^{b,s}$  by $\Gamma_{n}^{b,s}$.

\vspace{.4cm}

\subsection{Random Gibbsian measure on directed paths}
Next we define a random Gibbs measure on the space $\Gamma^{b,s}_n$ of directed paths.  Let $\omega_{h}  $   be an i.i.d.\ family of   random variables labeled by $h\in E_n^{b,s}$ and  having mean zero, variance one, and finite exponential moments, $\mathbb{E}\big[\exp\{\beta \omega_h\}\big]$ for $\beta\geq 0$. Given an inverse temperature value $\beta\in [0,\infty)$, we define a random path measure on directed paths such that the weight assigned to  $p\in \Gamma^{b,s}_n$ is given by
\begin{align*}
\mathbf{M}^{\omega}_{\beta, n}(p) \,   = \,\frac{1}{|\Gamma^{b,s}_n|} \frac{  e^{\beta  H_{n}^{\omega}(p) }    }{ \mathbb{E}\big[ e^{\beta  H_{n}^{\omega}(p) }     \big]   }\quad \quad \text{ for path energy }\quad \quad  H_{n}^{\omega}(p) \, :=  \, \sum_{a\in p} \omega_{a}   \, ,  
\end{align*}
where  $a\in p$ means that the edge $a\in E^{b,s}_n$ lies along the path $p$.  At infinite temperature ($\beta=0$), $\mathbf{M}^{\omega}_{\beta, n}$ is a uniform probability measure on $\Gamma_n^{b,s}$.  We denote the total mass of $\mathbf{M}^{\omega}_{\beta, n} $ by
\begin{align}\label{Partition}
 W_{n}^\omega(\beta)\,:=\, \mathbf{M}^{\omega}_{\beta, n}\big(\Gamma^{b,s}_n\big)\,=\,  \frac{1}{|\Gamma_{n}^{b,s}|  }\sum_{p\in \Gamma_{n}^{b,s}  }\prod_{a\in p} \frac{e^{\beta  \omega_{a}  } }{ \mathbb{E}[e^{\beta  \omega_{a}  } ]}    
\end{align}
in terms of the disorder variables $\omega_a$.  The recursive construction of the diamond graphs  implies the following distributional recursive relation for the partition functions $W_{n}^{\omega}(\beta)$:
\begin{align}\label{PartHierSymm}W_{n+1}^{\omega}(\beta)  \, \stackrel{ d }{ =}  \,  \frac{1}{b}\sum_{i=1}^{b} \prod_{j=1}^{s}  W_{n}^{(i,j)}(\beta)     \, ,
\end{align}
where the $W_{n}^{(i,j)}(\beta)$'s are independent copies of the random variable $W_{n}^{\omega}(\beta)$. The variances $\varrho_n(\beta):=\textup{Var}\big(W_{n}^{\omega}(\beta)\big)$ are recursively related as $  \varrho_{n+1}(\beta)=M_{b,s}\big(\varrho_{n}(\beta)\big)$ with  $M_{b,s}: [0,\infty)\rightarrow  [0,\infty) $ defined as
\begin{align}\label{RecurVar}
   M_{b,s}(x)\,:=\,&\frac{1}{b}\Big[(1+x)^s\,-\,1   \Big]\ . \nonumber
\intertext{Notice that the map $M_{b,s}$ has a fixed point at $x=0$ and for $0<x\ll 1$ }
\,=\, & \begin{cases} \frac{s}{b}x +\mathcal{O}(x^2) & \quad  s\neq b\,,  \\  x+\frac{b-1}{2}x^2+\mathit{O}(x^3)     & \quad  s=b\,.  \end{cases} 
\end{align}   
Thus the fixed point is linearly attractive when $b>s$, linearly repelling when $b<s$, and marginally repelling when $b=s$.

\subsection{High-temperature scaling limits for the Gibbs measure}
Our focus is on high-temperature (i.e., weak-disorder) scaling limits  in which the hierarchical level parameter, $n$, grows as the inverse temperature $\beta=\beta(n)$ decays under an appropriate tuning in $n$ such that the random path measures $\mathbf{M}^{\omega}_{\beta, n}$ converge in distribution to a limiting random measure on paths. This article concerns only  the total mass of the measures while~\cite{Clark3} extends this limit analysis to the full measures and discusses some delicate properties of the limiting path measures. High-temperature scaling limits are only of interest in the cases  $b<s$ and $b=s$    for which $x=0$ is a repelling fixed point of the variance map $M_{b,s}$. The article~\cite{US} contains a limit theorem for $W_{n}^{\omega}(\beta)$ in the case $b<s$, where for a fixed parameter value $r\in \R_+$ the inverse temperature  $\beta\equiv\beta_{n,r}^{b,s}$ has the large $n$ asymptotic form
\begin{align}\label{BetaFormb<s}
\beta_{n,r}^{b,s}\,=\,\sqrt{r}\Big(\frac{b}{s}  \Big)^{n/2}\,+\,\mathit{o}\bigg(\Big(\frac{b}{s}  \Big)^{n/2}\bigg)\,.
\end{align}
The sequences of random variables $\{W_{n}^{\omega}(\beta_{n,r}^{b,s})\}_{n\in \mathbb{N}}$ converge in distribution as $n\rightarrow \infty$ to a family of limit laws $\mathbf{W}_{r}$ supported on $(0,\infty)$ that satisfy the distributional recursion relation
\begin{align*}\mathbf{W}_{\frac{s}{b}r} \, \stackrel{ d }{ =}  \,  \frac{1}{b}\sum_{i=1}^{b} \prod_{j=1}^{s}  \mathbf{W}_{r}^{(i,j)}     \, ,
\end{align*}
where $\mathbf{W}_{r}^{(i,j)}  $ are i.i.d.\ copies of $\mathbf{W}_{r}$.  The variance, $R_{b,s}(r)$, of $\mathbf{W}_{r}$ satisfies $M_{b,s}\big(  R_{b,s}(r)  \big)\,=\,R_{b,s}(\frac{s}{b}r) $.  Of course, the exponential form of the inverse temperature scaling~(\ref{BetaFormb<s}) corresponds to the linear repelling~(\ref{RecurVar}) of the map $M_{b,s}$ from $x=0$ that occurs in the $b<s$ case.   

The main result of the current article is a proof of  an analogous limit theorem for $W_{n}^{\omega}(\beta)$ in the $b=s$ case.   An inverse temperature scaling\textemdash see below in~(\ref{BetaForm})\textemdash was proposed in~\cite{Clark1} although the results therein were confined to proving convergence of the positive integer moments.\footnote{The scaling~(\ref{BetaForm}) includes a correction pointed out by an anonymous referee that ensures  consistency with the variance asymptotics~(\ref{VarAsym}) below; see Appendix~\ref{AppendBetaScale} for an outline of the computation determining~(\ref{BetaForm}) from~(\ref{VarAsym}).  The variance asymptotics is what plays a direct role in all subsequent analysis.}  Although the convergence of the positive integer moments implies the existence of subsequential distributional limits, it does not imply convergence  in law because the higher limiting moments increase super-factorially; see (III) of Theorem~\ref{ThmHM} below.   For fixed $b\in \{2,3,4,\ldots\}$ and  $r\in \R$, let the sequence $(\beta_{n, r}^{(b)})_{n\in \mathbb{N}}$ have the large $n$ asymptotics
\begin{align}\label{BetaForm}
\beta_{n, r}^{(b)}\, :=\,  \frac{\kappa_{b}}{\sqrt{n}}\,-\,\frac{ \kappa_{b}^2 \tau}{2n}\,+\,\frac{\kappa_{b}\eta_{b}\log n}{2n^{\frac{3}{2}}}\,+\,\frac{\kappa_{b}r+\kappa_b^3(\frac{5}{4} \tau^2-\frac{7}{12}\tau'- \frac{1}{2} ) }{2n^{\frac{3}{2}}}\, +\,\mathit{o}\Big( \frac{1}{n^{\frac{3}{2}}} \Big) \,,
\end{align}
where $\tau:=\mathbb{E}[\omega_a^3]$ and $\tau':=\mathbb{E}[\omega_a^4]-3$ are respectively the third and fourth cumulants of the disorder variables, $\omega_a$, and the constants $\kappa_{b},\eta_{b}>0$ are defined as
\begin{align}
 \kappa_{b}\,:=\,\sqrt{\frac{2}{b-1}}  \hspace{1.2cm} \text{and}\hspace{1.2cm} \eta_{b}:=\frac{b+1}{3(b-1) }\,.\label{Constants}
\end{align}
If we let $M_{b,b}^n$ denote the $n$-fold composition of $M_{b,b}$, the variance, $ \varrho_n\big(\beta_{n, r}^{(b)}\big) $,  of $W_{n}^{\omega}\big(\beta_{n,r}^{(b)}\big)$ can be written explicitly as
\begin{align}
\varrho_n\big(\beta_{n, r}^{(b)}\big)\,=\,&M_{b,b}^n\big( \varrho_0(\beta_{n, r}^{(b)}) \big)\,, \text{ where $\varrho_0\big(\beta_{n, r}^{(b)}\big)$ has the large $n$ asymptotics  } \nonumber \\    \varrho_0\big(\beta_{n, r}^{(b)}\big)\,:=\,&\textup{Var}\Bigg( \frac{  e^{\beta_{n,r}^{(b)}\omega  }  }{\mathbb{E}\big[  e^{\beta_{n,r}^{(b)}\omega  } \big] }   \Bigg)\,=\, \kappa_{b}^2 \bigg(\frac{ 1 }{n}  +\frac{ \eta_{b}\log n}{n^2}+\frac{r}{n^2}\bigg)\,+\,\mathit{o}\Big(\frac{1}{n^2}  \Big)\,.\label{VarAsym}
\end{align}
 The basic observations above combined with Lemma~\ref{LemVar} below imply that $\varrho_n\big(\beta_{n, r}^{(b)}\big)$ converges as $n\rightarrow \infty$ to a limit $R_b(r)$ for any $r\in \R$.\vspace{.2cm}

\begin{remark}\label{RemarkBetaScale} Let us  set the skewness, $\tau$, of the disorder variables to zero  for simplicity. Theorem 7.1 of~\cite{US} states that if $\beta_{n, r}^{(b)}$ is replaced by a coarser scaling of the form $\hat{\beta} /\sqrt{n}$ for a parameter $\hat{\beta}\in \R_+$, then $ W_{n}^{\omega}\big(\hat{\beta} /\sqrt{n}\big)$ has the distributional behaviors listed below depending on $\hat{\beta}$ as $n\rightarrow \infty$.  
\begin{align*}
&W_{n}^{\omega}\big(\hat{\beta} /\sqrt{n}\big)\,\stackrel{d}{\approx} \, 1\,+\,\frac{1}{\sqrt{n}}\cdot \mathcal{N}\bigg(0, \frac{1}{1/\hat{\beta}^{2} - 1/\kappa_b^{2} } \bigg) &\text{$\hat{\beta} <\kappa_b $}&  \\  &W_{n}^{\omega}\big(\hat{\beta} /\sqrt{n}\big)\,\stackrel{d}{\approx}    1\,+\,\frac{1}{\sqrt{\log n}}\cdot\mathcal{N}\Big(0, \frac{6}{b+1} \Big) &\text{$\hat{\beta} =\kappa_b $}&   \\
&\text{The variance of  $W_{n}^{\omega}(\hat{\beta} /\sqrt{n})$ blows up. }   &\hat{\beta} >\kappa_b &
\end{align*}
In the above, we use the notation $\stackrel{d}{\approx}$  heuristically to mean that the random variables  are ``close" in distribution. Thus $\kappa_b$ is a critical point for the parameter $\hat{\beta}$ in the moment behavior of  $W_{n}^{\omega}\big(\hat{\beta} /\sqrt{n}\big)$ when $n\gg 1$, and  $\beta_{n, r}^{(b)}$  falls within a critical window around $\kappa_b$.  The variance blowup after $\kappa_b$ coincides with the transition to strong disorder as can be seen in the limit model emerging under the scaling~(\ref{BetaForm}) as $n\rightarrow \infty$; see Remark~\ref{RemarkTrans}. 
\end{remark}

\begin{remark}\label{RemarkLog} The critical inverse temperature scaling for   (2+1)-dimensional directed polymers considered in~\cite{CSZ4} has the form  
$
\beta_{L,r}= \frac{ \sqrt{\pi}  }{(\log L)^{1/2}  }-\frac{\pi\tau }{2 \log L }+\frac{\sqrt{\pi} r+\pi^{3/2}(\frac{5}{4}\tau^2-\frac{7}{12}\tau'-\frac{1}{2}    )   }{ 2(\log L)^{3/2}  }+\mathit{o}\big( \frac{1}{ (\log L)^{3/2} }\big)$ for $L\gg 1$, where $L$ is the polymer length, $r\in \R$ is a parameter, and $\tau,\tau'$ are the third and fourth cumulants of the disorder variables; see~\cite[Remark 1.1]{CSZ4}.   In terms of the length $L=b^n$ of the diamond graph polymers, the asymptotic form~(\ref{BetaForm}) is fairly similar 
except for the inclusion of the term $\frac{\log \log L}{(\log L)^{3/2}}$.

\end{remark}

\subsection{Previous results on the centered moments}

The lemma and theorem below are results from~\cite{Clark1}.

\begin{lemma}[Variance function]\label{LemVar}For any $b\in \{2,3,\ldots\}$, there exists a unique continuously differentiable increasing function $R_b:\R\rightarrow \R_+$  satisfying the properties (I)-(III) below.
\begin{enumerate}[(I)]
\item Composition of $R_b(r)$ with the map $M_{b,b}$ translates the parameter $r$:\, $ M_{b,b}\big(R_b(r)\big)\,=\, R_b(r+1) $.

\item As $r\rightarrow \infty$, $R_b(r)$ diverges to $\infty$.  As $r\rightarrow -\infty$, $R_b(r)$ has the vanishing asymptotics 
$$  R_b(r)\,=\,-\frac{ \kappa_b^2 }{ r }\,+\, \frac{ \kappa_b^2\eta_b\log(-r) }{ r^2 }\,+\,\mathit{O}\bigg( \frac{\log^2(-r)}{|r|^3} \bigg)  \,. $$

\item The derivative $R_b'(r)$ admits the limiting form
$$ R'_b(r)\,=\,\lim_{n\rightarrow \infty}\frac{\kappa_b^2}{n^2}\prod_{k=1}^n\big(1+R_b(r-k)\big)^{b-1} \,.     $$

\end{enumerate}
Moreover,   if for some $r\in \R$ the sequence of positive real numbers $(x^{n,r})_{n\in \mathbb{N}} $ has the large $n$ asymptotics 
\begin{align}\label{xAssump}
 x^{n,r}  =  \kappa_{b}^2 \bigg(\frac{ 1 }{n}  +\frac{ \eta_{b}\log n}{n^2}+\frac{r}{n^2}\bigg)\,+\,\mathit{o}\Big(\frac{1}{n^2}  \Big) \,,
\end{align}
 then $ M^{n}_{b,b}(x^{n,r}   )$ converges as $n\rightarrow \infty$ to $R_b(r) $.

\end{lemma}

Appendix~\ref{AppendixVarCheck} contains an elementary but instructive calculation showing the consistency between properties (I) and (II) above.
The higher centered moments of $W_{n}^{\omega} \big(\beta_{n, r}^{(b)} \big)$ converge to limits $R_b^{(m)}(r)$ characterized as follows.

\begin{theorem}[Limiting higher moments] \label{ThmHM} Fix $b\in \{2, 3,\ldots\}$ and let $s=b$.  For each $m\in\{2,3,\ldots\}$ there is a continuous, increasing  function $R^{(m)}_b:\R \rightarrow [0,\infty)$ such that for any $r\in \R$ 
\begin{align}\label{ConvHM}
\mathbb{E}\Big[ \big(W_{n}^{\omega} \big(\beta_{n, r}^{(b)} \big)-1\big)^m\Big]\hspace{.5cm}\stackrel{n\rightarrow \infty}{\longrightarrow}\hspace{.5cm}  R^{(m)}_b(r)\,.
\end{align}
The  limit functions $R^{(m)}_b$ satisfy properties (I)-(III) below. 
\begin{enumerate}[(I)]
 \item There are  multivariate polynomials $P_m: \R^{m-1}\rightarrow \R$   with nonnegative coefficients such that for all $r\in \R$
$$ R_{b}^{(m)}(r+1)\,=\, P_m\big(R_{b}^{(2)}(r), R_{b}^{(3)}(r), \ldots, R_{b}^{(m)}(r)   \big)     \,. $$

\item $R^{(m)}_b(r)$ diverges to $\infty$ as $r\rightarrow \infty$ and  vanishes as $r\rightarrow -\infty$ with the asymptotics $ R^{(m)}_b(r)\sim\kappa_b^m \frac{m!}{2^{m/2}(m/2)!}|r|^{-m/2} $ for $m$ even and $ R^{(m)}_b(r)= \mathit{O}\big(   |r|^{- (m+1)/2 }\big)$ for $m$ odd. 


\item  There is a $c>0$ such that  $
\frac{ \log\log( R_{b}^{(m)}(r)  )  }{ m }> c $ holds for any fixed $r\in \R$ and large enough $m\in \mathbb{N}$.

\end{enumerate}

\end{theorem}

\begin{remark} The function $R_b(r)$ in the statement of Lemma~\ref{LemVar} is equal to $R^{(2)}_b(r)$   in the statement of Theorem~\ref{ThmHM}.

\end{remark}

\begin{remark}The quantity $\kappa_b^m \frac{m!}{2^{m/2}(m/2)!}|r|^{-m/2}$ in (II) for $m$ even agrees with the $m^{th}$ moment of a centered normal random variable with variance $\kappa^2_b /|r|$.

\end{remark}

\subsection{A first version of the main result}

As mentioned above, Theorem~\ref{ThmHM} does not imply that  $W_{n}^{\omega} \big(\beta_{n, r}^{(b)} \big)$ converges in law as $n\rightarrow \infty$ since $R^{(m)}_b(r)$ grows super-factorially with $m\in \mathbb{N}$ by (III) of Theorem~\ref{ThmHM}. Thus the  following theorem was left as a  conjecture in~\cite{Clark1}.

\begin{theorem}\label{ThmMain} Fix $b\in \{2,3,\ldots\}$ and $r\in \R$, and let the sequence $(\beta_{n,r}^{(b)})_{n\in \mathbb{N} }$ have the form~(\ref{BetaForm}).  When $s=b$ there is convergence in distribution as $n\rightarrow \infty$
$$ W_{n}^{\omega}\big(\beta_{n,r}^{(b)}\big)\hspace{1cm}\Longrightarrow \hspace{1cm} L_{r}^{(b)}    $$
to a family of limit laws $\big\{L_{r}^{(b)}\big\}_{r\in \R}$ uniquely determined by (I)-(IV) below.

\begin{enumerate}[(I)]

\item $L_{r}^{(b)}$ has mean $1$ and variance $R_b(r)$.

\item For $m\in \{3,4,\ldots\}$, the $m^{th}$ centered moment of $L_{r}^{(b)}$ is equal to $R^{(m)}_b(r)$.

\item Let $\mathbf{W}_{r}$ be a random variable with distribution $L_{r}^{(b)}$.  The centered variables $  \sqrt{-r} (\mathbf{W}_{r}-1) $ converge in law as $r\rightarrow -\infty$  to a centered normal with variance $\kappa_b^2$.

\item  If $\mathbf{W}^{(i,j)}_r$ are independent variables with distribution $L_{r}^{(b)}$, then there is equality in distribution
\begin{align*}
\mathbf{W}_{r+1} \,\stackrel{d}{=}\,\frac{1}{b}\sum_{1\leq i\leq b}  \prod_{1\leq j\leq b} \mathbf{W}^{(i,j)}_r   \,.    
\end{align*}

\end{enumerate}

\end{theorem}

\begin{remark} The  convergence in distribution of  $  \sqrt{-r} (\mathbf{W}_{r}-1) $ to $\mathcal{N}(0,\kappa^2_b)$  as $r\rightarrow -\infty$  follows  from the asymptotics for the centered moments $R^{(m)}_b(r)$ in (II) of Theorem~\ref{ThmHM}.

\end{remark}

\begin{remark}\label{RemarkTrans}The family of limit laws in Theorem~\ref{ThmMain} exhibits a transition from weak disorder to strong disorder as $r$ goes from $-\infty$ to $+\infty$ in the sense that the random variables $\mathbf{W}_{r}$ converge in probability to one as $r\rightarrow -\infty$ and to zero as $r\rightarrow \infty$, where the latter is proved in~\cite[Section~5]{Clark4} using a conditional Gaussian multiplicative chaos structure that we will describe at the end of Section~\ref{SecDiscussion}.
\end{remark}

\section{A similar limit theorem for the site-disorder model}\label{SecSiteDisorderThm}
Next we will state an  analogous result to Theorem~\ref{ThmMain} corresponding to when the environmental disorder is built into the partition function through the vertices of the diamond graphs rather than the edges.

For $n\in \mathbb{N}_0$ and $b,s\in \{2,3,\ldots\}$, let $V^{b,s}_n$ denote the set of vertices on the $n^{th}$ diamond graph $D_n^{b,s}$ with the  roots $A$ and $B$ excluded.  Thus $V^{b,s}_0=\emptyset$, and for $n\geq 1$ the number of non root vertices is given by $\big|V^{b,s}_{n}\big|=b(s-1)\frac{(bs)^n-1   }{bs-1  }   $.  The hierarchical construction of the sequence of diamond graphs in Section~\ref{SecDHG}  implies that $V^{b,s}_{n-1}$ is canonically identifiable with a subset of $V^{b,s}_{n}$ for each $n\in \mathbb{N}$, and we refer to $V^{b,s}_{n}\backslash V^{b,s}_{n-1}$ as the set of \textit{generation-$n$ vertices}. 

 As before, let $\{\omega_a \}_{a\in V_n^{b,s}}$ be an i.i.d.\ family of centered random variables with variance one and finite exponential moments. We define the partition function $\widehat{W}_n^{\omega}(\beta)$  in analogy to   $W_n^{\omega}(\beta)$ in~(\ref{Partition}) except with the product of  random variables $e^{\beta \omega_a}/\mathbb{E}[ e^{\beta \omega_a} ]$ running over all vertices $a\in V_n^{b,s}$  along the path $p\in \Gamma_n^{b,s}$:
\begin{align}
 \widehat{W}_n^{\omega}(\beta)\,:=\,      \frac{1}{|\Gamma_{n}^{b,s}|} \sum_{p \in \Gamma_{n}^{b,s} }    \prod_{  a\pmb{\in} p   } \frac{ e^{\beta\omega_a}}{\mathbb{E}[ e^{\beta\omega_a}] }            \,, 
 \end{align}
where the notation $a\pmb{\in} p$ is used for a vertex $a\in V_n^{b,s}$ and a path $p:\{1,\ldots,s^n\}\rightarrow E_n^{b,s}$ to indicate that one of the edges $p(k)\in E_n^{b,s}$ for $k\in \{2,\dots, s^n-1\}$   is incident to $a$.  When $n=0$ the partition function  $\widehat{W}_n^{\omega}(\beta)$ is simply equal to $1$ since   $V_0^{b,s}=\emptyset$, and the hierarchical symmetry of the model implies the following   distributional equality, which is similar to~(\ref{PartHierSymm}): 
\begin{align}\label{PartHierSymmII}
\widehat{W}_{n+1}^{\omega}(\beta)\,\stackrel{d}{=}\, \frac{1}{b}\sum_{i=1}^b \Bigg(\prod_{j=1}^{s}\widehat{W}_{n}^{(i,j)}(\beta) \Bigg)\Bigg( \prod_{\ell=1}^{s-1}\frac{e^{\beta \omega_{i,\ell}}}{\mathbb{E}\big[ e^{\beta \omega_{i,\ell}} \big]} \Bigg)    \,,
\end{align}
where $\widehat{W}_{n}^{(i,j)}(\beta)$ are i.i.d.\ copies of $\widehat{W}_{n}^{\omega}(\beta)$ and $\omega_{i,\ell}$ are i.i.d.\ copies of the disorder variable.  The terms $e^{\beta \omega_{i,\ell}}/\mathbb{E}[ e^{\beta \omega_{i,\ell}} ]$  correspond to the  generation-$1$ vertices of the diamond graph $D_{n+1}^{b,s}$.\vspace{.2cm}

The following theorem is the counterpart to Theorem~\ref{ThmMain} for the site-disorder model, and its proof is in Section~\ref{SecSiteDisorder}.  
\begin{theorem}\label{ThmMainSite} Fix $b\in \{2,3,\ldots\}$ and $r\in \R$, and assume $s=b$.  Define $\widehat{\kappa}_b := \frac{ \pi\sqrt{b}   }{\sqrt{2}(b-1)   } $, and let $\tau$ and $\eta_b$ be  defined as in~(\ref{BetaForm}).  If the sequence $\big\{\widehat{\beta}_{n,r}^{(b)}\big\}_{n\in \mathbb{N}}$ has the asymptotic form
\begin{align}\label{BetaForm2}
\widehat{\beta}_{n,r}^{(b)}\,=\,\frac{  \widehat{\kappa}_b }{ n }\,+\,\frac{  \widehat{\kappa}_b\eta_b \log n  }{ n^2 }\,+\,\frac{ \widehat{\kappa}_b r-\widehat{\kappa}_b^2\frac{\tau}{2} }{ n^2 }\,+\,\mathit{o}\Big( \frac{1}{n^2} \Big)\,,
\end{align} then   $\widehat{W}_{n}^{\omega}\big(\widehat{\beta}_{n,r}^{(b)}\big)$ converges in distribution as $n\rightarrow \infty$ to the limit law $\mathbf{W}_{r}$ of Theorem~\ref{ThmMain}.
\end{theorem}

\begin{remark}\label{RemarkBetaScaleII} Define  $\upsilon_b:\big[0,\widehat{\kappa}_b\big)\rightarrow [0,\infty)$ by $\upsilon_b(\hat{\beta}):=\hat{\beta}\frac{ \sqrt{2} }{\sqrt{b}  }\tan\big(\frac{\pi}{2}\frac{\hat{\beta}}{ \widehat{\kappa}_b}   \big)$.  In the case of  $s=b$,~\cite[Thm.\ 2.5]{US} states that  the partition function $ \widehat{W}_{n}^{\omega}(\hat{\beta} /n )$ has the large $n$  distributional behaviors listed below depending on the parameter $\hat{\beta}\geq 0$. 
\begin{align*}
&\widehat{W}_{n}^{\omega}\big(\hat{\beta} /n\big)\,\stackrel{d}{\approx} \, 1\,+\,\frac{1}{n}\cdot \mathcal{N}\big(0, \upsilon_b(\hat{\beta})\big) &\text{$\hat{\beta} <\widehat{\kappa}_b $}&  \\  &\widehat{W}_{n}^{\omega}\big(\hat{\beta} /n\big)\,\stackrel{d}{\approx}    1\,+\,\frac{1}{\sqrt{\log n}}\cdot\mathcal{N}\Big(0, \frac{6}{b+1} \Big) &\text{$\hat{\beta} =\widehat{\kappa}_b $}&   \\
&\text{The variance of  $\widehat{W}_{n}^{\omega}(\hat{\beta} /n)$ blows up. }   &\hat{\beta} >\widehat{\kappa}_b &
\end{align*}
We use $\stackrel{d}{\approx} $ in the same heuristic sense as in Remark~\ref{RemarkBetaScale}. Thus $\widehat{\kappa}_b$ is a critical point  for the large $n$ behavior of  $\widehat{W}_{n}^{\omega}\big(\hat{\beta} /n\big)$ that is analogous to $\kappa_b$ for $W_{n}^{\omega}\big(\hat{\beta} /\sqrt{n}\big)$ as described in Remark~\ref{RemarkBetaScale}.  
\end{remark}

\begin{remark}\label{RemarkProofPlan} Our proof of Theorem~\ref{ThmMainSite} proceeds by showing that $\widehat{W}_{n}^{\omega}\big(\widehat{\beta}_{n,r}^{(b)}\big)$ is close in  $L^2$ norm to a similarly-defined partition function in which the disorder variables $e^{\beta \omega_a}/\mathbb{E}[ e^{\beta \omega_a} ]$ are only attached to vertices of generation greater than $\lfloor \log n\rfloor$. This effectively reduces the generation-$n$ site-disorder model to a 
 generation-$\lfloor \log n\rfloor$ bond-disorder model. The results developed to prove Theorem~\ref{ThmMain} can then be applied to prove Theorem~\ref{ThmMainSite}.
\end{remark}

\section{Further discussion}\label{SecDiscussion}

As mentioned in Section~\ref{SectionIntro}, the $(d+1)$-dimensional polymer model is disorder relevant when $d=1$ and  marginally relevant when $d=2$.  In principle, disorder relevance opens up the possibility that there exists a continuum disorder model that emerges in a joint limit in which the polymer length, $L$, grows as the inverse temperature $\beta\equiv \beta(L)$ vanishes with an appropriate dependence on $L$.\footnote{The general relationship between disorder relevance and continuum limits is argued for in~\cite{CSZ3}.}  A rigorous mathematical result in this direction was  developed by  Alberts, Khanin, and Quastel in the  article~\cite{alberts}, which  proved that the partition function for (1+1)-dimensional  polymers converges in law to a nontrivial distributional limit, $\mathcal{Z}_{\hat{\beta}}$, as $L\nearrow \infty$ and the inverse temperature has the asymptotic form $\beta=\big(\hat{\beta}+\mathit{o}(1)  \big)L^{-1/4}$ for a fixed parameter value $\hat{\beta}\in \R_+$.  This  scaling limit is referred to as the \textit{intermediate disorder regime} since it magnifies a parameter region between the weak ($\beta=0$) and the strong ($\beta>0$) domains of disorder behavior for the $(1+1)$-dimensional polymer, and it amounts to a continuum/weak-disorder limiting regime in which the  polymers  are diffusively rescaled towards Brownian motion trajectories while the environmental disorder variables are renormalized towards a white noise field  $W\equiv W(t,x)$ on $[0,1]\times \R$.
 The authors construct the limiting partition functions $\mathcal{Z}_{\hat{\beta}}$ in terms of Wiener chaos expansions of the field $W(t,x)$  involving the one-dimensional heat kernel $\varrho(t',x';t,x)=\frac{1}{\sqrt{2\pi(t-t')}}\textup{exp}\big\{-\frac{(x-x')^2}{2(t-t') }  \big\} $.

A model of continuum directed polymers corresponding to the limiting partition function laws $\mathcal{Z}_{\hat{\beta}}$  in~\cite{alberts} was discussed more explicitely in~\cite{alberts2}, where  $\mathcal{Z}_{\hat{\beta}}$ is equal in distribution to the   total mass of a random measure on $C([0,1])$, i.e., the space of Brownian trajectories.   Moreover, the authors use the point-to-point form, $\mathcal{Z}_{\hat{\beta}}\equiv \mathcal{Z}_{\hat{\beta}}(t',x';t,x) $, of these limiting partition function laws to   construct a solution to  the one-dimensional stochastic heat equation (SHE):
\begin{align*}
\text{}\hspace{.5cm}\partial_t\mathcal{Z}_{\hat{\beta}}\,=\,\frac{1}{2}\partial_x^2\mathcal{Z}_{\hat{\beta}}\,+\,\hat{\beta} W \mathcal{Z}_{\hat{\beta}}\,, \hspace{1cm}\mathcal{Z}_{\hat{\beta}}(t,x';t,x)=\delta_0(x'-x)\,.
\end{align*} 
In the case where $\mathcal{Z}_{\hat{\beta}}\equiv \mathcal{Z}_{\hat{\beta}}(0,0   ; 1,*   )$ corresponds to the limit of point-to-line  partition functions for polymers starting at the origin, $\mathcal{Z}_{\hat{\beta}}$  is equal in law to the total mass of a random measure $M_{\hat{\beta}}$ on  $C([0,1])$ that can be formally expressed as
\begin{align}\label{GMCForm}
 \text{}\hspace{.1cm}  M_{\hat{\beta}}( dp) \,=\, e^{\hat{\beta}\widehat{W}(p)     -\frac{\hat{\beta}^2}{2}\mathbb{E}[\widehat{W}(p)   ]     }  \mathbf{P}(dp)\,\hspace{.7cm}\text{for}\hspace{.7cm}p\in C([0,1]) \,,  
 \end{align}
where $\mathbf{P}$ is the Wiener measure on $C([0,1])$ for a standard Brownian motion and $\widehat{W}(p):=\int_0^1 W(t,p_t) dt$ defines a Gaussian field\footnote{The field $\widehat{W}(p) $ yields a Gaussian random variable   when integrated against a test function $\psi\in L^2\big(C([0,1]), \mathbf{P} \big)$.} over $C([0,1])$ with correlation kernel given by the intersection time between paths: $T(p,q)=\mathbb{E}\big[ \widehat{W}(p)\widehat{W}(q)   \big]=\int_0^1\delta(p_t-q_t)dt   $. Random measures formally expressed in terms of exponentials of Gaussian fields as in~(\ref{GMCForm}) are the  focus of the theory of \textit{Gaussian multiplicative chaos} (GMC), and $M_{\hat{\beta}}$  is a subcritical GMC for any $\hat{\beta}\in \R_+$ that can be understood through the general approach to GMC theory in~\cite{Shamov}. The random measures $M_{\hat{\beta}}$ are a.s.\ mutually singular to $\mathbf{P}$ and satisfy 
\begin{align}\label{First_Sec}
\mathbb{E}\big[ M_{\hat{\beta}}( dp)  \big]  \,=\,\mathbf{P}(dp) \hspace{1cm}\text{and}\hspace{1cm}  \mathbb{E}\big[ M_{\hat{\beta}}( dp)M_{\hat{\beta}}( dq)\big]\,=\,e^{\beta^2 T(p,q)     }  \mathbf{P}(dp)\mathbf{P}(dq)\,.
\end{align}
In particular, $\mathbb{E}[ M_{\hat{\beta}}\times M_{\hat{\beta}}]$  is absolutely continuous with respect to $\mathbf{P}\times \mathbf{P}$, which is a necessary feature of subcritical GMCs.\footnote{See Lemma 34 of~\cite{Shamov}.} 

Weak-disorder limits analogous to~\cite{alberts} for the marginally relevant $(2+1)$-dimensional polymer involve fundamental new mathematical difficulties and are not as well understood as the weak-disorder regime for the $(1+1)$-polymer despite significant progress in a series of articles~\cite{CSZ0,CSZ1, CSZ2,  CSZ3, CSZ4}  by  Caravenna, Sun, and Zygouras. In~\cite{CSZ1} the authors proved that  the partition function $Z_{L,\beta}$ for (2+1)-dimensional polymers has the following distributional limit behavior as $L\nearrow \infty$ when the inverse temperature tends to zero as  $\beta\equiv \beta_L=  \frac{\sqrt{\pi}}{ (\log L)^{1/2}}\big(\hat{\beta}+ \mathit{o}(1) \big)$ for fixed $\hat{\beta}\in \R_+$:
\begin{align}\label{ZLog}
Z_{L,\beta_L }\quad \stackbin[L\rightarrow \infty]{\mathcal{L}}{\Longrightarrow} \quad \mathcal{Z}_{\hat{\beta}}\,:=\, \begin{cases}\textup{exp}\big\{\sigma_{\hat{\beta}}\chi-\frac{1}{2} \sigma_{\hat{\beta}}^2  \big\}  & \hat{\beta}< 1 \,, \\   0      &  \hat{\beta}\geq  1  \,, \end{cases}
\end{align}
where $\chi$ is a standard normal random variable and $\sigma_{\hat{\beta}}^2:=\log\big(\frac{1}{1-\hat{\beta}^2  }\big)$.  In other terms, for $\hat{\beta}<1$ the limit law, $\mathcal{Z}_{\hat{\beta}}$, is a mean-one lognormal that converges in probability to zero (while having exploding variance) as $ \hat{\beta}\nearrow 1$.  Thus a phase transition from weak disorder to strong disorder occurs at $\hat{\beta}=1$ within this weak-coupling limit regime.

A further study of the (2+1)-dimensional directed polymer around the critical point $\hat{\beta}=1$ within the weak-disorder limit is undertaken in~\cite{CSZ4} by choosing the more refined  inverse temperature scaling $\beta\equiv\beta_{L,r}$ in Remark~\ref{RemarkLog}, which depends on a fixed parameter value $r\in \R$.  This  scaling satisfies  $\beta_{L,r}= \frac{\sqrt{\pi}}{ (\log L)^{1/2}}\big(1+ \mathit{o}(1) \big)$ for $L\gg 1$, i.e., falls within the critical window of the phase transition~(\ref{ZLog}) and is determined by the requirement that the variance of $\textup{exp}\{\beta_{L,r} \omega   \}/\mathbb{E}\big[  \textup{exp}\{\beta_{L,r} \omega   \}  \big]$, where $\omega$ is a disorder variable, has the large $L$ asymptotic form $\frac{\pi}{\log L   }+\frac{\pi r}{\log^2 L}+\mathit{o}\big( \frac{1}{\log^2 L}  \big) $.\footnote{The parameter $r\in \mathbb{R}$ is related to the parameter $\vartheta\in \R $ used in~\cite{CSZ2,CSZ4}  through $r=\vartheta-\alpha$ for  $\alpha$ defined below~(\ref{CSZCov}).}  For a time parameter $t\geq 0$, the authors define the following  random  measures $ \mathscr{Z}_{Lt,\beta_{L,r}}$ on $\R^2$:
\begin{align}\label{CSZRandomized}
\mathscr{Z}_{Lt,\beta_{L,r}}(dx)\,:=\,\frac{1}{L}\sum_{y\in \frac{1}{\sqrt{L}}\Z^2   }Z_{Lt,\beta_{L,r}}(y\sqrt{L})\delta_{y}(x)\,,
\end{align}
where $Z_{L,\beta}(x)$  is the partition function for length $L$ polymers starting from position $x\in \Z^2$.  Using a tightness argument involving bounds for the third  moments of the variables $Z_{Lt,\beta_{L,r}}(\phi):=\int_{\R^2}\phi(x) \mathscr{Z}_{Lt,\beta_{L,r}}(dx) $ for $\phi\in C_c(\R^2)$, the authors  prove the existence of subsequential limits as $L\rightarrow \infty$ such that  $\mathscr{Z}_{Lt,\beta_{L,r}}$ converges in law to  a random measure $\mathcal{Z}_{t,r}$ on $\R^2$ satisfying
\begin{align}\label{CSZCov}
\mathbb{E}\Bigg[\bigg(\int_{\R^2}\phi(x)\mathcal{Z}_{t,r}(dx)\bigg)^2\Bigg]\,=\,\int_{\R^2\times \R^2}\phi(z)\phi(z')K_{t,r+\alpha}(z-z')dzdz'\,,
\end{align}
where $\alpha:=\gamma+\log 16 -\pi$ for the Euler-Mascheroni constant $\gamma$,   and  $K_{t,r}(z-z')$ is a correlation kernel with logarithmic blowup around its diagonal  from Bertini and Cancrini's article~\cite{BC} on the two-dimensional SHE.  The above is related to a recent breakthough on the moments of the two-dimensional SHE at criticality by  Gu, Quastel, and Tsai~\cite{GQT}. When $t=1$ the form~(\ref{CSZCov}) is consistent with the existence of a $(2+1)$-dimensional continuum random polymer measure $M_{r}^{\phi}(dp)$ on $C([0,1],\R^2) $, with total mass equal in distribution to the random variable $\int_{\R^2}\phi(x)\mathcal{Z}_{1,r}(dx)$, that is analogous to the (1+1)-dimensional case in~\cite{alberts2} when the  starting point of the polymer  has an appropriate probability density $\phi:\R^2\rightarrow [0,\infty)$ (i.e., diffuse initial position).  If $\mathbf{P}^{\phi}$ denotes Wiener measure on $C([0,1],\R^2) $ for trajectories starting with initial position density $\phi$, then  two independently chosen trajectories will a.s.\ not intersect.  In other words, the product Wiener measure $\mathbf{P}^{\phi}\times  \mathbf{P}^{\phi}$ assigns probability zero to the set of pairs of trajectories that intersect.  If a continuum disordered polymer measure $M_{r}^{\phi}$ exists,  $\mathbb{E}[M_{r}^{\phi}\times M_{r}^{\phi}] $ would  not be absolutely continuous with respect to  $\mathbf{P}^{\phi}\times  \mathbf{P}^{\phi}$, unlike the continuum $(1+1)$-dimensional polymer case~(\ref{First_Sec}).

Next we outline the rough analogy between models for directed polymers in a random environment  on diamond hierarchical graphs and on  rectangular lattices. 
 Hierarchical graphs (``lattices") are a frequent setting for statistical mechanical toy models because they  may retain key  characteristics of interest from their non-hierarchical analogs while providing a decomposability in terms of renormalization transformations; see for instance~\cite{GLT,Goldstein1,Hambly,HamblyII,Lacoin,Ruiz,Wehr} for recent mathematical work.    By the nature of their recursive construction,  hierarchical models embed copies of themselves after a change in the controlling parameters for the embedded copies. The articles~\cite{Cook,Derrida} were the first to study models of directed polymers in a random environment on diamond hierarchical graphs.\footnote{This assertion about the history of directed polymers on the diamond lattice is from~\cite[Page 73]{Comets}.} 
 In~\cite{Lacoin}, Lacoin and Moreno analyzed the phase diagram of polymers on diamond graphs when the disorder variables are placed on the vertices, showing that     
\begin{itemize}

\item strong disorder holds for any $\beta>0$ when $b\leq s$, and 

\item when $b> s$ there is a critical inverse temperature $\beta_c>0$ for which weak disorder holds when $\beta\leq \beta_c$ and strong disorder holds for $\beta$ above $\beta_c$.

\end{itemize}
In terms of their disorder relevance, the cases $b<s$, $b=s$, and $b>s$ are analogous respectively to the $d=1$, $d=2$, and $d\geq 3$ cases of (d+1)-dimensional polymers on the rectangular lattice.  In the disorder relevant $b<s$ case,~\cite{US} proves a limit theorem   for the  partition functions  in an intermediate disorder regime analogous to~\cite{alberts}, and~\cite{Clark2} defines a  continuum polymer model similar to~\cite{alberts2}, although using GMC for the construction rather than Wiener chaos.

When the model is altered by placing disorder variables on the edges of the graphs rather than the vertices (as in Section~\ref{SectionMainResult}), the analysis in~\cite{Lacoin} goes through essentially unchanged when $b<s$ or $b>s$, but for the marginal case of $b=s$ there is a basic combinatorial difference: for two directed polymers $p$ and $q$ chosen independently and uniformly at random, 
\begin{itemize}
\item the expected number of vertices shared by $p$ and $q$ has order $\log L$ for $L\gg 1$, where $L$ is the length\footnote{In terms of the parameter $s$, the polymer length has the form $L=s^n=b^n$.} of the polymers, and   

\item  the expected number of edges shared by $p$ and $q$ is exactly $1$, independent of $L$.  A closer look shows that when $L\gg 1$ the polymers will share no edges at all with a probability  $1-\mathit{O}(1/\log L)$,  and that the expected number of common edges will be of order  $\log L$ in the complementary event.   

\end{itemize}
Thus, when $b=s$, the diamond graph polymer model with edge disorder   is similar  to  the polymer measures underlying the mollified partition functions in~(\ref{CSZRandomized}) in the sense that two independent two-dimensional SSRW trajectories of length $L$ with initial spatial probability densities spread out on the order of $\sqrt{L}$  have a  probability of intersecting that vanishes with order $ 1/\log L $ and, when conditioned on the event that the paths  intersect, an expected number of intersections on the order of $\log L$.

We will briefly summarize the continuum polymer model defined in~\cite{Clark3} and its conditional Gaussian multiplicative chaos structure~\cite{Clark4}.  The limiting partition function law, $\mathbf{W}_r$, derived in later sections is equal in distribution to the total mass of a random  measure $\mathbf{M}_r$ on the space $\Gamma$ of directed paths crossing  a compact diamond fractal, $D$, having Hausdorff dimension two.  Each directed path $p\in \Gamma$ is an isometric embedding of the unit interval $[0,1]$ into the fractal, and there is a natural ``uniform" probability measure $\mu$ on $\Gamma$ (serving as the analog of Wiener measure for the continuum (1+1)-dimensional polymer) for which $\mathbb{E}[ \mathbf{M}_r ]=\mu  $.  For directed paths $p,q\in \Gamma$, the set of intersection times is $\mathcal{I}_{p,q}:=\{t\in [0,1]\,|\,p(t)=q(t)\}$, and  two paths  chosen uniformly at random, i.e., according to the product measure $\mu\times \mu$, have a finite (trivial) number of intersections with probability one.  In contrast,  the random product measures $\mathbf{M}_r\times \mathbf{M}_r$ a.s.\ assign positive weight to the set of pairs $(p,q)\in \Gamma\times \Gamma$ for which $ \mathcal{I}_{p,q}$ is uncountable, albeit of Hausdorff dimension zero. The size of typical $ \mathcal{I}_{p,q}$ can be characterized through the exponent $\frak{h}=1$ case of the generalized Hausdorff  measure $\mathcal{H}^{\textup{log}}_{\frak{h}}$ on $[0,1]$ of the form 
\begin{align}\label{Hausdorff}
\mathcal{H}^{\textup{log}}_{\frak{h}}(S)\,:=\,\lim_{\delta\searrow 0}  \mathcal{H}^{\textup{log}}_{\frak{h},\delta}(S)\hspace{1cm}\text{for}\hspace{1cm}  \mathcal{H}^{\textup{log}}_{\frak{h},\delta}(S)\,:=\,\inf_{\substack{ S\subset \cup_{k} I_k  \\ |I_k|<\delta  }   } \sum_{k } \frac{1}{|\log(\frac{1}{|I_k|})|^{\frak{h}} }   \,, 
\end{align}
where $S\subset [0,1]$, and the infimum is over all coverings of $S$ by intervals $I$ of length $|I|$ less than $\delta>0$; see the monograph~\cite{Rogers} for a discussion of the   general theory of Hausdorff  measures.

The  qualitative difference (trivial to nontrivial) between the typical behavior of the intersection-times set $I_{p,q}$ under the pure measure $\mu\times \mu$ and  realizations of the disordered product measure $\mathbf{M}_r\times \mathbf{M}_r$ is a strong localization property that is not present in the subcritical continuum models~\cite{alberts2,Clark2}. To compare with the (1+1)-dimensional continuum polymer measures $M_{\hat{\beta}}$ discussed above, the set of intersection times $I_{p,q}$ is appropriately measured by $T(p,q)=\int_0^1\delta_{0}(p_t-q_t)dt$---which is closely related to the dimension-$1/2$ Hausdorff measure of $I_{p,q}$---for both the product Wiener measure $\mathbf{P}\times \mathbf{P}$ and realizations of $M_{\hat{\beta}}\times M_{\hat{\beta}}$.
Secondly, in contrast with~(\ref{First_Sec}), the expectation of $\mathbf{M}_r\times \mathbf{M}_r$ has Lebesgue decomposition with respect to $\mu\times\mu$ given by
$$\mathbb{E}\big[\mathbf{M}_r\times\mathbf{M}_r\big]\,=\,\mu\times\mu\,+\,\varpi_r \,,$$
where the measure  $\varpi_r$ assigns full weight to the set of pairs $(p,q)\in \Gamma\times \Gamma$ such that $\mathcal{H}^{\textup{log}}_{\frak{h}}(I_{p,q})=\infty$  for all $\frak{h}<1$ and  $\mathcal{H}^{\textup{log}}_{\frak{h}}(I_{p,q})= 0$ for all $\frak{h}>1$, in other terms, for which $I_{p,q}$ has \textit{log-Hausdorff exponent} one.  The fact that $\mathbb{E}\big[\mathbf{M}_r\times\mathbf{M}_r\big]$ is not absolutely continuous with respect to $\mathbb{E}[\mathbf{M}_r]\times\mathbb{E}[\mathbf{M}_r]=\mu\times \mu$ implies   that $\mathbf{M}_r$ is not a subcritical GMC.

The random measure $\mathbf{M}_r$ is also not a ``critical" GMC since the expectation $\mathbb{E}\big[  \mathbf{M}_r ]=\mu $ is a probability measure and thus $\sigma$-finite. The family of random measure laws  $(\mathbf{M}_r)_{r\in \R}$, however, has a conditional interrelational GMC 
 structure wherein for any $a\in \R_+$ the law of the random measure $\mathbf{M}_{r+a}$ can be constructed from $\mathbf{M}_{r}$ as
\begin{align}\label{CondGMC}
\mathbf{M}_{r+a}(dp)\,\stackrel{\mathcal{L}}{=}\, e^{ \sqrt{a}\widehat{W}_{\mathsmaller{\mathsmaller{\mathbf{M}_{r}}}}(p)-\frac{a}{2}\mathbb{E}[ \widehat{W}^2_{\mathsmaller{\mathsmaller{\mathbf{M}_{r}}}}(p)]} \mathbf{M}_{r}(dp)   \,,\hspace{1.3cm}p\in \Gamma\,,
\end{align}
where  $\widehat{W}_{\mathbf{M}_{r}}(p)$ is a field over $(\Gamma, \mathbf{M}_{r})$ that is Gaussian when conditioned on  $\mathbf{M}_{r}$ and has a correlation kernel $T(p,q)=\mathbb{E}\big[\widehat{W}_{\mathbf{M}_{r}}(p)\widehat{W}_{\mathbf{M}_{r}}(q)\,|\, \mathbf{M}_{r}\big]$ roughly equivalent to the generalized Hausdorff measure with exponent $\frak{h}=1$, $\mathcal{H}^{\textup{log}}_{1}(\mathcal{I}_{p,q})$, of the set of intersection times.  Because the random measures $\mathbf{M}_{r}$ converge in law to the pure measure $\mu$ as $r\searrow-\infty$, the above formally implies that an infinite field strength is required to generate $\mathbf{M}_r$ as a GMC on $\mu$.

\section{Notation and organization}

\noindent \textbf{Notation:} In the remainder of the article, we refer exclusively to  the case when the branching parameter and the segmenting parameter  of the diamond graphs are equal ($b=s$). The dependence of  all previously defined expressions on the parameter $b\in \{2,3, \ldots\}$ will be suppressed as in the following list of notational identifications:
\begin{align*}
D^{b,b}_n\,\equiv\, D_n\,,\,\, \Gamma^{b,b}_n\,\equiv\, \Gamma_n\,,\,\, \beta_{n,r}^{(b)}\,\equiv \,\beta_{n,r}\,,\,\, M_{b,b}(x)\,\equiv\,M(x)\,,\,\,  R_b^{(m)}(r)\,\equiv\,R^{(m)}(r)\,,\,\, \kappa_b\,\equiv\, \kappa\,,\,\, \eta_b\equiv \eta\,.
\end{align*}
$\mathbb{N}$ denotes the positive integers and $\mathbb{N}_0:=\mathbb{N}\cup\{0\}$. In heuristic discussions, we write $X\stackrel{d}{\approx} Y $ for random variables $X$ and $Y$ that are ``close" in distribution. 

\vspace{.2cm}

\noindent \textbf{Article organization:}
\begin{itemize}
\item Section~\ref{SecMainThmOutline} builds up to the statement of Theorem~\ref{ThmUnique} (bond-disorder  \#2), which is a slightly strengthened version of Theorem~\ref{ThmMain} (bond-disorder \#1) that is couched in the language used in the proofs. Theorem~\ref{ThmSharpUnique} (bond-disorder  \#3) is a third version of this type of distributional convergence result that leverages more stringent moment conditions for greater control of the rate of convergence.

\item Taken together, Sections~\ref{SectExist} \&~\ref{SecOutlineMain} complete the proof of Theorem~\ref{ThmUnique} (bond-disorder \#2) after stating the key technical results in Proposition~\ref{PropFinalPush} and Lemmas~\ref{LemI}-\ref{LemIII} that support the proof.

\item  Sections~\ref{SecTemplate} \&~\ref{SecCentralLimit} contain the proofs of  Proposition~\ref{PropFinalPush} \& Lemmas~\ref{LemI}-\ref{LemIII} with some of the relatively routine elements delayed to Section~\ref{SecMiscProofs}.

\item Theorem~\ref{ThmSharpUnique} (bond-disorder \#3)  is proved in Section~\ref{SectionSharpRegProof}.

\item Theorem~\ref{ThmMainSite} (site-disorder) is proved in Section~\ref{SecSiteDisorder}.

\item Proofs of propositions that are technical variations of results from~\cite{Clark1} are placed in~Section~\ref{SecMiscII}.

\item Appendix~\ref{AppendBetaScale} derives the inverse temperature scaling~(\ref{BetaForm}) from the variance scaling~(\ref{VarAsym}), Appendix~\ref{AppendixVarCheck} carries through an instructive consistency check between (I) and (II) of Lemma~\ref{LemVar}, and Appendix~\ref{AppendixGoldstein} provides some background on the zero bias approach~\cite{Goldstein2} to Stein's method.

\end{itemize}

\section{Reformulation in terms of arrays and Wasserstein distance }\label{SecMainThmOutline}

This section defines the notation and terminology needed for the statement of  Theorem~\ref{ThmUnique}, which is a more flexible version of Theorem~\ref{ThmMain}.  The language defined here is used throughout the remainder of the article.

\subsection{Edge-labeled array notation}

The recursive construction of the diamond hierarchical graphs outlined in Section~\ref{SecDHG} implies a  canonical one-to-one correspondence between the set of edges, $E_k$, of the $k^{th}$-generation diamond graph $D_k$ and the $2k$-fold product set $(\{1,\ldots,b\}\times\{1,\ldots,b\}\big)^k$; see the diagram below illustrating this correspondence in the first- and second-generation graphs when $b=2$.   The hierarchical structure of the graphs also implies that for $ l, k\in \mathbb{N}_0$ with $l<k$ each element $\mathbf{a}\in E_l$ is canonically identifiable with a $b^{2(k-l)}$-element subset of $E_k$.
\begin{center}
\begin{minipage}{0.45\textwidth}
\includegraphics[width=\textwidth]{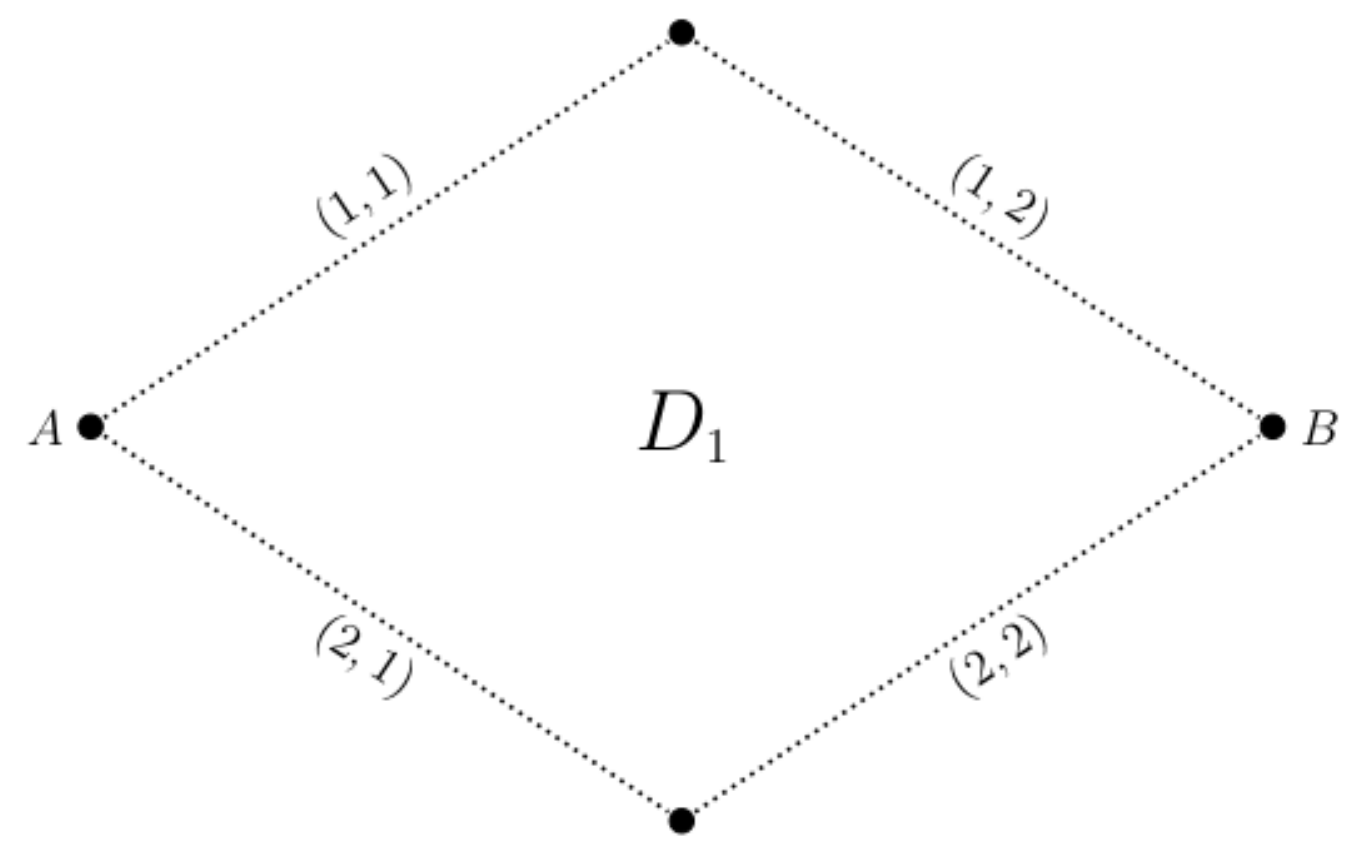}
\end{minipage}\hspace{.5cm}\begin{minipage}{0.45\textwidth}
\includegraphics[width=\textwidth]{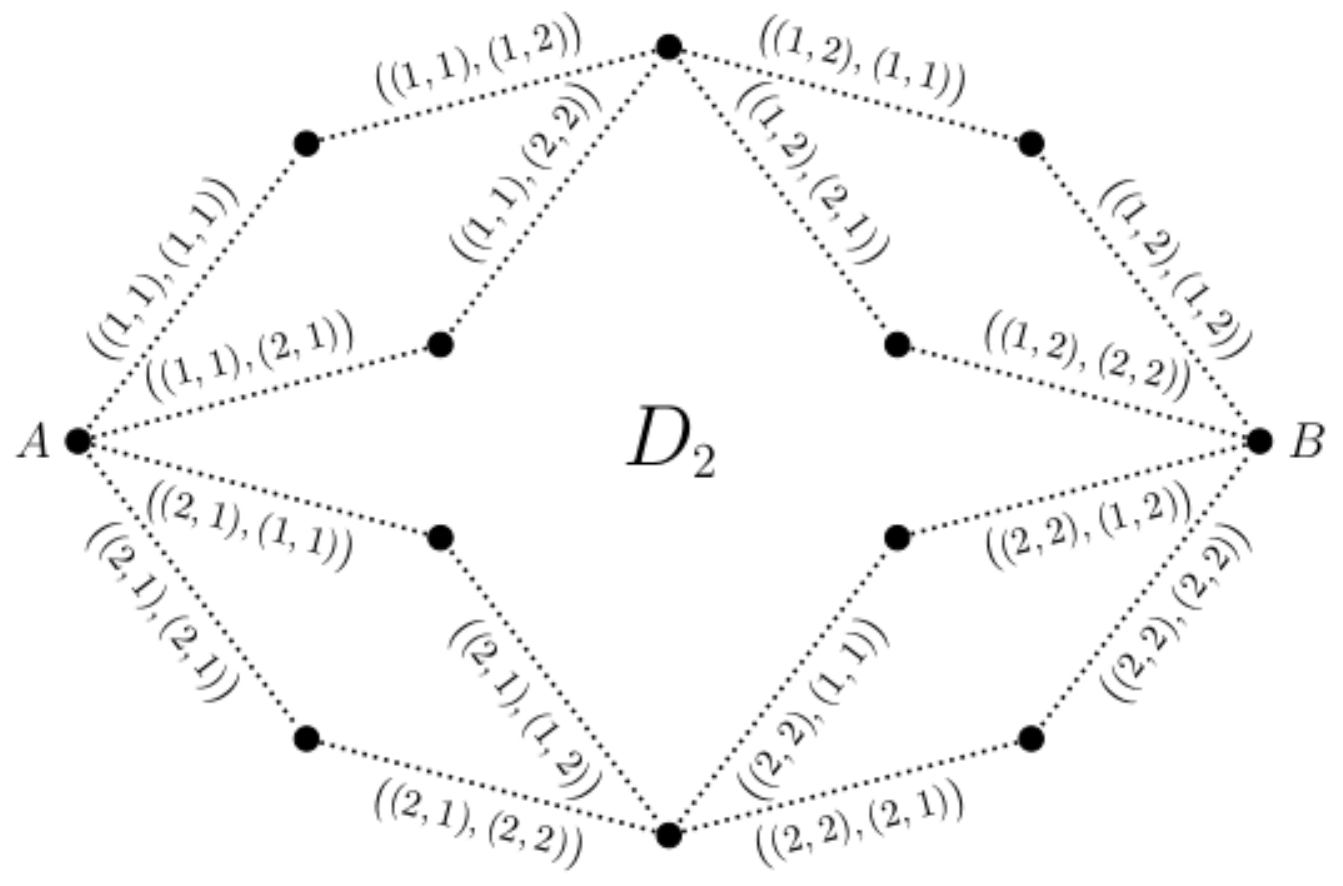}
\end{minipage}
\end{center}

\begin{notation}[Arrays]\label{NotationArray} \textup{Let $x_a$ be real numbers  labeled by $ E_k$ for some $k\in \mathbb{N}_0$. }
\begin{itemize}
\item  \textup{The notation $\{x_{a}\}_{a\in E_k}$ denotes an element of $\R^{b^{2k}}$, which we refer to as an} \textit{array}. 

\item  \textup{If $\mathbf{a}\in E_l$ for some $l\in \mathbb{N}$ with $l \leq  k$, then  $\{x_{a}\}_{a\in \mathbf{a}\cap E_k}$ denotes an element in $\R^{b^{2(k-l)}}$, where we have abused notation by identifying $\mathbf{a}$ with its canonically corresponding subset of  $ E_k$.}

\end{itemize}

\end{notation}

Next we define an operation on edge-labeled arrays that can be used (see Proposition~\ref{PropPartition}) to express  the partition function~(\ref{Partition}).
\begin{definition}[Array maps]\label{DefArrayMap}
For $k\in \mathbb{N}_0$ and  $a\in  E_k$, define $a{\times} (i,j)$ for $i,j\in \{1,\ldots, b\}$ as the element in $E_{k+1}$ corresponding to the $j^{th}$ segment along the $i^{th}$ branch  of the embedded copy of  $D_1$ in $D_{n+1}$ identified with $a$.\footnote{This is to be understood in the context of the recursive construction of $D_{n+1}$ from $D_n$ in Section~\ref{SecDHG}.}  
\begin{itemize}
\item We define $\mathcal{Q}$  as the  map that sends an array of real numbers $\{x_{a}\}_{a\in E_k}$ to the contracted array
\begin{align*}
 \{w_a\}_{a\in E_{k-1}  }&\,:=\,\mathcal{Q}\{x_{a}\}_{a\in E_k}\,\hspace{.5cm}\text{ for  }\hspace{.5cm} w_a \,:=\,\frac{1}{b}\sum_{i=1}^b\bigg( \prod_{j=1}^b \big(1+x_{a{\times}(i,j)}\big) \,-\,1\bigg) \,.
\end{align*}

\item We define $\mathcal{L}$ as the linearization of $\mathcal{Q}$ around the zero array: 
\begin{align*}
\{y_a\}_{a\in E_{k-1}  }&\,:=\,\mathcal{L}\{x_{a}\}_{a\in E_k}\,\hspace{.5cm}\text{ for }\hspace{.5cm} y_a\,:= \,\frac{1}{b}\sum_{1\leq i, j\leq b} x_{a\times (i,j)}\,.
\end{align*}

\item  We define $\mathcal{E} :=\mathcal{Q}-\mathcal{L}$, i.e., the ``error" of the linearization.

\item For $N\in \mathbb{N}_0$, $\mathcal{Q}^{N}$ and $\mathcal{L}^{N}$ refer to the $N$-fold composition of the maps  $\mathcal{Q}$ and $\mathcal{L}$, respectively.

\end{itemize}

\end{definition}

\begin{remark} Note  the ambiguity of the notations $\mathcal{Q}$,  $\mathcal{L}$, $\mathcal{E}$  since we  use them to denote maps from $\R^{E_k}$ to $\R^{E_{k-1}}$ for any $k\in \mathbb{N}$.
\end{remark}

\begin{remark} For $a\in E_k$, our notational conventions imply that $$a\cap E_{k+1}=\{a{\times} (i,j)\,|\, i,j\in \{1,\ldots, b\}\big\}\,. $$ 
\end{remark}

The following proposition relates the array map $\mathcal{Q}$ to the partition function  $W^{\omega}_n (  \beta )$. The proof is placed in Section~\ref{SecMiscZero}.
\begin{proposition}\label{PropPartition}\label{RemarkW} The partition function $W^{\omega}_n (  \beta )$ in~(\ref{Partition}) can be written in terms of the map $\mathcal{Q}$ as
\begin{align}\label{WRemark}
W^{\omega}_n\big(  \beta \big)\,=\,1\,+\,\mathcal{Q}^n\big\{ X_{h}^{(n)} \big\}_{h\in E_{n}}\, \hspace{1cm}\text{for}\hspace{1cm}  X_{h}^{(n)}\,:=\,  \frac{ e^{ \beta\omega_{h} }  }{\mathbb{E}\big[   e^{ \beta\omega_{h} } \big]   }\,-\,1  \,.  
\end{align}
\end{proposition}

\begin{remark}\label{RemarkArrayVar} Let $\{ x_a \}_{a\in E_k }$ be an array of i.i.d.\ centered random variables with variance $\sigma^2$.  
\begin{enumerate}[(i)]
\item  $\mathcal{Q}\{x_{a}\}_{a\in E_k}$ and $\mathcal{L}\{x_{a}\}_{a\in E_k}$ are i.i.d.\ arrays of centered random variables with variance $M(\sigma^2)$ and $\sigma^2$, respectively.  In particular, the operation $\mathcal{L}$ preserves the variance of the array variables.

\item For $\{y_{a}\}_{a\in E_{k-1}}:= \mathcal{L}\{x_{a}\}_{a\in E_k}$ and  $\{z_{a}\}_{a\in E_{k-1}}:=\mathcal{E}\{x_{a}\}_{a\in E_k}$,  the random variables $y_{a}$ and $z_{a}$ are uncorrelated. Thus the variables in the array $\mathcal{E}\{x_{a}\}_{a\in E_k}$ have variance $M(\sigma^2)-\sigma^2$.

\item  Moreover, the random variable $\mathcal{Q}^k\{ x_a \}_{a\in E_k }$ can be written as the following sum of uncorrelated terms: $ \mathcal{Q}^k\{ x_a \}_{a\in E_k }= \mathcal{L}^k\{ x_a \}_{a\in E_k }\,+\,\sum_{l=1}^k  \mathcal{L}^{l-1} \mathcal{E}\mathcal{Q}^{k-l}\{ x_a \}_{a\in E_k } $.

\end{enumerate}

\end{remark}

The  lemma below generalizes (iii) in Remark~\ref{RemarkArrayVar} and identifies the main source of uncorrelated terms found in this article.  The proof  follows easily from the multilinear polynomial forms of the maps  $\mathcal{Q}$, $\mathcal{E}$, $\mathcal{L}$.

\begin{lemma}\label{LemUnCor}  Let $\{ x_a \}_{a\in E_k }$ be an array of independent centered random variables with finite second moments. If $A_l,B_l\in \{\mathcal{Q},\mathcal{E},\mathcal{L}\}$ for $l\in \{1,\ldots, k\}$, then the random variables $A_1\cdots A_k \{ x_a \}_{a\in E_k }$ and $B_1\cdots B_k \{ x_a \}_{a\in E_k }$ are uncorrelated  when at least one of the following sets is nonempty: 
$$S_A\,:=\,\big\{l\,\big|\,A_l= \mathcal{E}\,\,\&\,\, B_l=\mathcal{L} \big\} \hspace{1cm} \text{and} \hspace{1cm}  S_B\,:=\,\big\{l\,\big|\,B_l= \mathcal{E} \,\,\&\,\, A_l=\mathcal{L}    \big\}\,.$$
\end{lemma}

\begin{proof}Suppose that $\ell \in S_A$.  The multilinear polynomial $A_1\cdots A_k \{ x_a \}_{a\in E_k }$ is    a  linear combination of monomials $\prod_{a\in U}x_{a}$ for which the set $U\subset E_k$ must contain a pair $a_1,a_2 \in U$ satisfying the following property: there exist $f_1,f_2\in E_{\ell}$  and $e\in E_{\ell-1}$ such that $a_1\in f_1$, $a_2\in f_2$, $f_1\neq f_2$, and $f_1,f_2\in e$. On the other hand, the multilinear polynomial $B_1\cdots B_k \{ x_a \}_{a\in E_k }$   does not contain any monomials of this type, so $A_1\cdots A_k \{ x_a \}_{a\in E_k }$ and $B_1\cdots B_k \{ x_a \}_{a\in E_k }$ are uncorrelated. 
\end{proof}

\begin{remark}\label{RemarkSum} Note that if $\{ x_h \}_{h\in E_n }$ is an array of i.i.d.\ centered random variables with  variance $\sigma^2$, then $\mathcal{L}^n\{ x_h \}_{h\in E_n }=\frac{1}{b^n}\sum_{h\in E_n}x_h$ has the form of a central limit-type normalized sum since $b^n=|E_n|^{1/2}$.  More generally, if $n\geq k$, then $\{z_a\}_{a\in E_{k}  }:=\mathcal{L}^{n-k}\{x_{h}\}_{h\in E_n}$ is an array of central limit-type normalized sums $z_a= \frac{1}{b^{n-k}}\sum_{h\in a\cap E_n}x_h $ since $b^{n-k}=|a\cap E_{n} |^{1/2}$.
\end{remark}

In the following, we define terminology for the multilayer arrays determined by  repeated application of  $\mathcal{Q}$ when starting from a given edge-labeled array.
\begin{definition}\label{DefQPyramid} Let $\mathcal{Q}$ be defined as in Definition~\ref{DefArrayMap} and $n\in \mathbb{N}_0$. 
\begin{itemize}
 \item A  \textit{$\mathcal{Q}$-pyramidic array}  is a finite sequence in $k=0,1,\ldots , n$ of arrays of real numbers $\{x_{a}^{(k,n)}\}_{a\in E_k}$  satisfying $\big\{x_{a}^{(k-1,n)}\big\}_{a\in E_{k-1}}=\mathcal{Q}\big\{x_{a}^{(k,n)}\big\}_{a\in E_k}$ for all $k\neq 0$.

\item    When $k=n$ we condense the superscript as $ x_{h}^{(n,n)}\equiv x_{h}^{(n)}$ for  $h\in E_n$.  Moreover, $\big\{x_{a}^{(k,n)}\big\}_{a\in E_k}=Q^k\{ x_{h}^{(n)} \}_{h\in E_n} $ is referred to as the $\mathcal{Q}$-\textit{pyramidic array generated from}  $\big\{ x_{h}^{(n)} \big\}_{h\in E_n}$.

\end{itemize} 

\end{definition}

\begin{remark}When $k=0$ we remove the subscript from $  x_a^{(0,n)}\equiv x^{(0,n)}$ since $|E_0|=1$.

\end{remark}

\begin{remark} To distinguish the entire  $\mathcal{Q}$-pyramidic array from one of its subarray layers, $\big\{x_{a}^{(k,n)}\big\}_{a\in E_k}$, we will sometimes write $\big\{x_{a}^{(*,n)}\big\}_{a\in E_*}$.
\end{remark}

\subsection{Regular sequences of $\mathcal{Q}$-pyramidic arrays of random variables}

Next we narrow our focus to sequences of $\mathcal{Q}$-pyramidic arrays of random variables.  The following definition characterizes the assumptions that we use in our limit theorem in the next subsection. 
\begin{definition}\label{DefRegular} A sequence $\big(\{ X_a^{(*,n)} \}_{a\in  E_{*} }\big)_{n\in \mathbb{N}} $ of $\mathcal{Q}$-pyramidic arrays of random variables taking values in $[-1,\infty)$ will be said to be \textit{regular with parameter $r\in \R$} if the sequence of generating arrays $\big(\{ X_h^{(n)} \}_{h\in  E_{n} }\big)_{n\in \mathbb{N}} $ satisfies the properties below.
\begin{enumerate}[(I)]

\item For each $n\in \mathbb{N}$, the random variables in the array $\{ X_h^{(n)} \}_{h\in  E_{n} } $ are centered and i.i.d.

\item The variance of the random variables in the array $\{ X_h^{(n)} \}_{h\in  E_{n} } $ has the large $n$ asymptotics~
\begin{align}\label{VARAsym}
\textup{Var}\big(  X_h^{(n)}  \big)\,=\,\kappa^2 \bigg(\frac{ 1 }{n}  +\frac{ \eta\log n}{n^2}+\frac{r}{n^2}\bigg)\,+\,\mathit{o}\Big(\frac{1}{n^2}  \Big)\,.
\end{align}

\item For each $m\in \{4,6,\ldots\}$, the $m^{th}$ moment of the random variables in the array $\{ X_h^{(n)} \}_{h\in  E_{n} } $  vanishes as $n\rightarrow \infty$.

\end{enumerate}
Moreover, $\big(\{ X_a^{(*,n)} \}_{a\in  E_{*} }\big)_{n\in \mathbb{N}} $ is \textit{minimally regular} if (I)-(II) hold, but (III) is only assumed for $m=4$.
\end{definition}

\begin{remark}\label{RemarkExample}
The first example of a regular sequence $\big(\big\{X_{a}^{(*,n)}\big\}_{a\in E_*}\big)_{n\in \mathbb{N}}$  of $\mathcal{Q}$-pyramidic arrays that we have in mind is when the random variables in the generating arrays $\big\{X_{h}^{(n)}\}_{h\in E_n}$ are defined as in~(\ref{WRemark}) with $\beta\equiv\beta_{n,r}$ having the large $n$ asymptotics~(\ref{BetaForm}) for some   $r\in \R$.  The variance criterion (II) in Definition~\ref{DefRegular} holds by~(\ref{VarAsym}) and the higher even moment criterion (III) merely follows from the fact that $\beta_{n,r}$ vanishes as $n\rightarrow \infty$.   
 \end{remark}

Proposition~\ref{PropHigherMom} generalizes the result~(\ref{ConvHM}) in Theorem~\ref{ThmHM} about the convergence of the higher centered moments of $W^{\omega}_n(\beta_{n,r})$.  We omit the proof, which is the same as that of part (i) of  Theorem 3.3 of~\cite{Clark1}, or said differently, the proof of part (i) of  Theorem 3.3 of~\cite{Clark1} proceeds by implicitly proving Proposition~\ref{PropHigherMom}.

\begin{proposition} \label{PropHigherMom}
Let  $\big(\big\{X_{a}^{(*,n)}\big\}_{a\in E_*}\big)_{n\in \mathbb{N}}$ be a  sequence of $\mathcal{Q}$-pyramidic arrays of random variables generated from a sequence of arrays $\big(\{ X_h^{(n)} \}_{h\in  E_{n} }\big)_{n\in \mathbb{N}} $ satisfying properties (I)-(II) in Definition~\ref{DefRegular} for some $r\in \R$.   If the $\frak{p}^{th}$  even moment of the random variables in the array $\big(\{ X_h^{(n)} \}_{h\in  E_{n} }\big)_{n\in \mathbb{N}} $ vanishes as $n\rightarrow \infty$, then for each $m\in \{2,3,\ldots,2\frak{p}\}$ the $m^{th}$ moment of the random variables $X^{(0,n)}= \mathcal{Q}^n\big\{X_{h}^{(n)}\big\}_{h\in E_n}$ converges to $R^{(m)}(r)$ as $n\rightarrow \infty$, where  $R^{(m)}:\R\rightarrow [0,\infty)$ is the function in Theorem~\ref{ThmHM}.
\end{proposition}

The statement of the following  lemma is formulated to emphasize the connection with the properties (I)-(III) in Theorem~\ref{ThmExist} below that we use to characterize the limit law emerging as $n\rightarrow \infty$.
\begin{lemma}\label{LemmaMom} The statements below hold for any regular sequence $\big(\big\{X_{a}^{(*,n)}\big\}_{a\in E_*}\big)_{n\in \mathbb{N}}$ of $\mathcal{Q}$-pyramidic arrays with parameter $r\in \R$. 
\begin{enumerate}[(I)]
\item For each  $n$ and $k$, the variables in the array   $ \big\{X_{a}^{(k,n)}\big\}_{a\in E_n}$ are i.i.d.

\item For each $n$ and $k\geq 1$, the array  $\mathcal{Q} \big\{X_{a}^{(k,n)}\big\}_{a\in E_k}$ is equal to  $ \big\{X_{a}^{(k-1,n)}\big\}_{a\in E_{k-1}}$.

\item For each $n$ and $k$, the variables in the array $\{X_{a}^{(k,n)}\}_{a\in E_{k}}$ are centered, and the variables have finite $m^{th}$ moment that converges to $ R^{(m)}(r-k)$ as $n\rightarrow \infty$ for every $k$ and $m\in \{2, 3, \ldots\}$.
\end{enumerate}
The above  hold for  minimally regular sequences except the  convergence in (III) is only for $m\in \{2,3,4\}$.

\end{lemma}
\begin{proof}\label{Restricted} Statements (I) and (II) of  Lemma~\ref{LemmaMom} are immediate consequences of the definition of the variable arrays $\big\{X_{a}^{(k,n)}\big\}_{a\in E_k}$.  To see  (III), note that for $a\in E_k$ we have $X_{a}^{(k,n)}=\mathcal{Q}^{n-k}\big\{X_{h}^{(n)}\big\}_{h\in a\cap E_n}$.
By definition, the random variables $X_{h}^{(n)}$ have variance satisfying the large $n$ asymptotics~(\ref{VarAsym}), which we can rewrite in the form
\begin{align}\label{Shifted}
\textup{Var}\big(X_{h}^{(n)}\big)\,=\, &\kappa^2 \bigg(\frac{ 1 }{n-k}  +\frac{ \eta\log (n-k)}{(n-k)^2}+\frac{r-k}{(n-k)^2}\bigg)\,+\,\mathit{o}\Big(\frac{1}{n^2}  \Big)\,.
\end{align}
Notice that~(\ref{Shifted}) has the form~(\ref{VarAsym}) with $n$ and $r$ replaced by $n-k$ and $r-k$, respectively.  It follows from Proposition~\ref{PropHigherMom} that the $m^{th}$ moment of $X_{a}^{(k,n)}=\mathcal{Q}^{n-k}\big\{X_{h}^{(n)}\big\}_{h\in a\cap E_n}$ converges to $R^{(m)}(r-k)$ as $n\rightarrow \infty$ for each $m\in \{2,3,\ldots\}$.\end{proof}

\subsection{A limit theorem for $\mathcal{Q}$-pyramidic arrays}

Theorems~\ref{ThmExist} \& \ref{ThmUnique} below are the main technical results of this article, and they are jointly proved in Section~\ref{SecThmUnique}. Theorem~\ref{ThmExist} characterizes the limiting law for the distributional convergence statement in Theorem~\ref{ThmUnique}. 

\begin{theorem}[Limit law]\label{ThmExist} For any  $r\in \R$, there exists a unique law on sequences in $k\in \mathbb{N}_0$ of edge-labeled arrays of random variables, $\big\{ \mathbf{X}_a^{(k)} \big\}_{a\in  E_{k}  }$, taking values in $[-1,\infty)$ and  holding the properties (I)-(III) below. 
\begin{enumerate}[(I)]

\item For each $k\in \mathbb{N}_0$, the variables in the array $\big\{ \mathbf{X}_a^{(k)} \big\}_{a\in E_{k}  }$ are i.i.d.

\item For each $k\in \mathbb{N}$, the array   $\big\{ \mathbf{X}_a^{(k-1)} \big\}_{a\in E_{k-1}  }$ is equal to $\mathcal{Q}\big\{ \mathbf{X}_a^{(k)} \big\}_{a\in E_{k}  }$.

\item  For each $k\in \mathbb{N}_0$,  the variables in the array $\big\{ \mathbf{X}_a^{(k)} \big\}_{a\in E_{k}  }$ are centered and have $m^{th}$ moment equal to $R^{(m)}(r-k)$ for all $m\in \{2, 3,\ldots\}$.

\end{enumerate}

\end{theorem}
\begin{notation} \label{NotationX} In the $k=0$ case of the random variables  $\mathbf{X}_a^{(k)} $ from Theorem~\ref{ThmExist}, i.e., the peak of the infinite  $\mathcal{Q}$-pyramidic array of random variables, we will drop the scripts $a$ \textup{\&} $(k)$ and optionally attach the parameter $r\in \R$ as a subscript: $\mathbf{X}_a^{(0)}\equiv\mathbf{X}\equiv\mathbf{X}_r $.
\end{notation}
\begin{remark}\label{RemarkHS} By hierarchical symmetry, the random variables in the arrays $\big\{ \mathbf{X}_a^{(k)} \big\}_{a\in  E_{k}  }$ from Theorem~\ref{ThmExist} with parameter $r\in \R$ are equal in distribution to $\mathbf{X}_{r-k}$.
\end{remark}

\begin{remark}\label{RemarkX2W} The limit law $\mathbf{W}_{r}$ in Theorem~\ref{ThmMain} is equal in distribution to  $1+\mathbf{X}_r$. 
\end{remark}

\begin{remark}\label{RemarkLimitRegular} Let $\big(\big\{ \mathbf{X}_a^{(k)} \big\}_{a\in  E_{k} }\big)_{k\in \mathbb{N}_0} $ be a sequence of arrays of random variables satisfying the properties in the statement of Theorem~\ref{ThmExist}.  For the purpose of proving the uniqueness in Theorem~\ref{ThmExist}, it will be useful to make the trivial observation  that the sequence  of $\mathcal{Q}$-pyramidic arrays $\big(\big\{ \mathbf{X}_a^{(*,n)} \big\}_{a\in  E_{*} }\big)_{n\in \mathbb{N}} $ defined by $\big\{ \mathbf{X}_a^{(k,n)} \big\}_{a\in  E_{k} }\equiv \big\{ \mathbf{X}_a^{(k)} \big\}_{a\in  E_{k} }$ for $0\leq k\leq n$ is regular with parameter $r$.
\end{remark}


In the sequel we will evaluate the distance between measures on $\R$ using Wasserstein-$1$ \& -$2$ metrics. 
\begin{definition}[Wasserstein distance]
For two Borel probability measures $\mu$ and $\nu$ on $\R$, let $\mathcal{M}_{\mu,\nu}$ be the set of  joint measures $J(dx,dy)$ on $\R^2$ with marginals $\mu$ and $\nu$.  For $p\geq 1$ assume that $\mu$ and $\nu$ satisfy $\int_{\R} |x|^p\mu(dx)<\infty$ and $\int_{\R} |x|^p\nu(dx)<\infty$.  We define the \textit{Wasserstein-$p$ distance} between $\mu$ and $\nu$ as
\begin{align*}
\rho_p(\mu,\nu)\,:=\,\inf_{J\in \mathcal{M}_{\mu,\nu} } \bigg(\int_{\R^2}  |x-y|^pJ(dx,dy)   \bigg)^{\frac{1}{p}}\,.
\end{align*}
If $X$ and $Y$ are random variables with distributional measures $\mu$ and $\nu$, respectively, then we extend our notation through the interpretation $\rho_p(X,Y)\equiv\rho_p(\mu,\nu)$.
\end{definition}

We prove the following proposition on the distributional continuity of  $r\,\mapsto \,\mathbf{X}_r $ in Section~\ref{SecMiscZero}.
\begin{proposition}\label{PropCont}  Let $\mathbf{X}_r$ be defined as in Notation~\ref{NotationX}.  The law of $\mathbf{X}_r$ is a locally $\frac{1}{2}$-H\"older continuous function of $r\in \R$ with respect to the Wasserstein-$2$ metric.  
\end{proposition}

By  Remark~\ref{RemarkExample} the limit  theorem below implies Theorem~\ref{ThmMain}.
\begin{theorem}\label{ThmUnique} Let $\big(\{ X_a^{(*,n)} \}_{a\in  E_{*} }\big)_{n\in \mathbb{N}} $ be a minimally regular  sequence  of $\mathcal{Q}$-pyramidic arrays of random variables with parameter  $r\in \R$.  For any $k\in \mathbb{N}_0$ and $a\in  E_{k}$, the  Wasserstein-2 distance between $ X_a^{(k,n)} $   and $  \mathbf{X}_a^{(k)}$   vanishes as $n\rightarrow \infty$, and, in particular, the i.i.d.\ array $\big\{ X_a^{(k,n)} \big\}_{a\in  E_{k} } $ (viewed as taking values in $\R^{b^{2k}}$) converges in law to $\big\{ \mathbf{X}_a^{(k)} \big\}_{a\in  E_{k}  }  $ for each $k\in \mathbb{N}_0$.
\end{theorem}


\begin{remark}\label{RemarkReduce} The hierarchical symmetry of the model implies that it is sufficient to prove   Theorem~\ref{ThmUnique} for the case $k=0$ in which the arrays  $\big\{ X_a^{(k,n)} \big\}_{a\in  E_{k} } $ and $\big\{ \mathbf{X}_a^{(k)} \big\}_{a\in  E_{k} } $ are  single random variables $X^{(0,n)}$ and $\mathbf{X}$, respectively.  The proof of Theorem~\ref{ThmUnique} involves writing $X^{(0,n)}=\mathcal{Q}^N\big\{ X_e^{(N,n)} \big\}_{e\in  E_{N} }$ and $\mathbf{X}=\mathcal{Q}^N\big\{ \mathbf{X}_e^{(N)} \big\}_{e\in  E_{N} }$ for  $N\in \mathbb{N}$ with $1\ll N\ll n$ and introducing arrays of random variables $\big\{ \mathbf{\widetilde{X}}_e^{(N)} \big\}_{e\in  E_{N} }$ (Definition~\ref{DefXs}) for which we show that $X_e^{(N,n)}\stackrel{d}{\approx} \mathbf{\widetilde{X}}_e^{(N)} $ and $\mathbf{X}_e^{(N)}\stackrel{d}{\approx} \mathbf{\widetilde{X}}_e^{(N)} $ in an appropriately strong sense that is characterized in Proposition~\ref{PropFinalPush}.    
\end{remark}

\section{Rate of convergence under stricter moment assumptions}\label{SectionConvRate}

In this section we will state an alternative version of the limit result in  Theorem~\ref{ThmUnique}  that offers  more explicit  rates of distributional convergence as $n\rightarrow \infty$ under stronger moment assumptions on the arrays of random variables from which the $\mathcal{Q}$-pyramidic arrays are generated.  The conditions of the limit theorem  easily translate into conditions for checking that a family of regular sequences of $\mathcal{Q}$-pyramidic arrays of random variables depending on an auxiliary parameter $s\in S$ is uniformly convergent with respect to the Wasserstein-$2$ metric (Corollary~\ref{CorSharpReg}).  The following definition characterizes our new assumptions.

\begin{definition}\label{DefSharpRegular} Fix some $\alpha\in (0,1)$.  A regular sequence of $\mathcal{Q}$-pyramidic arrays of random variables $\big(\{ X_a^{(*,n)} \}_{a\in  E_{*} }\big)_{n\in \mathbb{N}} $ with parameter $r\in \R$ is said to be $\alpha$-\textit{sharply} \textit{regular} if the sequence of generating arrays $\big\{ X_h^{(n)} \big\}_{h\in  E_{n} }$ satisfies the following more restrictive forms of (II) and (III) in Definition~\ref{DefRegular}:
\begin{enumerate}[(I*)]
\setcounter{enumi}{1}
\item The variance of the random variables in the array $\big\{ X_h^{(n)} \big\}_{h\in  E_{n} }$ has the asymptotics~(\ref{VARAsym})  with $\mathit{o}\big(\frac{1}{n^2}  \big)$ replaced by $\mathit{O}\big(\frac{1}{n^{2+\alpha}}  \big)$.

\item For each  $m\in\{4,6,\ldots\}$, the $m^{th}$ moment of the random variables in the array $\big\{ X_h^{(n)} \big\}_{h\in  E_{n} }$ is $\mathit{O}(n^{-m/2})$ as $n\rightarrow \infty$.

\end{enumerate}

\end{definition}

\begin{remark}\label{RemarkExampleII}
The sequence $\big(\{ X_a^{(*,n)} \}_{a\in  E_{*} }\big)_{n\in \mathbb{N}} $ of $\mathcal{Q}$-pyramidic arrays
 generated by arrays $\big\{X_{h}^{(n)}\}_{h\in E_n}$  defined as in~(\ref{WRemark}) where $\beta\equiv\beta_{n,r}$ has the large $n$ asymptotics~(\ref{BetaForm}) with $\mathit{o}\big(\frac{1}{n^{3/2}}\big)$ replaced by $\mathit{O}\big(\frac{1}{n^{3/2+\alpha}}\big)$ is $\alpha$-sharply regular.  Property (III*) holds since $\beta_{n,r}$ is $\mathit{O}\big( \frac{1}{n^{1/2}} \big)$  as $n\rightarrow \infty$ and property (II*) follows from the computation in Appendix~\ref{AppendBetaScale}. 
 \end{remark}

The following  theorem, which we prove in Section~\ref{SectionSharpRegProof}, implies that if  $\big(\{ X_a^{(*,n)} \}_{a\in  E_{*} }\big)_{n\in \mathbb{N}} $ is  an  $\alpha$-sharply regular sequence of $\mathcal{Q}$-pyramidic arrays of random variables  with parameter $r\in \R$, then the Wasserstein-$2$ distance between  $X^{(0,n)}$ (i.e, the peak of the $n^{th}$ $\mathcal{Q}$-pyramidic array in the sequence) and the limit law $\mathbf{X}_r$ vanishes with order $n^{-\upsilon}$ as $n\rightarrow \infty$ for any choice of  $\upsilon\in (0, \alpha/9)$. By  hierarchical symmetry, this generalizes to the convergence  of the random variables $\big\{X^{(k,n)}_{a}\big\}_{a\in E_k}$ in  the higher generation (i.e., $k\geq 1$) array layers.  The statement of Theorem~\ref{ThmSharpUnique} is formulated to provide easily verifiable conditions under which a family of $\alpha$-sharply regular sequences of $\mathcal{Q}$-pyramidic arrays of random variables can be shown to be uniformly convergent in law; see Corollary~\ref{CorSharpReg}.

\begin{theorem}\label{ThmSharpUnique} Fix  $\mathbf{v},\varkappa>0$, $\alpha\in (0,1)$,  $\upsilon\in (0, \alpha/9)  $, and a bounded interval $\mathcal{I} \subset \R$.  Define $\frak{p}:= \lceil \frac{2\alpha}{\alpha-9\upsilon}\rceil   +1  $. There exists a positive number $C\equiv C(\mathcal{I}, \mathbf{v},\varkappa,\alpha,\upsilon) $ such that for any $r\in \mathcal{I}$, $n\in \mathbb{N}$, and  i.i.d.\ array of centered random variables  $\big\{ X_h^{(n)} \big\}_{h\in  E_{n} } $  satisfying 
\begin{enumerate}[(I)]

\item $\left|\textup{Var}\big( X_h^{(n)}  \big)\,-\,\kappa^2\big(\frac{1}{n}+\frac{\eta\log n}{n^2}+\frac{r}{n^2}  \big)   \right| \,< \,\frac{ \mathbf{v} }{ n^{2+\alpha} } $ and

\item $  \mathbb{E}\left[ \big| X_h^{(n)}\big|^{2\frak{p}}\right]\,< \, \frac{ \varkappa}{ n^\frak{p} } $,

\end{enumerate}
the peak,   $X^{(0,n)}$, of the  $\mathcal{Q}$-pyramidic array, $\big\{ X_a^{(*,n)} \big\}_{a\in  E_{*} }$, generated by $\big\{ X_h^{(n)} \big\}_{h\in  E_{n} } $ has distance  less than $Cn^{-\upsilon}$ from $\mathbf{X}_r$ with respect to the Wasserstein-2 metric.
\end{theorem}

\begin{remark} Our proof of Theorem~\ref{ThmSharpUnique} follows essentially the same track as the proof of Theorem~\ref{ThmUnique} except for the use of technical lemmas that fit with this particular formulation of the distributional convergence.  Through a different proof method, it may be possible to extend the range of the exponent $\upsilon$ to a larger interval, e.g., $(0, \alpha/6)$.  
\end{remark}

The next corollary  is a direct consequence of Theorem~\ref{ThmSharpUnique}.

\begin{corollary}\label{CorSharpReg} Fix  $\mathbf{v},\varkappa>0$, $\alpha\in (0,1)$, $\upsilon\in (0, \alpha/9)  $, and a bounded interval $\mathcal{I} \subset \R$.  Let $\frak{r}$ be a function from a  set $S$ into $\mathcal{I}$.  For some $n\in \mathbb{N}$ and  all $s\in S$, let $\big\{ X_h^{(n)}(s) \big\}_{h\in  E_{n} }$ be an i.i.d.\ array of random variables satisfying conditions (I)-(II) in Theorem~\ref{ThmSharpUnique} with parameter $r\equiv \frak{r}(s)$. The  inequality below holds for the $C\equiv C(\mathcal{I}, \mathbf{v},\varkappa, \alpha,\upsilon)$ in Theorem~\ref{ThmSharpUnique}.
\begin{align*}
\sup_{s\in S}\rho_2\left(\mathcal{Q}^n\big\{ X_h^{(n)}(s) \big\}_{h\in  E_{n} },\, \mathbf{X}_{\frak{r}(s)}   \right)\, \leq  \,\frac{C}{n^{\upsilon}}  
\end{align*} 
\end{corollary}


Fix $T>0$ and $r\in \mathbb{R}$.  The following example applies Corollary~\ref{CorSharpReg} to uniformly approximate the random variables $\mathbf{X}_{r+t}  $ for $t$ in the interval  $ [0,T]$  by $\mathcal{Q}^n$ applied to an i.i.d.\ array $\{X_{h}^{(n)}(r,t) \}_{h\in E_n }$, where the variables $X_{h}^{(n)}(r,t)$ are  log-normal perturbations of the variables $\mathbf{X}_{h}^{(n)}$ from Theorem~\ref{ThmExist}.  The construction below is used in the proof of Proposition~\ref{PropCont} and is closely related to the Gaussian multiplicative chaos construction in~(\ref{CondGMC}).

\begin{example}\label{Example}  Let the  array of random variables $\{\mathbf{X}_{h}^{(n)}\}_{h\in E_n}$ be defined as in Theorem~\ref{ThmExist} for some parameter value $r\in \R$ and  $\{ \mathbf{B}^{h}\}_{h\in E_n}$ be an array of independent standard Brownian motions independent of   $\{\mathbf{X}_{h}^{(n)}\}_{h\in E_n}$.  For $t\in [0,T]$   define 
\begin{align}\label{Examp1}
X^{\mathbf{B}}_{n,r,t}\,:=\,\mathcal{Q}^n\big\{X_{h}^{(n)}(r,t) \big\}_{h\in E_n }\hspace{.7cm}\text{for}\hspace{.7cm} X_{h}^{(n)}(r,t)\,:=\, \big(1+\mathbf{X}_{h}^{(n)}\big)e^{\frac{\kappa}{n}\mathbf{B}^h_{t} - \frac{\kappa^2}{2n^2}t   }\,-\,1 \,.
\end{align}
Note that when $t=0$ the random variable $X^{\mathbf{B}}_{n,r,t}$ is equal in distribution to  $\mathbf{X}_r$ by (II) of Theorem~\ref{ThmExist}.  The variance of $X_{h}^{(n)}(r,t)$ has the large $n$ asymptotic form
\begin{align}\label{Examp2}
\textup{Var}\Big( X_{h}^{(n)}(r,t) \Big)\,=\,\big(1+R(r-n)  \big)e^{\frac{\kappa^2}{n^2}t   }\,-\,1\,=\,\kappa^2\bigg(\frac{1}{n}+\frac{\eta\log n  }{ n^2 }+\frac{ r+t  }{ n^2 }\bigg)+\mathit{O}\bigg( \frac{\log^2 n  }{n^3  }  \bigg)\,,
\end{align}
where we have used (II) of Lemma~\ref{LemVar}.  Moreover,  the error term  is uniformly bounded by a single multiple of $\frac{\log^2 n  }{n^3  }$ for all $t\in [0,T]$.  
By writing $X_{h}^{(n)}(r,t) $ as a sum of $\mathbf{X}_{h}^{(n)}$ and $\big(1+\mathbf{X}_{h}^{(n)}\big)\big(e^{\frac{\kappa}{n}\mathbf{B}^h_{t} - \frac{\kappa^2}{2n^2}t   } -1\big)  $, the $\frak{p}^{th}$ even moment of $X_{h}^{(n)}(r,t) $ can be shown to be $\mathit{O}\big(\frac{1}{n^{\frak{p}}}\big) $  using that
\begin{align*}
\mathbb{E}\Big[ \big(   \mathbf{X}^{(n)}_{h} \big)^{2\frak{p}} \Big]\,=\,R^{(2\frak{p})}(r-n)\,\sim\,\frac{1}{2^{\frak{p}}}{2\frak{p}\choose\frak{p}}\Big(\frac{\kappa^2}{n}\Big)^{\frak{p}}\,\hspace{.5cm}\text{and}\hspace{.5cm}\mathbb{E}\bigg[ \Big(e^{\frac{\kappa}{n}\mathbf{B}^h_{t} - \frac{\kappa^2}{2n^2}t   } -1\Big)^{2\frak{p}}\bigg]\,\sim\,\frac{1}{2^{\frak{p}}}{2\frak{p}\choose\frak{p}}\Big(\frac{\kappa^2 t}{n^2}\Big)^{\frak{p}}\,.
\end{align*}
The approximation above  for $R^{(2\frak{p})}(s)$ when $-s\gg 1$ is from (II) of Theorem~\ref{ThmHM}.  It follows that the arrays $\big\{X_{h}^{(n)}(r,t)\big\}_{h\in E_n}$ satisfy the conditions (I)-(II) of Theorem~\ref{ThmSharpUnique} for any fixed $\alpha\in (0,1)$ and all $n\in \mathbb{N}$ and $t\in [0,T]$ for large enough $\mathbf{v},\varkappa>0$.  By Corollary~\ref{CorSharpReg}, the  random variables $X^{\mathbf{B}}_{n,r,t}$ converge uniformly  to $\mathbf{X}_{r+t}$ over $t\in [0,T]$ with respect to the Wasserstein-$2$ metric  as $n\rightarrow \infty$. 
\end{example}

\section{Existence of a limiting $\mathcal{Q}$-pyramidic array of random variables}\label{SectExist}

In this section we prove the existence of the infinite $\mathcal{Q}$-pyramidic  array of random variables described in Theorem~\ref{ThmExist}. The proof is based on a routine  tightness argument involving nested subsequences.

\begin{proof}[Proof of Theorem~\ref{ThmExist} (existence)] Let $\big( \big\{ X_a^{(*,n)} \big\}_{a\in E_{*}} \big)_{n\in \mathbb{N}}$ be a regular sequence of $\mathcal{Q}$-pyramidic  arrays of random variables with parameter $r\in \R$, e.g., of the form in   Remark~\ref{RemarkExample}. For any $k\in \mathbb{N}_0$,  and  $a\in E_{k}$, the variance of $X_a^{(k,n)}$ converges   to $R(r-k)$ as $n\rightarrow \infty$  by  Lemma~\ref{LemmaMom}.  In particular, for any fixed $k$ the sequence $\{ X_a^{(k,n)} \}_{a\in E_k}$ of random arrays indexed by $n\in \mathbb{N}$, viewed as a random  vector in $\mathbb{R}^{b^{2k}}$, is tight.    We define  $\xi_{n}^{(k)}\in \mathbb{N}$  inductively in $k\in \mathbb{N}_0$  as a nested sequence of subsequences as follows:
\begin{itemize}
\item  Let  $(\xi_{n}^{(0)})_{n \in \mathbb{N}}$ be a subsequence of  $n=1,2,3,\ldots$ such that the single-element array $  \big\{ X_a^{(0,\,\xi_{n}^{(0)})} \big\}_{a\in E_{0}  }  $ converges in law as $n\rightarrow \infty$ to a limit  $  \big\{ \mathbf{X}_a^{(0)} \big\}_{a\in E_{0}  }  $.

\item If for $k\in \mathbb{N}_0$ the sequence $(\xi_{n}^{(k)})_{n \in \mathbb{N}}$ has been chosen  so that the array $  \big\{ X_a^{(k,\, \xi_{n}^{(k)})} \big\}_{a\in E_{k}  } $ converges in law as $n\rightarrow \infty$ to a limiting array $\big\{\mathbf{X}_a^{(k)}\big\}_{a\in E_{k}  }  $, then  we choose  $(\xi_{n}^{(k+1)})_{n \in \mathbb{N}}$ to be a subsequence  of $(\xi_{n}^{(k)})_{n \in \mathbb{N}}$ such that $  \big\{ X_a^{(k+1,\, \xi^{(k+1)}_{n})} \big\}_{a \in E_{k+1}  } $  converges in law to some limit $\big\{\mathbf{X}_a^{(k+1)} \big\}_{ a\in E_{k+1} } $.

\end{itemize}
With the sequence in $k\in \mathbb{N}_0$ of limiting array laws $\{\mathbf{X}_a^{(k)} \}_{ a\in E_{k} }$  constructed above, we will next consider properties (I)-(III).  When it comes to  property (II), we will first verify the equality in a distributional sense---see~(\ref{DistSense})---because the arrays $\{\mathbf{X}_a^{(k)} \}_{ a\in E_{k} }$ constructed above may be defined on different probability spaces for different $k\in \mathbb{N}_0$.  

 Property (I) follows immediately from the construction since all of the arrays,  $\big\{ X_a^{(k,n)} \big\}_{a\in E_k}$, used in the construction are i.i.d.   For property (II) notice that for any $k\in \mathbb{N}$
\begin{align}\label{DistSense}
 \big\{\mathbf{X}_a^{(k-1)} \big\}_{ a\in E_{k-1} }\,\stackrel{d}{=}\,\lim_{n\rightarrow \infty}\big\{ X_a^{(k-1,\,\xi_{n}^{(k-1)})} \big\}_{a\in E_{k-1}  }\,=\,&\lim_{n\rightarrow \infty}\mathcal{Q}\big\{ X_a^{( k,\,\xi_{n}^{(k)})} \big\}_{a\in E_{k}  } \,\stackrel{d}{=}\, \mathcal{Q}\big\{ \mathbf{X}_a^{(k)} \big\}_{a\in E_{k}  }\,, \  
 \end{align}
 where the second equality follows from part (II) of Lemma~\ref{LemmaMom}, and the third holds by the continuity of the map $\mathcal{Q}$.  It follows that for each $k\in \mathbb{N}$ the $\mathcal{Q}$-pyramidic array generated from $\big\{ \mathbf{X}_a^{(k-1)} \big\}_{a\in E_{k-1}  } $ is equal in distribution to the top $k-1$ layers of the $\mathcal{Q}$-pyramidic array generated by $\big\{ \mathbf{X}_a^{(k)} \big\}_{a\in E_{k}  }$.  By the Kolmogorov extension theorem,  the sequence in $k\in \mathbb{N}_0$ of arrays of random variables $ \big\{\mathbf{X}_a^{(k)} \big\}_{ a\in   E_{k} }$ can be defined on a single probability space such that $\big\{\mathbf{X}_a^{(k)} \big\}_{ a\in E_{k} }$ is a.s.\ equal to $\mathcal{Q}\big\{ \mathbf{X}_a^{(k-1)} \big\}_{a\in E_{k-1}  }$.
  For property (III),  Lemma~\ref{LemmaMom} implies that the  $m^{th}$ moment of $X_{a}^{(k,n)}$ converges to the limit  $R^{(m)}(r-k)$  for any $a\in E_k$ and $m\in \{2,3,\ldots\}$.  Since this holds for all $m$, we have that $\mathbb{E}\big[(\mathbf{X}_a^{(k)})^m\big]=R^{(m)}(r-k)$ for all $m$ by uniform integrability.

 The limiting random variables $\{\mathbf{X}_a^{(k)} \}_{ a\in E_{k} }$  take values in $[-1,\infty)$ since the random variables $\big\{1+X_h^{(n)}\big\}_{h\in E_n}$ are nonnegative by their definition~(\ref{WRemark}), and the form of the map $\mathcal{Q}$ implies that   the arrays $\big\{1+X_{a}^{(k,n)}\big\}_{ a\in E_k }$ for $\big\{X_{a}^{(k,n)}\big\}_{ a\in E_k }:=\mathcal{Q}^{n-k}\big\{X_{h}^{(n)}\big\}_{ h\in E_n }  $  must also be nonnegative.
 \end{proof}

\section{Uniqueness of the limiting $\mathcal{Q}$-pyramidic array and  universality}\label{SecOutlineMain}

The goal of this section is to prove Theorem~\ref{ThmUnique} and, simultaneously, the uniqueness part of Theorem~\ref{ThmExist} after stating the key propositions that enter into the proof.  
 Section~\ref{SecFinalPush} contains the statement of Proposition~\ref{PropFinalPush}, which is central to the  organization of  our analysis.  In Section~\ref{SecMotivate}, we heuristically motivate the definitions of the arrays of random variables that have a role in the proof of Theorem~\ref{ThmUnique}, which is in Section~\ref{SecThmUnique}.

\subsection{$\mathbf{L^2}$-bound for a contractive dynamics on arrays of random variables   }\label{SecFinalPush}
 The following proposition provides a condition template by which we can show that the random variables $\mathcal{Q}^{N}\big\{ U_e^{(N)} \big\}_{ e\in E_N  }  $ and $\mathcal{Q}^{N}\big\{ V_e^{(N)} \big\}_{ e\in E_N  }  $ are close together under the $L^2$ metric on random variables provided that  $\big\{ \big(U_e^{(N)},V_e^{(N)} \big) \big\}_{ e\in E_N  } $ is an i.i.d.\  array of ($\R^2$-valued)   random variables and the variables $U_e^{(N)} $ and $ V_e^{(N)}$ are close together in $L^2$.  In loose terms, we are bounding the sensitivity of the ``dynamics" on arrays generated by the map $\mathcal{Q}$  to the initial conditions.

\begin{proposition}\label{PropFinalPush}  Fix some  $s\in \R$, and  let $N\in\mathbb{N}$.  There exist $\delta>0$ and $C>0$ depending only on $s\in \R$ such that the statements (i)-(ii) below hold for any i.i.d.\ array  $\big\{ \big(U_e^{(N)},V_e^{(N)}\big) \big\}_{e\in E_N} $ of centered $\R^2$-valued   random variables for which $U_e^{(N)}$ has the variance bound
 \begin{align}\label{VarCond}
     \mathbb{E}\Big[ \big(  U_e^{(N)}  \big)^2  \Big]   \,< \,R(-N) \,+\,\frac{\kappa^2 s }{ N^2 } \, .  
     \end{align}

\begin{enumerate}[(i)]

\item If $\mathbb{E}\big[  \big(V_e^{(N)}-U_e^{(N)}\big)^2\big]< \delta /N^4$, then
\begin{align*}
\mathbb{E}\bigg[ \Big(\mathcal{Q}^{N}\big\{ V_e^{(N)} \big\}_{ e\in E_N  }\,-\, \mathcal{Q}^{N}\big\{ U_e^{(N)} \big\}_{ e\in E_N  }\Big)^2\bigg]^{\frac{1}{2}}\,\leq \, CN^{2}\mathbb{E}\Big[  \big(V_e^{(N)}-U_e^{(N)}\big)^2\Big]^{\frac{1}{2}} \,.
\end{align*}

\item If $\mathbb{E}\big[  \big(V_e^{(N)}-U_e^{(N)}\big)^2\big]<\delta/N^2$ and the variables  $U_e^{(N)}$ and $V_e^{(N)} -U_e^{(N)}$ are uncorrelated,  then
\begin{align*}
\mathbb{E}\bigg[ \Big(\mathcal{Q}^{N}\big\{ V_e^{(N)} \big\}_{ e\in E_N  }\,-\, \mathcal{Q}^{N}\big\{ U_e^{(N)} \big\}_{ e\in E_N  }\Big)^2\bigg]^{\frac{1}{2}}\,\leq \, CN\mathbb{E}\Big[  \big(V_e^{(N)}-U_e^{(N)}\big)^2\Big]^{\frac{1}{2}} \,.
\end{align*}

\end{enumerate}

\end{proposition}

\begin{remark} In particular, if $\big\{ \big(U_e^{(N)},V_e^{(N)}\big) \big\}_{e\in E_N} $ is a sequence in $N\in \mathbb{N}$ of arrays of random variables satisfying the conditions of Proposition~\ref{PropFinalPush} and  $\mathbb{E}\big[  \big(V_e^{(N)}-U_e^{(N)}\big)^2\big]=\mathit{o}\big(1/N^4)$, then the $L^2$ distance between $\mathcal{Q}^{N}\big\{ V_e^{(N)} \big\}_{ e\in E_N  }$ and $\mathcal{Q}^{N}\big\{ U_e^{(N)} \big\}_{ e\in E_N  }$ vanishes with large $N$.

\end{remark}

\begin{remark}\label{RemarkPropCond} By the asymptotics for $R(r)$ as $r\rightarrow -\infty$ in (II) of  Lemma~\ref{LemVar}, the right side of~(\ref{VarCond}) is equal to $R(s-N)+\mathit{o}\big(\frac{1}{N^2}\big)$.  The statement of  Proposition~\ref{PropFinalPush} is equivalent if $R(-N) +\frac{\kappa^2 s }{ N^2 }$ is replaced by $R(s-N)$.
\end{remark}

\subsection{Defining intermediary distributional approximations}\label{SecMotivate}

After the heuristic discussion below, we will state Definition~\ref{DefXs}, which defines the arrays of random variables appearing in the proof of Theorem~\ref{ThmUnique}.  Lemmas~\ref{LemI}-\ref{LemIII} in the next subsection state bounds for the $L^2$ distance/Wasserstein-$2$ distance between the random variables in these arrays, providing opportunities to apply Proposition~\ref{PropFinalPush}.

Let  $\{X^{(*,n)}\}_{a\in E_*}   $ be a minimally regular sequence in  $n\in \mathbb{N}$ of $\mathcal{Q}$-pyramidic arrays of random variables.  Proposition~\ref{PropFinalPush} combined with Remark~\ref{RemarkReduce} suggests a path for proving Theorem~\ref{ThmUnique} by showing that  for $1\ll N\ll n$ and $e\in E_N$ the $L^2$ distance between the random variables $X^{(N,n)}_e$ and $\mathbf{X}_e^{(N)}$ is small for some coupling of the variables.  To help orient the reader towards the framework of the analysis in coming sections, we will motivate the definitions of  three distributional  approximations  for the  random variable $ X^{(N,n)}_e$  that have roles in the proof of Theorem~\ref{ThmUnique}.  The analysis will be founded on the introduction of intermediary generational scales $\mathbf{n}(N), \mathbf{\widehat{n}}(N)\in \mathbb{N}$ between $N$ and $n$ that allow us to identify two sources of central limit-type renormalized sums---see (I) and (II) below---within an approximation for $X^{(N,n)}_e$. It suffices for us to take 
\begin{align}\label{DefLiln}
\mathbf{\widehat{n}}(N)\,:=\,N\,+\,\lfloor \frak{m}\log N \rfloor   \hspace{1cm}\text{and}\hspace{1cm}  \mathbf{n}(N)\,:=\,N\,+\,\lfloor 2\frak{m}\log N \rfloor  
\end{align}
for a large enough choice of $\frak{m}>0$.\footnote{For the purpose of proving Theorem~\ref{ThmUnique},   $\frak{m}\log N$ can also be replaced by $N^{\epsilon}$ for any choice of $0 <\epsilon<1/2$ in the definitions of  $\mathbf{\widehat{n}}(N)$ and $\mathbf{n }(N)$, however, this is not optimal for Theorem~\ref{ThmSharpUnique}.} In particular, when $1\ll  N\ll n$
\begin{align*}
N\,<\, \mathbf{\widehat{n}}(N)\,<\,\mathbf{n}(N)\,\ll \,n \,, \hspace{.7cm}  1\,\ll \,\mathbf{n}(N)-\mathbf{\widehat{n}}(N) \,, \hspace{.7cm} \text{and}\hspace{.7cm}  1\,\ll \,\mathbf{\widehat{n}}(N)-N  \,.
\end{align*}
For notational neatness, we will suppress the dependence of these generational parameters on $N$:   $\mathbf{\widehat{n}}(N) \equiv \mathbf{\widehat{n}} $  and  $\mathbf{n}(N)\equiv \mathbf{n}$.
\begin{remark}\label{Remark4Gen} To enable the reader to distinguish at a glance between arrays having the four distinct generational parameters  $ N<\mathbf{\widehat{n}}<\mathbf{n} \ll n$,   we will maintain a rigid indexing convention  in which the arrays with generation numbers $N$, $\mathbf{\widehat{n}}$, $\mathbf{n}$,   $n$   are respectively dummy indexed by the letters $e$, $f$, $g$,  $h$:
$$ \{x_e\}_{e\in E_N},\hspace{.8cm} \{x_f\}_{f\in E_{\mathbf{\widehat{n}}}}, \hspace{.8cm} \{x_g\}_{g\in E_{\mathbf{n}}}, \hspace{.8cm} \{x_h\}_{h\in E_{n}} \,.    $$
\end{remark}

Recall from (ii) of Notation~\ref{NotationArray} that given an array $\{ x_a \}_{a\in E_k}$  and some  $\mathbf{a}\in E_\ell$ with $0\leq \ell\leq k$, then $\{ x_a \}_{a\in \mathbf{a}\cap E_k}$ refers to the subarray labeled by all $a\in E_k$ canonically embedded in $\mathbf{a}$.  From Definition~\ref{DefQPyramid} we can write   $ X^{(N,n)}_e =\mathcal{Q}^{n-N}\big\{ X_h^{(n)} \big\}_{h\in  e\cap E_{n}} $.  For any $\mathbf{n}$ defined as above with $n\geq\mathbf{n}$, this equality can be rewritten using the identity $\mathcal{Q}=\mathcal{L}+\mathcal{E}$ as
\begin{align}\label{XApprx}
X^{(N,n)}_e\,=\,&\mathcal{L}^{\mathbf{n}-N}\mathcal{Q}^{n-\mathbf{n}}\big\{ X_h^{(n)} \big\}_{h\in e \cap E_{n}}\,+\,\sum_{k=1}^{\mathbf{n}-N} \mathcal{L}^{k-1}\mathcal{E}\mathcal{Q}^{n-N-k}\big\{ X_h^{(n)} \big\}_{h\in e \cap E_{n}}\,.\nonumber
\intertext{Lemma~\ref{LemI}, which is stated  in Section~\ref{SecThmUnique} and proven in Section~\ref{SecProofLemI},  provides an estimate by which the term $\mathcal{Q}^{n-N-k}$ in the expression above can be approximated by the partial linearization $\mathcal{L}^{\mathbf{n}-N-k}\mathcal{Q}^{n-\mathbf{n}}$:}
\,\approx \,&\mathcal{L}^{\mathbf{n}-N}\mathcal{Q}^{n-\mathbf{n}}\big\{ X_h^{(n)} \big\}_{h\in e \cap E_{n}}\,+\,\sum_{k=1}^{\mathbf{n}-N} \mathcal{L}^{k-1}\mathcal{E}\mathcal{L}^{\mathbf{n}-N-k}\mathcal{Q}^{n-\mathbf{n}}\big\{ X_h^{(n)} \big\}_{h\in e \cap E_{n}}\,=:\,\widehat{X}^{(N,n)}_e \,.
\intertext{More specifically,  the proof of Lemma~\ref{LemI} shows that the variance of the difference $X^{(N,n)}_e-\widehat{X}^{(N,n)}_e$ is of order $(\mathbf{n}-N)^2 \frac{1}{N^3} \approx  \frak{m}^2\frac{\log^2 N}{ N^{3}} $ when $1\ll N\ll n$. Furthermore, the random variables $\widehat{X}^{(N,n)}_e$ and $X^{(N,n)}_e-\widehat{X}^{(N,n)}_e$ are uncorrelated by Lemma~\ref{LemUnCor}, and thus the $L^2$ distance between  $\mathcal{Q}^N\big\{X^{(N,n)}_e\big\}_{e\in E_N   }$ by $\mathcal{Q}^N\big\{\widehat{X}^{(N,n)}_e\big\}_{e\in E_N   }$ can be shown to be small when $N$ and $n$ are large  using (ii) of Proposition~\ref{PropFinalPush}.  Since $ \mathbf{\widehat{n}}$ is between $ N$ and $\mathbf{n}  $,  we can rearrange the above as }
\,= \,&\bigg(\mathcal{L}^{\mathbf{\widehat{n}}-N}\,+\,\sum_{k=1}^{\mathbf{\widehat{n}}-N} \mathcal{L}^{k-1}\mathcal{E}\mathcal{L}^{\mathbf{\widehat{n}}-N-k}\bigg)\underbrace{\Big(\mathcal{L}^{\mathbf{n}-\mathbf{\widehat{n}}}  \mathcal{Q}^{n-\mathbf{n}}\big\{ X_h^{(n)} \big\}_{h\in e \cap E_{n}}\Big)}_{\text{(I)}}\nonumber
\\ &\,+\,\underbrace{\mathcal{L}^{\mathbf{\widehat{n}}-N }\bigg( \sum_{k=1}^{\mathbf{n}-\mathbf{\widehat{n}}} \mathcal{L}^{k-1}\mathcal{E}\mathcal{L}^{\mathbf{n}-\mathbf{\widehat{n}}-k} \mathcal{Q}^{n-\mathbf{n}}\big\{ X_h^{(n)} \big\}_{h\in e \cap E_{n}}\bigg)}_{\text{(II)}} \,.\label{2ndLabel}
\end{align}
The braced expressions above are central limit-type normalized sums (recall Remark~\ref{RemarkSum}), and thus admit Gaussian approximations when $\mathbf{n}-\mathbf{\widehat{n}}\gg 1$ and $\mathbf{\widehat{n}}-N\gg 1$:
\begin{enumerate}[(I)]

\item For $e\in E_N$  the variables in the array $\big\{ Y_f^{N,n}\big\}_{f\in e \cap E_{\mathbf{\widehat{n}}}} :=\mathcal{L}^{\mathbf{n}-\mathbf{\widehat{n}}}  \mathcal{Q}^{n-\mathbf{n}}\big\{ X_h^{(n)} \big\}_{h\in e \cap E_{n}}$
are approximately distributed  as 
\begin{align}\label{DefY}
Y_f^{N,n}\,\stackrel{ d }{\approx  }\, \mathbf{Y}_f^{(N)}\,\sim\, \mathcal{N}\big(0, R(r-\mathbf{n}) \big) \,
\end{align}
because the variables in the array $\mathcal{Q}^{n-\mathbf{n}}\big\{ X_h^{(n)} \big\}_{h\in e \cap E_{n}}$ have variance approximately equal to $R(r-\mathbf{n})$ when $n\gg 1$ by Lemma~\ref{LemmaMom}.

\item For  $\displaystyle Z_f^{N,n}:=     \sum_{k=1}^{\mathbf{n}-\mathbf{\widehat{n}}} \mathcal{L}^{k-1}\mathcal{E}\mathcal{L}^{\mathbf{n}-\mathbf{\widehat{n}}-k}\mathcal{Q}^{n-\mathbf{n}}\big\{ X_h^{(n)} \big\}_{h\in f \cap E_{n}} $, the variable $\displaystyle 
 \widebar{Z}_e^{N,n}:=\mathcal{L}^{\mathbf{\widehat{n}}-N }\big\{ Z_f^{N,n} \big\}_{f\in e \cap E_{\mathbf{\widehat{n}}}}$   has approximate distribution
\begin{align}\label{DefZ}
\widebar{Z}_e^{N,n}\,\stackrel{ d }{\approx  }\, \mathbf{Z}_e^{(N)}\,\sim\, \mathcal{N}\big(0, \varsigma^2_N \big)\,, \hspace{.5cm}\text{where}\hspace{.5cm}\varsigma_N^2\, :=\, (\mathbf{n}-\mathbf{\widehat{n}})\big( R(r-\mathbf{n}+1)\,-\, R(r-\mathbf{n}) \big)\,. 
\end{align}
The variance $\varsigma_N^2$ is the asymptotic variance of $Z_e^{N,n}$ as $n\rightarrow \infty$ as will be shown in Lemma~\ref{LemBasic}.  
\end{enumerate}
The above line of heuristic reasoning suggests that variables in the array  $\big\{ X^{(N,n)}_e \}_{e\in E_{N}}$ are close in distribution to variables in the array $\big\{ \mathbf{\widetilde{X}}^{(N)}_e \}_{e\in E_{N}}$ defined in (iii) of Definition~\ref{DefXs} below.   The random variables $\widehat{X}^{N,n}_e $ and $\mathbf{\widehat{X}}^{N,n}_e $   in (i) \& (ii) of Definition~\ref{DefXs} serve as distributional intermediaries between $X^{(N,n)}_e $ and $\mathbf{\widetilde{X}}^{(N)}_e $; see the Wasserstein-$2$ bounds for their differences in Lemmas~\ref{LemI}-\ref{LemIII}.  Note that $\widehat{X}^{N,n}_e$ in (i) is merely a different way of writing~(\ref{XApprx}).
\begin{definition}\label{DefXs} 
Let  $\mathbf{\widehat{n}},\mathbf{n}\in \mathbb{N}$ be defined as in~(\ref{DefLiln}) for a given value of $\frak{m}>0$, and let the i.i.d.\ arrays of random variables  $\big\{ Y_f^{N,n}\big\}_{f\in   E_{\mathbf{\widehat{n}}}}$,  $
\big\{ \widebar{Z}_e^{N,n}\big\}_{e\in   E_{N}}$,    $\big\{ \mathbf{Y}_f^{(N)} \big\}_{f\in  E_{\mathbf{\widehat{n}}}}$ and   $\big\{\mathbf{Z}_e^{(N)}\big\}_{e\in  E_{N}}$ be   defined as in (I) and (II) above. 

\begin{enumerate}[(i)] 

\item We define variables in the array $\big\{\widehat{X}^{N,n}_e\big\}_{e\in  E_{N}}$ as
\begin{align*}
\widehat{X}^{N,n}_e \, :=\, \,&\mathcal{L}^{\mathbf{\widehat{n}}-N}\big\{ Y_f^{N,n} \big\}_{f\in e\cap E_{\mathbf{\widehat{n}}}}\,+\,\sum_{k=1}^{\mathbf{\widehat{n}}-N} \mathcal{L}^{k-1}\mathcal{E}\mathcal{L}^{\mathbf{\widehat{n}}-N-k}\big\{ Y_f^{N,n} \big\}_{f\in e\cap E_{\mathbf{\widehat{n}}}}\,+\,\widebar{Z}_e^{N,n}\,.
\end{align*}

\item For  $\big\{ Y_f^{N,n}\big\}_{f\in   E_{\mathbf{\widehat{n}}}}$ and   $\big\{\mathbf{Z}_e^{(N)}\big\}_{e\in  E_{N}}$ independent, we define the i.i.d.\ array $\big\{\mathbf{\widehat{X}}^{N,n}_e\big\}_{e\in  E_{N}}$ to have variables with distribution
\begin{align*}\mathbf{\widehat{X}}^{N,n}_e\,:\stackrel{d}{=}\,&\mathcal{L}^{\mathbf{\widehat{n}}-N}\big\{ Y_f^{N,n} \big\}_{f\in e\cap E_{\mathbf{\widehat{n}}}}\,+\, \sum_{k=1}^{\mathbf{\widehat{n}}-N} \mathcal{L}^{k-1}\mathcal{E}\mathcal{L}^{\mathbf{\widehat{n}}-N-k}\big\{ Y_f^{N,n} \big\}_{f\in e\cap E_{\mathbf{\widehat{n}}}}\,+\, \mathbf{Z}_e^{(N)} \,.
\end{align*}

\item  For  $\big\{ \mathbf{Y}_f^{(N)} \big\}_{f\in  E_{\mathbf{\widehat{n}}}}$ and   $\big\{\mathbf{Z}_e^{(N)}\big\}_{e\in  E_{N}}$ independent, we define the i.i.d.\ array $\big\{\mathbf{\widetilde{X}}_e^{(N)}\big\}_{e\in  E_{N}}$ to have variables with distribution
\begin{align*}
\mathbf{\widetilde{X}}_e^{(N)}\,:\stackrel{d}{=}\,\mathcal{L}^{\mathbf{\widehat{n} }-N}\big\{ \mathbf{Y}_f^{(N)} \big\}_{f\in e\cap E_{\mathbf{\widehat{n}}}}\,+\,\sum_{k=1}^{\mathbf{\widehat{n}}-N}\mathcal{L}^{k-1}\mathcal{E}\mathcal{L}^{\mathbf{\widehat{n}}-N-k}\big\{ \mathbf{Y}_f^{(N)} \big\}_{f\in e\cap E_{\mathbf{\widehat{n}}}} \,+\, \mathbf{Z}_e^{(N)}\,.
\end{align*}

\end{enumerate}

\end{definition}

\begin{remark} The superscripts of the variables $\widehat{X}^{N,n}_e$, $\mathbf{\widehat{X}}^{N,n}_e$, $\mathbf{\widetilde{X}}_e^{(N)}$,   $ Y_f^{N,n}$, $\mathbf{Y}_f^{(N)}$, $Z_f^{N,n}$, $\widebar{Z}_e^{N,n}$,  and $\mathbf{Z}_e^{(N)}$ refer to  their dependence on the underlying generational parameters $N,n\in \mathbb{N}$ with $\mathbf{n}\leq n$, whereas the superscript of $X^{(N,n)}_e$ (with the parenthesis and two indices) denotes more specifically that the random variable $X^{(N,n)}_e$ is an element of the $N^{th}$ layer of a $\mathcal{Q}$-pyramidic array generated from a generation-$n$ array, $\{X^{(n)}_h\}_{h\in E_n}$.

\end{remark}

\subsection{Proof of Theorem~\ref{ThmUnique}  } \label{SecThmUnique}
We will prove Theorem~\ref{ThmUnique} and the uniqueness part of Theorem~\ref{ThmExist} after stating the crucial Lemmas~\ref{LemI}-\ref{LemIII}, whose proofs in Section~\ref{SecCentralLimit} form the core of our technical analysis.

 For $N, n\in \mathbb{N}$ with $n\geq \mathbf{n}$ and  $e\in E_N$,
let the random variables  $\widehat{X}^{N,n}_e $, $\mathbf{\widehat{X}}_e^{N,n}$, $ \mathbf{\widetilde{X}}_e^{(N)}  $  be defined as in Section~\ref{SecMotivate} for a minimally regular sequence,  $\big(\{ X_a^{(*,n)} \}_{a\in  E_{*} }\big)_{n\in \mathbb{N}} $, of $\mathcal{Q}$-pyramidic arrays with parameter $r\in \R$ and a choice of the parameter $\frak{m}>0$ in the equations~(\ref{DefLiln}) defining  $\mathbf{n}$ and $\mathbf{\widehat{n}}$.  The  lemmas below imply that the pairs  $\big(X^{(N,n)}_e ,\widehat{X}^{N,n}_e\big) $, $\big(\widehat{X}^{N,n}_e ,\mathbf{\widehat{X}}_e^{N,n}\big)$, and $\big(\mathbf{\widehat{X}}_e^{N,n},\mathbf{\widetilde{X}}^{(N)}_e\big)$ satisfy the conditions  (i) or (ii) of Proposition~\ref{PropFinalPush} when $\frak{m}\geq\frac{5}{\log b}$ after  appropriate  couplings of the variables for the latter two pairs.  The constants $\mathbf{c}>0$ in the statements of the next three lemmas depend on $\frak{m}>0$ and the sequence $\big(\{ X_a^{(*,n)} \}_{a\in  E_{*} }\big)_{n\in \mathbb{N}} $.   \vspace{.15cm}

Lemma~\ref{LemI}, which is proved in Section~\ref{SecProofLemI}, bounds the  error in $L^2$ resulting from the partial linear approximation in~(\ref{XApprx}).
\begin{lemma}\label{LemI}
 The random variables  $ X^{(N,n)}_e -\widehat{X}^{N,n}_e $ and $\widehat{X}^{N,n}_e $ are uncorrelated.  There is a positive number $\mathbf{c}$ such that for any  $N\in \mathbb{N}$  the  inequality below holds for all large enough $n\in \mathbb{N}$.
 $$\mathbb{E}\Big[ \big( X^{(N,n)}_e-\widehat{X}^{N,n}_e   \big)^2\Big]^{\frac{1}{2}}\, < \,  \mathbf{c}\frac{ \log (N+1) }{N^{\frac{3}{2}}}  $$
\end{lemma}

Lemma~\ref{LemII} provides a bound for the error, when measured in terms of the Wasserstein-2 distance,  of the Gaussian approximation heuristically motivated in (I) of Section~\ref{SecMotivate}.  The proof is in  Section~\ref{SecProofLemII} and  uses a  perturbative generalization of Stein's method that is discussed in Section~\ref{SecWasserstein}.
\begin{lemma}\label{LemII}
There exists a positive number $\mathbf{c}$   such that for any  $N\in \mathbb{N}$  the  inequality below holds  for all large enough $n\in \mathbb{N}$. 
 $$\rho_2 \big( \widehat{X}^{N,n}_e ,\mathbf{\widehat{X}}^{N,n}_e   \big)\,  < \,  \mathbf{c}\frac{\log^{-\frac{1}{6}}(N+1)}{ N^{\frac{\frak{m}}{3}\log b  +\frac{1}{3} }  } $$
 \end{lemma}

Lemma~\ref{LemIII} bounds the Wasserstein-2 distance error  resulting from the Gaussian approximation heuristically motivated in (II) of Section~\ref{SecMotivate}.  The proof is in Section~\ref{SecProofLemIII} and uses a bound (Lemma~\ref{LemNorm}) that follows from the zero bias approach to Stein's method, which is discussed in Appendix~\ref{AppendixGoldstein}.
\begin{lemma}\label{LemIII} 
 There exists a positive number $\mathbf{c}$  such that for any  $N\in \mathbb{N}$  the  inequality below holds for all large enough $n\in \mathbb{N}$. 
$$\rho_2\big(   \mathbf{\widehat{X}}_e^{N,n} ,    \mathbf{\widetilde{X}}_e^{(N)}  \big)\, < \,  \frac{\mathbf{c}}{ N^{\frac{\frak{m}}{3}\log b +\frac{1}{2} }  }  $$
\end{lemma}

\begin{remark} By definition of $\rho_2$, Lemmas~\ref{LemII} \& \ref{LemIII} imply that there are couplings $\big(\widehat{X}_e^{N,n} , \mathbf{\widehat{X}}_e^{N,n}\big) $ and $\big(\mathbf{\widehat{X}}_e^{N,n}, \mathbf{\widetilde{X}}_e^{(N)}\big) $ such that $\mathbb{E}\big[ \big( \widehat{X}_e^{N,n} -\mathbf{\widehat{X}}^{N,n}_e   \big)^2\big]$ and      $\mathbb{E}\big[ \big( \mathbf{\widehat{X}}_e^{N,n} -\mathbf{\widetilde{X}}^{(N)}_e   \big)^2\big]$ are   $ < \mathbf{c}N^{-\frac{\frak{m}}{3}\log b-\frac{1}{3}  }     $ for large $n$.
\end{remark}
\begin{remark}\label{RemarkMinAssump} When  applying  Proposition~\ref{PropFinalPush} in the proof of Theorem~\ref{ThmUnique}, we only need that the bounds $\mathbf{a}_N:=\mathbf{c} N^{-3/2}\log (N+1) $,  $\mathbf{b}_N:=\mathbf{c}N^{-1/3-\frac{\frak{m}}{3} \log b } \log^{-1/6 } (N+1)$, and  $\mathbf{c}_N:=\mathbf{c}N^{-1/2-\frac{\frak{m}}{3} \log b }$  in Propositions~\ref{LemI}-\ref{LemIII} are respectively $\mathit{o}(N^{-1})$, $\mathit{o}(N^{-2})$, and $\mathit{o}(N^{-2})$ for which it is sufficient to assume that $\frak{m}\geq 5/\log b$ for $\mathbf{b}_N$ and $\mathbf{c}_N$. 
\end{remark}

The following easy corollary of Lemmas~\ref{LemI}\,-\,\ref{LemIII} verifies the condition~(\ref{VarCond}) in the statement of Proposition~\ref{PropFinalPush} for the pairs of random variables discussed above, and its proof is in Section~\ref{SecMiscFirst}.

\begin{corollary}\label{CorollaryTriv} Define $\frak{m}:=5/\log b$. For any $s\in (r,\infty)$ the inequality  $\mathbb{E}\big[  \big(U_{e}^{(N)}\big)^2 \big] <  R(-N)+\frac{\kappa^2 s  }{N^2}$ holds for  $U_{e}^{(N)}$ equal to $\widehat{X}^{N,n}_e$, $\mathbf{\widehat{X}}^{N,n}_e$, and $\mathbf{\widetilde{X}}^{(N)}_e$ for  large enough $N$ and  $n\geq \mathbf{n}$.
\end{corollary}

\begin{remark}
The relevant sense of a given statement holding ``for large enough $N$ and $n$" will always be that there exists a constant $\lambda>0$ and an increasing function $\Lambda:\mathbb{N}\rightarrow (0,\infty)$ such that the statement is true whenever $N>\lambda$  and $n>\Lambda(N)$.

\end{remark} \vspace{.11cm}

Let us temporarily assume Proposition~\ref{PropFinalPush},  Lemmas~\ref{LemI}\,-\,\ref{LemIII}, and Corollary~\ref{CorollaryTriv} to complete the remainder of the proof of Theorem~\ref{ThmUnique}. As in Corollary~\ref{CorollaryTriv},  we will define $\frak{m}:=5/\log b$ for the reason explained in Remark~\ref{RemarkMinAssump}.

\begin{proof}[Proof of Theorem~\ref{ThmUnique} and Theorem~\ref{ThmExist} (uniqueness part)] Let $\big(\{ X_a^{(*,n)} \}_{a\in  E_{*} }\big)_{n\in \mathbb{N}} $ be a  minimally regular sequence of $\mathcal{Q}$-pyramidic arrays of random variables with parameter $r\in \R$.  By Remark~\ref{RemarkReduce} it suffices for us to focus on distributional convergence in the case $k=0$ in which the array $\big\{X^{(k,n)}_a\big\}_{a\in E_k}$  consists of a single random variable, $ X^{(0,n)}$. We have divided the analysis below into parts (a)-(d).\vspace{.2cm}

\noindent \textbf{(a) Setting up:}  For $n\geq \mathbf{n}$ let the arrays of random variables $\big\{ \widehat{X}_e^{N,n}\big\}_{e\in E_{N}}$, $\big\{ \mathbf{\widehat{X}}_e^{N,n}\big\}_{e\in E_{N}} $, and $\big\{ \mathbf{\widetilde{X}}_e^{(N)}\big\}_{e\in E_{N}} $ be defined as in Definition~\ref{DefXs}.  We will show that the Wasserstein-$2$ distance between $ X^{(0,n)}$ and   $  \mathcal{Q}^{N}\big\{\mathbf{\widetilde{X}}_e^{(N)}\big\}_{e\in E_{N}}$ converges to zero as  $N$ and $n$ grow.   Writing $ X^{(0,n)}= \mathcal{Q}^{N}\big\{ X_e^{(N,n)}\big\}_{e\in E_{N}}  $  and applying the triangle inequality yields
\begin{align}
\rho_2\Big( X^{(0,n)},   \mathcal{Q}^{N}\big\{ \mathbf{\widetilde{X}}_e^{(N)}\big\}_{e\in E_{N}}  \Big)\,\,\leq \,&\rho_2\Big(  \mathcal{Q}^{N}\big\{ X_e^{(N,n)}\big\}_{e\in E_{N}} ,   \mathcal{Q}^{N}\big\{ \widehat{X}_e^{N,n}\big\}_{e\in E_{N}}  \Big) \nonumber \\  &\,+\,\rho_2\Big(  \mathcal{Q}^{N}\big\{ \widehat{X}_e^{N,n}\big\}_{e\in E_{N}} ,   \mathcal{Q}^{N}\big\{ \mathbf{\widehat{X}}_e^{N,n}\big\}_{e\in E_{N}}  \Big) \nonumber   \\ & \,+\,\rho_2\Big(  \mathcal{Q}^{N}\big\{ \mathbf{\widehat{X}}_e^{N,n}\big\}_{e\in E_{N}} ,   \mathcal{Q}^{N}\big\{ \mathbf{\widetilde{X}}_e^{(N)}\big\}_{e\in E_{N}}  \Big)\,. \nonumber 
\intertext{For any particular couplings of the above three pairs of random variables, we have    }
 \leq \,&
\mathbb{E}\bigg[ \Big( \mathcal{Q}^{N}\big\{ X_e^{(N,n)}\big\}_{e\in E_{N}} \,-\,   \mathcal{Q}^{N}\big\{ \widehat{X}_e^{N,n}\big\}_{e\in E_{N}}  \Big)^2  \bigg]^{\frac{1}{2}} \nonumber   \\ &\,+\, \mathbb{E}\bigg[  \Big( \mathcal{Q}^{N}\big\{ \widehat{X}_e^{N,n}\big\}_{e\in E_{N}} \,-\,   \mathcal{Q}^{N}\big\{ \mathbf{\widehat{X}}_e^{N,n}\big\}_{e\in E_{N}}  \Big)^2  \bigg]^{\frac{1}{2}} \nonumber   \\ & \,+\, \mathbb{E}\bigg[ \Big( \mathcal{Q}^{N}\big\{ \mathbf{\widehat{X}}_e^{N,n}\big\}_{e\in E_{N}} \,-\,   \mathcal{Q}^{N}\big\{ \mathbf{\widetilde{X}}_e^{(N)}\big\}_{e\in E_{N}}  \Big)^2     \bigg]^{\frac{1}{2}}
.  \label{Couplings}
\end{align}
The random variables  $\mathcal{Q}^{N}\big\{ \widehat{X}_e^{N,n}\big\}_{e\in E_{N}}$ and $\mathcal{Q}^{N}\big\{ X_e^{(N,n)}\big\}_{e\in E_{N}} $ are already defined  in the same probability space, and we will not require any special coupling between them.  Notice that the expressions on the right side above have the form of those expressions bounded in Proposition~\ref{PropFinalPush}. \vspace{.25cm}

\noindent \textbf{(b) Verifying the conditions of Proposition~\ref{PropFinalPush}:}   By Lemma~\ref{LemI} the variables $X_e^{(N,n)}-\widehat{X}_e^{N,n}$ and $\widehat{X}_e^{N,n}$ are uncorrelated, and there is a positive sequence $\{ \mathbf{a}_N\}_{N\in \mathbb{N}}$ with $ \mathbf{a}_N=\mathit{o}(N^{  -1} )$ such that $$\mathbb{E}\Big[ \big(  X_e^{(N,n)} -  \widehat{X}_e^{N,n} \big)^2\Big]< \mathbf{a}_N^2$$ for any fixed $N$ and large enough $n$.  By Lemmas~\ref{LemII} \&~\ref{LemIII} and Remark~\ref{RemarkMinAssump}, there is a positive sequence $\{ \mathbf{b}_N\}_{N\in \mathbb{N}}$  with $ \mathbf{b}_N=\mathit{o}(N^{  -2} )$ and i.i.d.\ couplings  $\big\{\big(\widehat{X}_e^{N,n}, \mathbf{\widehat{X}}_e^{N,n}  \big)\big\}_{e\in E_N}  $  and $\big\{\big(\mathbf{\widehat{X}}_e^{N,n}, \mathbf{\widetilde{X}}_e^{(N)}  \big)\big\}_{e\in E_N}  $ such that 
\begin{align*}
\mathbb{E}\Big[ \big(  \widehat{X}_e^{N,n} -  \mathbf{\widehat{X}}_e^{N,n} \big)^2\Big]\,<\, \mathbf{b}_N^2 \hspace{1cm}\text{and}  \hspace{1cm} \mathbb{E}\Big[ \big(  \mathbf{\widehat{X}}_e^{N,n} -  \mathbf{\widetilde{X}}_e^{(N)} \big)^2\Big]\,<\, \mathbf{b}_N^2
\end{align*}
for any fixed $N$ and large enough $n\geq \mathbf{n}$.  Corollary~\ref{CorollaryTriv} implies that the arrays $\big\{ \widehat{X}_e^{N,n}\big\}_{e\in E_{N}}$, $\big\{ \mathbf{ \widehat{X}}_e^{N,n}\big\}_{e\in E_{N}}$, $\big\{ \mathbf{ \widetilde{X}}_e^{(N)}\big\}_{e\in E_{N}}$ satisfy  condition~(\ref{VarCond}) of Proposition~\ref{PropFinalPush} for any $s\in (r,\infty)$  and large $N$ and $n$.  Moreover,  the above considerations imply that  for  large enough $N$ and  $n$ we have the following: 
\begin{itemize}
\item the array $ \big\{ \big( \widehat{X}_e^{N,n}, X_e^{(N,n)}\big)\big\}_{e\in E_{N}} $  satisfies the conditions for part (ii) of  Proposition~\ref{PropFinalPush} with $\big(\widehat{X}_e^{N,n}, X_e^{(N,n)}  \big)=\big( U_e^{(N)},V_e^{(N)} \big)   $, 
\item  the  arrays $\big\{\big( \widehat{X}_e^{N,n}, \mathbf{\widehat{X}}_e^{N,n}\big)\big\}_{e\in E_{N}} $ satisfy the conditions for part (i) of  Proposition~\ref{PropFinalPush}, and
\item the arrays $\big\{ \big(\mathbf{\widehat{X}}_e^{N,n},  \mathbf{\widetilde{X}}_e^{(N)}\big)\big\}_{e\in E_{N}} $  satisfy the conditions for part (i) of  Proposition~\ref{PropFinalPush}. 
\end{itemize}
 \vspace{.25cm}

\noindent \textbf{(c) Returning to~(\ref{Couplings}):} Therefore with three applications of Proposition~\ref{PropFinalPush} to the right side of~(\ref{Couplings}) there is a $C>0$ such that for large enough $N$ and $n\geq \mathbf{n}$ we have the first inequality below.
\begin{align*}
\rho_2\Big( &X^{(0,n)},   \mathcal{Q}^{N}\big\{ \mathbf{\widetilde{X}}_e^{(N)}\big\}_{e\in E_{N}}  \Big)\\
 \,\leq \,& CN\mathbb{E}\Big[ \big(  X_e^{(N,n)} -  \widehat{X}_e^{N,n} \big)^2\Big]^{\frac{1}{2}} \,+\, CN^{2} \mathbb{E}\Big[ \big(  \widehat{X}_e^{N,n}  -    \mathbf{\widehat{X}}_e^{N,n}  \big)^2\Big]^{\frac{1}{2}} \,+\, CN^{2} \mathbb{E}\Big[\big(  \mathbf{\widehat{X}}_e^{N,n}  -  \mathbf{\widetilde{X}}_e^{(N)}  \big)^2\Big]^{\frac{1}{2}}\,
\\  \, < \,&  \mathbf{c}CN\mathbf{a}_N   \,+\,\mathbf{c}CN^{2}\mathbf{b}_N\,+\,\mathbf{c}CN^{2}\mathbf{b}_N 
\end{align*}
The second inequality holds by Lemmas~\ref{LemI}\,-\,\ref{LemIII}.
As $N\rightarrow \infty$ the above goes to zero by the asymptotic properties of $\mathbf{a}_N$ and $\mathbf{b}_N$. \vspace{.25cm}

\noindent \textbf{(d) Connecting with the random array constructed in Section~\ref{SectExist}:} We have established that the Wasserstein-$2$ distance between  $X^{(0,n)} $ and   $\mathcal{Q}^{N}\big\{ \mathbf{\widetilde{X}}_e^{(N)}\big\}_{e\in E_{N}} $  vanishes as $n$ and $N$ grow.  Let  $\big\{\mathbf{X}^{(k)}_a\big\}_{a\in E_k}$  be the sequence in $k\in \mathbb{N}_0$ of arrays of random variables for parameter $r\in \R$ constructed in Section~\ref{SectExist} through  subsequential distributional limits of $\big\{ X_a^{(k,n)}\big\}_{a\in E_{k}}$ as $n\rightarrow \infty$.  
As mentioned in Remark~\ref{RemarkLimitRegular},  the arrays  $\big\{\mathbf{X}^{(k)}_a\big\}_{a\in E_k}$  form a parameter-$r$ regular sequence of $Q$-pyramidic arrays of random variables with no $n\in \mathbb{N}$ dependence. Thus we can apply the distributional convergence result that we have just proved to the special case  $\big\{ X_a^{(k,n)}\big\}_{a\in E_{k}} :=\big\{ \mathbf{X}_a^{(k)}\big\}_{a\in E_{k}} $ to get that the Wassertstein-$2$ distance between  $\mathbf{X} $ and  $\mathcal{Q}^{N}\big\{\mathbf{\widetilde{X}}_e^{(N)}\big\}_{e\in E_{N}} $  converges to zero as $N\rightarrow \infty$.  Therefore,  $\rho_2\big(X^{(0,n)}  , \mathbf{X} \big)$ vanishes with large $n$ and the law of $\mathbf{X} $ must be unique.  
\end{proof}

\section{Proof of Proposition~\ref{PropFinalPush}}\label{SecTemplate}

\begin{proof}   By Remark~\ref{RemarkPropCond}, the condition~(\ref{VarCond}) is equivalent to assuming that the variance of $U_e^{(N)}$ is smaller than $R(s-N)$. For  $0\leq k \leq N$, define the i.i.d.\ arrays of random variables
\begin{align*}
\big\{ U_a^{(k,N)} \big\}_{ a\in E_k  }\,:=\,\mathcal{Q}^{N-k}\big\{ U_e^{(N)} \big\}_{ e\in E_N  }\hspace{1cm}\text{and}\hspace{1cm}\big\{ V_a^{(k,N)} \big\}_{ a\in E_k  }\,:=\,\mathcal{Q}^{N-k}\big\{ V_e^{(N)} \big\}_{ e\in E_N  }\,
\end{align*}
and $ W_a^{(k,N)}\,:=\,V_a^{(k,N)}\,-\,U_a^{(k,N)}$.    The variables $U_a^{(k,N)}$, $V_a^{(k,N)}$,  $W_a^{(k,N)}$ have mean zero, and $U_a^{(k,N)}$ has variance 
\begin{align}\label{SigR}
 \big(\sigma^{(N)}_k\big)^2\,:=\,\textup{Var}\big(U_a^{(k,N)}\big)\,= \, M^{N-k}\big( (\sigma^{(N)})^2\big) \,\leq \,M^{N-k}\big( R(s-N)   \big) \,=\,R(s-k) 
 \end{align}  for  $\big(\sigma^{(N)} \big)^2:= \mathbb{E}\big[ \big(U_e^{(N)}\big)^2   \big]  $, where the second equality above holds by Remark~\ref{RemarkArrayVar} (note that $U_a^{(k,N)}$ has the same law as a generation $N-k$ partition function).  The inequality uses our assumption that the variance of $U_e^{(N)}$ is smaller than $R(s-N)$, and the last equality is property (I) of Lemma~\ref{LemVar}.

We have the following recursive relation for the variables $W_a^{(k,N)}$
\begin{align}
W_a^{(k,N)}  \,=\,&\frac{1}{b}   \sum_{i=1}^b  \prod_{j=1}^b \Big(1+ U_{a\times (i,j)}^{(k+1, N)}+W_{a\times (i,j)}^{(k+1,N)}   \Big)\,-\,\frac{1}{b}  \sum_{i=1}^b \prod_{j=1}^b \Big(1+ U_{a\times (i,j)}^{(k+1,N)} \Big)\,.\nonumber \intertext{Expanding the products on the left and cancelling  yields  }  \,=\,&\frac{1}{b}  \sum_{i=1}^b \Bigg( \sum_{j=1}^b  W_{a\times (i,j)}^{(k+1,N)}     \,+\,   \sum_{\substack{1\leq j, J\leq b \\  j\neq J  }   }    W_{a\times (i,j)}^{(k+1,N)}U_{a\times (i,J)}^{(k+1,N)}\nonumber \\ & \hspace{1.3cm}\,+\, \sum_{ \substack{A,B\subset \{1,\ldots , b\} \\ A\cap B=\emptyset \text{ and }|A|\geq 1 \\ |A|\geq 2 \text{ or }  |B|\geq 2   } }\prod_{j\in A} W_{a\times (i,j)}^{(k+1,N)}  \prod_{J\in B} U_{a\times (i,J)}^{(k+1,N)}    \Bigg)\,.\label{RecurW}
\end{align}
Since the arrays are i.i.d.\ and centered, the recursive formula above shows, by induction, that if $W_e^{(N)}:=V_e^{(N)}-U_e^{(N)}$ is uncorrelated with $U_e^{(N)}$  for  $e\in E_N$ then $W_a^{(k,N)}$ is uncorrelated with $U_a^{(k,N)}$ for all $0\leq k< N$ and $a\in E_k$. In particular if  $U_e^{(N)}$ and $V_e^{(N)}-U_e^{(N)}$ are uncorrelated, then   $\mathcal{Q}^{N}\big\{ U_e^{(N)} \big\}_{ e\in E_N  }$  and $\mathcal{Q}^{N}\big\{ V_e^{(N)} \big\}_{ e\in E_N  }-\mathcal{Q}^{N}\big\{ U_e^{(N)} \big\}_{ e\in E_N  }$ are uncorrelated.

Define the multivariate  polynomial
\begin{align*}
 P\big(x  , y , z \big) \,:=\,\sum_{ \substack{A,B\subset \{1,\ldots , b\} \\ A\cap B=\emptyset \text{ and }|A|\geq 1 \\ |A|\geq 2 \text{ or }  |B|\geq 2   } }\sum_{u=0}^{\textup{min}(|A|,|B|)} {|A| \choose u} {|B| \choose u}  x^{|A|-u  } y^{|B|-u} z^{2u}       \,.
\end{align*} The form of the polynomial $P$ implies that there exists a $\mathbf{c}>0$ such that 
\begin{align}\label{PBound}
\big|P\big(x  , y , z \big)\big|\,\leq \, \mathbf{c} x\big( x+y^2   \big)
\end{align}
 for all $(x,y,z)$ with $ 0\leq x,y \leq 1$ and $|z|\leq \sqrt{xy}$. To see the above inequality, notice that a single term  $x^{|A|-u  } y^{|B|-u} z^{2u} $ has absolute value bounded by $x^{|A|  } y^{|B|}$, and we  have  $x^{|A|  } y^{|B|}\leq x^2$ when $|A|\geq 2$ and   $x^{|A|  } y^{|B|}\leq x y^2$ when $|B|\geq 2$ since  $|A|\geq 1$.

 Let $ \big( \varrho_{k}^{(N)}\big)^2$ denote the second moment of $W_a^{(k,N)}$, and define $u_{k}^{(N)}:=\mathbb{E}\big[U_a^{(k,N)}   W_a^{(k,N)} \big]$. Taking the second moment of~(\ref{RecurW}) yields
\begin{align}
  \big( \varrho_{k}^{(N)}\big)^2 \,=\,& \big(\varrho_{k+1}^{(N)}\big)^2  \,+\, (b-1)\big(\varrho_{k+1}^{(N)}\big)^2 \big(\sigma_{k+1}^{(N)}\big)^2\,+\,(b-1)\big(u_{k+1}^{(N)}\big)^2 \, +\,P\Big( \big(\varrho_{k+1}^{(N)}\big)^2  , \big(\sigma_{k+1}^{(N)}\big)^2  , u_{k+1}^{(N)} \Big),\nonumber
  \intertext{where the middle two terms on the right side above correspond to the middle term on the right side  of~(\ref{RecurW}).  Define $\epsilon\in \{0,2\}$ as $\epsilon=0$  when $U_e^{(N)}$ and  $W_e^{(N)} =V_e^{(N)}  -  U_e^{(N)} $ are uncorrelated and as $\epsilon =2$ otherwise. Using that $\kappa^2:=\frac{2}{b-1}$ and applying Cauchy-Schwarz to bound $u_{k+1}^{(n)}   $ by $\varrho_{k+1}^{(N)} \sigma_{k+1}^{(N)}$  yields }
  \,\leq  \,& \big(\varrho_{k+1}^{(N)}\big)^2  \,+\, \frac{2+\epsilon}{\kappa^2}\big(\varrho_{k+1}^{(N)}\big)^2 \big(\sigma_{k+1}^{(N)}\big)^2\, +\,P\Big( \big(\varrho_{k+1}^{(N)}\big)^2  , \big(\sigma_{k+1}^{(N)}\big)^2  , u_{k+1}^{(N)} \Big)\,. \nonumber \\
\intertext{Next we can apply~(\ref{PBound}) to bound $P\Big( \big(\varrho_{k+1}^{(N)}\big)^2  , \big(\sigma_{k+1}^{(N)}\big)^2  , u_{k+1}^{(N)} \Big)$ and get} 
   \leq \,&\big(\varrho_{k+1}^{(N)} \big)^2\,+\, \frac{2+\epsilon}{\kappa^2}\big(\varrho_{k+1}^{(N)} \big)^2\big(\sigma_{k+1}^{(N)}\big)^2  \,+\,\mathbf{c}\big(\varrho_{k+1}^{(N)} \big)^2\Big(\big(\varrho_{k+1}^{(N)} \big)^2\,+\,\big(\sigma_{k+1}^{(N)} \big)^4    \Big)  \nonumber \\
    \leq \,&\big(\varrho_{k+1}^{(N)} \big)^2\,+\, \frac{2+\epsilon}{\kappa^2}\big(\varrho_{k+1}^{(N)} \big)^2 R(s-k-1)  \,+\,\mathbf{c}\big(\varrho_{k+1}^{(N)} \big)^2\Big(\big(\varrho_{k+1}^{(N)} \big)^2\,+\,\big(R(s-k-1) \big)^2   \Big)\,,\label{Return}
\end{align}
where the last inequality follows from~(\ref{SigR}).

In the following analysis, we will temporarily assume that $s<-1$ and that $s$ is sufficiently far in the negative direction so that $R(r) >\frac{\kappa^2}{-r}$ for all $r\in (-\infty,s]$, which is possible by the asymptotics for $R(r)$ as $r\rightarrow -\infty$ in (II) of Lemma~\ref{LemVar}. These assumptions ensure that the terms $R(s-\ell)-\frac{\kappa^2}{\ell-s} $ in the sums over $\ell\in \mathbb{N}$  below are positive and that the denominator $\ell-s$ is bounded away from zero.    Recall that $\big(\varrho_{N}^{(N)}\big)^2:= \mathbb{E}\Big[ \big( V_e^{(N)}-U_e^{(N)}  \big)^2  \Big] $ for $e\in E_N$.  Suppose that $\big(\varrho_{N}^{(N)}\big)^2 < \delta/N^{2+\epsilon}  $, where
\begin{align*}
\delta\,:=\,&\Bigg(\inf_{\substack{ N\in \mathbb{N}  \\  0\leq k\leq N}  }   \frac{ \big(R(s-k)\big)^2 N^{2+\epsilon}  }{ \big( \frac{1}{k-s}\big)^{2+\epsilon} (N-s)^{2+\epsilon}  }   \Bigg) \textup{exp}\Bigg\{-\frac{2+\epsilon}{\kappa^2}\sum_{\ell=1}^\infty \bigg(R(s-\ell)\,-\,\frac{\kappa^2}{\ell-s}    \bigg)  \,-\,2\mathbf{c}\sum_{\ell=1}^\infty \big(R(s-\ell) \big)^2      \Bigg\}   \,. 
\end{align*}
Note that $\delta>0$ because property (II) in Lemma~\ref{LemVar} implies that the series $ \sum_{\ell=1}^{\infty}\big(R(s-\ell)  \,-\,\frac{ \kappa^2 }{\ell-s}\big)$ and  $  \sum_{\ell=1}^\infty\big(R(s-\ell)   \big)^2    $
are summable and because the asymptotics  $R(-r)\sim \frac{\kappa^2}{r}$ for $r\gg 1$ implies that the infimum above is finite.

Let $\mathbf{k}^{(N)}$ be the smallest $k\in \mathbb{N}_0$ such that  $\big(\varrho_{k}^{(N)} \big)^2 \leq \big(R(s-k) \big)^2$.  Note that the inequality  $\big(\varrho_{N}^{(N)} \big)^2 \leq \big(R(s-N) \big)^2$ holds by the assumption $\big(\rho_{N}^{(N)}\big)^2<\delta/N^{2+\epsilon}  $ and the definition of $\delta$, and thus we must have $\mathbf{k}^{(N)}\leq N$.
 For $k\in \mathbb{N}_0$ with $ k+1 \in \big[\mathbf{k}^{(N)}, N\big] $, we have the inequality
\begin{align*}
\big(\varrho_{k}^{(N)} \big)^2 
\leq \,&\big(\varrho_{k+1}^{(N)} \big)^2\,+\,\frac{2+\epsilon}{\kappa^2}\big(\varrho_{k+1}^{(N)} \big)^2 R(s-k-1)    \,+\,2\mathbf{c}\big(\varrho_{k+1}^{(N)} \big)^2\big(R(s-k-1)    \big)^2     \\
\leq \,& \big(\varrho_{k+1}^{(N)} \big)^2\textup{exp}\bigg\{\frac{2+\epsilon}{\kappa^2}R(s-k-1) \,+\,2\mathbf{c}\big(R(s-k-1)    \big)^2      \bigg\}\,.
\intertext{Applying the above recursively  and rearranging  yields  }
\leq \,& \big(\varrho^{(N)}_N \big)^2\textup{exp}\bigg\{\frac{2+\epsilon}{\kappa^2}\sum_{\ell=k+1}^N R(s-\ell)  \,+\,2\mathbf{c}\sum_{\ell=k+1}^N\big(R(s-\ell) \big)^2     \bigg\}
\\
= \,& \big(\varrho^{(N)}_N \big)^2\textup{exp}\Bigg\{ \big(2+\epsilon\big)\sum_{\ell=k+1  }^{N  } \frac{1}{\ell-s}    \,+\,  \frac{2+\epsilon}{\kappa^2}\sum_{\ell=k+1 }^{N  } \Big(R(s-\ell)  \,-\,\frac{ \kappa^2 }{\ell-s}\Big)\\ &\hspace{2cm} \,+\,2\mathbf{c}\sum_{\ell=k+1}^N\big(R(s-\ell)   \big)^2     \Bigg\}\,.
\intertext{The sum $\sum_{\ell=k+1}^N \frac{1}{\ell-s}$ is a Riemann lower bound for $\int_{k}^{N}\frac{1}{t-s}dt =  \log\big(\frac{N-s}{k-s}\big) $, so the above is smaller than}
\leq \,& \big(\varrho^{(N)}_N \big)^2\Big(\frac{N-s}{k -s}  \Big)^{2+\epsilon}\textup{exp}\Bigg\{\frac{2+\epsilon}{\kappa^2}\sum_{\ell=k+1}^N\Big(R(s-\ell)  \,-\,\frac{ \kappa^2 }{\ell-s}\Big)\,+\,2\mathbf{c}\sum_{\ell=k+1}^N\big(R(s-\ell)   \big)^2     \Bigg\}\,.
\intertext{By definition of $\delta>0$  }
\leq \,&  \frac{ \big(R(s-k)\big)^2  }{ \delta  } N^{2+\epsilon} \big(\varrho^{(N)}_N \big)^2\,.
\end{align*}
Notice that  
$\big(\varrho_{k}^{(N)} \big)^2$ is smaller than $\big(R(s-k)\big)^2$ because $\big(\varrho^{(N)}_N \big)^2<\delta/N^{2+\epsilon}$. Hence, $k\geq \mathbf{k}^{(N)}$ and by induction on $k$ we can deduce that $\mathbf{k}^{(N)}=0$.  Therefore we can apply the above inequality with $k=0$ to get
\begin{align*}
\mathbb{E}\bigg[\Big( \mathcal{Q}^{N}\big\{ V_e^{(N)} \big\}_{ e\in E_N  }\,-\,\mathcal{Q}^{N}\big\{ U_e^{(N)} \big\}_{ e\in E_N  }   \Big)^2\bigg]\,=:\,\big(\varrho_{0}^{(N)} \big)^2\,\leq \,C^2N^{2+\epsilon}\big(\varrho^{(N)}_N \big)^2 \,,
\end{align*}
where $C:=R(s)/\delta^{\frac{1}{2}}$.  Since $\big(\varrho^{(N)}_N \big)^2:=\mathbb{E}\big[ \big( V_e^{(N)}-U_e^{(N)}  \big)^2  \big]$, the proof is complete in the case when $s\in (-\infty,-1)$ is sufficiently far in the negative direction, i.e., for all $s\in (-\infty,\theta]$ for some $\theta<-1$. 

For the general case of $s\in \R$, pick $n\in \mathbb{N}$ large enough so that $s-n\leq \theta$.   Our previous result for $s\in (-\infty,\theta]$ implies that there exist $\delta', C' >0$ such for any $N>n$, 
\begin{align}\label{Old}
\big(\varrho^{(N)}_N \big)^2\,<\,\frac{\delta'}{(N-n)^{2+\epsilon}}\hspace{.3cm}\implies \hspace{.3cm}  \big(\varrho_{n}^{(N)} \big)^2\,\leq \,C'(N-n)^{2+\epsilon}\big(\varrho^{(N)}_N \big)^2  \,.    
\end{align}
The above uses that  $\big(\varrho^{(N)}_n\big)^2:=\mathbb{E}\big[\big(U_a^{(n,N)}-V_a^{(n,N)}\big)^2\big]$, where $U_a^{(n,N)}$, $V_a^{(n,N)}$ are distributed as generation $N-n$ partition functions and that  we can write the argument of $R$ in the inequality $\mathbb{E}\big[ \big(U_e^{(N)}\big)^2 \big]<R(s-N)$ in the form $s-N:=s'-(N-n) $ for $s':=s-n$ with $s'\leq \theta$. Through iterating~(\ref{Return}) $n$ times, we can get an inequality of the form $  \big(\varrho_{0}^{(N)} \big)^2\,\leq\, Q_s\big( \big(\varrho_{n}^{(N)}\big)^2\big) $ for a degree $2^n$ polynomial $Q_s$ with nonnegative coefficients that depend on $s\in \R$ and no constant term.    Since $Q_s$ is differentiable and $Q_s(0)=0$, there are $\lambda, \mathbf{c}>0$ such that $Q_s(x)\leq \mathbf{c}x$ for all $x\in [0,\lambda]$.  Combining~(\ref{Old}) with this inequality for $Q_s$ yields that for any $N > n$, 
$$ \big(\varrho^{(N)}_N \big)^2\,<\,\frac{\min(\delta',\frac{\lambda}{C'}) }{N^{2+\epsilon}}\hspace{.3cm}\implies \hspace{.3cm}  \big(\varrho_{0}^{(N)} \big)^2\,\leq \,\mathbf{c}C'N^{2+\epsilon}\big(\varrho^{(N)}_N \big)^2  \,.    $$
This  implies that the desired inequalities hold in the general case of $s\in \R$.
\end{proof}

\section{The three approximation lemmas}\label{SecCentralLimit}

In this section, we will  prove Lemmas~\ref{LemI}-\ref{LemIII}.   Recall from Sections~\ref{SecMotivate} \&~\ref{SecThmUnique} that Lemma~\ref{LemI} involves bounding the error of the partial linearization~(\ref{XApprx}) and Lemmas~\ref{LemII} \&~\ref{LemIII} are  Gaussian approximations of the terms (I) and (II) in~(\ref{2ndLabel})  driven by central limit-type normalized sums   that occur  at different generational scales.

As before, let $\big(\big\{X_{a}^{(*,n)}\big\}_{ a\in E_* }\big)_{n\in \mathbb{N}}  $ denote a minimally regular sequence of $\mathcal{Q}$-pyramidic arrays with parameter $r\in \R$.  For $0\leq k\leq n$, $a\in E_k$, and $h\in E_n$, we will frequently use the notation
\begin{align}\label{DefLittleSigma}
 \sigma_{k,n}^2\,:=\,\textup{Var}\big(  X_{a}^{(k,n)}\big) \,=\,M^{n-k}\big( \sigma_n^2  \big)\,  \hspace{1cm}\text{for} \hspace{1cm}   \sigma_n^2 \,:=\, \textup{Var}\big(  X_{h}^{(n)}\big)  \,.   
 \end{align}
Note that $\sigma_{k,n}^2\rightarrow R(r-k)$ as $n\rightarrow \infty$ by (III) of Lemma~\ref{LemmaMom} with $m=2$.


\subsection{Proof of Lemma~\ref{LemI}}\label{SecProofLemI}

\begin{proof}[Proof of Lemma~\ref{LemI}]
The variables $X^{(N,n)}_e- \widehat{X}^{N,n}_e$ and $\widehat{X}^{N,n}_e$ are uncorrelated by Lemma~\ref{LemUnCor} and have mean zero, so the square of the $L^2$ distance between $X^{(N,n)}_e$ and $\widehat{X}^{N,n}_e$ can be written as 
\begin{align}
\mathbb{E}\Big[  \big( & X_e^{(N,n)}-  \widehat{X}_e^{N,n}\big)^2   \Big] \nonumber \\  \, =\,&\mathbb{E}\Big[ \big( X_e^{(N,n)}\big)^2 \Big]\,-\,  \mathbb{E}\Big[ \big(\widehat{X}_e^{N,n}\big)^2  \Big]\,, \nonumber
\intertext{and by definition of $\widehat{X}_e^{N,n}$ the above is  equal to  } \nonumber
 \, =\,&\mathbb{E}\Big[ \big( X_e^{(N,n)}\big)^2  \Big]   \,-\,\mathbb{E}\Bigg[\bigg(\mathcal{L}^{\mathbf{n}-N}\big\{ X_g^{(\mathbf{n},n)} \big\}_{g\in e \cap E_{\mathbf{n}}}\,+\,\sum_{k=1}^{\mathbf{n}-N} \mathcal{L}^{k-1}\mathcal{E}\mathcal{L}^{\mathbf{n}-N-k}\big\{ X_g^{(\mathbf{n},n)} \big\}_{g\in e \cap E_{\mathbf{n}}}\bigg)^2\Bigg]\,.\nonumber 
 \intertext{By Lemma~\ref{LemUnCor}, the random variables in the sum above are uncorrelated, and thus we have the equality}
  \, =\,&\mathbb{E}\Big[ \big( X_e^{(N,n)}\big)^2  \Big]   \,-\,\Bigg(\mathbb{E}\bigg[\Big(\mathcal{L}^{\mathbf{n}-N}\big\{ X_g^{(\mathbf{n},n)} \big\}_{g\in e \cap E_{\mathbf{n}}}\Big)^2\bigg]\,+\,\sum_{k=1}^{\mathbf{n}-N} \mathbb{E}\bigg[\Big(\mathcal{L}^{k-1}\mathcal{E}\mathcal{L}^{\mathbf{n}-N-k}\big\{ X_g^{(\mathbf{n},n)} \big\}_{g\in e \cap E_{\mathbf{n}}}\Big)^2\bigg]\Bigg)\,. \nonumber
  \intertext{Let $\sigma^2_{k,n}  $ be defined as in~(\ref{DefLittleSigma}). By Remark~\ref{RemarkArrayVar}, we can write the above as   }
 \, =\, & \sigma^2_{N,n} \,-\,   \sigma^2_{\mathbf{n},n} \,-\,(\mathbf{n}-N)\Big(M\big(\sigma^2_{\mathbf{n},n}\big)\,-\,\sigma^2_{\mathbf{n},n}   \Big)   \,. \label{SigmaForm}
\intertext{For any fixed $k\in \mathbb{N}_0$ the sequence $\sigma^2_{k,n}$ converges as $n\rightarrow \infty$ to $R(r-k)$ by (III) of Lemma~\ref{LemmaMom} with $m=2$.  It follows that for each $N\in \mathbb{N}$ there is a sequence $\{ \xi_N(n)\}_{n\in \mathbb{N}}$ such that $\xi_N(n)$ vanishes with large $n$ and  the above is equal to    }
 \, =\, &R(r-N)  \,-\,   R(r-\mathbf{n}) \,-\, (\mathbf{n}-N)\Big(M\big(R(r-\mathbf{n})\big)\,-\,R(r-\mathbf{n})  \Big)   
 \,+\,\xi_N(n)\,. \label{RForm}
  \intertext{Using that $M\big(R(s)\big)=R(s+1)$ for any $s\in \R$ by (I) of Lemma~\ref{LemVar}, we can rewrite the above as a telescoping sum } \, =\, & \sum_{k=1}^{\mathbf{n}-N  } \bigg( \Big(  M\big(R(r-N-k)\big)  \,-\,  R(r-N-k)\Big) \,-\,\Big(M\big(R(r-\mathbf{n})\big)\,-\,R(r-\mathbf{n})  \Big)    \bigg)   
 \,+\,\xi_N(n) \,.\nonumber 
\intertext{Since $ M(x)=x+\frac{b-1}{2}x^2+\mathit{O}\big(x^3\big) $ for $x\ll 1$, $\kappa^2:=\frac{2}{b-1}$, and $R(s)=\frac{\kappa^2}{-s}+\mathit{O}\Big(\frac{\log (-s)  }{ s^2}   \Big)$ as $s\rightarrow -\infty$ by (II) of Lemma~\ref{LemVar}, we get}
 \, =\, & \sum_{k=1}^{\mathbf{n}-N  } \Bigg( \bigg(\frac{\kappa^2}{(N+k-r)^2}+\mathit{O}\bigg(\frac{\log (N+k)  }{ (N+k)^3}   \bigg)\bigg) \,-\,\bigg(\frac{\kappa^2}{(\mathbf{n}-r)^2}+\mathit{O}\bigg(\frac{\log \mathbf{n}  }{ \mathbf{n}^3}   \bigg) \bigg)  \Bigg)   
 \,+\,\xi_N(n) \nonumber 
 \, . \intertext{Since  $
\mathbf{n}-N := \lfloor 2\frak{m}\log N\rfloor$, there is a $\frak{c}>0$ such that for all $N\in \mathbb{N}$ and $n\geq \mathbf{n}$  }
 \, \leq\, & \frak{c}\frac{\log^2 (N+1)}{N^{3}}   \,+\,\xi_N(n) \,. \label{XiOne}
\end{align}
The proof is complete since $\xi_N(n)\rightarrow 0$ as $n\rightarrow \infty$.
\end{proof}

\subsection{A generalization of Stein's auxiliary functions}\label{SecWasserstein}

Before moving  to the proof of Lemma~\ref{LemII} we will discuss a generalized version of the auxiliary functions used in Stein's method~\cite{Stein},  which  is a general strategy for proving the central limit theorem under the Wasserstein-$1$ metric.  For random variables $X$ and  $Y$ with $\mathbb{E}[|X|],\mathbb{E}[|Y|]<\infty$, the Wasserstein-$1$ distance has the dual form  
\begin{align*}
 \rho_{1}(X,Y)\,=\, \sup_{H\in \textup{Lip}_1  }\Big(\mathbb{E}\big[H(X)  \big]\,-\,   \mathbb{E}\big[H(Y)  \big] \Big)   \,,
\end{align*}
where $\textup{Lip}_1$ is the collection of all  Lipshitz functions on $\R$ with  Lipshitz constant $\leq 1$.  Given $H\in \textup{Lip}_1$ define the auxiliary function $f:\R \rightarrow \R$ 
\begin{align}
 f(x)\,:=\,  e^{\frac{x^2}{2}}\int_{-\infty}^{x}\big(H(t) \,-\,\widehat{H} \big) e^{-\frac{1}{2}t^2}  dt\,, \hspace{.5cm}\text{where}\hspace{.5cm}\widehat{H}\,:=\,\int_{-\infty}^{\infty}H(r)\frac{ e^{-\frac{1}{2}r^2  } }{\sqrt{2\pi}}dr\,.
\end{align} 
The function $f$ solves the differential equation  
\begin{align}\label{Aux}
H(x)\,-\,\widehat{H} \,:=\, f'(x)\,-\,xf(x)
\end{align}
and has the following convenient uniform bounds on its first two derivatives:
\begin{align}\label{Uniform}
\sup_{x\in \mathbb{R}}\big|f'(x)\big|\,\leq \, 1  \hspace{1cm}\text{and}  \hspace{1cm}\sup_{x\in \mathbb{R}}\big|f''(x)\big|\,\leq \, 2\,. 
\end{align}
Thus if $X$ is a random variable with finite variance and $\mathcal{X}\sim \mathcal{N}(0,1)$ then 
\begin{align}
\mathbb{E}\big[ H(X)   \big]\,-\,\mathbb{E}\big[ H(\mathcal{X})   \big]\,=\, \mathbb{E}\big[ f'(X)   \big]\,-\,\mathbb{E}\big[ X f(X) \big]\,.
\end{align}
A  useful feature of the auxiliary function, $f$, is that the Wasserstein-$1$ distance between the distributions of  $ X$ and $\mathcal{X}$ can be reduced to a quantity only involving $X$.

We will require a perturbative  generalization of Stein's method that bounds the Wasserstein-$1$ distance between random variables of the form $X:=Y+Z$ and $\mathcal{X}:=Y+\mathcal{Z}$ for variables  $Y$, $Z$, $\mathcal{Z}$ satisfying that $Z$ is centered with $\textup{Var}(Z)=1  $ and $\mathcal{Z}\sim \mathcal{N}(0,1)$ is independent of $Y$.  In other words, we would like to show how to bound the error of replacing the random variable $Z$ with a  standard normal $\mathcal{Z}$ independent of $Y$.  In this case  we will define an auxiliary function $F:\R^2\rightarrow \R$ for a given $H\in \textup{Lip}_1$ that satisfies the following partial differential equation analogous to~(\ref{Aux}):
\begin{align}\label{AuxII}
H(y+z)\,-\,\int_{\R}H(y+r)\frac{ e^{-\frac{r^2}{2} }}{\sqrt{2\pi}   }dr \,:=\, \partial_z F(y,z)\,-\,zF(y,z)\,.
\end{align}
The following proposition, whose proof  is in Section~\ref{SecProofCLT}, provides bounds for the first- and second-order partial derivatives of $F$  in analogy to~(\ref{Uniform}).

\begin{proposition}\label{PropAltStein}
 Define $F:\R^2\rightarrow \R$ for $H\in \textup{Lip}_1$ through the formula
$$  F(y,z)\,:=\, e^{\frac{z^2}{2}}\int_{-\infty}^{z} \bigg( H(y+t)\,-\,\int_{\R}H(y+r)\frac{e^{-\frac{r^2}{2}}}{\sqrt{2\pi}}dr \bigg)  e^{-\frac{t^2}{2}}dt\,. $$
For all $(y,z)\in \R^2$, 
$$\big|\partial_y F(y,z)\big|\,\leq \,\sqrt{\pi/2}\,,\quad  \big|\partial_z F(y,z)\big|\,\leq \,1 \, ,  \quad \text{and}  \hspace{.5cm} \big|\partial_y^2 F(y,z)\big|,\,\big|\partial_z\partial_y F(y,z)\big|,\,\big|\partial_z^2 F(y,z)\big| \leq \,2 \,.  $$

\end{proposition}

The trivial corollary below  generalizes Proposition~\ref{PropAltStein} to arbitrary variance $\sigma^2>0$. 
\begin{corollary}\label{CorrAltStein}  Define $F_{\sigma}:\R^2\rightarrow \R$ for  $H\in \textup{Lip}_1$ through the formula
$$  F_{\sigma}(y,z)\,:=\,\frac{1}{\sigma} e^{\frac{z^2}{2\sigma^2}}\int_{-\infty}^{z} \bigg( H(y+t)\,-\,\int_{\R}H(y+r)\frac{e^{-\frac{r^2}{2\sigma^2}}}{\sqrt{2\pi\sigma^2}}dr \bigg)  e^{-\frac{t^2}{2\sigma^2}}dt\,. $$
The function $  F_{\sigma}(y,z)$  solves the partial differential equation
$$ H(y+z)\,-\,\int_{\R}H(y+r)\frac{e^{-\frac{r^2}{2\sigma^2}}}{\sqrt{2\pi\sigma^2}}dr \,=\,\sigma \frac{\partial}{\partial z}  F_{\sigma}(y,z)\,-\,\frac{z}{\sigma}F_{\sigma}(y,z) \,, $$
and for all $(y,z)\in \R^2$, 
$$\big|\partial_y F_{\sigma}(y,z)\big| \, \leq \, \sqrt{\pi/2}\,,\quad  \big|\partial_z F_{\sigma}(y,z)\big|\,\leq \,1 \,, \quad \text{and}  \hspace{.5cm} \big|\partial_y^2 F_\sigma(y,z)\big|,\,\big|\partial_z\partial_y F_{\sigma}(y,z)\big|,\,\big|\partial_z^2 F_\sigma(y,z)\big| \leq \,\frac{2}{\sigma} \,.  $$

\end{corollary}

\begin{proof} Define $\widehat{F}_\sigma(y,z)\,:=\,\frac{1}{\sigma}F_{\sigma}(\sigma y, \sigma z)$ and $\widehat{H}_\sigma(z):=\frac{1}{\sigma}H(\sigma z)$.  Notice that we can write $\widehat{F}_\sigma$ as
$$  \widehat{F}_{\sigma}(y,z)\,=\, e^{\frac{z^2}{2}}\int_{-\infty}^{z} \bigg( \widehat{H}_{\sigma}(y+t)\,-\,\int_{\R}\widehat{H}_{\sigma}(y+r)\frac{e^{-\frac{r^2}{2}}}{\sqrt{2\pi}}dr \bigg)  e^{-\frac{t^2}{2}}dt\,. $$
Since $\widehat{H}_\sigma(z)\in \textup{Lip}_1$, it follows that the first- and second-order derivatives of $\widehat{F}_{\sigma}$ have the bounds in Proposition~\ref{PropAltStein}. From the equation $F_\sigma(y,z)=\sigma \widehat{F}_{\sigma}(\frac{y}{\sigma},\frac{z}{\sigma})$ we see that the derivatives of $F_\sigma$ have the desired bounds.  
\end{proof}

\subsection{Proof of Lemma~\ref{LemII} }\label{SecProofLemII}

For $N,n\in \mathbb{N}$ with $n\geq \mathbf{n}$, we will maintain the usual convention that  $ e\in E_N  $,  $f\in E_{\mathbf{\widehat{n}}} $, and $g\in E_{\mathbf{n}} $.   Recall that  $Y_f^{ N,n}$ is defined  in~(\ref{DefY}), $Z^{N,n}_f $ is defined above~(\ref{DefZ}), and $ \widehat{X}^{N,n}_e $, $  \mathbf{\widehat{X}}^{N,n}_e $  are defined in Definition~\ref{DefXs}.  We will need the following lemma, which collects some statements about the second and fourth moments of these  random variables.  The proof is in Section~\ref{SecProofCLT}.

\begin{lemma}\label{LemBasic} Let the random variables $Y_f^{ N,n}$, $Z^{N,n}_f $, $ \widehat{X}^{N,n}_e $,  $  \mathbf{\widehat{X}}^{N,n}_e $ 
 be defined in terms of a minimally regular sequence of $\mathcal{Q}$-pyramidic arrays of random variables    $\big(\big\{ X^{(*,n)}_{a}  \big\}_{a\in E_* }   \big)_{n\in \mathbb{N}} $ with parameter $r\in \R$.
\begin{enumerate}[(i)]

\item The variance of $Y_f^{N,n}:=  \mathcal{L}^{\mathbf{n}-\mathbf{\widehat{n}} }\big\{ X_g^{( \mathbf{n},n)} \big\}_{g\in f\cap E_{\mathbf{n}}}$ is $\sigma_{\mathbf{n},n}^2 $, and $\displaystyle \lim_{n\rightarrow \infty}\sigma_{\mathbf{n},n}^2=R(r- \mathbf{n})$. Moreover, $\sigma_{\mathbf{n},n}^2$ is bounded from above and below by constant multiples of $\frac{1}{N}$ for all $n,N\in \mathbb{N}$ with $n\geq \mathbf{n}$.

\item The variance of $ Z_f^{N,n}:=\sum_{k=1}^{\mathbf{n}-\mathbf{\widehat{n}}} \mathcal{L}^{k-1}\mathcal{E}\mathcal{L}^{\mathbf{n}-\mathbf{\widehat{n}}-k}\big\{ X_g^{( \mathbf{n},n)} \big\}_{g\in f\cap E_{\mathbf{n}}} $ has the large $n$ convergence
$$ \varsigma_{N,n}^2\,:=\,\textup{Var}\big( Z_f^{N,n}   \big)  \hspace{.4cm}\stackrel{n\rightarrow \infty  }{\longrightarrow}\hspace{.4cm}\varsigma_{N}^2\,:=\,(\mathbf{n}-\mathbf{\widehat{n}})\big(R(r-\mathbf{n}+1)\,-\,R(r-\mathbf{n}) \big)\, .    $$
Moreover, $\varsigma_{N,n}^2$ is bounded by a constant multiple of $\frac{\log(N+1)}{N^2}$ for all $n, N\in \mathbb{N}$ with $n\geq \mathbf{n}$.

\item  There is a $C>0$ such that the fourth moments of the random variables $Y_f^{N,n}$ and $Z_f^{N,n}$ are respectively bounded by $\frac{C}{N^2}$ and $ C\frac{\log^2(N+1)}{N^4  }$   for all  $N,n\in \mathbb{N}$ with $n\geq 2\mathbf{n}$.

\item     There is a $C>0$ such that the fourth moments of the random variables $X^{(\mathbf{n},n)}_g$, $\widehat{X}^{N,n}_e$,  $\mathbf{\widehat{X}}^{N,n}_e$ are bounded by $\frac{C}{N^2}$  for all $N,n\in \mathbb{N}$ with $n\geq 2\mathbf{n}$.

\end{enumerate}

\end{lemma}
\begin{remark}\label{RemarkVarSigma} For  (ii) of Lemma~\ref{LemBasic},  note that $\varsigma_{N}^2$ is bounded from below by a constant multiple $c>0$ of $\frac{\log (N+1)}{N^2}$ for all $N\in \mathbb{N}$ as a consequence of (II) of Lemma~\ref{LemVar} and since $\mathbf{n}\sim N$ for $N\gg 1$ and $\mathbf{n}-\mathbf{\widehat{n}}\propto \log N $.
\end{remark}

The lemma below, whose proof is in Section~\ref{SecProofCLT}, follows easily from Holder's inequality and the definition of Wasserstein-$p$ distance.  
\begin{lemma}\label{Lemma1to2} For $m\in \mathbb{N}$, let $X$ and $Y$  be random variables with finite $(m+1)^{th}$ absolute  moments.  We have the following bound on the Wasserstein-$2$ distance between $X$ and $Y$ using the Wasserstein-$1$ distance:
\begin{align*}
\rho_2\big(  X, Y \big)\,\leq \, 2^{\frac{m+1}{2m}}\big(\rho_1 ( X , Y ) \big)^{\frac{m-1}{2m}} \Big(  \mathbb{E}\big[|  X|^{m+1}\big]^{\frac{1}{2m}}\,+\,\mathbb{E}\big[ |    Y  |^{m+1}\big]^{\frac{1}{2m}} \Big) \,.
\end{align*}

\end{lemma}

\begin{proof}[Proof of Lemma~\ref{LemII}] This proof is divided into parts (a)-(g).\vspace{.2cm}

\noindent \textbf{(a) Notation:}   For $e\in E_N$ we can write  $\widehat{X}^{N,n}_e $ and $\mathbf{\widehat{X}}^{N,n}_e$ in the forms
\begin{align}\label{TheXs}
\widehat{X}^{N,n}_e \,=\, \widebar{X}_e^{N,n}  \,+\,\widebar{Y}_e^{N,n}\,+\,\widebar{Z}_e^{N,n} \hspace{1cm}\text{and}\hspace{1cm}\mathbf{\widehat{X}}^{N,n}_e\,=\, \widebar{X}_e^{N,n}\,+\,\widebar{Y}_e^{N,n}\,+\,\mathbf{Z}^{(N)}_e \,, 
\end{align}
where the random variables $ \widebar{X}_e^{N,n}$,  $\widebar{Y}_e^{N,n}$,  $\widebar{Z}_e^{N,n}$ are defined as
\begin{align*}
\widebar{X}_e^{N,n}\,:=\,&\mathcal{L}^{\mathbf{\widehat{n}}-N }\big\{ Y^{N,n}_f \big\}_{f\in e\cap E_{\mathbf{\widehat{n}}}} \,, \nonumber  \\
\widebar{Y}_e^{N,n}\,:=\,&\sum_{k=1}^{\mathbf{\widehat{n}}-N} \mathcal{L}^{k-1}\mathcal{E}\mathcal{L}^{\mathbf{\widehat{n}}-N-k}\big\{ Y^{N,n}_f \big\}_{f\in e\cap E_{\mathbf{\widehat{n}}}}\,, \nonumber  \\
\widebar{Z}_e^{N,n}\,:=\, &\mathcal{L}^{\mathbf{\widehat{n}}-N}\big\{ Z^{N,n}_f \big\}_{f\in e\cap E_{\mathbf{\widehat{n}}}}\,, \nonumber 
\end{align*}  
and recall that $\mathbf{Z}^{(N)}_e$ is the normal random variable  (independent of $\widebar{X}_e^{N,n}$ and $\widebar{Y}_e^{N,n}$)  defined  in~(\ref{DefZ}).\vspace{.3cm}

\noindent \textbf{(b) Stein's method:} Next we will use Stein's method to bound the  Wasserstein-$1$ distance between  $ \widehat{X}^{N,n}_e $ and $ \mathbf{\widehat{X}}^{N,n}_e $.  By definition of Wasserstein-$1$ distance,
\begin{align}\label{RhoOne}
\rho_1\big( \widehat{X}^{N,n}_e , \mathbf{\widehat{X}}^{N,n}_e  \big)\,=\,&\sup_{H\in \textup{Lip}_1}\Big|\mathbb{E}\big[H\big( \widehat{X}^{N,n}_e  \big)\big]\,-\,\mathbb{E}\big[H\big(   \mathbf{\widehat{X}}^{N,n}_e    \big) \big]\Big| \,.
\end{align}
For a given $H:\R\rightarrow \R$ with Lipschitz constant less than $1$, define $F:\R^2\rightarrow \R$  as in Corollary~\ref{CorrAltStein} with $\sigma:=\varsigma_N$.  Then $F$ is a solution to the partial differential equation
\begin{align}\label{DefH}
H(x+z)\,-\,\mathbb{E}\big[ H\big(x+\mathbf{Z}^{(N)}_e \big) \big]\,=\,\varsigma_{N} \partial_2 F(x, z)\,-\,\frac{z}{\varsigma_{N}} F(x, z)\,,
\end{align}
where the expectation is w.r.t.\  $\mathbf{Z}^{(N)}_e\sim \mathcal{N}\big(0, \varsigma_{N}^2 \big) $. By Corollary~\ref{CorrAltStein}, the first-order partial derivatives of $F$ are bounded by $\sqrt{\pi/2}$ and the second-order partial derivatives are bounded by $2/\varsigma_N$.

By~(\ref{TheXs}) and~(\ref{DefH}), to  bound the expression in the supremum of~(\ref{RhoOne}), we must  bound the absolute value of
\begin{align}\label{Steiny}
&\mathbb{E}\bigg[\varsigma_{N}\partial_2 F\left(\widebar{X}_e^{N,n}  +\widebar{Y}_e^{N,n},\,\widebar{Z}_e^{N,n}\right)\,-\,\frac{\widebar{Z}_e^{N,n}}{\varsigma_{N}} F\left(\widebar{X}_e^{N,n}  +\widebar{Y}_e^{N,n},\, \widebar{Z}_e^{N,n}\right)\bigg]\,. \nonumber 
\intertext{Since  $ \widebar{Z}_e^{N,n}$ is  a sum over $f\in e\cap E_{\mathbf{\widehat{n}}}$ of terms $\frac{1}{b^{\mathbf{\widehat{n}}-N } } Z^{N,n}_f   $, the above can be written as   }
&\,=\,\underbrace{\varsigma_{N}\mathbb{E}\left[\partial_2 F\big( \widebar{X}_e^{N,n}\,+\,\widebar{Y}_e^{N,n},\,\widebar{Z}^{N,n}_e \big)\right]}_{(\textup{I})}\,-\,  \underbrace{\frac{\varsigma_{N}^{-1}}{b^{\mathbf{\widehat{n}}-N } }\sum_{f\in e\cap E_{\mathbf{\widehat{n}}}  } \mathbb{E}\left[ Z^{N,n}_f F\big(\widebar{X}_e^{N,n}  +\widebar{Y}_e^{N,n},\, \widebar{Z}_e^{N,n}\big)\right]}_{(\textup{II})}\,,\nonumber 
\intertext{and with  $\widebar{V}^{N,n}_e\,:=\,\big( \widebar{X}_e^{N,n}\,+\,\widebar{Y}_e^{N,n},\,\widebar{Z}^{N,n}_e \big) $ we have the compact form}
&\,=\,\varsigma_{N}\mathbb{E}\big[\partial_2 F\big(\widebar{V}^{N,n}_e\big)\big]\,-\,  \frac{\varsigma_{N}^{-1}}{b^{\mathbf{\widehat{n}}-N } }\sum_{f\in e\cap E_{\mathbf{\widehat{n}}}  } \mathbb{E}\Big[ Z^{N,n}_f F\big(\widebar{V}^{N,n}_e\big)\Big]\,.
\end{align}
As in the usual implementation of Stein's method, we would like to tease out cancellations between (I) and (II) by writing the random variable $\widebar{Z}^{N,n}_e$ in (II) as a sum of a ``large" term, $ \widebar{Z}^{N,n}_e- \frac{1}{b^{\mathbf{\widehat{n}}-N } } Z^{N,n}_f $, and a ``small" term,  $\frac{1}{b^{\mathbf{\widehat{n}}-N } } Z^{N,n}_f $, and then Taylor expanding (II). The complicating feature here is that $\widebar{X}_e^{N,n}  +\widebar{Y}_e^{N,n}$ is not independent of $Z^{N,n}_f$.\vspace{.3cm}

\noindent \textbf{(c) Identifying the dependent factors:}  Next we seek to separate out the  dependence of the random variables $\widebar{X}_e^{N,n}$ and  $\widebar{Y}_e^{N,n}$ on the random variable $Z^{N,n}_f $ for a given $f\in e\cap  E_{\mathbf{\widehat{n}}} $.  More precisely, we can define a  term $B_f^{N,n}$ such that the statements (i)-(iii) below hold for the $\R^2$-valued random variable $\Delta_f^{N,n}:=\frac{1}{b^{\mathbf{\widehat{n}}-N}}\big(Y_f^{N,n}+ Y_f^{N,n} B_f^{N,n },\,Z^{N,n}_f \big)$. 
\begin{enumerate}[(i)]
\item The random variables $Z^{N,n}_f $, $Y_f^{N,n}$,  $B_f^{N,n}$ have mean zero. 

\item The random vector $\big(\widebar{V}^{N,n}_e -\Delta_f^{N,n}, B_f^{N,n}\big)$ is independent of  $\big(Y_f^{N,n}, Z^{N,n}_f  \big)$.

\item  The random variables $Y_f^{N,n}$ and  $Z^{N,n}_f $ are uncorrelated.  Thus with (ii) the random variables $Y_f^{N,n}$,  $Z^{N,n}_f $, $B_f^{N,n}$ are pairwise uncorrelated.

\end{enumerate}
For $f\in  e\cap E_{\mathbf{\widehat{n}}} $ the definition of $B_f^{N,n}$ is as follows:
\begin{align*}
B_f^{N,n}\,:=\,  b^{\mathbf{\widehat{n}}-N}\frac{\partial\mathcal{F}}{\partial y_f} \big\{ Y^{N,n}_{\widehat{f}} \big\}_{\widehat{f}\in e\cap E_{\mathbf{\widehat{n}}}} \,,
\end{align*}
where the function $\mathcal{F}  $, which maps arrays $\{y_a \}_{a\in  E_{\mathbf{\widehat{n}}-N}}$ into $\R$, is defined below.\footnote{Recall that for $e\in E_N$ the indexing set $e\cap E_{\mathbf{\widehat{n}}}$ is canonically identifiable with $E_{\mathbf{\widehat{n}}-N}$.}   The variable $\widebar{Y}_e^{N,n}$ is a multilinear function, $\mathcal{F}$, of the array $\big\{ Y^{N,n}_f \big\}_{f\in e\cap E_{\mathbf{\widehat{n}}}} $, where
\begin{align*}
\mathcal{F}\big\{&y_a \big\}_{a\in E_{\mathbf{\widehat{n}}-N}}\\ 
\,:=\,& \sum_{k=1}^{\mathbf{\widehat{n}}-N} \mathcal{L}^{k-1}\mathcal{E}\mathcal{L}^{\mathbf{\widehat{n}}-N-k}\big\{y_a \big\}_{a\in E_{\mathbf{\widehat{n}}-N}}
\\
=\,& \sum_{k=1}^{\mathbf{\widehat{n}}-N}  \frac{1}{b^{k}}\sum_{ \mathbf{a}\in   E_{k-1}  }\Bigg[ \sum_{i=1}^b \Bigg(\prod_{j=1}^b \bigg(1+\frac{1}{b^{\mathbf{\widehat{n}}-N-k} }    \sum_{ a\in (\mathbf{a}\times (i,j))\cap  E_{\mathbf{\widehat{n}}-N}  }y_a\bigg)\,-\,1\Bigg) \,-\,\frac{1}{b^{\mathbf{\widehat{n}}-N-k} }  \sum_{ a\in \mathbf{a}\cap  E_{\mathbf{\widehat{n}}-N}  }y_a\Bigg] \,.
\end{align*}
Moreover, the partial derivative of  $\mathcal{F}$ with respect to $y_\alpha$ has the form
\begin{align}\label{FormOfF}
\frac{\partial \mathcal{F}}{\partial y_\alpha}\big\{y_a \big\}_{a\in  E_{\mathbf{\widehat{n}}-N}}\,=\,\frac{1}{b^{\mathbf{\widehat{n}}-N}}\sum_{k=1}^{\mathbf{\widehat{n}}-N}\Bigg( \prod_{ \mathbf{\widehat{a}}\in  E_{k}^{\updownarrow \alpha}  } \bigg(1+  \frac{1}{b^{\mathbf{\widehat{n}}-N-k} }    \sum_{ a\in \mathbf{\widehat{a}}\cap E_{\mathbf{\widehat{n}}-N}  }y_a\bigg)  \,-\,1\Bigg)\,,
\end{align}
where $E_{k}^{\updownarrow \alpha}  $ is the $(b-1)$-element subset of $E_{k}  $ consisting of elements $\mathbf{\widehat{a}}  $ with the following three $\mathbf{(1)}$-$\mathbf{(3)}$ restrictions:  $\mathbf{(1)}$ $\alpha \not\subset\mathbf{\widehat{a}}$, $\mathbf{(2)}$ there is a path in $\Gamma_k$ that  passes over both $ \alpha$ and $\mathbf{\widehat{a}}$, and $\mathbf{(3)}$ there is an element in $E_{k-1}$ that contains both $\alpha$ and $\mathbf{\widehat{a}}$.\footnote{The  elements $\mathbf{\widehat{a}}\in E_k^{\updownarrow \alpha} $ correspond to the $\mathbf{a}\times (i,j)\in E_k$ in the above expression for $\mathcal{F}\big\{y_a \big\}_{a\in E_{\mathbf{\widehat{n}}-N}}$.}

Next we justify statements (i)-(iii).  For statement (i), note that the variables $X_f^{N,n}$, $Y_f^{N,n}$, $B_f^{N,n}$ are multilinear polynomials of the array $\big\{ X^{(\mathbf{n},n)}_g  \big\}_{g\in e\cap E_{\mathbf{n}}   } $ that have no constant term, and consequently  these variables have mean zero. Statement (iii) follows from Lemma~\ref{LemUnCor} because the random variables  have the forms $Y_f^{N,n}:=\mathcal{L}^{\mathbf{n}-\mathbf{\widehat{n}}}\big\{ X^{(\mathbf{n},n)}_g  \big\}_{g\in f\cap E_{\mathbf{n}}   }  $  and $Z^{N,n}_f :=\sum_{k=1}^{\mathbf{n}-\mathbf{\widehat{n}}} \mathcal{L}^{k-1}\mathcal{E}\mathcal{L}^{\mathbf{n}-\mathbf{\widehat{n}}-k}\big\{ X_g^{( \mathbf{n},n)} \big\}_{g\in f\cap E_{\mathbf{n}}}$. Note, in particular, that $Y_f^{N,n}$ and $Z_f^{N,n}$  are functions of the random variables  $ X^{(\mathbf{n},n)}_g $ with $ g\in f\cap E_{\mathbf{n}}   $.    The form~(\ref{FormOfF}) of the multilinear polynomial $\frac{\partial \mathcal{F}}{\partial y_f}\{y_{\hat{f}} \}_{\hat{f}\in  e\cap E_{\mathbf{\widehat{n}}}}$ implies that $B_f^{N,n}$ only depends on variables in the array $\big\{ X^{(\mathbf{n},n)}_g  \big\}_{g\in e\cap E_{\mathbf{n}}   } $ with $g\notin f\cap E_{\mathbf{n}}$.  Hence $B_f^{N,n}$ is independent of $\big(Y_f^{N,n}, Z^{N,n}_f\big)$.  By using that $\widebar{Y}_e^{N,n}= \mathcal{F}\big\{ Y^{N,n}_f \big\}_{f\in e\cap E_{\mathbf{\widehat{n}}}} $, the difference between the $\R^2$-valued random variables $\widebar{V}^{N,n}_e$ and $\Delta_f^{N,n}$ can be written as
\begin{align*}
\widebar{V}^{N,n}_e -\Delta_f^{N,n}=\frac{1}{b^{\mathbf{\widehat{n}}-N}}\sum_{\substack{\widehat{f}\in   e\cap E_{\mathbf{\widehat{n}}}\\ \widehat{f}\neq f }}\big(Y_{\widehat{f}}^{N,n},\,Z^{N,n}_{\widehat{f}} \big)+\Big(\mathcal{F}\big\{Y_{\widehat{f}}^{N,n}  \big\}_{\widehat{f}\in   e\cap E_{\mathbf{\widehat{n}}}}-Y_{f}^{N,n} \frac{\partial \mathcal{F}}{\partial y_f}\big\{Y_{\widehat{f}}^{N,n} \big\}_{\widehat{f}\in   e\cap E_{\mathbf{\widehat{n}}}} ,\,0\Big)\,.
\end{align*}
The  multilinearity of  $\mathcal{F}$ implies that   $\mathcal{G}\big\{y_{\widehat{f}} \big\}_{\widehat{f}\in e\cap E_{\mathbf{\widehat{n}}}}:=\mathcal{F}\big\{y_{\widehat{f}} \big\}_{\widehat{f}\in e\cap E_{\mathbf{\widehat{n}}}}-y_{f}\frac{\partial \mathcal{F}}{\partial y_f}\big\{y_{\widehat{f}} \big\}_{\widehat{f}\in  e\cap E_{\mathbf{\widehat{n}}}} $ does not depend on the variable $y_f$.  The right side of the display above  is a function of the variables $\big(Y_{\widehat{f}}^{N,n},\,Z^{N,n}_{\widehat{f}} \big)$ with $\widehat{f}\in  e\cap E_{\mathbf{\widehat{n}}}$ and $\widehat{f}\neq f$, and  thus $\widebar{V}^{N,n}_e -\Delta_f^{N,n}$ is independent of  $\big(Y_{f}^{N,n},\,Z^{N,n}_{f} \big)$.  In fact, these observations imply that $B_f^{N,n}$ and $\widebar{V}^{N,n}_e -\Delta_f^{N,n}$ are
 are jointly independent of  $\big(Y_f^{N,n}, Z^{N,n}_f  \big)$, i.e., (ii).

 With~(\ref{FormOfF}) and the triangle inequality,   we can bound the $L^2$ norm of $B^{N,n}_f$ for $f\in e\cap E_{\mathbf{\widehat{n}}}$ by
\begin{align}\label{BBound} 
\mathbb{E}\Big[\big|B^{N,n}_f\big|^2\Big]^{\frac{1}{2}}\,\leq  \,&\sum_{k=N+1}^{\mathbf{\widehat{n}}}\mathbb{E}\left[\Biggl( \prod_{ \mathbf{\widehat{a}}\in (e\cap E_{k})^{\updownarrow f}  } \Bigg(1+  \frac{1}{b^{\mathbf{\widehat{n}}-k} }    \sum_{ \widehat{f}\in \mathbf{\widehat{a}}\cap E_{\mathbf{\widehat{n}}}  }Y_{\widehat{f}}^{N,n}\Bigg)  \,-\,1\Biggr)^2\right]^{\frac{1}{2}}\,.\nonumber
\intertext{Since the random variables $Y_{\widehat{f}}^{N,n}$ have variance $\sigma^2_{\mathbf{n},n}$ by  part (i) of Lemma~\ref{LemBasic}, the above  is equal to   }
\,= \,&(\mathbf{\widehat{n}}-N)  \Big(\big(1+\sigma_{\mathbf{n},n }^2\big)^{b-1}-1   \Big)^{\frac{1}{2}}\,\leq \,C\frac{\log (N+1)}{ N^{1/2}}\,.
\end{align}
The inequality holds for some $C>0$ and all $n\geq \mathbf{n}$ as a consequence of part (i) of Lemma~\ref{LemBasic}.\vspace{.3cm}

\noindent \textbf{(d) Stein analysis:} Now we are ready to begin an analysis of the expression~(\ref{Steiny}). By Taylor's theorem to second-order, the expression inside the expectation in  (II) has the form
\begin{align}\label{Taylor}
  Z^{N,n}_f F\big( \widebar{V}^{N,n}_e\big) \nonumber 
     \,=\, & Z^{N,n}_f F\Big(  \widebar{V}^{N,n}_e - \Delta_f^{N,n}\Big) 
\\ & \,+\, Z^{N,n}_f \Delta_f^{N,n} \mathlarger{\mathlarger{\mathbf{\cdot}}} \nabla F\Big(  \widebar{V}^{N,n}_e   - \Delta_f^{N,n}\Big) 
\nonumber   \\  & \,+\,\frac{1}{2}Z^{N,n}_f \big( \Delta_f^{N,n}  \big)^{\otimes^2}\mathlarger{\mathlarger{\mathbf{\cdot}}} (\mathbf{D}_2 F)\Big(  \widebar{V}^{N,n}_e - \mathbf{r}_f \Delta_f^{N,n} \Big) \,,
\end{align}
where $\mathbf{D}_2$ is the 2-tensor of second-order derivatives and  $\mathbf{r}_f$ is some value between $0$ and $1$ depending on $  \widebar{V}^{N,n}_e  $ and $\Delta_f^{N,n}$. 
The expectation of the first expression on the right side of~(\ref{Taylor}) is zero by observations (i)-(iii) in part (c) above.  By definition of $ \Delta_f^{N,n}$,  the second term on the right side of~(\ref{Taylor}) can be written as
\begin{align}\label{TermTwo}
Z^{N,n}_f \Delta_f^{N,n} &\mathlarger{\mathlarger{\mathbf{\cdot}}} \nabla F\Big(  \widebar{V}^{N,n}_e   - \Delta_f^{N,n}\Big)\nonumber \\  \,=\,&\frac{1}{b^{\mathbf{\widehat{n}}-N}}Z^{N,n}_f\Big( Y_f^{N,n}\,+\,Y_f^{N,n} B_f^{N,n}   \Big) (\partial_1 F)\Big(  \widebar{V}^{N,n}_e   - \Delta_f^{N,n}\Big)\nonumber \\ & \,+\,\frac{1}{b^{\mathbf{\widehat{n}}-N}}\big(Z^{N,n}_f\big)^2(\partial_2 F)\Big(  \widebar{V}^{N,n}_e   - \Delta_f^{N,n}\Big) \,.
\end{align}
Again by observations (i)-(iii) in part (c), the expectation of the first expression on the right side of~(\ref{TermTwo}) is zero.

As a consequence of the above remarks, taking the expectation of~(\ref{Taylor}) leaves us with
\begin{align}\label{LefTOver}
 \mathbb{E}\Big[ Z_f^{N,n}F\big(  \widebar{V}^{N,n}_e \big)\Big]
=\,&\frac{1}{b^{\mathbf{\widehat{n}}-N}}\mathbb{E}\Big[\big(Z^{N,n}_f\big)^2 \Big]\mathbb{E}\Big[ (\partial_2 F)\Big( \widebar{V}^{N,n}_e  - \Delta_f^{N,n}\Big)\Big]\nonumber \\ &  +\,\underbrace{\frac{1}{2}\mathbb{E}\Big[Z^{N,n}_f \big( \Delta_f^{N,n}  \big)^{\otimes^2}\mathlarger{\mathlarger{\mathbf{\cdot}}} (\mathbf{D}_2 F)\Big(  \widebar{V}^{N,n}_e - \mathbf{r}_f \Delta_f^{N,n} \Big) \Big]}_{\text{(III)}} \,,
\end{align}
where we have used that  $Z^{N,n}_f$ is independent of $\widebar{V}^{N,n}_e  - \Delta_f^{N,n}$ to factor the first expectation on the right.  The right-most expectation on the top line of~(\ref{LefTOver}) is equal to 
\begin{align}
\mathbb{E}\Big[ (\partial_2 F)\Big( \widebar{V}^{N,n}_e  - \Delta_f^{N,n}\Big)\Big]  
=\,
\mathbb{E}\Big[ (\partial_2 F)\big(  \widebar{V}^{N,n}_e  \big)\Big]\,-\,\underbrace{\mathbb{E}\bigg[\int_0^{  1  }   \Delta_f^{N,n} \mathlarger{\mathlarger{\mathbf{\cdot}}}   (\nabla\partial_2 F)\Big( \widebar{V}^{N,n}_e \,-\,r \Delta_f^{N,n}  \Big) dr \bigg]}_{\text{(IV)}}\,.\label{Scanga}
\end{align}

For $\varsigma_{N,n} := \mathbb{E}\big[\big(Z^{N,n}_f\big)^2 \big]^{1/2}$, combining~(\ref{LefTOver}) and~(\ref{Scanga}) with~(\ref{Steiny}) yields the equality
\begin{align}\label{Almost}
 \mathbb{E}\bigg[ \varsigma_{N} (\partial_2 F)\big(   \widebar{V}^{N,n}_e\big)\,-\,&\frac{\widebar{Z}_e^{N,n}}{\varsigma_{N} } F\big(  \widebar{V}^{N,n}_e \big) \bigg]\nonumber \\ \,= \,&\bigg(\varsigma_{N}-\frac{\varsigma_{N,n}^2 }{\varsigma_{N} }\bigg)\mathbb{E}\Big[  (\partial_2 F)\big(   \widebar{V}^{N,n}_e\big)\Big]\,-\, b^{\mathbf{\widehat{n}}-N}\frac{1 }{\varsigma_{N} } \cdot \text{(\textup{III})}\,+\,\frac{\varsigma_{N,n}^2 }{\varsigma_{N} } \cdot \text{(\textup{IV})}  \,.
\end{align}
In the above we have used that the expressions (III) and (IV) do not depend on the choice of $f\in e\cap E_{\mathbf{\widehat{n}}} $ and that there are $b^{2(\mathbf{\widehat{n}}-N)}$ elements in $e\cap E_{\mathbf{\widehat{n}}} $.   The first term on the right side of~(\ref{Almost}) vanishes as $n \rightarrow \infty$ because  $\partial_2 F$ is bounded by $\sqrt{\pi/2}$ and $\varsigma_{N,n}\rightarrow \varsigma_{N} $ by part (ii) of Lemma~\ref{LemBasic}.  We will bound the last two terms on the right side of~(\ref{Almost}) in (e) and (f) below. \vspace{.3cm}

\noindent \textbf{(e) Second term on the right side of~(\ref{Almost}):}  For any $(x,z)\in \R^2$, the  norm of the 2-tensor $\mathbf{D}_2F(x,z) $  is bounded by  $4/\varsigma_{N}$ since its components  are smaller than  $2/\varsigma_{N}$ as a consequence of Corollary~\ref{CorrAltStein}.  Thus we have the second inequality below.
\begin{align*}
 b^{\mathbf{\widehat{n}}-N}\frac{1 }{\varsigma_{N} } \cdot \big|\text{(III)}\big| &\,\leq \, \frac{1}{2\varsigma_{N}}  b^{\mathbf{\widehat{n}}-N}\mathbb{E}\bigg[\Big| Z^{N,n}_f \big( \Delta_f^{N,n}  \big)^{\otimes^2} \mathlarger{\mathlarger{\mathbf{\cdot}}} (\mathbf{D}_2 F)\Big( \widebar{V}^{N,n}_e  - \mathbf{r}_f\Big) \Big| \bigg] \\ &\,\leq \,\frac{2}{\varsigma_{N}^2} b^{\mathbf{\widehat{n}}-N} \mathbb{E}\Big[\big|Z^{N,n}_f\big|\, \big\|  \Delta_f^{N,n}  \big\|^2  \Big]
 \intertext{By definition of $\Delta_f^{N,n}  $, the above is equal to   }
  &\,= \,\frac{2}{\varsigma_{N}^2 b^{\mathbf{\widehat{n}}-N} } \mathbb{E}\Big[\big|Z^{N,n}_f\big|\,\Big( \big|  Y^{N,n}_f \,+\, Y^{N,n}_f  B^{N,n}_f \big|^2\,+\,\big|  Z^{N,n}_f \big|^2 \Big) \Big]\,.
  \intertext{Foiling the products and using that  $B^{N,n}_f$ has mean zero and is independent of  $\big(Y^{N,n}_f,Z^{N,n}_f\big)$, we get }
   &\,= \,\frac{2}{\varsigma_{N}^2 b^{\mathbf{\widehat{n}}-N} } \bigg(\mathbb{E}\Big[\big|Z^{N,n}_f\big|\,\big|  Y^{N,n}_f\big|^2\Big]\Big(1\,+\,\mathbb{E}\Big[\big|B^{N,n}_f\big|^2\Big]\Big)\,+\,\mathbb{E}\Big[ \big|  Z^{N,n}_f \big|^3  \Big]   \bigg)\,.
   \intertext{Applying the Cauchy-Schwarz  inequality to each term above yields that  }
 &\,\leq \,\frac{2}{\varsigma_{N}^2b^{\mathbf{\widehat{n}}-N}}\mathbb{E}\Big[\big| Z^{N,n}_f \big|^2\Big]^{\frac{1}{2}}\bigg( \mathbb{E}\Big[\big|Y^{N,n}_f \big|^4\Big]^{\frac{1}{2}}\Big(1\,+\,\mathbb{E}\Big[\big|B^{N,n}_f \big|^2\Big]    \Big)+ \mathbb{E}\Big[\big|Z^{N,n}_f \big|^4\Big]^{\frac{1}{2}}\bigg)\,. 
\end{align*}
By~(\ref{BBound}), Lemma~\ref{LemBasic}, and Remark~\ref{RemarkVarSigma}, the above is bounded for all $n,N\in \mathbb{N}$ with $n\geq 2\mathbf{n}$ by
\begin{align*}
\frac{2 N^2 }{cb^{\mathbf{\widehat{n}}-N}\log(N+1)}\bigg( \frac{C\log(N+1)}{  N^2} \bigg)^{\frac{1}{2}}  \Bigg(   \frac{\sqrt{C}}{N}\bigg(1\,+ \,  \frac{C^2\log^2 (N+1)}{N } \bigg)\,+\, \frac{\sqrt{C}\log(N+1) }{N^2} \Bigg) \,.
\end{align*}
As $N\rightarrow \infty$ the above is asymptotically proportional to $ \frac{ \log^{-\frac{1}{2} } (N+1)}{N^{\frak{m}\log b } }$  since $\mathbf{\widehat{n}}-N\sim \frak{m}\log N$.

\vspace{.3cm}

\noindent \textbf{(f) Third  term on the right side of~(\ref{Almost}):} To bound the third term on the right side of~(\ref{Almost}), we can use that the vector  $(\nabla\partial_2 F)(x,z)$ has  norm less $\sqrt{2}$ times $2/\varsigma_N$, i.e., the bound for the second-order partial derivatives of $F$, and apply the Cauchy-Schwarz inequality to get
\begin{align*}
\frac{\varsigma_{n,N}^2 }{\varsigma_{N} } \cdot \big|\text{(IV)} \big| &\,:=\,\frac{\varsigma_{n,N}^2}{\varsigma_{N}} \bigg|\mathbb{E}\bigg[\int_0^{  1  }   \Delta_f^{N,n} \mathlarger{\mathlarger{\mathbf{\cdot}}}   (\nabla\partial_2 F)\Big( \widebar{V}^{N,n}_e \,-\,r \Delta_f^{N,n}  \Big) dr \bigg]\bigg|   \,\leq \,2^{\frac{3}{2}}\frac{\varsigma_{n,N}^2}{\varsigma_{N}^2}\mathbb{E}\Big[ \big\|  \Delta_f^{N,n}  \big\|  \Big] \,.\intertext{By Jensen's inequality, the above is smaller than  }
& \,\leq \,2^{\frac{3}{2}}\frac{\varsigma_{n,N}^2}{\varsigma_{N}^2}\mathbb{E}\Big[ \big\|  \Delta_f^{N,n}  \big\|^2  \Big]^{\frac{1}{2}}\,=\,\frac{2^{\frac{3}{2}}}{b^{\mathbf{\widehat{n}}-N}}\frac{\varsigma_{n,N}^2}{\varsigma_{N}^2}\mathbb{E}\Big[ \big(Y^{N,n}_f\,+\,Y^{N,n}_f  B^{N,n}_f    \big)^2\,+\,\big(  Z^{N,n}_f  \big)^2           \Big]^{\frac{1}{2}}\,.
\intertext{Since $Y^{N,n}_f$ and   $B^{N,n}_f$ are independent and      $B^{N,n}_f$ has mean zero, }
 &\,=\,\frac{2^{\frac{3}{2}}}{b^{\mathbf{\widehat{n}}-N}}\frac{\varsigma_{n,N}^2}{\varsigma_{N}^2}\bigg( \mathbb{E}\Big[\big|Y^{N,n}_f \big|^2\Big]\,+\,\mathbb{E}\Big[\big|Y^{N,n}_f\big|^2\Big] \mathbb{E}\Big[\big| B^{N,n}_f \big|^2\Big]  +\mathbb{E}\Big[\big|Z^{N,n}_f \big|^2  \Big]\bigg)^{\frac{1}{2}} \,.
\end{align*}
By~(\ref{BBound}), Lemma~\ref{LemBasic}, and Remark~\ref{RemarkVarSigma}, the above is bounded for all $n,N\in \mathbb{N}$ with $n\geq 2\mathbf{n}$ by
\begin{align*}
\frac{2^{\frac{3}{2}}C}{cb^{\mathbf{\widehat{n}}-N}}\bigg(\frac{C}{N }\,+\,\frac{C^3\log^2 (N+1)  }{ N^{2}} \,+\, \frac{C\log(N+1) }{N^2} \bigg)^{\frac{1}{2}} \,.
\end{align*}
As $N\rightarrow \infty$ the above is asymptotically proportional to $\frac{ 1}{N^{\frak{m}\log b +\frac{1}{2}} } $.

\vspace{.3cm}

\noindent \textbf{(g) Extension to the Wasserstein-$\mathbf{2}$ distance:}  Our results in parts (b)-(f) can be summarized by stating that there is a $\frak{c}>0$ such that for all large $n,N\in \mathbb{N}$ with $n\geq 2\mathbf{n}$
\begin{equation} \label{Xi2}
 \rho_1\big( \widehat{X}^{N,n}_e , \mathbf{\widehat{X}}^{N,n}_e  \big)\,\leq \,\frak{c}\frac{ \log^{-\frac{1}{2} } (N+1)}{N^{\frak{m}\log b } } \,+\, \xi_{N}'(n)   \,,     
 \end{equation}
where $\xi_{N}'(n):=\sqrt{\frac{\pi}{2}}\big|\frac{\varsigma_{N,n}^2 }{\varsigma_{N} }-\varsigma_{N}\big| $.  As mentioned below~(\ref{Almost}), $\xi_{N}'(n)$ vanishes as $n\rightarrow \infty$ for any fixed $N$.  By applying Lemma~\ref{Lemma1to2} with $m=3$, we have that
\begin{align*}
\rho_2\big( \widehat{X}^{N,n}_e , \mathbf{\widehat{X}}^{N,n}_e  \big)\,\leq \,& 2^{\frac{2}{3}}\Big(\rho_1  \big( \widehat{X}^{N,n}_e , \mathbf{\widehat{X}}^{N,n}_e  \big) \Big)^{\frac{1}{3}} \bigg(  \mathbb{E}\Big[\big|   \widehat{X}^{N,n}_e \big|^4\Big]^{\frac{1}{6}}\,+\,\mathbb{E}\Big[\big|    \mathbf{\widehat{X}}^{N,n}_e   \big|^4\Big]^{\frac{1}{6}} \bigg) \,.      
\end{align*}
The limit superior  of the above as $n\rightarrow \infty$ is bounded by a constant multiple of $\frac{  \log^{-\frac{1}{6} } (N+1)  }{ N^{\frak{m}\frac{\log b}{3} +\frac{1}{3}}}   $    by~(\ref{Xi2}) and part (iv) of Lemma~\ref{LemBasic}.
 \end{proof}

\subsection{Proof of Lemma~\ref{LemIII}}\label{SecProofLemIII}

The following lemma is a central limit theorem in which the distance between a normalized sum of i.i.d.\ random variables and a centered normal random variable of the same variance is measured in terms of the Wasserstein-1 distance.    We include a proof using the zero bias transformation of Goldstein and Reinert~\cite{Goldstein2} in Appendix~\ref{AppendixGoldstein}.

\begin{lemma}\label{LemNorm}
Let $X_1,\ldots,X_{n}$ be i.i.d.\ centered random variables with  variance $\sigma^2$ and finite third absolute moment.   Then for $\widebar{X}_n  := \frac{X_1+\cdots +X_n}{\sqrt{n}} $ and $\mathcal{X}\sim \mathcal{N}(0,\sigma^2)$
\begin{align*}
\rho_1\big( \widebar{X}_n, \mathcal{X} \big)\, \leq \,\frac{3}{\sigma^2 \sqrt{n}}\mathbb{E}\big[ |X_1|^3  \big]  \,.  
\end{align*}
\end{lemma}
The next corollary applies Lemma~\ref{Lemma1to2} to the above result. The  proof is at the end of Appendix~\ref{AppendixGoldstein}.
\begin{corollary}\label{CorNorm} Let us take the conditions of Lemma~\ref{LemNorm} and assume in addition that the fourth moment of the random variables is finite.  Then for any $n\in \mathbb{N}$
\begin{align*}
\rho_2\big( \widebar{X}_n, \mathcal{X} \big)\, \leq \,\frac{6 }{\sigma^{\frac{2}{3}}  n^{\frac{1}{6}} }  \mathbb{E}\big[X_1^4\big]^{\frac{5}{12}}  \,.  
\end{align*}
\end{corollary}

\begin{proof}[Proof of Lemma~\ref{LemIII}] For $e\in E_{N}$ the variables $\mathbf{\widehat{X}}_e^{N,n} $ and $\mathbf{\widetilde{X}}_e^{(N)} $
have the form
\begin{align*}
\mathcal{L}^{\mathbf{\widehat{n}}-N}\big\{ Y_f \big\}_{f\in e\cap E_{\mathbf{\widehat{n}}}}\,+\,\sum_{k=1}^{\mathbf{\widehat{n}}-N}\mathcal{L}^{k-1}\mathcal{E}\mathcal{L}^{\mathbf{\widehat{n}}-N-k}\{ Y_f \}_{f\in e\cap E_{\mathbf{\widehat{n}}}} \,+\, \mathbf{Z}_e^{(N)}
\end{align*}
for $Y_f:=Y_f^{N,n}$ and $Y_f:=\mathbf{Y}_f^{(N)}$, respectively, where $\big\{ Y_f^{N,n}\big\}_{f\in e\cap E_{\mathbf{\widehat{n}}}}$  and $\big\{ \mathbf{Y}_f^{(N)}\big\}_{f\in e\cap E_{\mathbf{\widehat{n}}}}$ are defined as in~(\ref{DefY}) and independent of $\mathbf{Z}_e^{(N)}$. In the analysis below, we bound the Wasserstein-$2$ distance between $\mathbf{\widehat{X}}_e^{N,n} $ and $\mathbf{\widetilde{X}}_e^{(N)} $ after choosing i.i.d.\ couplings  $\big(Y_f^{N,n}, \mathbf{Y}_f^{(N)}\big)$ for $f\in e\cap E_{\mathbf{\widehat{n}}} $.\vspace{.25cm}

\noindent \textbf{(a) Using i.i.d.\ couplings to bound the Wasserstein-$2$ distance:}
For each $f\in e\cap E_{\widehat{n}}$, let  $ (Y_f^{N,n},\mathbf{Y}_f^{(N)})$ be a coupling of the variables $Y^{N,n}_f $ and $\mathbf{Y}_f^{(N)}$ such that
\begin{align}\label{ELL2}  \rho_2\big( Y_f^{N,n},\mathbf{Y}_f^{(N)} \big)\, =\, \mathbb{E}\Big[ \big(  Y_f^{N,n}  \,-\, \mathbf{Y}_f^{(N)}  \big)^2  \Big]^{\frac{1}{2}}\,. 
\end{align}
With this coupling, we can bound the Wasserstein-$2$ distance between  $\mathbf{\widehat{X}}_e^{N,n} $ and $\mathbf{\widetilde{X}}_e^{(N)} $ as follows:
\begin{align}\label{Rho2}
\Big(\rho_2\big(\mathbf{\widehat{X}}_e^{N,n}  &,\mathbf{\widetilde{X}}_e^{(N)} \big)\Big)^2 \nonumber \\ \,\leq\,&  \mathbb{E}\Bigg[\bigg| \mathcal{L}^{\mathbf{\widehat{n}}-N}\Big(\big\{ Y^{N,n}_f  \big\}_{f\in e\cap E_{\mathbf{\widehat{n}}}}\,-\,  \big\{ \mathbf{Y}_f^{(N)} \big\}_{f\in e\cap E_{\mathbf{\widehat{n}}}}\Big) \nonumber \\  &\,\,\,\,+\,\sum_{k=1}^{\mathbf{\widehat{n}}-N}\mathcal{L}^{k-1}\Big( \mathcal{E}\mathcal{L}^{\mathbf{\widehat{n}}-N-k}\big\{ Y^{N,n}_f \big\}_{f\in e\cap E_{\mathbf{\widehat{n}}}}\,-\,\mathcal{E}\mathcal{L}^{\mathbf{\widehat{n}}-N-k}\big\{ \mathbf{Y}_f^{(N)}\big\}_{f\in e\cap E_{\mathbf{\widehat{n}}}} \Big) \bigg|^2    \Bigg]\,. \nonumber
\intertext{Since the terms summed above are uncorrelated by Lemma~\ref{LemUnCor} and the operations $\mathcal{L}^{k-1}$ act on i.i.d.\ arrays of mean zero random variables, we have the equality }
\,=\,& \sum_{f\in e\cap E_{\mathbf{\widehat{n}}}}\frac{1}{b^{2(\mathbf{\widehat{n}}-N)}} \mathbb{E}\Big[\big|  Y^{N,n}_f  \,-\,   \mathbf{Y}_f^{(N)} \big|^2 \Big] \nonumber  \\ &\,+\,\sum_{k=1}^{\mathbf{\widehat{n}}-N}\frac{1}{b^{2(k-1)}}\sum_{\mathbf{e}\in e\cap E_{N+k-1}  }  \mathbb{E}\bigg[\Big| \mathcal{E}\big\{ \widetilde{Y}^{N,n}_{\mathbf{e}\times (i,j)} \big\}_{1\leq  i,j \leq b}\,-\,\mathcal{E}\big\{ \mathbf{\widetilde{Y}}^{(N)}_{\mathbf{e}\times (i,j)} \big\}_{1\leq  i,j \leq b}\Big|^2   \bigg]\,,
\end{align}
where for  $\mathbf{e}\in e\cap E_{N+k-1}$ the arrays within the expectations above are defined as
$$ \big\{ \widetilde{Y}^{N,n}_{\mathbf{e}\times (i,j)} \big\}_{1\leq  i,j \leq b}\,:=\, \mathcal{L}^{\mathbf{\widehat{n}}-N-k}\big\{ Y^{N,n}_f \big\}_{f\in \mathbf{e}\cap E_{\mathbf{\widehat{n}}}} \hspace{.4cm}\text{and}\hspace{.4cm}  \big\{ \mathbf{\widetilde{Y}}^{(N)}_{\mathbf{e}\times (i,j)} \big\}_{1\leq  i,j \leq b} \, :=\, \mathcal{L}^{\mathbf{\widehat{n}}-N-k}\big\{ \mathbf{Y}_f^{(N)} \big\}_{f\in \mathbf{e}\cap E_{\mathbf{\widehat{n}}}} \,. \vspace{.13cm} $$

\noindent \textbf{(b) Bounding the inner summand  on the second line of~(\ref{Rho2}):} Recall from (i) of  Lemma~\ref{LemBasic} and~(\ref{DefY}), respectively,  that the variables $Y^{N,n}_f$ and  $\mathbf{Y}^{(N)}_{f} $ have variances $\textup{Var}\big(Y^{N,n}_f\big)=\sigma_{\mathbf{n},n}^2$ and $\textup{Var}\big(\mathbf{Y}^{N,n}_f\big)=R(r-\mathbf{n})$.  Consequently, elements in the above arrays have variances $
\textup{Var}\big( \widetilde{Y}^{N,n}_\mathbf{f}  \big) = \sigma_{\mathbf{n},n}^2$ and $ \textup{Var}\big( \mathbf{\widetilde{Y}}^{(N)}_\mathbf{f}   \big)\,=\,R(r-\mathbf{n})$ since  $\mathcal{L}$ preserves the variance of the array variables.  For any   $1\leq k \leq \mathbf{\widehat{n}}-N$, we can write the summand in~(\ref{Rho2}) in the form   
\begin{align}
\mathbb{E}\bigg[\Big| \mathcal{E}\big\{ \widetilde{Y}^{N,n}_{\mathbf{e}\times (i,j)} \big\}_{1\leq  i,j \leq b}\,-\,\mathcal{E}\big\{ \mathbf{\widetilde{Y}}^{(N)}_{\mathbf{e}\times (i,j)} \big\}_{1\leq  i,j \leq b} \Big|^2    \bigg]\,=\,&\sum_{i=1}^{b}\frac{1}{b^2} \sum_{\substack{A\subset \{1,\ldots,b\} \nonumber  \\  |A|\geq 2  }}  \mathbb{E}\Bigg[\bigg|\prod_{j\in A  } \widetilde{Y}^{N,n}_{\mathbf{e}\times (i,j)} \,-\,\prod_{ j\in A  } \mathbf{\widetilde{Y}}^{(N)}_{\mathbf{e}\times (i,j)}  \bigg|^2    \Bigg]
\nonumber\,\nonumber 
\intertext{because the operation $\mathcal{E}=\mathcal{Q}-\mathcal{L}$ returns  $\frac{1}{b}\sum_i \big(\prod_j(1+a_{i,j})   -1-\sum_{j}a_{i,j}\big)$  when it acts on an array $\{a_{i,j}\}_{1\leq i,j\leq b}$. By writing $\widetilde{Y}^{N,n}_{\mathbf{f}}= \mathbf{\widetilde{Y}}^{(N)}_{\mathbf{f}} +\big(  \widetilde{Y}^{N,n}_{\mathbf{f}}- \mathbf{\widetilde{Y}}^{(N)}_{\mathbf{f}} \big) $ for each $\mathbf{f}=\mathbf{e}\times (i,j)$ in the products above and foiling, we get  }
  \,= \,&\mathbb{E}\Big[\big| \widetilde{Y}^{N,n}_\mathbf{f}\, -\,\mathbf{\widetilde{Y}}_\mathbf{f} ^{(N)} \big|^2\Big] U\big(  \sigma_{\mathbf{n},n}^2,\, R(r-\mathbf{n})\big) \,,\nonumber
\intertext{  
where  $U(y_1,y_2)$ is a  degree-$b$ polynomial with nonnegative coefficients and no constant term. 
The equality $\widetilde{Y}^{N,n}_\mathbf{f}-\mathbf{\widetilde{Y}}_\mathbf{f} ^{(N)}=\mathcal{L}^{\mathbf{\widehat{n}}-N-k   }\big\{Y^{N,n}_f  -\mathbf{Y}_f^{(N)} \big\}_{f\in \mathbf{f}\cap E_{\mathbf{\widehat{n}}  }} $ implies that the $L^2$ distance between $\widetilde{Y}^{N,n}_\mathbf{f} $ and $\mathbf{\widetilde{Y}}_\mathbf{f} ^{(N)}$ is equal to the $L^2$ distance between $ Y^{N,n}_f $ and $\mathbf{Y}_f^{(N)}$ for any given $f\in E_{\mathbf{\widehat{n}}  }$, so by~(\ref{ELL2}) the above can be written as   }
   \,=  &\,\Big( \rho_2\big( Y^{N,n}_f ,\mathbf{Y}_f^{(N)} \big)\Big)^2 U\big(  \sigma_{\mathbf{n},n}^2,\, R(r-\mathbf{n})\big) \,. \nonumber 
\intertext{By (i) of Lemma~\ref{LemBasic},  $\sigma_{\mathbf{n},n}^2$  is bounded by a multiple of $\frac{1}{N}$ for all $n,N$ with $n\geq \mathbf{n}$. Similarly,  $R(r-\mathbf{n})$ is bounded by a multiple of $\frac{1}{N}$ for all $N$ as a consequence of $N\sim \mathbf{n}$ and (II) of Lemma~\ref{LemVar}.  Thus, since the polynomial $U$  has no constant term, there is a $\mathbf{c}>0$ such that for all $N$ and $n$ with $n\geq\mathbf{n}$}
   \,\leq  &\,\frac{\mathbf{c}}{N}\Big( \rho_2\big( Y^{N,n}_f ,\mathbf{Y}_f^{(N)} \big)\Big)^2  \,.\label{Rho4}
\end{align}

\noindent \textbf{(c) Going back to (\ref{Rho2}):}  The first term on the right side of~(\ref{Rho2}) is equal to $\big(\rho_2\big( Y^{N,n}_f  ,\mathbf{Y}_f^{(N)}\big)\big)^2$ for any representative $f\in e\cap E_{\mathbf{\widehat{n}}} $ by definition of how the couplings in~(\ref{ELL2}) are defined and since $\big|e\cap E_{\mathbf{\widehat{n}}}|=b^{2(\mathbf{\widehat{n}}-N)  }$.  Similarly, as a consequence of~(\ref{Rho4}), the second term on the right side of~(\ref{Rho2}) is bounded from above by $\mathbf{c}\frac{\mathbf{\widehat{n}}-N}{N} \big(\rho_2\big( Y^{N,n}_f  ,\mathbf{Y}_f^{(N)}\big)\big)^2$.  Thus for all $n,N\in \mathbb{N}$ with $n\geq \mathbf{n}$
\begin{align}
\Big(\rho_2\big(\mathbf{\widehat{X}}_e^{N,n}  ,\mathbf{\widetilde{X}}_e^{(N)} \big)\Big)^2 
 \,\leq \, \Big(\rho_2\big( Y^{N,n}_f  ,\mathbf{Y}_f^{(N)}\big)\Big)^2\bigg(1 \,+\,\mathbf{c}\frac{\mathbf{\widehat{n}}-N}{N }\bigg)\, \leq \,\mathbf{C} \Big(\rho_2\big( Y^{N,n}_f  ,\mathbf{Y}_f^{(N)}\big)\Big)^2 \,,\label{2Mult}
\end{align}
where the  second inequality holds for some $\mathbf{C}>0$ since $\mathbf{\widehat{n}}:=N+\lfloor \frak{m}\log N\rfloor$.  Thus we have shown that $\rho_2\big(\mathbf{\widehat{X}}_e^{N,n}  ,\mathbf{\widetilde{X}}_e^{(N)} \big)$ is bounded by a constant multiple  of  $\rho_2\big( Y^{N,n}_f  ,\mathbf{Y}_f^{(N)}\big)$.\vspace{.25cm}

\noindent \textbf{(d) Bounding the right side of~(\ref{2Mult}):} Next we focus on bounding $\rho_2\big( Y^{N,n}_f   ,\mathbf{Y}_f^{(N)}\big)$.     Since $Y^{N,n}_f$ has variance $\sigma_{\mathbf{n},n }^2$ and $\mathbf{Y}_f^{(N)}$ has variance $R(r-\mathbf{n})$, it will be convenient to use the triangle inequality to get
\begin{align}\label{Triangle}
\rho_2\big( Y^{N,n}_f   ,\mathbf{Y}_f^{(N)}\big)\,\leq \,  \rho_2\bigg( Y^{N,n}_f   ,\,\frac{\sigma_{\mathbf{n},n } }{\sqrt{R(r-\mathbf{n})}}  \mathbf{Y}_f^{(N)}\bigg)\,+\,\rho_2\bigg(\frac{\sigma_{\mathbf{n},n } }{\sqrt{R(r-\mathbf{n})}}  \mathbf{Y}_f^{(N)},\,   \mathbf{Y}_f^{(N)} \bigg)\,.
\end{align}
Using that $\textup{Var}\big( \mathbf{Y}_f^{(N)}\big)=R(r-\mathbf{n})  $,   the first term on the right side of~(\ref{Triangle}) can  simply be bounded by
\begin{align}\label{RhoBound}
\rho_2\bigg(\frac{\sigma_{\mathbf{n},n } }{\sqrt{R(r-\mathbf{n})}}  \mathbf{Y}_f^{(N)},\,   \mathbf{Y}_f^{(N)} \bigg)\,\leq\,\big| \sigma_{\mathbf{n},n }\,-\,\sqrt{R(r-\mathbf{n})}  \big|\,=:\,\varsigma_{N}''(n)  \,.
\end{align}
By definition, $ Y^{N,n}_f $ is a sum of the i.i.d.\ random variables $\frac{1}{b^{\mathbf{n}-\mathbf{\widehat{n}}  }} X^{(\mathbf{n},n)}_g$  over $g\in f\cap E_{\mathbf{n}}$, which contains $b^{2(\mathbf{n}-\mathbf{\widehat{n}} ) } $ elements.     Hence, by Corollary~\ref{CorNorm} we have the  inequality below for the first term on the right side of~(\ref{Triangle}).
\begin{align}
\rho_2\bigg( Y^{N,n}_f   ,\, \frac{\sigma_{\mathbf{n},n } }{\sqrt{R(r-\mathbf{n})}}  \mathbf{Y}_f^{(N)}\bigg) \,=\,&\rho_2\Bigg( \frac{1}{b^{\mathbf{n}-\mathbf{\widehat{n}}  } } \sum_{ g\in f\cap E_{\mathbf{n}}  }X^{(\mathbf{n},n)}_g   ,\,\frac{\sigma_{\mathbf{n},n } }{\sqrt{R(r-\mathbf{n})}}  \mathbf{Y}_f^{(N)}\Bigg) \nonumber \\
 \,\leq \,& \frac{6}{b^{ \frac{1}{3}( \mathbf{n}- \mathbf{\widehat{n}) } }}\frac{\mathbb{E}\big[ \big| X^{(\mathbf{n},n)}_g \big|^4  \big]^{\frac{5}{12}}}{ \sigma_{\mathbf{n},n}^{\frac{2}{3}}  } \nonumber 
\intertext{Part (i) of Lemma~\ref{LemBasic} implies that  $\sigma^2_{\mathbf{n},n}$ is bounded from below by a constant multiple of $\frac{1}{N}$, so there is a $c>0$ such that} 
 \,\leq \,&\frac{N^{\frac{1}{3}}}{c b^{ \frac{1}{3}( \mathbf{n}- \mathbf{\widehat{n}) } }}\mathbb{E}\big[ \big| X^{(\mathbf{n},n)}_g \big|^4  \big]^{\frac{5}{12}}\,\leq \,\frac{C}{N^{\frac{1}{2}} b^{ \frac{1}{3}( \mathbf{n}- \mathbf{\widehat{n}) }}}  \,.\label{RhoY}
\end{align}
The second inequality holds for some $C>0$ since $\mathbb{E}\big[ \big| X^{(\mathbf{n},n)}_g \big|^4  \big]$ is bounded from above by a constant multiple of $\frac{1}{N^2}$ for all $n,N\in \mathbb{N}$ with $n\geq 2\mathbf{n}$ by (iv) of Lemma~\ref{LemBasic}.  The last term  in~(\ref{RhoY}) is asymptotically proportional to $ N^{-\frak{m}\frac{\log b}{3} -\frac{1}{2}  } $ as $N\rightarrow \infty$ since $ \mathbf{n}- \mathbf{\widehat{n}}\approx \frak{m}\log N  $. \vspace{.25cm}

\noindent \textbf{(e) Conclusion:}
The inequalities~(\ref{2Mult})-(\ref{RhoY})  show that there is a $\frak{c}>0$ such that for all $N,n\in \mathbb{N}$ with  $n\geq 2\mathbf{n}$
\begin{align}\label{XiThree}
\rho_2\big(\mathbf{\widehat{X}}_e^{N,n}  ,\mathbf{\widetilde{X}}_e^{(N)} \big)\,\leq \frac{\frak{c}}{N^{\frak{m}\frac{\log b}{3} +\frac{1}{2}  }  }\,+\,\xi_{N}''(n) \,,
\end{align}
where $\varsigma_{N}''(n) :=\frak{c}\big| \sigma_{\mathbf{n},n }-\sqrt{R(r-\mathbf{n})}  \big|  $.  The term  $\varsigma_{N}''(n) $ vanishes as $n\rightarrow \infty$ since $\sigma_{\mathbf{n},n } ^2\rightarrow R(r-\mathbf{n})$ by (III) of Lemma~\ref{LemmaMom} with $m=2$, and hence the proof is complete.
\end{proof}

\section{Miscellaneous proofs from Sections~\ref{SecMainThmOutline}, \ref{SecOutlineMain}, \& \ref{SecCentralLimit}}\label{SecMiscProofs}

\subsection{Proofs from Section~\ref{SecMainThmOutline}}\label{SecMiscZero}

\begin{proof}[Proof of Proposition~\ref{PropPartition}] We will prove the identity~(\ref{WRemark}) using induction starting from $n=0$.  When $n=0$, the set $E_n$ contains a single element $h$, and the identity follows immediately from the definitions:  
$$W^{\omega}_{0}(  \beta )\,=\,1\,+\,\big(W^{\omega}_{0}(  \beta )\,-\,1\big)\,= \, 1\,+\,\bigg(\frac{ e^{ \beta\omega_{h} }  }{\mathbb{E}\big[   e^{ \beta\omega_{h} } \big]   }\,-\,1\bigg)   \,=\,1\,+\,\mathcal{Q}^{0}\big\{ X_{h}^{(0)} \big\}_{h\in E_0   }\,.  $$
Suppose that the identity~(\ref{WRemark}) holds for some $n\in \mathbb{N}_0$. The hierarchical nesting that defines the sequence $\{D_n\}_{n\in \mathbb{N}_0}$ of diamond graphs implies that  there is a one-to-one correspondence between the set of generation-$(n+1)$ paths, $p\in \Gamma_{n+1}$, crossing $D_{n+1}$ and the set  of $(b+1)$-tuples $(i,p_1,\ldots, p_b)$ with  $i\in \{1,\ldots, b\}$ and $p_j\in \Gamma_n$.  Within this identification, $i\in \{1,\ldots, b\}$ labels the branch of  $D_{n+1}$ that $p\in \Gamma_{n+1}$ traces over  and $p_j\in \Gamma_n$ for $j\in\{1,\ldots, b\}$ is the trajectory of $p$ through the $j^{th}$ copy of $D_n$ along the branch.  In particular, it follows that $|\Gamma_{n+1}|=b|\Gamma_n|^b$. Using this bijection, we can rewrite the partition function $W^{\omega}_{n+1}(  \beta )$ as
\begin{align*}
W^{\omega}_{n+1}(  \beta )\,:=\,&\frac{1}{|\Gamma_{n+1}|}\sum_{p\in \Gamma_{n+1}}\prod_{\ell=1  }^{b^{n+1}}\frac{ e^{ \beta\omega_{p(\ell)} }  }{\mathbb{E}\big[   e^{ \beta\omega_{p(\ell)} } \big]   }\nonumber  \\ \,=\,&\frac{1}{b|\Gamma_{n}|^b}\sum_{i=1}^b\sum_{(p_1,\ldots, p_b)\in  \Gamma_{n}^b}\prod_{j=1 }^b\prod_{\ell_{j}=1  }^{b^{n}}\frac{ e^{ \beta\omega_{(i,j)\times p_j(\ell_j)} }  }{\mathbb{E}\big[   e^{ \beta\omega_{(i,j)\times p_j(\ell_j)} } \big]   }\,,
\intertext{where   $(i,j){\times} h$  denotes the element of $E_{n+1}$ corresponding to the element  $h\in E_n$ within the $(i,j)$-labeled subcopy of $E_n$ in $E_{n+1}$.  The above expression factors yielding}   
\,=\,&\frac{1}{b}\sum_{i=1}^b\prod_{j=1 }^b \Bigg(\frac{1}{|\Gamma_{n}|}\prod_{\ell_{j}=1  }^{b^{n}}\frac{ e^{ \beta\omega_{(i,j)\times p_j(\ell_j)} }  }{\mathbb{E}\big[   e^{ \beta\omega_{(i,j)\times p_j(\ell_j)} } \big]   }\Bigg) \,.
\intertext{The quantity above in  brackets is a generation-$n$ partition function, so by our induction assumption}
 \,=\,&\frac{1}{b}\sum_{i=1}^b\prod_{j=1 }^b\Big(1+ \mathcal{Q}^n\big\{ X_{h }\big\}_{h\in (i,j)\cap E_{n+1}   }\Big)\,=:\,1\,+\,\mathcal{Q}^{n+1}\big\{ X_{h }\big\}_{h\in E_{n+1}   }\,.
\end{align*}
Hence the identity~(\ref{WRemark}) holds for all $n\in \mathbb{N}_0$ by induction.\end{proof}

\begin{proof}[Proof of Proposition~\ref{PropCont}]  We will prove that the law of $\mathbf{X}_{r}$ is a locally $\frac{1}{2}$-H\"older continuous function of $r\in \R$ with respect to the Wasserstein-$2$ metric by showing that for all $r$ and $t\geq 0$
\begin{align}\label{RhoLips}
\rho_2\big( \mathbf{X}_{r}, \mathbf{X}_{r+t}   \big)\,\leq \, \sqrt{R(r+t)\,-\,R(r)}     \,,
\end{align}
where the function $R:\R\rightarrow (0,\infty)$ has a continuous---and thus locally bounded---derivative by  Lemma~\ref{LemVar}.  For any $n\in \mathbb{N}$ we can construct $\mathbf{X}_{r}$ as $\mathbf{X}_{r}=\mathcal{Q}^n\{\mathbf{X}_{h}^{(n)}\}_{h\in E_n}$, where the array of random variables $\{\mathbf{X}_{h}^{(n)}\}_{h\in E_n}$  is defined as in Theorem~\ref{ThmExist} for parameter $r\in \R$.  Let  $\{ \mathbf{B}^{h}_t\}_{h\in E_n}$ be an array of independent  normal random variables with mean $0$ and variance $t>0$ that is independent of $\{\mathbf{X}_{h}^{(n)}\}_{h\in E_n}$.
 Define $X_{n,r,t}^{\mathbf{B}}:=\mathcal{Q}^n\big\{X_{h}^{(n)}(r,t) \big\}_{h\in E_n }$ for $X_{h}^{(n)}(r,t):= \big(1+\mathbf{X}_{h}^{(n)}\big)\textup{exp}\big\{\frac{\kappa}{n}\mathbf{B}^h_t - \frac{\kappa^2}{2n^2}t \big\}-1 $, i.e., as in Example~\ref{Example}.  By the triangle inequality, we can bound the Wasserstein-$2$ distance between $\mathbf{X}_{r}$ and $ \mathbf{X}_{r+t} $ by
\begin{align}\label{Tri}
\rho_2\big( \mathbf{X}_{r}, \mathbf{X}_{r+t}   \big)\,\leq\, \rho_2\big( \mathbf{X}_{r}, X^{\mathbf{B}}_{n,r,t}   \big)+\rho_2\big( X^{\mathbf{B}}_{n,r,t}, \mathbf{X}_{r+t}   \big)\,\leq\, \sqrt{\mathbb{E}\Big[\big( \mathbf{X}_{r}- X^{\mathbf{B}}_{n,r,t}   \big)^2\Big]}+\rho_2\big( X^{\mathbf{B}}_{n,r,t}, \mathbf{X}_{r+t}   \big)\,.
\end{align}
The second term on the right side of~(\ref{Tri}) converges to zero as $n\rightarrow \infty$ by the discussion in Example~\ref{Example}.  The random variables $  X^{\mathbf{B}}_{n,r,t}-\mathbf{X}_{r} $ and $\mathbf{X}_{r}$ are uncorrelated since  $\mathbf{X}_{r}=\mathcal{Q}^n\{\mathbf{X}_{h}^{(n)}\}_{h\in E_n}$ is the conditional expectation of $ X^{\mathbf{B}}_{n,r,t}$ given $\{\mathbf{X}_{h}^{(n)}\}_{h\in E_n}$.  Thus, since $\mathbf{X}_{r}$ and $ X^{\mathbf{B}}_{n,r,t} $ have mean zero,
\begin{align}
\mathbb{E}\Big[\big( \mathbf{X}_{r}- X^{\mathbf{B}}_{n,r,t}   \big)^2\Big]\,=\,\textup{Var}\big( X^{\mathbf{B}}_{n,r,t}   \big)\,-\,\textup{Var}\big(\mathbf{X}_{r}\big)\,=\, \textup{Var}\big( X^{\mathbf{B}}_{n,r,t}   \big)\,-\,R(r)\,.
\end{align}
To see that $\textup{Var}\big( X^{\mathbf{B}}_{n,r,t}   \big)$ converges to $R(r+t)$ as $n\rightarrow \infty$, notice that 
\begin{align*}
\textup{Var}\big( X_{n,r,t}^{\mathbf{B}}\big)= M^n\Big(\textup{Var}\big( X_{h}^{(n)}(r,t) \big)   \Big)    = M^n\bigg(\kappa^2\Big(\frac{ 1 }{n}+\frac{\eta\log n}{n^2}+\frac{r+t  }{n^2}\Big) +\mathit{o}\Big(\frac{1}{n^2}\Big)\bigg)=R(r+t)+\mathit{o}(1)  \,, 
\end{align*}
where the first and third equalities hold by part (i) of Remark~\ref{RemarkArrayVar} and  Lemma~\ref{LemVar}, respectively. The second equality above follows from~(\ref{Examp2}). Therefore we have established the inequality~(\ref{RhoLips}). 
\end{proof}

\subsection{Proofs from Section~\ref{SecOutlineMain}}\label{SecMiscFirst}

\begin{proof}[Proof of Corollary~\ref{CorollaryTriv}]
 The random variables $ X^{(N,n)}_e -\widehat{X}^{N,n}_e $ and $\widehat{X}^{N,n}_e $ are uncorrelated as a consequence of Lemma~\ref{LemUnCor}, and thus
\begin{align}\label{Tbit}
\mathbb{E}\Big[\big(\widehat{X}^{N,n}_e \big)^2  \Big]  \,\leq \,\mathbb{E}\Big[\big( X^{(N,n)}_e \big)^2  \Big] \,\,\,\stackrel{n\rightarrow \infty}{\longrightarrow} \,\,\,R(r-N)\,=\,R(-N)+\frac{\kappa^2 r}{N^2}+\mathit{o}\Big( \frac{1}{N^2} \Big)\,, 
\end{align}
where the convergence holds by (III) of Lemma~\ref{LemmaMom} with $m=2$. The equality holds for $N\gg 1$ by the asymptotics for $R(r)$ as $r\rightarrow -\infty$ in (II) of Lemma~\ref{LemVar}.  If $s>r$, then the right side above is smaller than $R(-N)+\frac{\kappa^2s}{N^2}$ for $N\gg 1$.  Thus we have verified the desired condition in the case $U_{e}^{(N)}:=\widehat{X}^{N,n}_e$ for any $s\in (r,\infty)$ and large enough $N, n\in \mathbb{N}$.   

Next we extend our result to the case $U_{e}^{(N)}:=\mathbf{\widehat{X}}^{N,n}_e$. By Lemma~\ref{LemII} and Remark~\ref{RemarkMinAssump}, there are couplings between the random variables $\widehat{X}^{N,n}_e$ and $\mathbf{\widehat{X}}^{N,n}_e$ such that the limit superior as $n\rightarrow \infty$ of $\mathbb{E}\big[ \big( \widehat{X}^{N,n}_e-\mathbf{\widehat{X}}^{N,n}_e   \big)^2\big]$  is $\mathit{o}\big(\frac{1}{N^4} \big)$   for  $N\gg 1$. By foiling and applying
Cauchy-Schwarz, we get
  \begin{align*} \mathbb{E}\Big[\big(\mathbf{\widehat{X}}^{N,n}_e \big)^2  \Big]&\,=\,\mathbb{E}\Big[\big(\widehat{X}^{N,n}_e \big)^2  \Big]\,+\,2\mathbb{E}\Big[ \widehat{X}^{N,n}_e \big(\mathbf{\widehat{X}}^{N,n}_e - \widehat{X}^{N,n}_e  \big)\Big]\,+\,\mathbb{E}\Big[ \big( \mathbf{\widehat{X}}^{N,n}_e -\widehat{X}^{N,n}_e  \big)^2\Big] \\ &\,\leq \,\mathbb{E}\Big[\big(\widehat{X}^{N,n}_e \big)^2  \Big]\,+\,2\mathbb{E}\Big[ \big(\widehat{X}^{N,n}_e\big)^2\Big]^{\frac{1}{2}} \mathbb{E}\Big[\big(\mathbf{\widehat{X}}^{N,n}_e - \widehat{X}^{N,n}_e  \big)^2\Big]^{\frac{1}{2}}\,+\,\mathbb{E}\Big[ \big( \mathbf{\widehat{X}}^{N,n}_e -\widehat{X}^{N,n}_e  \big)^2\Big]   \, .
  \end{align*}
Since $\limsup_{n\rightarrow \infty} \mathbb{E}\big[(\widehat{X}^{N,n}_e )^2  \big]\leq R(r-N)  $  by~(\ref{Tbit}) and $R(r-N)$ is $\mathit{O}\big(\frac{1}{N}\big)$ for $N\gg 1$ as a consequence of  (II) of Lemma~\ref{LemVar}, the limit superior of the middle term above as $n\rightarrow \infty$  is $\mathit{o}\big( \frac{1}{N^{5/2}} \big) $  with large $N$.  Thus $\limsup_{n\rightarrow \infty} \mathbb{E}\big[(\mathbf{\widehat{X}}^{N,n}_e )^2  \big]$ is bounded by $\limsup_{n\rightarrow \infty}\mathbb{E}\big[(\widehat{X}^{N,n}_e )^2  \big]+\mathit{o}\big( \frac{1}{N^{5/2}} \big) $, which is smaller than $R(-N)+\frac{\kappa^2s}{N^2}$ when $N\gg 1$ for any choice of $s\in (r,\infty)$. Hence we have extended our result to the case $U_{e}^{(N)}:= \mathbf{\widehat{X}}^{N,n}_e $, and the same reasoning applies to $U_{e}^{(N)}:=\mathbf{\widetilde{X}}^{(N)}_e$.
\end{proof}

\subsection{Proofs from Section~\ref{SecCentralLimit}}\label{SecProofCLT} 

\begin{proof}[Proof of Proposition~\ref{PropAltStein}] The bounds $\sup_{y,z\in \R}|\partial_z F(y,z)|\leq 1$ and $\sup_{y,z\in \R}|\partial_z^2 F(y,z)|\leq 2$ are equivalent to~(\ref{Uniform}), so we can focus on the partial derivatives $\partial_y$, $\partial_y^2$, and  $\partial_y\partial_z $.  Define $\displaystyle \phi_-(t):=\int_{-\infty}^t\frac{1}{\sqrt{2\pi}}e^{-\frac{r^2}{2}}dr$ and $\phi_+(t):=1-\phi_-(t)$.    We can rewrite $H$ in terms of $H'$ as
\begin{align}\label{Ach}
H(z)\,-\,\int_{\R}H(r)\frac{e^{-\frac{r^2}{2}}}{\sqrt{2\pi}}dr \,=\,\int_{-\infty}^{z}H'(t)\phi_-(t) dt\,-\, \int_{z}^{\infty}H'(t)\phi_+(t) dt\, .
\end{align}
Moreover, we can rewrite $F$ in the form
\begin{align*}
  F(y,z)\,=\,& \frac{1}{2}e^{\frac{z^2}{2}}\int_{-\infty}^{z} \bigg(H(y+t)\,-\,\int_{\R}H(y+r)\frac{e^{-\frac{r^2}{2}}}{\sqrt{2\pi}}dr\bigg)e^{-\frac{t^2}{2}}dt\\  &\,-\,  \frac{1}{2}e^{\frac{z^2}{2}}\int_{z}^{\infty}  \bigg(H(y+t)\,-\,\int_{\R}H(y+r)\frac{e^{-\frac{r^2}{2}}}{\sqrt{2\pi}}dr\bigg)e^{-\frac{t^2}{2}}dt\,,
  \intertext{and using the identity~(\ref{Ach}) we have}
   \,=\,& e^{\frac{z^2}{2}}\int_{-\infty}^{z} \bigg(\int_{-\infty}^{t}H'(y+r)\phi_-(r) dr\,-\, \int_{t}^{\infty}H'(y+r)\phi_+(r) dr   \bigg)e^{-\frac{t^2}{2}}dt
\\   
   \,&\,-\, e^{\frac{z^2}{2}}\int_{z}^{\infty} \bigg(\int_{-\infty}^{t}H'(y+r)\phi_-(r) dr\,-\, \int_{t}^{\infty}H'(y+r)\phi_+(r) dr   \bigg)e^{-\frac{t^2}{2}}dt   \,.
  \intertext{Swapping the order of integration yields}
 \,=\,& \int_{\R}G(z,r)H'(y+r)dr  \,=\, \int_{\R}G(z,r-y)H'(r)dr
 \,,
  \end{align*}
where $G:\R^2\rightarrow \R$ is the kernel
$$G(z,r)\,:=\,\begin{cases} -\sqrt{2\pi}e^{\frac{z^2}{2}}\phi_{-}(z)  \phi_+(r)&  z<r \, , \\ -\sqrt{2\pi}e^{\frac{z^2}{2}} \phi_{+}(z)  \phi_-(r) &  z\geq r \, . \end{cases}       $$
The results will follow by bounding  $  \sup_{z\in \R}\int_{\R}\big|(\mathbf{d}G)(z,r) \big|dr   $
for the derivatives $\mathbf{d}\in \big\{\partial_r,\partial_r^2,\partial_{z}\partial_r\big\}$.

The first partial derivative with respect to $r$ has the form
\begin{align*}
\partial_{ r} G(z,r)\,=\,\begin{cases} e^{\frac{z^2}{2}}\phi_{-}(z)  e^{-\frac{r^2}{2}}&  z<r  \,, \\ -e^{\frac{z^2}{2}} \phi_{+}(z)  e^{-\frac{r^2}{2}} &  z\geq r \,. \end{cases} 
\end{align*}
For any $z\in \R$, the equality $\int_{\R}\big|\partial_{ r} G(z,r)\big|dr\,=\,2 \sqrt{2\pi} e^{\frac{z^2}{2}} \phi_{-}(z)  \phi_{+}(z) $ holds, and the right side attains its maximum value, $\sqrt{\pi/2}  $, when $z=0$.

The second-order partial derivatives involving $r$ have the forms $\partial_r^2 G(z,r)= -\delta(z-r)+A_{1}(z,r)  $ and $\partial_z\partial_r G(z,r)= -\delta(z-r)+A_{2}(z,r)  $, where
\begin{align*}
&A_1(z,r)\,:=\,\begin{cases}- e^{\frac{z^2}{2}}\phi_{-}(z)  re^{-\frac{r^2}{2}}&  z<r \,,  \\ e^{\frac{z^2}{2}} \phi_{+}(z) r e^{-\frac{r^2}{2}} &  z\geq r \,, \end{cases} \hspace{.5cm}
A_2(z,r)\,:=\,\begin{cases} \big( \sqrt{2\pi} z e^{\frac{z^2}{2}}\phi_{-}(z) +1\big)\frac{ e^{-\frac{r^2}{2}}}{\sqrt{2\pi}}&  z<r \,,  \\ -\big(\sqrt{2\pi}z e^{\frac{z^2}{2}} \phi_{+}(z) -1\big)  \frac{ e^{-\frac{r^2}{2}}}{\sqrt{2\pi}} &  z\geq r \,. \end{cases}  
\end{align*}
Notice that $1+\sqrt{2\pi} ze^{\frac{z^2}{2}}\phi_{-}(z)$ and $1 - \sqrt{2\pi}ze^{\frac{z^2}{2}}\phi_{+}(z)$ are nonnegative for all $z\in \R$, and thus we simply have 
\begin{align*}
\int_{\R}\big|A_1(z,r)\big|dr\,=\,& \phi_{-}(z) +  \phi_{+}(z)    \,=\,1\,,\intertext{and}
\int_{\R}\big|A_2(z,r)\big|dr\,=\,& \Big(1+ \sqrt{2\pi} z e^{\frac{z^2}{2}}\phi_{-}(z) \Big)\phi_{+}(z)  \,+\,\Big(1- \sqrt{2\pi} z e^{\frac{z^2}{2}}\phi_{+}(z) \Big)\phi_{-}(z) \,=\,1   \,.
\end{align*}
Therefore $  \sup_{z\in \R}\int_{\R}\big|(\mathbf{d}G)(z,r) \big|dr  \leq 2 $ for $\mathbf{d}=\partial_z\partial_r$ and $\mathbf{d}=\partial_r^2$.
\end{proof}

\vspace{.4cm}

  The following proposition gives uniform bounds for the second and fourth moments of random variables from a minimally regular sequence of $\mathcal{Q}$-pyramidic arrays.  We prove Proposition~\ref{PropUnif} in  Section~\ref{SecLemUnif} using techniques and an inequality  from~\cite{Clark1}.
\begin{proposition}\label{PropUnif} Let $\big(  \big\{  X^{(*,n)}_a \big\}_{a\in E_{*} }  \big)_{n\in \mathbb{N}}$ be a minimally regular sequence of $\mathcal{Q}$-pyramidic arrays of random variables.  
\begin{enumerate}[(i)]
\item The variances of the random variables $X^{(k,n)}_a$ are bounded from above and below by positive multiples of $\frac{1}{k+1}$  for all $n\in \mathbb{N}$ and $k\in \{0,\ldots, n\}$.

\item The fourth moments of the random variables  $X^{(k,n)}_a$ are bounded from above by a  multiple of $\frac{1}{(k+1)^2}$  for all $n\in \mathbb{N}$ and $k\in \big\{0,\ldots, \lfloor n/2 \rfloor \big\}$.
\end{enumerate}

\end{proposition}

We will prove the next lemma in Section~\ref{LemmaFourTerms}.  In a basic sense, the proof  uses the same idea as the proof of Lemma~\ref{LemUnCor} although the analysis is made more complex by the fourth moment.  
\begin{lemma}\label{LemFourTerms}
For $n\in \mathbb{N}$,  let   $\{ x_a \}_{a\in E_n }$ be an array of i.i.d.\ centered random variables with finite fourth moment.   Define $Y_{\ell}:= \mathcal{L}^{\ell-1}\mathcal{E}\mathcal{L}^{n-\ell}\{ x_a \}_{a\in E_n } $ for $\ell\in \{1, \ldots, n\}$. There is a $C>0$ not depending on the distribution of the variables $ x_a$ such that the following inequality holds for all $n\in \mathbb{N}$:  
$$  \mathbb{E}\Bigg[ \bigg( \sum_{\ell=1}^n Y_{\ell}   \bigg)^4  \Bigg]\,\leq\, Cn \sum_{\ell=1}^n \mathbb{E}\big[Y_{\ell}^4    \big] \,.  $$
\end{lemma}

\begin{proof}[Proof of Lemma~\ref{LemBasic}] Part
(i): For $f\in E_{\mathbf{\widehat{n}}}  $ the variance of   $Y_f^{N,n}\,=\,  \mathcal{L}^{\mathbf{n}-\mathbf{\widehat{n}} }\big\{ X_g^{( \mathbf{n},n)} \big\}_{g\in f\cap E_{\mathbf{n}}}$ is   $\sigma_{\mathbf{n}, n}^{2}:=\textup{Var}\big(  X_g^{( \mathbf{n},n)}  \big) $ since the operation $\mathcal{L}$ preserves the variance of the random  variables in the array by Remark~\ref{RemarkArrayVar}.   The convergence of $\sigma_{\mathbf{n},n}^{2}$ to $R(r-\mathbf{n})$ as $n\rightarrow \infty$ holds by (III) of Lemma~\ref{LemmaMom} with $m=2$. Finally, $\sigma_{\mathbf{n}, n}^{2}$ is bounded from  above and below by constant multiples of $\frac{1}{N}$ for all $N,n\in \mathbb{N}$ with $n\geq \mathbf{n}$ by Proposition~\ref{PropUnif} since $N\sim \mathbf{n}:=N+\lfloor 2 \frak{m}\log N\rfloor $.\vspace{.3cm}

\noindent Part (ii): Since terms in the sum $Z^{N,n}_f=\sum_{k=\mathbf{\widehat{n}}+1}^\mathbf{n}  \mathcal{L}^{k-\mathbf{\widehat{n}}-1}\mathcal{E}\mathcal{L}^{\mathbf{n}-k}\big\{ X_g^{( \mathbf{n},n)} \big\}_{g\in f\cap E_{\mathbf{n}}}$ are uncorrelated by Lemma~\ref{LemUnCor}, we have the second equality below.
\begin{align}
\varsigma_{N,n}^2\,:=\,\textup{Var}\big( Z_f^{N,n}   \big)\,=\,& \sum_{k=\mathbf{\widehat{n}}+1}^\mathbf{n} \textup{Var}\Big( \mathcal{L}^{k-\mathbf{\widehat{n}}-1}\mathcal{E}\mathcal{L}^{\mathbf{n}-k}\big\{ X_g^{( \mathbf{n},n)} \big\}_{g\in f\cap E_{\mathbf{n}}}  \Big)\nonumber 
\intertext{By part (ii) of Remark~\ref{RemarkArrayVar}, the above is equal to   }
 \,=\,&(\mathbf{n}-\mathbf{\widehat{n}})\big(M(x)-x  \big)\Big|_{x=M^{n-\mathbf{n} }(\sigma_n^2)} 
 \,=\,(\mathbf{n}-\mathbf{\widehat{n}})\big(M(\sigma_{\mathbf{n},n}^{2})-\sigma_{\mathbf{n},n}^{2} \big)\,.\label{VarsigmaForm}
\intertext{Since $\sigma_{\mathbf{n},n}^{2} $ converges to $R(r-\mathbf{n})$ with large $n$ by Lemma~\ref{LemmaMom} and $M\big(R(s)\big)=R(s+1)$ for all $s\in \R$ by  Lemma~\ref{LemVar}, there  is a sequence $\{\xi_{N}(n)\}_{n\in \mathbb{N}}$ that vanishes as $n\rightarrow \infty$ and  for which~(\ref{VarsigmaForm}) is equal to}
 \,=\,&(\mathbf{n}-\mathbf{\widehat{n}})\big(R(r-\mathbf{n}+1)\,-\,R(r-\mathbf{n}) \big)\,+ \,\xi_{N}(n)\,.\nonumber 
\end{align}
By definition of $\varsigma_{N}^2$, the above expression has the form $\varsigma_{N}^2 + \xi_{N}(n)$.  

Next we argue that $\varsigma_{N,n}^2$ is bounded from above by a constant multiple of $  \frac{\log N}{N^2}  $.  By~(\ref{VarsigmaForm}), we have that $\varsigma_{N,n}^2:=(\mathbf{n}-\mathbf{\widehat{n}})S\big(\sigma_{\mathbf{n},n}^2\big)$, where the polynomial  $S(x):=M(x)-x$ has no constant or linear terms. Since the lowest-order nonzero term in the polynomial $S(x)$ is quadratic, part (i) of Proposition~\ref{PropUnif} implies that $S\big(\sigma_{\mathbf{n},n}^2\big)$ is bounded by a constant multiple of  $\frac{1}{N^2}$ for all $n,N\in \mathbb{N}$ with $n\geq \mathbf{n}$.  The result then follows because $\mathbf{n}-\mathbf{\widehat{n}}\sim \frak{m}\log N$ for $N\gg 1$.

\vspace{.3cm}

\noindent Part (iii):  For  $g\in E_{\mathbf{n}}$, define $\sigma^{(4)}_{\mathbf{n},n}\,:=\,\mathbb{E}\Big[ \big(X_g^{(\mathbf{n},n)}\big)^4  \Big]$.  Also, for  $m\in \{2,4\}$ and  $a\in E_{k}$ with $k\in \{0,\ldots, \mathbf{n}\} $, we define
$$ \widetilde{\sigma}^{(m)}_{k,\mathbf{n},n}\,:=\, \mathbb{E}\bigg[ \Big( \mathcal{L}^{ \mathbf{n}-k }\big\{ X_g^{(\mathbf{n},n)} \big\}_{g\in a\cap E_{\mathbf{n}}    } \Big)^m  \bigg]  \,=\,\mathbb{E}\Bigg[ \bigg( \frac{1}{b^{\mathbf{n}-k }  }\sum_{g\in a\cap E_{\mathbf{n}}     } X_g^{(\mathbf{n},n)}\bigg)^m \Bigg]\,. $$
Note that $\widetilde{\sigma}^{(2)}_{k,\mathbf{n},n}=\textup{Var}\big(X_g^{(\mathbf{n},n)}\big)=:\sigma^{2}_{\mathbf{n},n} $, and
 Jensen's inequality implies that
\begin{align}\label{AsInHere}
\widetilde{\sigma}^{(4)}_{k,\mathbf{n}, n} 
\,=\,\frac{1}{b^{2(\mathbf{n}- k) }  }\sigma^{(4)}_{\mathbf{n},n}\,+\,3\frac{b^{2(\mathbf{n}- k) }-1}{b^{2(\mathbf{n}- k) } }   \big(\sigma^2_{\mathbf{n},n}\big)^2\, \leq  \,3 \sigma^{(4)}_{\mathbf{n},n} \,\leq \,\frac{C}{N^2}\,.
\end{align}
The second inequality above holds for some $C>0$ and all $n,N\in \mathbb{N}$ with $n\geq 2\mathbf{n}$ by  (ii) of Proposition~\ref{PropUnif} and since $\mathbf{n}\sim N$ for $N\gg 1$.  Applying~(\ref{AsInHere}) with $k=\mathbf{\widehat{n}}$ yields our desired bound for 
$\mathbb{E}\big[ \big(Y_f^{N,n}\big)^4  \big] 
\,=\,\widetilde{\sigma}^{(4)}_{\mathbf{\widehat{n}},\mathbf{n}, n}$. \vspace{.2cm}

Let $f\in E_{ \mathbf{\widehat{n}}}$.  By Lemma~\ref{LemFourTerms} the fourth moment of $Z_f^{N,n}$ has the bound
\begin{align}
\mathbb{E}\Big[ \big( Z_f^{N,n} \big)^4 \Big]
\, =\,& \mathbb{E}\Bigg[ \bigg(    \sum_{k=\mathbf{\widehat{n}}+1}^{\mathbf{n}} \mathcal{L}^{k-\mathbf{\widehat{n}}-1}\mathcal{E}\mathcal{L}^{\mathbf{n}-k}\big\{ X_g^{(\mathbf{n}, n)} \big\}_{g\in f \cap E_{\mathbf{n}}}  \bigg)^4\Bigg] \nonumber \\ 
\leq \, & C(\mathbf{n}-\mathbf{\widehat{n}}) \sum_{k=\mathbf{\widehat{n}}+1}^{\mathbf{n}} \mathbb{E}\bigg[  \Big(\mathcal{L}^{k-\mathbf{\widehat{n}}-1}\mathcal{E}\mathcal{L}^{\mathbf{n}-k}\big\{ X_g^{(\mathbf{n}, n)} \big\}_{g\in f \cap E_{\mathbf{n}}}  \Big)^4\bigg] \,. \label{Break}
\end{align}
 For $\mathbf{\widehat{n}}< k\leq \mathbf{n} $, define   $\big\{ \check{X}_a^{N, n}\big\}_{a\in f \cap E_{k}} := \mathcal{L}^{\mathbf{n}-k}\big\{ X_g^{(\mathbf{n}, n)} \big\}_{g\in f \cap E_{\mathbf{n}}}   $. 
 A single term from the sum in~(\ref{Break}) has the bound
\begin{align}
 \mathbb{E}\bigg[  \Big(\mathcal{L}^{k-\mathbf{\widehat{n}}-1}\mathcal{E}\big\{ \check{X}_a^{N,n} \big\}_{a\in f \cap E_{k}} \Big)^4\bigg]\nonumber  \,
\leq \,& 3\mathbb{E}\Bigg[ \Bigg( \underbrace{\frac{1}{b}\sum_{i=1}^b\prod_{j=1}^b \Big(1+\check{X}_{\mathbf{a}\times (i,j)}^{N,n}\Big)\,-\,1 \,-\,\frac{1}{b}\sum_{1\leq i,j\leq b}  \check{X}_{\mathbf{a}\times (i,j)}^{N,n}}_{ =\mathcal{E}\big\{ \check{X}_{\mathbf{a}\times (i,j)}^{N,n} \big\}_{(i,j)\in \{1,\ldots,b\}^2  }    } \Bigg)^4\Bigg]\nonumber
\intertext{for  any representative $\mathbf{a}\in f\cap  E_{k-1} $,  where we have used that $\mathcal{L}^{k-\mathbf{\widehat{n}}-1}\mathcal{E}\big\{ \check{X}_a^{N,n} \big\}_{a\in f \cap E_{k}}$ is a sum of $b^{2(k-\mathbf{\widehat{n}}-1)}$ independent mean zero random variables having the braced form and applied Jensen's inequality as in~(\ref{AsInHere}). By foiling the products in the above expression and using that random variables $ \check{X}_{\mathbf{a}\times (i,j)}^{N,n} $ for $i,j\in \{1,\ldots, b\}$ are independent and centered, we can see that there is a degree-$b$ polynomial $T(x,y)$ with nonnegative coefficients  and having the form $a_1x^2\,+\,a_2xy^2\,+\,a_3 y^4 $ plus higher-order terms  such that the above is equal to  }
 =\,&  T\Big( \widetilde{\sigma}^{(4)}_{k,\mathbf{n}, n}, \widetilde{\sigma}^{(2)}_{k,\mathbf{n}, n}  \Big)\leq T\Big( \widetilde{\sigma}^{(4)}_{k,\mathbf{n}, n}, \big(\widetilde{\sigma}^{(4)}_{k,\mathbf{n}, n}\big)^{\frac{1}{2}}  \Big) \leq T\bigg( \frac{C}{N^2}, \frac{\sqrt{C}}{N} \bigg)   \,.\label{Yip}
\end{align}
The  first inequality above is Jensen's, and the second inequality holds for all $n,N$ with $n\geq \mathbf{n}$ by~(\ref{AsInHere}). Thus by~(\ref{Break}),~(\ref{Yip}), and the form of the polynomial $T(x,y)$, the fourth moment of $ Z_f^{N,n} $ is bounded from above by a multiple of $\frac{(\mathbf{n}-\mathbf{\widehat{n}})^2 }{N^2}\sim \frak{m}^2\frac{\log^2 (N+1)  }{N^2} $.

 \vspace{.3cm}

\noindent Part (iv): Since $\mathbf{n}\sim N$ for $N\gg 1$, an  application of  (ii) of Proposition~\ref{PropUnif} with $k=\mathbf{n}$ yields that the fourth moment of  $X_g^{(\mathbf{n},n)}$ is bounded by a constant multiple of $\frac{C}{N^2}$ for all $n, N\in \mathbb{N}$ with $n\geq 2\mathbf{n}$.  The fourth moment bounds for $\widehat{X}_e^{N,n}$ and $\mathbf{\widehat{X}}_e^{N,n}$ can be proven using the techniques in the proof of (iii).\footnote{Also, see the proof of part (iv) of Lemma~\ref{LemBasic} in Section~\ref{SecLemmaLittle}, which  is an analogous result for general even moments under  $\alpha$-sharp regularity-type assumptions.} \end{proof}

\begin{proof}[Proof of Lemma~\ref{Lemma1to2}] Let $ (X, Y)$ be a coupling such that the $L^1$-distance between the variables $ X$ and  $Y $  is equal to $\rho_1( X , Y )$.   Since $\rho_2( X , Y )$ is an infimum of the $L^2$ distance over couplings, 
\begin{align*}
\rho_2( X , Y )\,\leq  \,\mathbb{E}\big[|   X - Y  |^2\big]^{\frac{1}{2}}
\,= \,&\mathbb{E}\Big[\big|    X - Y    \big|^{\frac{m-1}{m}}\big|    X - Y    \big|^{\frac{m+1}{m}}\Big]^{\frac{1}{2}}\,. \\ \intertext{Applying Holder's inequality with $(p,q)=\big(\frac{m}{m-1},m\big)$ and the triangle inequality yields }
 \,\leq \,& \mathbb{E}\big[|    X - Y  |\big]^{\frac{m-1}{2m}} \mathbb{E}\big[|    X - Y   |^{m+1}\big]^{\frac{1}{2m}}  \\
 \,\leq \,& \big(\rho_1 ( X , Y ) \big)^{\frac{m-1}{2m}}  \Big(  \mathbb{E}\big[|  X|^{m+1}\big]^{\frac{1}{m+1}}\,+\,\mathbb{E}\big[ |    Y  |^{m+1}\big]^{\frac{1}{m+1}} \Big)^{\frac{m+1}{2m}}\,.
 \intertext{Finally, using that $(x+y)^{a}\leq 2^a(x^a+y^a)$ for $x,y\geq 0$  with  $a=\frac{m+1}{2m}$ gives us    }
 \,\leq \,& 2^{\frac{m+1}{2m}}\big(\rho_1 ( X , Y ) \big)^{\frac{m-1}{2m}} \Big(  \mathbb{E}\big[|  X|^{m+1}\big]^{\frac{1}{2m}}\,+\,\mathbb{E}\big[ |    Y  |^{m+1}\big]^{\frac{1}{2m}} \Big) \,.     
\end{align*} 
\end{proof}

\section{Sharp regularity and rate of convergence}\label{SectionSharpRegProof}

Next we focus on proving Theorem~\ref{ThmSharpUnique}. To do this, we will use analogous technical  results to those in Lemmas~\ref{LemI}-\ref{LemIII}---see (i)-(iii) of Lemma~\ref{LemmaRe} below---that assume sharp regularity-type conditions and provide bounds in terms of functions of the ``microscopic" parameter $n\in \mathbb{N}$ rather than the ``mesoscopic" parameter $N\in \mathbb{N}$.  With Lemma~\ref{LemmaRe} in hand, the proof of Theorem~\ref{ThmSharpUnique} carries through with only minor modifications of the proof of Theorem~\ref{ThmUnique}.   We prove Lemma~\ref{LemmaRe} in Section~\ref{SecLemDis}, and in Section~\ref{SecLemmaLittle} we prove an analog of Lemma~\ref{LemBasic}.

\subsection{Proof of Theorem~\ref{ThmSharpUnique}   }\label{SubSectionSharpRegProof}

We will prove Theorem~\ref{ThmSharpUnique} after stating two preliminary lemmas.   Lemma~\ref{LemmaRe} bounds  the same quantities as in  Lemmas~\ref{LemI}-\ref{LemIII}, and its proof is in the next subsection.
\begin{lemma}\label{LemmaRe}    Fix  $\mathbf{v},\varkappa>0$, $\alpha\in(0,1)$, $\upsilon\in (0,\alpha/9)$, and a bounded interval $\mathcal{I}\subset \R$. Define $\frak{p}=\lceil \frac{2\alpha}{\alpha-9\upsilon}\rceil  +1 $ and  $N\equiv N(n):=\lfloor n^{2\alpha/9}\rfloor$ for  $n\in \mathbb{N}$.  There exists a positive number $\mathbf{c}\equiv \mathbf{c}(\mathcal{I},\mathbf{v},\varkappa, \alpha, \upsilon)$ such that for any $r\in \mathcal{I}$, $n\in \mathbb{N}$, and  i.i.d.\ array of centered random variables $\big\{ X_{h}^{(n)} \big\}_{h\in  E_{n} }$ satisfying
\begin{enumerate}[(I)]

\item $\left| \textup{Var}\big(X_h^{(n)}\big)-\kappa^2\big(\frac{1}{n}+\frac{\eta \log n  }{n^2}+\frac{r}{n^2}\big)\right|<\frac{\mathbf{v}}{n^{2+\alpha}}$  and

\item $\mathbb{E}\Big[\big|X_h^{(n)}\big|^{2\frak{p}}\Big] <\frac{\varkappa}{n^{\frak{p}}}$,

\end{enumerate}
 the following inequalities hold:
\begin{enumerate}[(i)]

\item $\mathbb{E}\Big[ \big( X^{(N,n)}_e-\widehat{X}^{N,n}_e   \big)^2\Big]^{1/2}\, <  \, \mathbf{c} \frac{ \log (n+1)}{n^{\alpha/3}}  $\,,

\item $\rho_2 \big( \widehat{X}^{N,n}_e ,\mathbf{\widehat{X}}^{N,n}_e   \big)\,  < \, \frac{\mathbf{c}}{n^{4\alpha/9+\upsilon }}\displaystyle  $\,, 

\item $\rho_2\big(   \mathbf{\widehat{X}}_e^{N,n} ,    \mathbf{\widetilde{X}}_e^{(N)}  \big)\, < \,\frac{\mathbf{c} }{n^{8\alpha/9}}\,, \displaystyle $

\end{enumerate}
where $\big\{X^{(N,n)}_e\big\}_{e\in E_N}$ is the $N^{th}$ generation layer of the  $\mathcal{Q}$-pyramidic array generated from  $\big\{ X_{h}^{(n)} \big\}_{h\in  E_{n} }$, and $\{\widehat{X}^{N,n}_e\}_{e\in E_N} $, $\{\mathbf{\widehat{X}}_e^{N,n}\}_{e\in E_N}$,  $ \{\mathbf{\widetilde{X}}_e^{(N)}\}_{e\in E_N}  $ are defined as in Definition~\ref{DefXs}  with $\frak{m}:=\frac{21}{2\log b}$.    
\end{lemma}

\begin{remark} In  Lemma~\ref{LemmaRe}, any value of $\frak{m}$ greater than $\frac{21}{2\log b}$  yields the same result. 
\end{remark}

Recall that the random variables in the array $ \big\{\mathbf{X}_{h}^{(n)}\big\}_{h\in E_n}$ from Theorem~\ref{ThmExist} with parameter $r\in \R$ have $m^{th}$ moment given by  $R^{(m)}(r-n)$, where the function $R^{(2)}\equiv R$ is characterized in Lemma~\ref{LemVar} and  the functions $R^{(m)}$ for $m\geq 3$ are characterized in Theorem~\ref{ThmHM}.
The following trivial lemma implies that the conditions of Lemma~\ref{LemmaRe} are satisfied by $ \{\mathbf{X}_{h}^{(n)}\}_{h\in E_n}$ for all $n\in \mathbb{N}$ and all $r$ in a bounded interval $\mathcal{I}$ when $\mathbf{v},\varkappa>0$ are  large enough.

\begin{lemma}\label{LemForR} Fix $\alpha\in (0,1)$, $\frak{p}\in \{2,3,\ldots\}$, and a bounded interval $\mathcal{I}\subset \R$.  There exist $\mathbf{v},\varkappa>0$ such that (I)-(II) below hold for all $r\in \mathcal{I}$ and $n\in \mathbb{N}$.
\begin{enumerate}[(I)]

\item $\left|R(r-n)-\kappa^2\big(\frac{1}{n}+\frac{\eta \log n  }{n^2}+\frac{r}{n^2}\big)\right|<\frac{\mathbf{v}}{n^{2+\alpha}}$ 

\item $R^{(2\frak{p})}(r-n) <\frac{\varkappa}{n^{\frak{p}}}$

\end{enumerate}

\end{lemma}
\begin{proof} The inequalities (I)-(II) above hold for large enough $\mathbf{v},\varkappa>0$ and all $r\in \mathcal{I}$ and $n\in \mathbb{N}$ as a consequence of the asymptotics in (II) of Lemma~\ref{LemVar} and (II) of Theorem~\ref{ThmHM}, respectively.
\end{proof}

\vspace{.4cm}

\begin{proof}[Proof of Theorem~\ref{ThmSharpUnique}] Let $\mathbf{v}$, $\varkappa $, $\alpha$, $\upsilon$,  $\mathcal{I}$, $\frak{p} $, and $N$ be as in Lemma~\ref{LemmaRe}.   By Lemma~\ref{LemmaRe}, there is a $\mathbf{c}\equiv \mathbf{c}(\mathcal{I},\mathbf{v},
\varkappa,\alpha,\upsilon)  $ such that if $r\in \mathcal{I}$, $n\in \mathbb{N}$, and $\big\{X^{(n)}_h\big\}_{h\in E_n}$ is an array of i.i.d.\ centered random variables satisfying conditions (I)-(II) in Theorem~\ref{ThmSharpUnique}, then 
\begin{align}
\mathbb{E}\Big[ \big( X^{(N,n)}_e-\widehat{X}^{N,n}_e   \big)^2\Big]^{\frac{1}{2}} <    \mathbf{c} \frac{\log (n+1)}{n^{\alpha/3}} \,,\hspace{.3cm} \rho_2 \big( \widehat{X}^{N,n}_e ,\mathbf{\widehat{X}}^{N,n}_e   \big)  <  \frac{\mathbf{c}}{n^{4\alpha/9+\upsilon } } \,, \hspace{.3cm} \rho_2\big(   \mathbf{\widehat{X}}_e^{N,n} ,    \mathbf{\widetilde{X}}_e^{(N)}  \big) <  \frac{\mathbf{c}}{n^{4\alpha/9+\upsilon } } \,, \nonumber
\end{align}
where  for the third inequality we have used that $n^{-8\alpha/9}$ is $\mathit{O}\big( n^{-4\alpha/9-\upsilon} \big)$ as $n\rightarrow \infty$ since $\upsilon< \alpha/9 $.  By the same reasoning as in parts (a)-(c) of the proof  of Theorem~\ref{ThmUnique}, there are i.i.d.\ families of pair couplings $\big\{\big( \widehat{X}^{N,n}_e ,\mathbf{\widehat{X}}^{N,n}_e   \big)\big\}_{e\in E_N} $ and $\big\{\big(   \mathbf{\widehat{X}}_e^{N,n} ,    \mathbf{\widetilde{X}}_e^{(N)}  \big)\big\}_{e\in E_N}$ such that the first two inequalities below hold:
\begin{align}
\rho_2\Big( &X^{(0,n)},   \mathcal{Q}^{N}\big\{ \mathbf{\widetilde{X}}_e^{(N)}\big\}_{e\in E_{N}}  \Big) \nonumber  \\
 \,\leq \,& CN\mathbb{E}\Big[ \big(  X_e^{(N,n)} -  \widehat{X}_e^{N,n} \big)^2\Big]^{\frac{1}{2}} \,+\, CN^{2} \mathbb{E}\Big[ \big(  \widehat{X}_e^{N,n}  -    \mathbf{\widehat{X}}_e^{N,n}  \big)^2\Big]^{\frac{1}{2}} \,+\, CN^{2} \mathbb{E}\Big[\big(  \mathbf{\widehat{X}}_e^{N,n}  -  \mathbf{\widetilde{X}}_e^{(N)}  \big)^2\Big]^{\frac{1}{2}}\,\nonumber   \\
  \, < \,&  \frac{\mathbf{c}CN\log (n+1)}{n^{\alpha/3}}  \,+\,\frac{\mathbf{c}CN^{2}}{n^{4\alpha/9+\upsilon } } \,+\,\frac{\mathbf{c}CN^{2}}{n^{4\alpha/9+\upsilon } }\,\leq \,\frac{\mathbf{C}}{n^{\upsilon}} \,, \label{Wass}
\end{align}
 where  $C>0$ arises from an application of Proposition~\ref{PropFinalPush}.  For  $\mathbf{C}:=\mathbf{c}C\big(2+\sup_{u\in \mathbb{N}}\frac{\log(u+1)  }{ u^{\alpha/9-\upsilon} }     \big)$, the third inequality simply uses that $N:=\lfloor n^{2\alpha/9}\rfloor$.

By~(\ref{Wass}) the Wasserstein-2 distance between $X^{(0,n)}=\mathcal{Q}^{n}\big\{X^{(n)}_h\big\}_{h\in E_n}$ and 
$ \mathcal{Q}^{N}\big\{ \mathbf{\widetilde{X}}_e^{(N)}\big\}_{e\in E_{N}}$ is bounded by a multiple $\mathbf{C}\equiv \mathbf{C}(\mathcal{I},\mathbf{v},\varkappa,\alpha,\upsilon)$ of $n^{-\upsilon}$ for any i.i.d.\ array 
 $\big\{X^{(n)}_h\big\}_{h\in E_n}$ satisfying properties (I)-(II) in the statement of Theorem~\ref{ThmSharpUnique}. Let the array of random variables $\big\{\mathbf{X}^{(n)}_h\big\}_{h\in E_n}$  be defined as in Theorem~\ref{ThmExist} for parameter $r $.  By property (III) in Theorem~\ref{ThmExist}, the $m^{th}$ positive integer moment of $\mathbf{X}^{(n)}_h$ is $R^{(m)}(r-n)$, and thus by Lemma~\ref{LemForR} the array $\big\{\mathbf{X}^{(n)}_h\big\}_{h\in E_n}$ satisfies  conditions (I)-(II) of Lemma~\ref{LemmaRe} for all $r\in \mathcal{I}$ and $n\in \mathbb{N}$ with possibly larger values of $\mathbf{v},\varkappa>0$.  By substituting $\big\{\mathbf{X}^{(n)}_h\big\}_{h\in E_n}$ for $\big\{X^{(n)}_h\big\}_{h\in E_n}$ in our above analysis, we get that the Wasserstein-$2$ distance between  $\mathbf{X}=\mathcal{Q}^{n}\big\{\mathbf{X}^{(n)}_h\big\}_{h\in E_n}$ and 
$ \mathcal{Q}^{N}\big\{ \mathbf{\widetilde{X}}_e^{(N)}\big\}_{e\in E_{N}}$  is bounded by a multiple  $\mathbf{C'}\equiv \mathbf{C'}(\mathcal{I},\alpha ,\upsilon) $ of $n^{-\upsilon}$ for all $n\in \mathbb{N}$ and $r\in \mathcal{I}$.  By the triangle inequality, we thus have the bound that we sought for the Wasserstein-2 distance between $X^{(0,n)}$ and  $\mathbf{X}$.\end{proof}

\subsection{Proof of Lemma~\ref{LemmaRe}}\label{SecLemDis}

Recall that there are steps in each of the proofs of Lemmas~\ref{LemI}-\ref{LemIII} in which we respectively identified sequences  $\{\xi_{N}(n)\}_{n\in \mathbb{N}} $, $\{\xi_{N}'(n)\}_{n\in \mathbb{N}} $, $\{\xi_{N}''(n)\}_{n\in \mathbb{N}} $ that  vanish as $n\rightarrow \infty$ for each fixed $N\in \mathbb{N}$ and for which  the inequalities ($\textup{i}'$)-($\textup{iii}'$) below hold for some $\frak{c}>0$ and all $N,n$ with $n\geq 2\mathbf{n}$.
\begin{enumerate}

\item[($\textup{i}'$)\hspace{-.1cm}] \hspace{.1cm}$\mathbb{E}\Big[ \big( X^{(N,n)}_e-\widehat{X}^{N,n}_e   \big)^2\Big]\, \leq     \, \frak{c}\frac{\log^2 (N+1)}{N^{3} }\,+\,\xi_{N}(n)   $

\item[($\textup{ii}'$)\hspace{-.1cm}] \hspace{.1cm}$\rho_1 \big( \widehat{X}^{N,n}_e ,\mathbf{\widehat{X}}^{N,n}_e   \big)\,  \leq  \, \frak{c}\frac{\log^{-\frac{1}{2}}(N+1)}{N^{\frak{m}\log b  }}\,+\,\xi_{N}'(n)\displaystyle  $

\item[ \,($\textup{iii}'$)\hspace{-.1cm}] \hspace{.1cm}$\rho_2\big(   \mathbf{\widehat{X}}_e^{N,n} ,    \mathbf{\widetilde{X}}_e^{(N)}  \big)\, \leq  \,  \frac{\frak{c}}{N^{\frac{\frak{m}}{3}\log b +\frac{1}{2} }}\,+\,\xi_{N}''(n)\displaystyle$

\end{enumerate} 
The inequalities ($\textup{i}'$)-($\textup{iii}'$) are from~(\ref{XiOne}),~(\ref{Xi2}), \&~(\ref{XiThree}).  Also recall that the proofs of ($\textup{ii}'$) \& ($\textup{iii}'$) rely on bounds from Lemma~\ref{LemBasic}.    The following lemma states analogous  results to those in  Lemma~\ref{LemBasic} under the conditions (I)-(II) of Lemma~\ref{LemmaRe}, and its proof is in Section~\ref{SecLemmaLittle}.  In the statement of Lemma~\ref{LemBasicII}, the random variables $ Y^{N,n}_f $ and  $Z^{N,n}_f $ are defined as in~(\ref{DefY})~\&~(\ref{DefZ}), $ \sigma_{N,n}^2$ is defined as in~(\ref{DefLittleSigma}), and  $\varsigma_{N,n}^2, \varsigma_{N}^2$  are defined as in Lemma~\ref{LemBasicII}.

\begin{lemma}\label{LemBasicII}     Fix  $\mathbf{v},\varkappa>0$, $\alpha\in (0,1)$, $\frak{p}\in \{2,3,\ldots\}$, and a bounded interval $\mathcal{I}\subset \R$. For $n\in \mathbb{N}$, define $N:=\lfloor n^{2\alpha/9}\rfloor$.  There exist positive numbers $\mathbf{c}\equiv \mathbf{c}(\mathcal{I},\mathbf{v}, \varkappa,\alpha, \frak{p})$ and $\lambda\equiv \lambda( \mathcal{I},\mathbf{v},\alpha )$ such that for any $r\in \mathcal{I}$, $n\in \mathbb{N}$, and i.i.d.\ array of centered random variables  $\big\{ X_{h}^{(n)} \big\}_{h\in  E_{n} }$  satisfying conditions (I)-(II) of Lemma~\ref{LemmaRe},  the inequalities below hold for the  random variables $ Y^{N,n}_f $, $Z^{N,n}_f $, $\widehat{X}^{N,n}_e $,  $\mathbf{\widehat{X}}_e^{N,n}$ and the variances $\sigma_{N,n}^2:=\textup{Var}\big(X_{e}^{(N,n)}  \big)$ \& $\varsigma_{N,n}^2:=\textup{Var}\big(Z^{N,n}_f\big)$ defined through the $\mathcal{Q}$-pyramidic array  $\big\{ X_{a}^{(*,n)} \big\}_{a\in  E_{*} }$ generated from $\big\{ X_{h}^{(n)} \big\}_{h\in  E_{n} }$.
\begin{enumerate}[(i)]
\item    $ \sigma_{N,n}^2$ is bounded from above by $\frac{\mathbf{c} }{N}$,  and  $ \sigma_{N,n}^2$ is bounded from below by  $\frac{\mathbf{c}^{-1}}{ N}$ provided that $n>\lambda$.

\item $\varsigma_{N,n}^2$ is bounded from above by $\mathbf{c}\frac{\log(N+1)}{N^2 }  $ and  satisfies the inequality 
$$ \bigg|\frac{\varsigma_{N,n}^2}{\varsigma_N }\,-\,\varsigma_N    \bigg|\leq \mathbf{c}\frac{\log^{1/2}(n+1)  }{ n^{\alpha} }\,.$$

\item  The fourth moments of the random variables $ Y^{N,n}_f $ and $Z^{N,n}_f $ are bounded by  $\frac{ \mathbf{c}}{N^2}$ and $ \mathbf{c}\frac{ \log^2(N+1)}{N^4}$, respectively.

\item The $(2\frak{p})^{th}$ moments of the random variables $ X^{(\mathbf{n},n)}_g$, $\widehat{X}^{N,n}_e $, and  $\mathbf{\widehat{X}}_e^{N,n}$ are bounded by $\frac{ \mathbf{c}}{N^{ \frak{p} }}$.

\end{enumerate}

\end{lemma}

The lemma below states that analogs of the inequalities  ($\textup{i}'$)-($\textup{iii}'$) hold for large enough $\frak{c}\equiv\frak{c}(\mathcal{I},\mathbf{v},\varkappa,\alpha)>0$  when $X^{(N,n)}_e$, $\widehat{X}^{N,n}_e $, and  $\mathbf{\widehat{X}}_e^{(N)}$  are defined in terms of an array  of random variables $\big\{ X_{h}^{(n)} \big\}_{h\in  E_{n} }$ satisfying the conditions of Lemma~\ref{LemmaRe}.  If we were only concerned with having a counterpart to the inequality ($\textup{i}'$), then the constant $\frak{c}$ would only depend  on the bounded interval $\mathcal{I}$ because the derivation of ($\textup{i}'$) in the proof of Lemma~\ref{LemI} is  entirely based on properties of the  function $R$ from Lemma~\ref{LemVar}.  The counterparts to ($\textup{ii}'$) \& ($\textup{iii}'$) can be shown by following the steps in the proofs of  ($\textup{ii}'$) \& ($\textup{iii}'$) and   replacing each application of (i)-(iv) from Lemma~\ref{LemBasic} by an application of (i)-(iv) from Lemma~\ref{LemBasicII}.  Thus we omit the proof of Lemma~\ref{LemStart}, which is a lengthy near-repetition of our previous line of arguments establishing ($\textup{i}'$)-($\textup{iii}'$) in Section~\ref{SecCentralLimit}.
\begin{lemma}\label{LemStart} Fix $\mathbf{v}, \varkappa, \frak{m}>0$, $\alpha\in (0,1)$, and a bounded interval $\mathcal{I}$. Define $N=\lfloor n^{2\alpha/9}\rfloor $.  There exists a positive number $\frak{c}\equiv\frak{c}(\mathcal{I},\mathbf{v},\varkappa,\alpha,\frak{m})$ such that for any $r\in \mathcal{I}$, $n\in \mathbb{N}$, and  i.i.d.\ array of centered random variables  $\big\{ X_{h}^{(n)} \big\}_{h\in  E_{n} }$ satisfying conditions (I)-(II) of Lemma~\ref{LemmaRe} for  $\frak{p}=2$, then the inequalities ($i'$)-($iii'$)  above hold, where   $\big\{X^{(N,n)}_e\big\}_{e\in E_N}$ is the $N^{th}$ generation layer of the  $\mathcal{Q}$-pyramidic array generated from  $\big\{ X_{h}^{(n)} \big\}_{h\in  E_{n} }$, and $\{\widehat{X}^{N,n}_e\}_{e\in E_N} $, $\{\mathbf{\widehat{X}}_e^{N,n}\}_{e\in E_N}$,  $ \{\mathbf{\widetilde{X}}_e^{(N)}\}_{e\in E_N}  $ are defined as in Definition~\ref{DefXs}.  
\end{lemma}

Lemma~\ref{LemVarApprox} offers some control for the rate of convergence in the $m=2$ case of Lemma~\ref{LemmaMom} under the $\alpha$-sharp regularity condition on the variance of the random variables in the generating array $\big\{ X_h^{(n)} \big\}_{h\in  E_{n} }$.  The proof, which is placed in  Section~\ref{SecSharpVC},  borrows a technical result from~\cite{Clark1}. 
\begin{lemma}\label{LemVarApprox} Fix $\mathbf{v}>0$, $\alpha\in (0,1)$, and a bounded interval $\mathcal{I}\subset \R$.  There exists a positive number  $C_{\mathcal{I},\mathbf{v}, \alpha} $ such that for any $r\in \mathcal{I}$, $n\in \mathbb{N}$, and  i.i.d.\ array of centered random variables $\big\{ X_h^{(n)} \big\}_{h\in  E_{n} }$   satisfying condition (I) of Lemma~\ref{LemmaRe},  the inequality below holds for all $k\in \{0,1,\ldots, n\}$:
$$\big| \textup{Var}\big(X_a^{(k,n)}\big)  \,-\, R(r-k) \big| \,\leq \,\frac{C_{\mathcal{I},\mathbf{v}, \alpha}}{n^{\alpha}}\,, $$ 
where $\big\{X_a^{(k,n)}\big\}_{a\in E_k}$ is the $k^{th}$ generation layer of the $\mathcal{Q}$-pyramidic array generated from $\big\{ X_h^{(n)} \big\}_{h\in  E_{n} }$.
\end{lemma}

\vspace{.2cm}

\begin{proof}[Proof of Lemma~\ref{LemmaRe}] Fix $\mathbf{v},\varkappa >0$, $\alpha\in (0,1)$, $\upsilon\in(0,\alpha/9)$, and a bounded interval $\mathcal{I}$. Define $N:=\lfloor n^{2\alpha/9}\rfloor$, $\frak{p}:=\lceil \frac{2\alpha}{\alpha-9\upsilon}\rceil  +1$, and $\frak{m}:=\frac{21}{2\log b}$.  Let $\big\{ X_{h}^{(n)} \big\}_{h\in  E_{n} }$ be an i.i.d. array of random variables satisfying  conditions (I)-(II) for $\mathbf{v}$,  $\varkappa$, $\alpha$, $\frak{p}  $, and some $r\in \mathcal{I}$.  Since $\frak{p}\geq 2$, Jensen's inequality  and condition (II) imply that
\begin{align}
\mathbb{E}\Big[\big|X_h^{(n)}\big|^{4}   \Big]\,\leq \,\mathbb{E}\Big[\big|X_h^{(n)}\big|^{2\frak{p}}   \Big]^{\frac{ 2 }{ \frak{p} }}\, < \,\Big( \frac{\varkappa}{n^{\frak{p}} }  \Big)^{\frac{ 2 }{ \frak{p} }}\,\leq \, \frac{ \max(1,\varkappa)  }{ n^{2}}\,.
\end{align}
Thus $\big\{ X_{h}^{(n)} \big\}_{h\in  E_{n} }$ satisfies condition (II) with $\frak{p}\mapsto 2$ and $\varkappa\mapsto  \max(1,\varkappa)$.  By Lemma~\ref{LemStart}, there is $\frak{c}\equiv\frak{c}(\mathcal{I},\mathbf{v},\varkappa,\alpha,\frak{m})$ such that the inequalities ($\textup{i}'$)-($\textup{iii}'$) hold. In parts (i)-(iii) below we will start from the inequalities ($\textup{i}'$)-($\textup{iii}'$), respectively, and focus on bounding the terms $\xi_N(n)$, $\xi_N'(n)$, $\xi_N''(n)$. \vspace{.2cm}

\noindent Part (i):  By  inequality ($\textup{i}'$),
\begin{align}\label{EX}
\mathbb{E}\Big[ \big( X^{(N,n)}_e-\widehat{X}^{N,n}_e   \big)^2\Big]\,\leq \,\frak{c}\frac{ \log^2 (N+1)  }{N^3  }\,+\,\xi_N(n)\, \leq  \,\frak{c}'\frac{ \log^2 (n+1)  }{n^{2\alpha/3}  }\,+\,\xi_N(n)  \,,
\end{align}
where $\xi_N(n)$ is  the error of the approximation of~(\ref{SigmaForm}) by the expression in~(\ref{RForm}), i.e.,
 \begin{align*}
 \xi_N(n)\,:=\,& \sigma^2_{N,n} \,-\,   \sigma^2_{\mathbf{n},n} \,-\,(\mathbf{n}-N)\Big(M\big(\sigma^2_{\mathbf{n},n}\big)\,-\,\sigma^2_{\mathbf{n},n}   \Big)\\  &\,-\,R(r-N)  \,+\,   R(r-\mathbf{n}) \,+\, (\mathbf{n}-N)\Big(M\big(R(r-\mathbf{n})\big)\,-\,R(r-\mathbf{n})  \Big) \,.
 \end{align*}
The second inequality in~(\ref{EX}) holds for some $\frak{c}'>0$ since $N:=\lfloor n^{2\alpha/9}\rfloor $. In the analysis below, we will show that  $\xi_N(n)$ is bounded by a multiple of $\frac{\log(n+1)  }{ n^{11\alpha /9} }$, and consequently  that the $L^2$ distance between $X^{(N,n)}_e$ and $\widehat{X}^{N,n}_e$ is bounded by a multiple of $ \frac{ \log (n+1)  }{n^{\alpha/3}  }$ by~(\ref{EX}).

  Define the polynomial $S(x):=M(x)-x$, in other words, as $M$ with the linear term removed. As in the proof of Lemma~\ref{LemI}, we can use telescoping sums to write 
$$ \sigma^2_{N,n} \,-\,   \sigma^2_{\mathbf{n},n} \,=\, \sum_{k=N+1}^{\mathbf{n}  }S\big( \sigma_{k,n}^2\big) \hspace{1cm}\text{and}\hspace{1cm}  R(r-N) \,-\,   R(r-\mathbf{n}) \,=\,  \sum_{k=N+1}^{\mathbf{n}  }S\big( R(r-k)\big)\,,  $$
where we have used the identities $M( \sigma_{k,n}^2)=\sigma_{k-1,n}^2$ and $M\big(R(r-k)\big)=R(r-k+1)$. Thus $\xi_N(n)$ can be written as
$$
\xi_N(n)\, =  \, \sum_{k=N+1}^{\mathbf{n}  }\Big(S\big( \sigma_{k,n}^2\big) \,-\, S\big( R(r-k)\big)    \Big) \,+\,(\mathbf{n}-N)\Big( S\big( \sigma_{\mathbf{n},n}^2\big) \,-\, S\big( R(r-\mathbf{n})\big)  \Big)\,.
 $$
It follows that 
\begin{align}\label{Tando}
\big|\xi_N(n)\big|\,
\leq  \,& 2(\mathbf{n}-N) \max_{N \leq  k\leq \mathbf{n}} \Big|S\big( \sigma_{k,n}^2\big) \,-\, S\big( R(r-k)\big)    \Big| \,. 
\end{align}
 By  Lemma~\ref{LemVarApprox}, there is a $C_{\mathcal{I},\mathbf{v},\alpha}>0$ such that for all $r\in \mathcal{I}$ and $n\in \mathbb{N}$
\begin{align}\label{SigmaR!}
\max_{0\leq k\leq n}\big|\sigma_{k,n}^2\,-\, R(r-k)  \big| \,\leq \,  \frac{C_{\mathcal{I},\mathbf{v},\alpha} }{ n^{\alpha}}    \,. 
\end{align}
The lowest-order nonzero term in the polynomial $S(x)$ is quadratic, and thus the following is finite:
$$ \mathbf{c}'\,\equiv\,\mathbf{c}'(\mathcal{I},\mathbf{v},\alpha)\,:=\, \sup_{ \substack{ 0\leq  y \leq R(\sup \mathcal{I})  \\ |x-y|\leq C_{\mathcal{I},\mathbf{v},\alpha}      }   }\frac{\big|S(x)-S(y)\big|}{|y| |x-y|  }\,. $$
 Since $ \frac{1 }{ n^{\alpha}}\leq 1 $ for $n\in \mathbb{N}$,~(\ref{SigmaR!}) implies that the distance between $S\big(\sigma_{k,n}^2\big)$ and $S\big(R(r-k)\big)$  is bounded by
\begin{align}
\big|S\big(\sigma_{k,n}^2\big)-S\big(R(r-k)\big)\big|\,\leq \,&\mathbf{c}'R(r-k)\big|\sigma_{k,n}^2\,-\,R(r-k)\big|\,\leq \,\frac{\mathbf{c}'C_{\mathcal{I},\mathbf{v},\alpha}}{n^{\alpha}} R(r-k) \,.\label{Tip!}
\end{align}
By applying~(\ref{Tip!}) to~(\ref{Tando}) and using that $R$ is an increasing function, we get that 
\begin{align}
\big|\xi_N(n)\big|\,\leq  \,& \frac{2(\mathbf{n}-N)\mathbf{c}'C_{\mathcal{I},\mathbf{v},\alpha}}{n^{\alpha}}R(r-N) \,\leq\,\frac{2(\mathbf{n}-N)\mathbf{c}'C_{\mathcal{I},\mathbf{v},\alpha}}{ N n^{\alpha}}   \sup_{ \substack{r\in \mathcal{I} \\ x\geq 0 }  } xR(r-x)  \,.  \label{Dim}
\end{align}
The supremum above is finite because $R(s)\sim \frac{\kappa^2}{-s}$ for $s\gg 1$ by Lemma~\ref{LemVar}. Since $N:=\lfloor n^{2\alpha/9}\rfloor $ and $\mathbf{n}:=N+\lfloor 2\frak{m}\log N \rfloor $, the inequality~(\ref{Dim}) implies  that $\big|\xi_N(n)\big|$ is bounded by a  multiple of $\frac{\log(n+1)  }{ n^{11\alpha /9} }$.  

\vspace{.5cm}

\noindent Part (ii): Since $\xi'_{N,n}:= \sqrt{\frac{\pi}{2}}\big|\frac{\varsigma_{N,n}^2 }{\varsigma_{N} }-\varsigma_{N}\big|  $, the first inequality below is ($\textup{i}''$):
\begin{align}
\rho_1\big( \widehat{X}^{N,n}_e , \mathbf{\widehat{X}}^{N,n}_e  \big)\,\leq \,\frak{c}\frac{ \log^{-\frac{1}{2}} N   }{N^{ \frak{m}\log b  }}\,+\,\sqrt{\frac{\pi}{2}}\bigg|\frac{\varsigma_{N,n}^2 }{\varsigma_{N} }-\varsigma_{N}\bigg|\,\leq\, \frac{ C\log^{1/2}(n+1) }{n^{\alpha}  } \,.\label{TT}
\end{align}
 The second inequality holds for some $C>0$ by   part (ii) of Lemma~\ref{LemBasicII} for the second term and since $N:=\lfloor n^{2\alpha/9}\rfloor$ and $\frak{m}:=\frac{21}{2\log b}$ for the first term.

As in the proof of Lemma~\ref{LemII}, we will use Lemma~\ref{Lemma1to2} to bound the Wasserstein-$2$ distance using the Wasserstein-$1$ distance.    Applying Lemma~\ref{Lemma1to2} with $m=2\frak{p}-1$ yields
\begin{align*}
\rho_2\big( \widehat{X}^{N,n}_e , \mathbf{\widehat{X}}^{N,n}_e  \big)\,\leq \,& 2^{\frac{\frak{p}}{2\frak{p}-1}}\Big(\rho_1\big( \widehat{X}^{N,n}_e , \mathbf{\widehat{X}}^{N,n}_e  \big)\Big)^{\frac{\frak{p}-1}{2\frak{p}-1}} \bigg(  \mathbb{E}\Big[\big|   \widehat{X}^{N,n}_e \big|^{2\frak{p}}\Big]^{\frac{1}{4\frak{p}-2}}\,+\,\mathbb{E}\Big[ \big|    \mathbf{\widehat{X}}^{N,n}_e   \big|^{2\frak{p}}\Big]^{\frac{1}{4\frak{p}-2}} \bigg)\,.
\intertext{By part (iv) of Lemma~\ref{LemBasicII}, the terms
$\mathbb{E}\Big[\big|   \widehat{X}^{N,n}_e \big|^{2\frak{p}}\Big]$ and $\mathbb{E}\Big[ \big|    \mathbf{\widehat{X}}^{N,n}_e   \big|^{2\frak{p}}\Big]$ are bounded by $\frac{\mathbf{c}}{N^{\frak{p}}}$.  Thus for $C':= 2^{\frac{\frak{4p-1}}{4\frak{p}-2} }C^{\frac{\frak{p}-1}{2\frak{p}-1}}\mathbf{c}^{ \frac{1}{4\frak{p}-2} }
$, we have the inequality}
\,\leq \,& C' \frac{\log^{ \frac{\frak{p}-1}{4\frak{p}-2} } (n+1)    }{ n^{ \alpha\frac{\frak{p}-1}{2\frak{p}-1}  } N^{\frac{ \frak{p} }{4\frak{p}-2  } } }   \,\leq\, C'\frac{ \log^{ \frac{\frak{p}-1}{4\frak{p}-2} } (n+1)  }{ n^{\alpha \frac{(\frak{p}-1)}{2\frak{p}-1} } n^{ \alpha \frac{ 2\frak{p}}{ 9(4\frak{p}-2)  }} }\,=\, C' 
\frac{ \log^{ \frac{\frak{p}-1}{4\frak{p}-2} } (n+1)  }{ n^{ \frac{5\alpha}{9}-\frac{ 4 \alpha}{9(2\frak{p}-1)  } }  }
 \,.
\end{align*}
The second inequality uses that $N=\lfloor n^{2\alpha/9}\rfloor$.
Note that the exponent $\frac{5\alpha}{9}-\frac{ 4\alpha }{9(2\frak{p}-1)} $ is strictly greater than $\frac{4\alpha}{9}+\upsilon$ since $\frak{p}:=\lceil \frac{2\alpha}{\alpha-9\upsilon}   \rceil  +1 $, and thus the above shows that the Wassertstein-$2$ distance between $\widehat{X}^{N,n}_e $  and  $\mathbf{\widehat{X}}^{N,n}_e $ is bounded by a multiple of $n^{-4\alpha/9-\upsilon}$.

\vspace{.5cm}

\noindent Part (iii): Since  $\xi''_{N,n}:= \frak{c}\big| \sigma_{\mathbf{n},n } - \sqrt{R(r-\mathbf{n})} \big|  $,  the inequality ($\textup{iii}'$) gives us that
\begin{align}\label{Ygo}
\rho_2\big(\mathbf{\widehat{X}}_e^{N,n}  ,\mathbf{\widetilde{X}}_e^{(N)} \big)\,\leq \,& \frac{ \frak{c} }{N^{\frac{\log b}{3}\frak{m}+\frac{1}{2}   }  }   \,+\,\frak{c}\Big| \sigma_{\mathbf{n},n } - \sqrt{R(r-\mathbf{n})} \Big|  \,.
  \end{align}
Since $\frak{m}:=\frac{21}{2\log b} $ and $N=\lfloor n^{2\alpha /9} \rfloor$, the first term on the right side of~(\ref{Ygo}) is bounded by a multiple of $n^{-8\alpha/9}$.  By Lemma~\ref{LemVarApprox}, we have the second inequality below:
\begin{align}
\left|  \sigma_{\mathbf{n},n }  - \sqrt{R(r-\mathbf{n})} \right|\,\leq \,\frac{\big|\sigma_{\mathbf{n},n}^2- R(r-\mathbf{n})\big|   }{ \sqrt{ R(r-\mathbf{n})}  }\,\leq\, & \frac{1}{ \sqrt{ R(r-\mathbf{n})} }\frac{ C_{\mathcal{I},\mathbf{v},\alpha}   }{n^{\alpha} } \,.\nonumber 
\intertext{Since $R(s)\sim -\frac{\kappa^2}{s}$ as $s\rightarrow -\infty$ by (II) of Lemma~\ref{LemVar} and the interval $\mathcal{I}$ is bounded, the supremized expression below is finite: }
\,\leq\, & \sqrt{ \mathbf{n} }\bigg(\sup_{r\in \mathcal{I}}\sup_{s\in [1,\infty)}\frac{ 1/\sqrt{s}  }{ \sqrt{ R(r-s)} }  \bigg)\frac{ C_{\mathcal{I},\mathbf{v},\alpha}  }{n^{\alpha}  }\,.
\end{align}
Since $ \mathbf{n}=N+\lfloor 2\frak{m}\log N\rfloor $ and $N=\lfloor n^{2\alpha/9}\rfloor$, the above is bounded by a multiple of $n^{-8\alpha/9}$.   
\end{proof}

\subsection{Proof of Lemma~\ref{LemBasic}}\label{SecLemmaLittle}

The  following is an analog of Proposition~\ref{PropUnif} that provides bounds for the  moments of the random variables in a $\mathcal{Q}$-pyramidic array $\big\{ X_a^{(k,n)} \big\}_{a\in  E_{k} }$ generated from an i.i.d.\ array $\big\{ X_h^{(n)} \big\}_{h\in  E_{n} }$ satisfying the conditions of Lemma~\ref{LemmaRe}.  The proof uses techniques from~\cite{Clark1} and is placed in Section~\ref{SecSharpHMC}.

\begin{proposition}\label{PropMomApprox} Fix $\mathbf{v},\varkappa>0$, $\alpha\in (0,1)$, $\frak{p}\in \{2,3,\ldots\}$, and a bounded interval $\mathcal{I}\subset \R$.  There exists a positive number  $C\equiv C(\mathcal{I},\mathbf{v},\varkappa, \alpha, \frak{p}) $ such that for any $r\in \mathcal{I}$, $n\in \mathbb{N}$, and i.i.d.\ array of centered random variables  $\big\{ X_h^{(n)} \big\}_{h\in  E_{n} }$ satisfying conditions (I)-(II) of Lemma~\ref{LemmaRe},  the inequality below holds for all $k\in \{0,1,\ldots, n\}$:
$$  \mathbb{E}\Big[ \big(X_a^{(k,n)}\big)^{2\frak{p}}  \Big] \ \,\leq \,\frac{C}{(k+1)^{\frak{p}}}  \,, $$
where $\big\{X_a^{(k,n)}\big\}_{a\in E_k}$ is the $k^{th}$ generation layer of the $\mathcal{Q}$-pyramidic array generated from $\big\{ X_h^{(n)} \big\}_{h\in  E_{n} }$.
\end{proposition}

\begin{proof}[Proof of Lemma~\ref{LemBasic}]  Part (i):  By Lemma~\ref{LemVarApprox}, there is a $C_{\mathcal{I},\mathbf{v},\alpha}>0$ such that $$\big|\sigma_{\mathbf{n},n}^2-R(r-\mathbf{n})\big|\,\leq\, \frac{C_{\mathcal{I},\mathbf{v},\alpha}}{n^{\alpha}}\,\leq \, \frac{C_{\mathcal{I},\mathbf{v},\alpha}}{N^{9/2}}$$ holds for any $r\in \mathcal{I}$, $n\in \mathbb{N}$, and i.i.d.\ array of random variables $\big\{ X_h^{(n)} \big\}_{h\in  E_{n} }$ with $n\geq \mathbf{n}$ and satisfying condition (I) of Lemma~\ref{LemmaRe}, where we have used that $N:=\lfloor n^{2\alpha/9}\rfloor   $ for the second inequality.  Since $R(s)\sim -\frac{\kappa^2}{s}$ when $-s\gg 1$ and $\mathbf{n}\sim N$ for $N\gg 1$, the supremum and infimum  of $R(r-\mathbf{n})$ for $r\in\mathcal{I}$ are respectively bounded from above and below by positive multiples $ C_{\mathcal{I}}$ and $c_{\mathcal{I}}$ of $\frac{1}{N}$:
$$ \frac{c_{\mathcal{I}}}{N}\,-\, \frac{C_{\mathcal{I},\mathbf{v},\alpha}}{N^{9/2}} \,\leq \, \inf_{r\in \mathcal{I}} R(r-\mathbf{n}) \,-\, \frac{C_{\mathcal{I},\mathbf{v},\alpha}}{N^{9/2}}\,\leq \,\sigma_{\mathbf{n},n}^2\,\leq \, \sup_{r\in \mathcal{I}} R(r-\mathbf{n}) \,+\, \frac{C_{\mathcal{I},\mathbf{v},\alpha}}{N^{9/2}}\,\leq\, \frac{C_{\mathcal{I}}}{N}\,+\, \frac{C_{\mathcal{I},\mathbf{v},\alpha}}{N^{9/2}}    \,. $$ 
Thus $\sigma^2_{\mathbf{n},n}$ is bounded from above by a constant multiple of $\frac{1}{N}$. When $N>\lambda:=\big(\frac{2C_{\mathcal{I},\mathbf{v},\alpha}  }{ c_{\mathcal{I}} }  \big)^{2/7}  $, then $\sigma^2_{\mathbf{n},n}$ is bounded from below by $\frac{c_{\mathcal{I}}}{2N}$.\vspace{.25cm}

 \noindent  Part (ii): Define  the polynomial $S(x):=M(x)-x$, in other terms, as $M$ with the linear term removed.  We can write  $\varsigma_{N,n}^2$ and $\varsigma_{N}^2$ in the forms below:
\begin{align}
\varsigma_{N,n}^2\,=\,&(\mathbf{n}-\mathbf{\widehat{n}})\left(M(\sigma_{k,n}^2)\,-\,\sigma_{k,n}^2   \right)\,=\,(\mathbf{n}-\mathbf{\widehat{n}})S( \sigma_{k,n}^2 )\,,\nonumber \\
\varsigma_{N}^2\,=\,&(\mathbf{n}-\mathbf{\widehat{n}})\left(M\big(R(r-\mathbf{n})\big) \,-\,R(r  -\mathbf{n}) \right)\,=\,(\mathbf{n}-\mathbf{\widehat{n}})S\big(R(r-\mathbf{n})\big)\,. \label{Trih}
\end{align}
The first equality on the top line above uses~(\ref{VarsigmaForm}) and that $\sigma_{k-1,n}^2=M(\sigma_{k,n}^2)$, and the first equality on the second line uses that $M\big(R(s)\big)=R(s+1)$ by part (I) of Lemma~\ref{LemVar}.  

We will first prove the bound for $\frac{1}{\varsigma_{N} }|\varsigma_{N,n}^2-\varsigma_{N}^2| $.   By the same reasoning as in~(\ref{Tip!}), there is a $\mathbf{C}\equiv \mathbf{C}(\mathcal{I},\mathbf{v},\alpha) $ such that the inequality below holds
\begin{align}
\big|\varsigma_{N,n}^2-\varsigma_{N}^2\big|\,=\,(\mathbf{n}-\mathbf{\widehat{n}})\big|S\big(\sigma_{\mathbf{n},n}^2\big)-S\big(R(r-\mathbf{n})\big)\big|\,\leq \,&\mathbf{C} (\mathbf{n}-\mathbf{\widehat{n}})\frac{R(r-\mathbf{n})}{ n^{\alpha}  }\,.\label{Tip}
\end{align}
From the relations~(\ref{Trih}) and~(\ref{Tip}), we get the first inequality below,
\begin{align*}
\frac{\big|\varsigma_{N,n}^2-\varsigma_{N}^2\big| }{\varsigma_{N} }\, \leq  \,&\mathbf{C}\frac{(\mathbf{n}-\mathbf{\widehat{n}})^{1/2} R(r-\mathbf{n})}{ n^{\alpha}  \sqrt{S\big(R(r-\mathbf{n})\big)}  } 
 \leq \,\mathbf{C}\frac{(\mathbf{n}-\mathbf{\widehat{n}})^{1/2}}{n^{\alpha}}\Bigg(\sup_{\substack{r\in \mathcal{I}\\ s\geq 1}  }\frac{ R(r-s) }{  \sqrt{S\big(R(r-s)\big)  }}\Bigg)  \,.
   \end{align*}
The supremum above is finite since the  lowest-order nonzero term in the polynomial $S$ is quadratic and $R(s)\sim -\frac{\kappa^2}{s}$ as $s\rightarrow -\infty$. The above shows that $\frac{1}{\varsigma_{N} }|\varsigma_{N,n}^2-\varsigma_{N}^2| $ is bounded by a constant multiple of $ n^{-\alpha} \log^{1/2} (n+1)  $  since $\mathbf{n}-\mathbf{\widehat{n}}\approx \frak{m}\log N \approx \frac{2\frak{m}}{9}\log n $.

Next we show that $\varsigma_{N,n}^2$ is bounded by a constant multiple of $\frac{\log (N+1)}{N^2}$.  By the triangle inequality and~(\ref{Tip}), we have that
\begin{align}\label{Pjr}
\varsigma_{N,n}^2\,\leq \, \varsigma_{N}^2\,+\,\big|\varsigma_{N,n}^2-\varsigma_{N}^2\big| \,\leq \,(\mathbf{n}-\mathbf{\widehat{n}})S\big(R(r-\mathbf{n})\big)\,+\,\mathbf{C} (\mathbf{n}-\mathbf{\widehat{n}})\frac{R(r-\mathbf{n})}{ n^{\alpha}  }\,.
\end{align}
 Note that $R(r  -\mathbf{n})\propto \frac{ \kappa^2 }{N } $ by (II) of Lemma~\ref{LemVar} since $\mathbf{n}\sim N$ as $N\gg 1$. Thus, since the lowest-order nonzero term of the polynomial $S$ is quadratic, the first term on the right side of~(\ref{Pjr})  is bounded by a constant multiple of $\frac{\log (N+1)  }{ N^2 }$.  The second term on the right side of~(\ref{Pjr}) is bounded by a constant multiple of $\frac{\log(N+1)   }{N^{11/2}  }$ because $n^{\alpha}\sim N^{9/2}$.  Thus $\varsigma_{N,n}^2$  has the stated bound. \vspace{.25cm}

\noindent Part (iii): The proof follows  through the same steps as   the proof of part (iii) of Lemma~\ref{LemBasic} with each application of Proposition~\ref{PropUnif} replaced by an application of Proposition~\ref{PropMomApprox}. \vspace{.25cm}

\noindent Part (iv): The bound for the $(2\frak{p})^{th}$ moment of $X_g^{(\mathbf{n},n)}$ follows from Proposition~\ref{PropMomApprox} since $\mathbf{n} \geq N $.  We will only prove the bound for the  $(2\frak{p})^{th}$ moment of $\widehat{X}^{N,n}_e $ since the analysis for  $ \mathbf{\widehat{X}}^{N,n}_e  $ is similar. By~(\ref{XApprx}) the random variable $\widehat{X}_{e}^{N,n}$ can be written in the form
\begin{align}
\widehat{X}_{e}^{N,n}\,=\,\underbrace{\mathcal{L}^{\mathbf{n}-N }\big\{ X_g^{(\mathbf{n},n)}\big\}_{g\in e\cap E_{\mathbf{n} }}}_{(\textbf{a})}\,+\,\underbrace{\sum_{\ell=1}^{\mathbf{n}-N} \mathcal{L}^{\mathbf{n}-N-\ell}\mathcal{E}\mathcal{L}^{\ell-1}\big\{ X_g^{(\mathbf{n},n)}\big\}_{g\in e\cap E_{\mathbf{n} }}}_{(\textbf{b})}\,.
\end{align}
It suffices to bound the $(2\frak{p})^{th}$ moment of each of the terms (a) and (b)   by a multiple of $N^{-\frak{p}}$.\vspace{.3cm}

\noindent \textbf{(a):} Fix $k\in \{0,\ldots, \mathbf{n}-N\}$, and let $\mathbf{a}\in E_{\mathbf{n}-k}$. Since $\mathcal{L}^{k}\big\{ X_g^{(\mathbf{n},n)}\big\}_{g\in \mathbf{a}\cap E_{\mathbf{n} }}$ is an i.i.d.\ sum of random variables $\frac{1}{b^{  k}}X_g^{(\mathbf{n},n)}$ indexed by $g\in \mathbf{a}\cap E_{\mathbf{n} }$, the Marcinkiewicz-Zygmund inequality gives us the first inequality below for some universal constant $B_{\frak{p}}>0$.
\begin{align}
\mathbb{E}\bigg[ \Big(  \mathcal{L}^{k}\big\{ X_g^{(\mathbf{n},n)}\big\}_{g\in \mathbf{a}\cap E_{\mathbf{n} }} \Big)^{2\frak{p} }\bigg]\,\leq\,& B_{\frak{p}}\mathbb{E}\left[ \Bigg(  \frac{1}{b^{2k} } \sum_{g\in \mathbf{a}\cap E_{\mathbf{n}}  }  \big(X_g^{(\mathbf{n},n)}\big)^2 \Bigg)^{\frak{p} }\right]\nonumber 
\intertext{Applications of Jensen's inequality over the sum $ \frac{1}{b^{2k} } \sum_{g\in \mathbf{a}\cap E_{\mathbf{n}}  }  $ and  Proposition~\ref{PropMomApprox} yield the first two inequalities below for any representative $g\in \mathbf{a}\cap E_{\mathbf{n}}  $.}
\,\leq\,&  B_{\frak{p}}\mathbb{E}\left[  \big(X_g^{(\mathbf{n},n)}\big)^{2\frak{p}} \right]\,
\leq\,\frac{CB_{\frak{p}}}{(\mathbf{n}+1)^{\frak{p}} }\,
\leq\,\frac{CB_{\frak{p}}}{N^{\frak{p}} }\label{MarcZyg}
\end{align}
The third inequality uses that $\mathbf{n}\geq N$.
Applying the inequality above with $k=\mathbf{n}-N$ yields the sought-after bound for the $(2\frak{p})^{th}$ moment of (a).\vspace{.3cm}

\noindent \textbf{(b):} For $\ell\in \{1,\ldots, \mathbf{n}-N\}$, define the array $\big\{ Y^{\ell,\mathbf{n},n}_{a}\big\}_{a\in e\cap E_{\mathbf{n}-\ell } }:=\mathcal{E}\mathcal{L}^{\ell-1}\big\{ X_g^{(\mathbf{n},n)}\big\}_{g\in e\cap E_{\mathbf{n} }}  $.  By the triangle inequality,
\begin{align}
\mathbb{E}\Bigg[ \bigg( \sum_{\ell=1}^{\mathbf{n}-N} \mathcal{L}^{\mathbf{n}-N-\ell}\big\{ Y^{\ell,\mathbf{n},n}_{a}\big\}_{a\in e\cap E_{\mathbf{n}-\ell  } } \bigg)^{2\frak{p}}  \Bigg]\,\leq \,&\left(\sum_{\ell=1}^{\mathbf{n}-N}  \mathbb{E}\left[ \Big( \mathcal{L}^{\mathbf{n}-N-\ell}\big\{ Y^{\ell,\mathbf{n},n}_{a}\big\}_{a\in e\cap E_{\mathbf{n}-\ell  } } \Big)^{2\frak{p}}  \right]^{\frac{1}{2\frak{p}}} \right)^{2\frak{p}}\\
\,\leq \,&\big(\mathbf{n}-N\big)^{2\frak{p}} \max_{1\leq \ell\leq \mathbf{n}-N} \mathbb{E}\left[ \Big(\mathcal{L}^{\mathbf{n}-N-\ell}\big\{ Y^{\ell,\mathbf{n},n}_{a}\big\}_{a\in e\cap E_{\mathbf{n}-\ell  } }\Big)^{2\frak{p}}  \right]\,.\nonumber 
\end{align}
We will show that the maximum above is bounded by a multiple of $N^{-2\frak{p}} $, which suffices to show that the $(2\frak{p})^{th}$ moment of (b) has order $N^{-\frak{p}} $ since $\mathbf{n}-N\approx \frak{m}\log N$.  Applying the Marcinkiewicz-Zygmund and Jensen inequalities as in~(\ref{MarcZyg}) yields the following inequality for a representative $a\in e\cap E_{\mathbf{n}-\ell  }  $:
\begin{align}
 \mathbb{E}\bigg[ \Big(\mathcal{L}^{\mathbf{n}-N-\ell}\big\{ Y^{\ell,\mathbf{n},n}_{a}\big\}_{a\in e\cap E_{\mathbf{n}-\ell  } }\Big)^{2\frak{p}}  \bigg]\,\leq \, & B_{\frak{p}}\mathbb{E}\left[  \big(Y^{\ell,\mathbf{n},n}_{a}\big)^{2\frak{p}} \right]\,.\nonumber\intertext{Define $\big\{X^{\ell,\mathbf{n},n}_{\mathbf{a}}\big\}_{\mathbf{a}\in e\cap E_{\mathbf{n}-\ell+1  }  }:= \mathcal{L}^{\ell-1}\big\{ X_g^{(\mathbf{n},n)}\big\}_{g\in e\cap E_{\mathbf{n} }}   $.
 Since $Y^{\ell,\mathbf{n},n}_{a}=(\mathcal{Q}-\mathcal{L})\big\{ X^{\ell,\mathbf{n},n}_{a\times (i,j)} \big\}_{i,j\in\{1,\ldots,b\}}  $ is a multilinear polynomial of the centered random variables   $X^{\ell,\mathbf{n},n}_{a\times (i,j)} $ with no constant or linear terms,   there is a polynomial $T_{\frak{p}}\big(y_2,\ldots, y_{2\frak{p}}\big)$  such that the above is equal to   } \,=\,& B_{\frak{p}}T_{\frak{p}}\left( \mathbb{E}\bigg[ \Big(  \mathcal{L}^{\ell-1}\big\{ X_g^{(\mathbf{n},n)}\big\}_{g\in \mathbf{a}\cap E_{\mathbf{n} }} \Big)^{j}\bigg]; \, 2\leq j \leq  2\frak{p}  \right)\,,\label{Dar}
\end{align}
where $T_{\frak{p}}(y_2,\ldots, y_{2\frak{p}})$ is a linear combination of monomials $y_{j_1}y_{j_2}\cdots y_{j_m}$ with $j_1+\cdots + j_m\geq 4\frak{p}$.  It follows that~(\ref{Dar}) is bounded by a multiple of $N^{-2\frak{p}}$ for all $\ell\in \{1,\ldots, \mathbf{n}-N\} $ since an application of Jensen's inequality and~(\ref{MarcZyg}) yields
$$\Bigg|\mathbb{E}\bigg[ \Big(  \mathcal{L}^{\ell-1}\big\{ X_g^{(\mathbf{n},n)}\big\}_{g\in \mathbf{a}\cap E_{\mathbf{n} }} \Big)^{j}\bigg]\Bigg|\,\leq \,\mathbb{E}\bigg[ \Big(  \mathcal{L}^{\ell-1}\big\{ X_g^{(\mathbf{n},n)}\big\}_{g\in \mathbf{a}\cap E_{\mathbf{n} }} \Big)^{2\frak{p}}\bigg]^{\frac{j}{2\frak{p}}}  \,\leq \,\Big(  \frac{CB_{\frak{p}}}{N^{\frak{p}} } \Big)^{\frac{j}{2\frak{p}}}\,= \, \frac{(CB_{\frak{p}})^{\frac{j}{2\frak{p}}}   }{N^{\frac{j}{2}} } \,. $$
\end{proof}

\section{The site-disorder model}\label{SecSiteDisorder}

The goal of this section is to prove Theorem~\ref{ThmMainSite}.  As mentioned in Remark~\ref{RemarkProofPlan}, the proof involves showing that $\widehat{W}_n^{\omega}\big(\widehat{\beta}_{n,r}\big)$ has a vanishing $L^2$ distance from a reduced partition function, $\widetilde{W}_n^{\omega}\big(\widehat{\beta}_{n,r}\big)$, for which the disorder variables corresponding to vertices of generation less than $  \log n $ have been integrated out (Lemma~\ref{LemCond}).  Moreover,  $\widetilde{W}_n^{\omega}\big(\widehat{\beta}_{n,r}\big)$ is the peak of a  $\mathcal{Q}$-pyramidic array of random variables with  $\lfloor \log n\rfloor$ layers  (Proposition~\ref{PropReduce}).   Lemmas~\ref{LemMtilde} \&~\ref{LemmaHM} respectively verify the  conditions (II) and (III) in Definition~\ref{DefRegular} for the large-$n$ behavior of the variance and higher moments of the random variables in the base layer of the $\mathcal{Q}$-pyramidic array.  We can then apply Theorem~\ref{ThmUnique} to conclude that  $\widetilde{W}_n^{\omega}\big(\widehat{\beta}_{n,r}\big)$---and consequently also  $\widehat{W}_n^{\omega}\big(\widehat{\beta}_{n,r}\big)$---converges in distribution  to $\mathbf{W}_r$ as $n\rightarrow \infty$.

\subsection{Proof of Theorem~\ref{ThmMainSite} }\label{SecProofThmMainSiteOutline}

We will prove Theorem~\ref{ThmMainSite} after stating the technical lemmas used in its proof. The proofs of the lemmas are placed in the next four subsections. 

 Recall that $V_{n-1}$ is canonically identifiable with a subset of $V_n$ and that under this identification $V_n\backslash V_{n-1}$ is referred to as the set of \textit{generation-$n$ vertices}. Thus, for $k\leq n$, the set $V_n\backslash V_k$ is  all vertices on the diamond graph $D_n$ of generation greater than $k$.  The elementary proposition below, whose proof is in Section~\ref{SecReducePart}, states that the conditional expectation of the site-disorder partition function $\widehat{W}_{n}^{\omega}(\beta)$ with respect to the $\sigma$-algebra generated by $\omega_a$ for $a\in V_n\backslash V_k$ can be expressed in terms of the array map $\mathcal{Q}$. 

\begin{proposition}\label{PropReduce} Let $k,n\in \mathbb{N}_0$, and assume $k \leq  n$.  Define the $\sigma$-algebra $\mathcal{F}_{n}^{k}:= \sigma\big\{\omega_{a}\,\big|\,a\in V_n\backslash V_k   \big\}$.  The conditional expectation of $\widehat{W}_{n}^{\omega}(\beta)$ with respect to $\mathcal{F}_{n}^{k}$ can be written in the form
$$
\mathbb{E}\Big[\widehat{W}_{n}^{\omega}(\beta)\,\Big|\, \mathcal{F}_{n}^{k}  \Big]\,=\,1\,+\,\mathcal{Q}^{k}\big\{ X_h(\beta) \big\}_{h\in E_{k}  }\,,  $$ 
where $\big\{ X_h(\beta) \big\}_{h\in E_{k}  }$ is an array of independent copies of 
$\widehat{W}_{n-k}^{\omega}(\beta)\,-\,1$.
\end{proposition}

Lemma~\ref{LemCond} states that  the partition function $\widehat{W}_n^{\omega}\big(\widehat{\beta}_{n,r} \big)$ is not changed much by integrating out the disorder variables labeled by vertices of generation less than $\log n$ when $n$ is large. The proof is in Section~\ref{SecLemCond}.
\begin{lemma}\label{LemCond} For fixed $r\in \R$, let the sequence  $\{\widehat{\beta}_{n,r}\}_{n\in \mathbb{N}}$ have the large $n$ asymptotics~(\ref{BetaForm2}).  The $L^2$ distance between $\widehat{W}_{n}^{\omega}(\widehat{\beta}_{n,r})$ and $\widetilde{W}_{n}^{\omega}(\widehat{\beta}_{n,r}):=\mathbb{E}\big[\widehat{W}_{n}^{\omega}(\widehat{\beta}_{n,r})\,\big|\, \mathcal{F}_{n}^{\lfloor \log n\rfloor}  \big]$ vanishes as $n\rightarrow \infty$. 
\end{lemma}

It follows from Proposition~\ref{PropReduce} and Lemma~\ref{LemCond} that the $L^2$ distance between $\widehat{W}_{n}^{\omega}\big(\widehat{\beta}_{n,r}\big)$ and $1+\mathcal{Q}^{\lfloor \log n \rfloor}\big\{ X_h(\widehat{\beta}_{n,r} )\big\}_{h\in E_{\lfloor \log n\rfloor}  }$ converges to zero as $n\rightarrow \infty$, where $\big\{ X_h(\widehat{\beta}_{n,r} )\big\}_{h\in E_{\lfloor \log n \rfloor}  } $ is an array of independent copies of $\widehat{W}_{n-\lfloor \log n\rfloor }^{\omega}\big(\widehat{\beta}_{n,r} \big)-1$.  The following lemma verifies the variance asymptotics  in condition (II) of Definition~\ref{DefRegular}---with $n$ replaced by $\lfloor \log n\rfloor$---for the sequence in $n\in \mathbb{N}$ of $\mathcal{Q}$-pyramidic arrays generated from the edge-labeled arrays  $\big\{X_h(\widehat{\beta}_{n,r} ) \big\}_{h\in E_{\lfloor \log n \rfloor  }}$.   Our proof, which is in Section~\ref{SecVarAnal}, refines an argument from the proof of~\cite[Lemma 5.16]{US}.
\begin{lemma}\label{LemMtilde} The variance of  $\widehat{W}_{n-\lfloor \log n\rfloor}^{\omega}\big(\widehat{\beta}_{n,r}\big)$ has the large $n$ asymptotics 
\begin{align}\label{Unflat}
\textup{Var}\Big( \widehat{W}_{n-\lfloor \log n\rfloor}^{\omega}\big(\widehat{\beta}_{n,r}\big) \Big)\,=\,\kappa^2\left(\frac{ 1  }{\lfloor \log n\rfloor   }\,+\,\frac{ \eta\log\lfloor \log n\rfloor }{  \lfloor \log n\rfloor^2 }\,+\, \frac{ r  }{ \lfloor \log n\rfloor^2  }\right) \,+\, \mathit{o}\left(\frac{1}{\log^2 n}  \right)\,.
\end{align}
\end{lemma}

Lemma~\ref{LemmaHM} verifies the vanishing higher moment condition (III) of Definition~\ref{DefRegular} for random variables in the array  $\big\{ X_h\big(\widehat{\beta}_{n,r}\big)\big\}_{h\in E_{\lfloor \log n \rfloor }} $. 
The proof is in Section~\ref{SecLemmaHM}.
\begin{lemma}\label{LemmaHM} For each $m\in \mathbb{N}$, the $m^{th}$ centered moment of  $\widehat{W}_{n-\lfloor \log n\rfloor}^{\omega}\big(\widehat{\beta}_{n,r}\big)$ vanishes as $n\rightarrow \infty$. 
\end{lemma}

\vspace{.2cm}

\begin{proof}[Proof of Theorem~\ref{ThmMainSite}] For  $\big\{X_{h}\big(\widehat{\beta}_{n,r}\big)\big\}_{h\in E_{\lfloor \log n\rfloor  }} $ defined as in Proposition~\ref{PropReduce}, the $L^2$ distance between the generation-$n$ vertex-disorder partition function $\widehat{W}_{n}^{\omega}\big(\widehat{\beta}_{n,r}\big)$ and the effectively generation-$\lfloor \log n\rfloor$ edge-disorder partition function given by
$$\widetilde{W}_{n}^{\omega}\big(\widehat{\beta}_{n,r}\big)\,:=\, \mathbb{E}\Big[ \widehat{W}_{n}^{\omega}\big(\widehat{\beta}_{n,r}\big)  \,\Big|\,\mathcal{F}_{n}^{\lfloor \log n \rfloor}  \Big]  \,=\,1\,+\,\mathcal{Q}^{\lfloor \log n\rfloor}\big\{ X_h\big(\widehat{\beta}_{n,r}\big) \big\}_{h\in E_{\lfloor \log n\rfloor}  }  $$ 
vanishes with large $n$ by Lemma~\ref{LemCond}, where the second equality above holds by Proposition~\ref{PropReduce}.  In particular, the Wasserstein-$2$ distance between $\widehat{W}_{n}^{\omega}\big(\widehat{\beta}_{n,r}\big)-1$ and $ \mathcal{Q}^{\lfloor \log n\rfloor}\big\{ X_h\big(\widehat{\beta}_{n,r}\big) \big\}_{h\in E_{\lfloor \log n\rfloor}  } $ vanishes as $n\rightarrow \infty$.  Thus it suffices to prove that the Wasserstein-$2$ distance between $ \mathcal{Q}^{\lfloor \log n\rfloor}\big\{ X_h\big(\widehat{\beta}_{n,r}\big) \big\}_{h\in E_{\lfloor \log n\rfloor}  } $ and $\mathbf{X}_r \stackrel{d}{=}\mathbf{W}_{r}-1$ converges to zero with large $n$.

 Notice that the statements (I)-(III) below hold.
\begin{enumerate}[(I)]
\item By Proposition~\ref{PropReduce}, the random variables in the array $\big\{ X_h\big(\widehat{\beta}_{n,r}\big) \big\}_{h\in E_{\lfloor \log n\rfloor}}$ are independent copies of $\widehat{W}_{n-\lfloor \log n\rfloor }^{\omega}\big(\widehat{\beta}_{n,r}\big)-1$.

\item By Lemma~\ref{LemMtilde} the variance of the random variable $\widehat{W}_{n-\lfloor \log n\rfloor }^{\omega}\big(\widehat{\beta}_{n,r}\big)$ has the large $n$ asymptotics 
$$ \textup{Var}\Big( \widehat{W}_{n-\lfloor \log n\rfloor }^{\omega}\big(\widehat{\beta}_{n,r}\big) \Big)\,=\,\kappa^2\left(\frac{1}{\lfloor \log n\rfloor}\,+\,\frac{\eta \log \lfloor \log n\rfloor    }{ \lfloor \log n\rfloor^2 } \,+\,\frac{ r  }{ \lfloor \log n\rfloor^2 } \right) \,+\,\mathit{o}\left( \frac{1}{ \log^2 n  }  \right)  \,.    $$

\item By Lemma~\ref{LemmaHM},  the $m^{th}$ centered moment of $\widehat{W}_{n-\lfloor \log n\rfloor }^{\omega}\big(\widehat{\beta}_{n,r}\big)$ vanishes as $n\rightarrow \infty$ for each $m\in \{4, 6,\ldots\}$.

\end{enumerate}
Statements (I)-(III) imply that the sequence in $n\in \mathbb{N}$ of edge-labeled arrays $\big\{ X_{h}\big(\widehat{\beta}_{n,r}\big)\big\}_{h\in E_{\lfloor \log n\rfloor}}$
satisfies the  conditions (I)-(III) in Definition~\ref{DefRegular}.  Thus, by Theorem~\ref{ThmUnique}, the Wasserstein-$2$ distance between $\mathbf{X}_r$ and  $ \mathcal{Q}^{\lfloor \log n\rfloor}\big\{ X_h\big(\widehat{\beta}_{n,r}\big) \big\}_{h\in E_{\lfloor \log n\rfloor}  } $ vanishes with large $n$.\footnote{Although the  definition of a ``regular" sequence of $\mathcal{Q}$-pyramidic arrays formulated in  Definition~\ref{DefRegular} assumes that the generation, $\frak{g}_n \in \mathbb{N}$, of the bottom layer of the $n^{th}$ $\mathcal{Q}$-pyramidic array is $\frak{g}_n=n$, the conclusions of Theorem~\ref{ThmUnique} remain valid when $(\frak{g}_n)_{n\in \mathbb{N}}$ is any sequence that  diverges to $\infty$, such as $\frak{g}_n=\lfloor \log n\rfloor $.}  Therefore, $\widehat{W}_{n}^{\omega}\big(\widehat{\beta}_{n,r}\big)$ converges in law to $\mathbf{W}_r$ as $n\rightarrow \infty$.
\end{proof}

\subsection{Proof of Proposition~\ref{PropReduce}}\label{SecReducePart}

As a preliminary, we will extend our observations and notations relating   to the structure of the diamond hierarchical graphs. For  $k \leq n$  recall that $V_n\backslash V_k$ is the set of vertices on the diamond graph $D_n$ of generation greater than $k$.  
\begin{enumerate}[(I)]

\item From the construction of the sequence of diamond graphs outlined in Section~\ref{SecDHG}, we can see that $D_n$ has $b^{2k}$  embedded copies of $D_{n-k}$, which are in canonical one-to-one correspondence with elements of $E_{k}$.  The  vertices in $V_k$---viewed as a subset of $V_n$---are  roots of the embedded copies of  $D_{n-k}$, and the remaining vertices in  $V_n\backslash V_k$ are internal (non root) to the embedded copies of  $D_{n-k}$.  We denote that set of internal vertices on the copy of $D_{n-k}$ associated with $h\in E_k$ by $h\cap V_n$.\footnote{This abuse of notation  is similar to our previous use of $h\cap E_n$ to denote a subset of $E_n$.} The collection $\{ h\cap V_n\,| \, h\in E_k  \}$ is a partition of the set $V_n\backslash V_k$.

\item For $h\in E_k$, let $\Gamma_{n}^{h}$ denote the set of functions $\mathbf{q} :\{1,\ldots, b^{n-k}\}\rightarrow h\cap E_n$ that are directed paths crossing the embedded copy of $D_{n-k}$ corresponding to $h$.  Thus each $\Gamma_{n}^{h}$ is a copy of $\Gamma_{n-k}$.

\item  For $a\in V_n $ and  $\mathbf{q}\in \Gamma_{n}^{h}$, we write $a\pmb{\in}  \mathbf{q}$ when $a$ sits internally (non endpoint) along the path $\mathbf{q}$, i.e., when $\mathbf{q}(j)\in h\cap E_n$ is incident to $a$ for some $j\in \{2,\dots, b^{n-k}-1\}$.  A vertex $a\in V_n$ is an element of $ V_n\backslash V_k$ if and only if there is an $h\in E_k$ and a $\mathbf{q} \in \Gamma_{n}^{h}  $ such that $a\pmb{\in}  \mathbf{q}$.\footnote{This  is equivalent to the remark in (I) that $a\in  V_n\backslash V_k$ iff $a$ is an internal vertex to one of the subcopies of $D_{n-k}$.}

\item There is a canonical  one-to-one correspondence between $\Gamma_n$ and the union of $b^k$-fold product sets given by $\bigcup_{q\in \Gamma_k} \prod_{\ell=1}^{b^k}\Gamma_n^{q(\ell)}$.  In this association, each $p\in \Gamma_n$ has a generation-$k$ coarse-graining $q\in \Gamma_k$ and the component $\mathbf{q}_{\ell}\in \Gamma_n^{q(\ell)}$  in the $b^k$-tuple  $(\mathbf{q}_1,\ldots, \mathbf{q}_{b^k})$ is the trajectory of $p$ through the embedded copy of $D_{n-k}$ corresponding to  $q(\ell)\in E_k $.

\end{enumerate}

The following defines a restricted partition function $\widehat{W}_n^{(h)}(\beta)$ for the embedded copy of $D_{n-k}$ within $D_n$ that corresponds to $h\in E_k$.  
\begin{definition}\label{DefMicroPart} Let $k,n\in \mathbb{N}_0$, and assume $k \leq  n$.  For $h\in E_k$, define the random variable 
$$ \widehat{W}_n^{h}(\beta)\,:=\,   \frac{1}{|\Gamma_{n-k}|} \sum_{\mathbf{q} \in  \Gamma_{n}^{h}  }    \prod_{ \substack{ a\pmb{\in}  \mathbf{q}  } } \frac{ e^{\beta\omega_a}}{\mathbb{E}[ e^{\beta\omega_a}] }  \, , $$
where the set $\Gamma_{n}^{h} $ and  the relation $ \pmb{\in}  $ are defined as in (II) and (III) above, respectively.
\end{definition}
\begin{remark}\label{RemarkSubPart} The random variable $\widehat{W}_n^{(h)}(\beta)$ in Definition~\ref{DefMicroPart} is equal in distribution to  $\widehat{W}_{n-k}(\beta)$.  
\end{remark}

\begin{proof}[Proof of Proposition~\ref{PropReduce}]  Taking the conditional expectation of $\widehat{W}_{n}^{\omega}(\beta)$ with respect to $\mathcal{F}_{n}^{k}$ is equivalent to integrating out the variables $\omega_a$ with $a\in V_k$:
\begin{align*}
 \mathbb{E}\Big[\widehat{W}_{n}^{\omega}(\beta)\,\Big|\, \mathcal{F}_{n}^{k}  \Big]\,=\,&\frac{1}{ |\Gamma_{n}|}\sum_{ p\in \Gamma_{n  }}\prod_{\substack{a\pmb{\in}  p \\ a\in V_n\backslash V_{k  }  }} \frac{ e^{\beta\omega_a}}{\mathbb{E}[ e^{\beta\omega_a}] } \,.  
\intertext{By (III), a vertex $a\in V_n$ is in $ V_n\backslash V_{k  }$  iff there is an $h\in E_k$ and a  $  \mathbf{q} \in \Gamma_n^h  $ such that $a \pmb{\in}  \mathbf{q} $.  Through the one-to-one correspondence between $\Gamma_n$ and $\bigcup_{q\in \Gamma_k} \prod_{\ell=1}^{b^k}\Gamma_n^{q(\ell)}$, the above can be written in the form}
\,=\,&\frac{1}{ |\Gamma_{k}|} \sum_{ q\in \Gamma_{k  } } \frac{1}{|\Gamma_{n-k}|^{b^k}}\sum_{(\mathbf{q}_1,\ldots, \mathbf{q}_{b^k}  )\in \prod_{\ell=1}^{b^k} \Gamma_n^{ q(\ell)  }   }   \prod_{\ell=1}^{b^k}  \prod_{a\pmb{\in}  \mathbf{q}_{\ell} } \frac{ e^{\beta\omega_a}}{\mathbb{E}[ e^{\beta\omega_a}] } \,. 
\intertext{The outer summand factors as       }
\,=\,&\frac{1}{ |\Gamma_{k}|} \sum_{ q\in \Gamma_{k  } } \prod_{\ell=1}^{b^k}\Bigg( \frac{1}{|\Gamma_{n-k}|} \sum_{\mathbf{q}_{\ell} \in \Gamma_{n-k}^{q(\ell)}  }    \prod_{a\pmb{\in} \mathbf{q}_{\ell} } \frac{ e^{\beta\omega_a}}{\mathbb{E}[ e^{\beta\omega_a}] } \Bigg)\,.
\intertext{The expression in brackets has the form of the random variable $\widehat{W}_{n}^{h }(\beta)$ from Definition~\ref{DefMicroPart} with $h=q(\ell)$, and thus the above is equal to}
\,=\,&\frac{1}{ |\Gamma_{k}|} \sum_{ q\in \Gamma_{k  } } \prod_{\ell=1}^{b^k}\widehat{W}_{n}^{q(\ell)  }(\beta)\,=\,1\,+\,\mathcal{Q}^{k}\Big\{ \widehat{W}_{n}^{h}(\beta)-1 \Big\}_{h\in E_{k}  }\,.
\end{align*}
The last equality is equivalent to what we proved in Proposition~\ref{PropPartition}.
\end{proof}

\subsection{Proof of Lemma~\ref{LemMtilde}}\label{SecVarAnal}

For $k\in \mathbb{N}_{0}$ and $\beta>0$, let  $\hat{\varrho}_{k}(\beta)$ denote the variance of the partition function $\widehat{W}_{k}(\beta)$.  As a consequence of the 
distributional identity~(\ref{PartHierSymmII}), the sequence  of variances $\big\{\hat{\varrho}_{k}(\beta)\big\}_{k\in \mathbb{N}_0}$ satisfies the recursive equation
\begin{align}\label{RecEqVar}
   \hat{\varrho}_{k+1}(\beta)\,=\,\widehat{M}_{V}\big(\hat{\varrho}_{k}(\beta)\big)  \hspace{1cm}\text{with}\hspace{1cm}\hat{\varrho}_{0}(\beta)\,=\,0\,,
\end{align}
where the map $\widehat{M}_{V}:[0,\infty)\rightarrow [0,\infty)$ is defined by
\begin{align}\label{DefMHat}
   \widehat{M}_{V}(x) \,:=\,\frac{1}{b}\Big[ (1+x)^b\big(1+ V \big)^{b-1}    \,-\,1   \Big] \hspace{.5cm}\text{for}\hspace{.5cm} V\,:=\,\textup{Var}\bigg( \frac{e^{\beta \omega}}{\mathbb{E}[ e^{\beta \omega} ]}  \bigg)     \,.
\end{align}
Of course, $\widehat{M}_{V}$ reduces to the map $M(x)=\frac{1}{b}\big[(1+x)^b-1\big]$ when $V=0$.

 The inverse temperature scaling~(\ref{BetaForm2}) results in the following  variance scaling:\footnote{A short computation at the end of Appendix~\ref{AppendBetaScale} verifies~(\ref{VarScaling2}) starting from~(\ref{BetaForm2}).}  
\begin{align}\label{VarScaling2}
 V_{n,r}\,:=\, \textup{Var}\left( \frac{ e^{\widehat{\beta}_{n,r}\omega}}{\mathbb{E}[ e^{\widehat{\beta}_{n,r}\omega}] } \right)\,=\, \widehat{\kappa}^2\bigg( \frac{1 }{n^2}    \,+\,\frac{2\eta\log n }{n^3}\,+\, \frac{2r }{n^3} \bigg)+\,\mathit{o}\Big( \frac{1}{n^3}  \Big)\,. 
 \end{align}
It will be convenient to write $ V_{n,r}$ in the form $ V_{n,r}= \frac{ \widehat{\kappa}^2}{\mathbf{n}_{n,r}^2} = \frac{b}{b-1}\frac{\pi^2 \kappa^2}{4\mathbf{n}_{n,r}^2}$ for  $\mathbf{n}_{n,r}:=  \frac{\pi \kappa}{2}\big(\frac{b}{b-1}\big)^{1/2}V_{n,r}^{-1/2}$, which has the large $n$ asymptotics
\begin{align}\label{NDef}
 \mathbf{n}_{n,r}\,=\,n-\eta\log n -r+\mathit{o}(1) \,. 
 \end{align}

\vspace{.1cm}

\begin{proof}[Proof of Lemma~\ref{LemMtilde}] We separate the proof into parts (a)-(h).\vspace{.15cm}

\noindent \textbf{(a) An approximation for the variance map:} Since the variance $\hat{\varrho}_k\big(\widehat{\beta}_{n,r}  \big)$ of $\widehat{W}_{k}\big(\widehat{\beta}_{n,r}\big)$ satisfies the recursive equation~(\ref{RecEqVar}) in $k\in \mathbb{N}_{0}$, we have that
\begin{align}
\textup{Var}\Big( \widehat{W}_{n-\lfloor \log n\rfloor}\big(\widehat{\beta}_{n,r}\big) \Big)\,=\, \widehat{M}_{n,r}^{n-\lfloor\log n\rfloor }(0)\,.\nonumber 
\end{align}
Let $\widetilde{M}_{n,r}:[0,\infty)\rightarrow [0,\infty)$ be defined  through an approximation of the expression for $\widehat{M}_{n,r}(x)$ in~(\ref{DefMHat}) around $(x,V_{n,r})=(0,0)$ that is third-order in $x$ and first-order in $V_{n,r}$:
\begin{align*}
        \widetilde{M}_{n,r}(x)\,:=\,&x\,+\,\frac{b-1}{2}x^2\,+\,\frac{(b-1)(b-2)}{6}x^3  \, +\,\frac{b-1}{b}V_{n,r} \,, \nonumber  
   \intertext{which we can rewrite in terms of $\kappa^2:=\frac{2}{b-1}$, $\eta :=\frac{b+1}{3(b-1)}$, and $\mathbf{n}_{n,r}:=\frac{\pi\kappa}{2}(\frac{b}{b-1})^{1/2}V_{n,r}^{-1/2}  $ as   }
         \,=\,&x\,+\,\frac{x^2}{\kappa^2}\,+\,(1-\eta)\frac{x^3}{\kappa^4}\,+\,\frac{\pi^2 \kappa^2  }{4 \mathbf{n}_{n,r}^2 }\,.
\end{align*}
Define $\mathscr{E}(x, \mathbf{n}_{n,r}  ):= \widehat{M}_{n,r}(x)-\widetilde{M}_{n,r}(x) $, in other terms, the error of the approximation of $\widehat{M}_{n,r}$ by $\widetilde{M}_{n,r}$.  The error term has the  bound below for   some $\mathbf{c}>0$ and all $n\in \mathbb{N}$ and  $0\leq x \leq 1$:
\begin{align}\label{ErrorTerm}
\mathscr{E}(x, \mathbf{n}_{n,r}  )\,\leq\,\mathbf{c}\Big( x^4 \,+\,\mathbf{n}_{n,r}^{-8/3} \Big)\,.
\end{align}
The above inequality follows by foiling the expression~(\ref{DefMHat}) in $x$ \& $V$ and then applying Young's inequality to the cross-terms, of which the lowest-order cross-term is $xV_{n,r}\propto x/\mathbf{n}_{n,r}^2$. \vspace{.25cm}

\noindent \textbf{(b) Transforming the variables:}
For $r\in \R$ and $n\in \mathbb{N}$, define the sequence $\big\{ \mathbf{r}_k^{(n,r)} \big\}_{k\in \mathbb{N}_0}$ of  numbers  in the interval $[0,1)$ as 
\begin{align}\label{FormR}
\mathbf{r}_k^{(n,r)}\,:=\, \frac{2}{\pi}\tan^{-1}\bigg(\frac{2\mathbf{n}_{n,r}}{\pi \kappa^2}\widehat{M}_{n,r}^k(0)\bigg)\,,\hspace{.3cm} \text{so that we have  }\hspace{.3cm} \frac{\pi \kappa^2}{2}\tan\Big( \frac{\pi}{2}\mathbf{r}_k^{(n,r)}\Big)\,=\,\mathbf{n}_{n,r} \widehat{M}_{n,r}^k(0)\,.
\end{align}
Note that $\mathbf{r}_0^{(n,r)}=0$ since $\widehat{M}_{n,r}^0(0)=0$.  For notational neatness, we will identify $\mathbf{r}_k^{(n,r)}\equiv \mathbf{r}_k$, i.e., suppress the dependence on the superscript variables.  The sequence $\big\{ \mathbf{r}_k^{(n,r)} \big\}_{k\in \mathbb{N}_0}$ converges monotonically to $1$ as $k\rightarrow \infty$, and it will suffice for us to show that 
\begin{align}\label{Flattened}
1\,-\,\mathbf{r}_{n-\lfloor\log n\rfloor}\,=\,\frac{\lfloor \log n\rfloor -\eta\log\log n\,-\,r    }{n}\,+\,\mathit{o}\Big(\frac{1}{n}\Big)\,.
\end{align}
To see the equivalence between~(\ref{Flattened}) and~(\ref{Unflat}), note that for large $n$---and thus small $1-\mathbf{r}_{n-\lfloor\log n \rfloor }$---we  get the second equality below  through second-order Taylor expansions of $f_1(x)=\sin\big( \frac{\pi}{2}x \big)$ and $f_2(x)=\cos\big( \frac{\pi}{2}x \big)$ at $x=1$:
\begin{align*}
 \mathbf{n}_{n,r} \widehat{M}_{n,r}^{n-\lfloor \log n\rfloor  }(0)\,=\,\frac{\pi \kappa^2}{2}\tan\Big( \frac{\pi}{2}\mathbf{r}_{n-\lfloor\log n\rfloor}\Big) \,=\, \kappa^2 \frac{1}{1-\mathbf{r}_{n-\lfloor\log n \rfloor }}\,+\,\mathit{O}\big( 1-\mathbf{r}_{n-\lfloor\log n\rfloor } \big) \,.
\end{align*} 
Finally, recall from~(\ref{NDef}) that $ \mathbf{n}_{n,r}=n+\mathit{O}( \log n  ) $ for large $n$.  Thus we only need to prove~(\ref{Flattened}). \vspace{.3cm}

\noindent \textbf{(c) Rewriting the increments of $\{ \mathbf{r}_k \}_{k\in \mathbb{N}_0}$ using Taylor's theorem:} By writing $\widehat{M}_{n,r}^{k+1}(0)=\widehat{M}_{n,r}\big(\widehat{M}_{n,r}^{k}(0)\big)  $ and splitting $\widehat{M}_{n,r}$ into a sum of  $\widetilde{M}_{n,r}$ and the error term $\mathscr{E}$, we get the equality
\begin{align*}
\mathbf{n}_{n,r} \widehat{M}_{n,r}^{k+1}(0)\,=\, &\mathbf{n}_{n,r} \widehat{M}_{n,r}^{k}(0) \,+\,\underbracket{\frac{1}{\kappa^2\mathbf{n}_{n,r} }\Big( \mathbf{n}_{n,r}\widehat{M}_{n,r}^{k}(0) \Big)^2}\,+\,\frac{1-\eta}{\kappa^2\mathbf{n}_{n,r}^2 }\Big(\mathbf{n}_{n,r} \widehat{M}_{n,r}^{k}(0) \Big)^3\,+\,\underbracket{\frac{\pi^2 \kappa^2  }{4 \mathbf{n}_{n,r} }}  \\ & \,+\, \mathbf{n}_{n,r}\mathscr{E}\left(\widehat{M}_{n,r}^{k}(0), \mathbf{n}_{n,r} \right)  \,.
\end{align*}
With~(\ref{FormR}), we can rewrite the equation  above in terms of the variables $\mathbf{r}_{k}$ and $\mathbf{r}_{k+1}$ as below, where the bracketed expressions have combined to form the $\sec^2$ term.
\begin{align}\label{TanOne}
\mathbf{r}_{k+1}\,=\,&\frac{2}{\pi}\tan^{-1}\Bigg(\tan\Big( \frac{\pi}{2}\mathbf{r}_k   \Big)   \,+\,\frac{\pi}{2\mathbf{n}_{n,r}}\sec^2\Big( \frac{\pi}{2}\mathbf{r}_k   \Big)\nonumber \\  &\text{}\hspace{1.5cm}\,+\, \underbrace{\frac{\pi^2}{4\mathbf{n}_{n,r}^2}(1-\eta)\tan^3\Big( \frac{\pi}{2}\mathbf{r}_k   \Big)\,+\,\mathbf{n}_{n,r} \mathscr{E}\bigg( \frac{\pi\kappa^2}{2 \mathbf{n}_{n,r}}\tan\Big(\frac{\pi}{2} \mathbf{r}_k \Big) , \mathbf{n}_{n,r}  \bigg)}_{\mathbf{(I)}}  \Bigg)\,.
\end{align}
If $\mathbf{r}_{k}< 1-1/\mathbf{n}_{n,r}$,  Taylor's theorem applied to  the function $g(x)=\tan\big(\frac{\pi}{2}x\big) $ around the point  $x=\mathbf{r}_{k}$ with second-order error implies there is an $\mathbf{r}_k^* \in [\mathbf{r}_{k},\mathbf{r}_{k}+1/\mathbf{n}_{n,r})$ such that
\begin{align}\label{TanTwo}
\mathbf{r}_{k}+\frac{1}{\mathbf{n}_{n,r}}\,=\,\frac{2}{\pi}\tan^{-1}\bigg(\tan\Big( \frac{\pi}{2}\mathbf{r}_k  \Big)   \,+\,\frac{\pi}{2\mathbf{n}_{n,r}}\sec^2\Big( \frac{\pi}{2}\mathbf{r}_k  \Big) \,+\, \underbrace{\frac{\pi^2}{4\mathbf{n}_{n,r}^2}\tan\Big( \frac{\pi}{2}\mathbf{r}_k^*   \Big) \sec^2\Big( \frac{\pi}{2}\mathbf{r}_k^*   \Big)}_{\mathbf{(II)}} \bigg)\,.
\end{align}
Define $\Delta_k$ as the difference between the terms $(\mathbf{II})$ and $(\mathbf{I})$:
\begin{align*}
\Delta_k\,:=\,\underbrace{\frac{\pi^2}{4\mathbf{n}_{n,r}^2}\bigg(\tan\Big( \frac{\pi}{2}\mathbf{r}_k^*   \Big) \sec^2\Big( \frac{\pi}{2}\mathbf{r}_k^*   \Big) \,-\,(1-\eta)\tan^3\Big( \frac{\pi}{2}\mathbf{r}_k  \Big) \bigg)}_{ \mathbf{(III)}}-\,\,\mathbf{n}_{n,r} \mathscr{E}\bigg( \frac{\pi\kappa^2}{2 \mathbf{n}_{n,r}}\tan\Big(\frac{\pi}{2} \mathbf{r}_k \Big) , \mathbf{n}_{n,r}  \bigg)\,.
\end{align*}
By Taylor's theorem applied to the function $h(x)=\frac{2}{\pi}\tan^{-1}(x)$ around the point $x=\tan\big(\frac{\pi}{2}\mathbf{r}_{k+1}  \big)$, there is an $\mathbf{r}_k^{**}$ between $\mathbf{r}_{k+1}$ and $\mathbf{r}_k+1/\mathbf{n}_{r,n}$ such that
\begin{align}
\mathbf{r}_k\,+\,\frac{1}{\mathbf{n}_{n,r}}\,=\,& \mathbf{r}_{k+1}\,+\,\frac{2}{\pi}\Delta_k\frac{1}{ 1+\tan^2\big(\frac{\pi}{2}\mathbf{r}_{k+1}   \big)    }\,-\,\frac{2}{\pi}\Delta_k^2\frac{\tan\big(\frac{\pi}{2}\mathbf{r}_{k}^{**}   \big)}{\big( 1+\tan^2\big(\frac{\pi}{2}\mathbf{r}_{k}^{**}   \big)\big)^2    } \nonumber  \\
\,=\,&\mathbf{r}_{k+1}\,+\,\frac{2}{\pi}\Delta_k \cos^2\Big(\frac{\pi}{2}\mathbf{r}_{k+1}   \Big)    \,-\,\frac{2}{\pi}\Delta_k^2\sin\Big( \frac{\pi}{2}\mathbf{r}_{k}^{**}  \Big)\cos^3\Big( \frac{\pi}{2}\mathbf{r}_{k}^{**}  \Big)\,.\label{Tintin}
\end{align}

\vspace{.3cm}

\noindent \textbf{(d) Bounds for the various terms in~(\ref{Tintin}):}    The inequalities below hold for some $C>0$ and all $k\in \mathbb{N}_0$ and $n\in \mathbb{N}$ such that $1-\mathbf{r}_{k}\geq \frac{\log n}{2 n } > 1/ \mathbf{n}_{n,r}  $.\footnote{The lower bound of $1-\mathbf{r}_{k}$ by $1/ \mathbf{n}_{n,r} $ ensures that $\mathbf{r}_{k}^*$ is well-defined by~(\ref{TanTwo}). When $n$ is sufficiently large, $\frac{\log n}{2 n } > 1/ \mathbf{n}_{n,r} $ holds as a consequence of~(\ref{NDef}).}  
\begin{enumerate}[(i)]
 \item $  0\,\leq \,\mathbf{n}_{n,r} \mathscr{E}\left( \frac{\pi\kappa^2}{2 \mathbf{n}_{n,r}}\tan\big(\frac{\pi}{2} \mathbf{r}_k \big) , \mathbf{n}_{n,r}  \right)\,\leq \, \frac{C}{ n^3(1-\mathbf{r}_k)^4 } \,+\,\frac{C}{ n^{5/3 } } $

\item $|\Delta_k| \,\leq \,\frac{ C  }{ n^2(1-\mathbf{r}_k)^3  }\,+\,\frac{C}{ n^{5/3 } }$

\item $\Big| \Delta_k -\frac{2\eta}{\pi n^2(1-\mathbf{r}_k)^3 } \Big| \,\leq \,\frac{ C  }{ n^3(1-\mathbf{r}_k)^4  }\,+\,\frac{C}{ n^{5/3 } }$

\item $\big| \mathbf{r}_{k} +\frac{1}{\mathbf{n}_{n,r}} - \mathbf{r}_{k+1}\big| \,\leq \,\frac{ C  }{ n^2(1-\mathbf{r}_k)  }\,+\,\frac{C}{ n^{5/3 } }$

\item $\Big|\frac{2}{\pi}\Delta_k^2\sin\big( \frac{\pi}{2}\mathbf{r}_{k}^{**}  \big)\cos^3\big( \frac{\pi}{2}\mathbf{r}_{k}^{**} \big) \Big| \,\leq \,\frac{ C  }{ n^4(1-\mathbf{r}_k)^3  }\,+\,\frac{C}{ n^{10/3 } }$

\item $\Big|\frac{2}{\pi}\Delta_k \cos^2\big(\frac{\pi}{2}\mathbf{r}_{k+1}   \big)\,-\,\frac{ \eta  }{ n  }\log\big( \frac{  1- \mathbf{r}_{k} }{ 1- \mathbf{r}_{k+1}  }  \big)   \Big| \,\leq \,\frac{ C  }{ n^3(1-\mathbf{r}_k)^2  }\,+\,\frac{C}{ n^{5/3 } }$

\end{enumerate}
The terms $\frac{C}{ n^{5/3 } }$ above  arise from~(\ref{ErrorTerm}) and are less important than the first bounding terms. 
Note that (vi) approximates the second term on the right side of~(\ref{Tintin}) by an expression that conveniently telescopes when summed over $k$, and (v) bounds the last term on the right side of~(\ref{Tintin}).

 The bound (i) follows from~(\ref{ErrorTerm}), that $\mathbf{n}_{n,r}\sim n$ for $n\gg 1 $ by~(\ref{NDef}), and the estimates below for $0\leq 1-x\ll 1$:
\begin{align}\label{TrigEstimates}
\cos\Big(\frac{\pi}{2}x   \Big) = \frac{\pi}{2}(1-x)\,+ \,\mathit{O}\big((1-x)^3\big)  \hspace{1cm}\text{and}\hspace{1cm}   \sin\Big(\frac{\pi}{2}x   \Big) = 1\,- \,\mathit{O}\big((1-x)^2\big)\,  .
\end{align}

The bound (ii) follows from (iii), so we will focus on (iii) next.    The inequality  $\mathbf{r}_k^*-\mathbf{r}_k < 1/\mathbf{n}_{n,r} $ and~(\ref{NDef}) imply the equalities below.
\begin{align}
\tan\Big( \frac{\pi}{2}\mathbf{r}_k^*   \Big) \sec^2\Big( \frac{\pi}{2}\mathbf{r}_k^*   \Big)\,=\,&\frac{(2/\pi)^3}{(1-\mathbf{r}_k^*)^3   }+ \mathit{O}\bigg(\frac{1}{1-\mathbf{r}_k^*   }\bigg) \,=\,\frac{(2/\pi)^3}{(1-\mathbf{r}_k)^3   }+ \mathit{O}\bigg(\frac{1}{n(1-\mathbf{r}_k)^4  }+\frac{1}{1-\mathbf{r}_k  }\bigg) \nonumber   \\   \tan^3\Big( \frac{\pi}{2}\mathbf{r}_k  \Big)\,=\,&\frac{(2/\pi)^3}{(1-\mathbf{r}_k)^3   }+ \mathit{O}\bigg(\frac{1}{1-\mathbf{r}_k  }\bigg)\label{TP}
\end{align} 
It follows from~(\ref{TP}) and~(\ref{NDef}) that the difference between $\frac{2\eta}{\pi n^2(1-\mathbf{r}_k)^3 } $ and the braced expression $\textbf{(III)}$ in part (d) is bounded by
\begin{align}
\bigg|\textbf{(III)} -\frac{2\eta}{\pi n^2(1-\mathbf{r}_k)^3 } \bigg|\,\leq \,&\frac{2\eta}{\pi (1-\mathbf{r}_k)^3 } \bigg|\frac{1}{ \mathbf{n}_{n,r}^2 } - \frac{1}{n^2} \bigg|\,+\,\frac{c}{n^3(1-\mathbf{r}_k)^4  }\,+\,\frac{c}{n^2(1-\mathbf{r}_k)  }  \nonumber 
\intertext{for some $c>0$.  Since $\mathbf{n}_{n,r}=n+\mathit{O}(\log n)  $ and $\frac{1}{n^2(1-\mathbf{r}_k ) }\leq                   \frac{3}{4 n^{5/3}  }+ \frac{1}{4n^3(1-\mathbf{r}_k)^4}   $ by Young's inequality, there is a $c'>0$ such that the above is bounded by}
 \,\leq \,& \frac{c'\log(n+1)}{n^3(1-\mathbf{r}_k)^3  }\,+\,\frac{c'}{n^3(1-\mathbf{r}_k)^4  }\,+\,\frac{c'}{n^{5/3}  }  \,\leq \,\frac{c''}{n^3(1-\mathbf{r}_k)^4  }\,+\,\frac{c''}{n^{5/3}  }  \, .\label{Rif}
\end{align}
The last inequality holds  by another application of
 Young's inequality to get $\frac{\log(n+1)}{n^3(1-\mathbf{r}_k)^3 }\leq \frac{\log^4(n+1)}{4n^3 }+\frac{3}{4 n^3(1-\mathbf{r}_k)^4 }    $ and since $ \frac{\log^4(n+1)}{4n^3 }\ll \frac{1}{n^{5/3}} $.  Finally, (iii) follows by combining~(\ref{Rif}) with (i).\vspace{.1cm}

Note that (iii) implies that $\Delta_k$ is positive for all $k$ with $1-\mathbf{r}_k\geq \frac{\log n}{2n}$ when $n$ is sufficiently large.  Thus~(\ref{TanOne})-(\ref{Tintin})  imply that $\mathbf{r}_k\leq \mathbf{r}_{k+1}\leq \mathbf{r}_{k}^{**}\leq  \mathbf{r}_k+1/\mathbf{n}_{n,r} $. The bound (iv) follows from applying (ii) to~(\ref{Tintin}) and using that $ \mathbf{r}_{k+1}$ and $ \mathbf{r}_{k}^{**}$ are within a distance of $1/\mathbf{n}_{n,r}\sim 1/n$ from $\mathbf{r}_k$.  The bounds (v) \& (vi) follow from (ii) and (iii), respectively, using  basic calculus estimates.

\vspace{.3cm}

\noindent \textbf{(e) A consequence of (iv):}  Before going to 
the estimates in part (f) below, we will point out an easy consequence of the bound (iv) in (d):  
  if $\ell \in  \mathbb{N}_0$ satisfies $1-\mathbf{r}_{\ell}\geq \frac{\log n}{2n}> 1/\mathbf{n}_{n,r} $, then the spacing between  the terms in the sequence $\{\mathbf{r}_{k}\}_{k=0}^{\ell}$ has the large $n$ form
$$  \mathbf{r}_{k+1}\,-\, \mathbf{r}_{k}\,=\,\frac{1}{\mathbf{n}_{n,r}} \,+\, \mathit{O}\bigg(\frac{1}{n\log^2 n}  \bigg)\,=\,\frac{1}{n} \,+\, \mathit{O}\bigg(\frac{1}{n\log^2 n}  \bigg)\,,$$
where the errors, $\mathit{O}\big(\frac{1}{n\log^2 n}\big)$, are uniformly bounded by a multiple of $\frac{1}{n\log^2 n}$ for all $0\leq k<\ell$ and $n\gg 1$.
 The second equality above holds since $\mathbf{n}_{n,r}=n+\mathit{O}\big(1/\log n\big)$. A Riemann sum approximation thus gives us
\begin{align}\label{4Sum}
\sum_{k=0}^{\ell-1}\frac{ 1  }{ n^3(1-\mathbf{r}_k)^2  }\,=\,\frac{1+\mathit{o}(1)}{n^2}\int_0^{\mathbf{r}_{\ell}}\frac{1}{(1-x)^2}dx\,=\,\frac{1+\mathit{o}(1)}{n^2} \bigg( \frac{1}{1-\mathbf{r}_{\ell}}  -1\bigg)\,=\,\mathit{o}\Big( \frac{1}{ n } \Big)\,.
\end{align}

\vspace{.2cm}

\noindent \textbf{(f) Applying the bounds to a key telescoping sum:} Assume that $\ell\in \mathbb{N}$ satisfies $\ell\leq n$ and that  $1-\mathbf{r}_{\ell}\geq \frac{\log n}{2n}> 1/ \mathbf{n}_{n,r} $ holds so that~(\ref{4Sum}) and the inequalities in part (d) are applicable.  Since $\mathbf{r}_0=0$, the equality below results from a telescoping sum:
\begin{align}
1\,-\,\mathbf{r}_{\ell}\,=\,&\bigg(1\,-\,\frac{\ell    }{ \mathbf{n}_{n,r}  }\bigg)\,+\,\sum_{k=0}^{\ell-1}\bigg(\mathbf{r}_{k}+\frac{1}{\mathbf{n}_{n,r}}-\mathbf{r}_{k+1}   \bigg)\,.\nonumber 
\intertext{Using the identity~(\ref{Tintin}) to rewrite the difference between $\mathbf{r}_{k}+1/\mathbf{n}_{n,r}$ and $\mathbf{r}_{k+1} $, we get that }
\,=\,&\bigg(1\,-\,\frac{\ell    }{ \mathbf{n}_{n,r}  }\bigg)\,+\,\frac{2}{\pi}\sum_{ k=0}^{\ell-1}\Delta_{k} \cos^2\Big( \frac{\pi}{2}\mathbf{r}_{k+1}  \Big)\,-\,\frac{2}{\pi}\sum_{k=0}^{\ell-1}\Delta_{k}^2 \sin\Big( \frac{\pi}{2}\mathbf{r}_{k}^{**}  \Big)\cos^3\Big( \frac{\pi}{2}\mathbf{r}_{k}^{**}  \Big)\,. \nonumber   
\intertext{By the bounds (v) and (vi) in part (d),  the above is equal to }
\,=\,&\bigg(1\,-\,\frac{\ell    }{ \mathbf{n}_{n,r}  }\bigg)\,+\,\sum_{k=0}^{\ell-1}\frac{ \eta  }{ n  }\log\bigg( \frac{  1- \mathbf{r}_{k} }{ 1- \mathbf{r}_{k+1}  }\bigg)\,+\,\mathit{O}\Bigg(\frac{  \ell }{ n^{5/3} } \,+\,  \sum_{k=0}^{\ell-1}\frac{ 1  }{ n^3(1-\mathbf{r}_k)^2  } \Bigg)\,.\nonumber
\intertext{The rightmost  term above is $\mathit{o}(1/n  )$ since $\ell\leq n $ and by~(\ref{4Sum}). The middle term  is a telescoping sum in which $ \mathbf{r}_{0} =0$, so we have  }
\,=\,&\bigg(1\,-\,\frac{\ell    }{ \mathbf{n}_{n,r}  }\bigg)\,+\,\frac{\eta}{n}\log\Big(\frac{1}{1-\mathbf{r}_{\ell}}   \Big)\,+\,\mathit{o}\Big(\frac{1}{n}\Big)\,.\nonumber 
\intertext{By adding and subtracting terms, we can rewrite the  above as}
\, =\,& \frac{ n-\lfloor \log n\rfloor -\ell   }{ \mathbf{n}_{n,r}  }\,+\,\underbracket{\frac{\lfloor \log n\rfloor + \mathbf{n}_{n,r} -n}{ \mathbf{n}_{n,r}  }\,+\,\frac{\eta\log n-\eta\log\log n}{n}}\,+\,\frac{\eta}{n}\log\bigg(\frac{\log n}{n(1-\mathbf{r}_{\ell})}   \bigg)\,+\,\mathit{o}\Big(\frac{1}{n}\Big) \,.\nonumber 
\intertext{Since $\mathbf{n}_{n,r}=n-\eta\log n-r+\mathit{o}(1)$ for $n\gg 1$, the difference between the sum of the bracketed terms above and  bracketed term below is $\mathit{o}(1/n) $:}
\, =\,& \frac{ n-\lfloor \log n\rfloor -\ell    }{ \mathbf{n}_{n,r}  }\,+\,\underbracket{\frac{ \lfloor \log n\rfloor \, -\, \eta\log \log n \,-\,r}{ n  }}\,+\,\frac{\eta}{n}\log\bigg(\frac{\log n}{n(1-\mathbf{r}_{\ell})}   \bigg)\,+\,\mathit{o}\Big(\frac{1}{n}\Big)\,.\label{Tib}
\end{align}

\vspace{.3cm}

\noindent \textbf{(g) How we can make use of~(\ref{Tib}):}  We will temporarily assume that   $1-\mathbf{r}_{n-\lfloor \log n\rfloor}\geq \frac{\log n}{2n}$ holds for sufficiently large $n$ to show that the asymptotics~(\ref{Flattened}) follows.   
If $1-\mathbf{r}_{n-\lfloor \log n\rfloor}\geq \frac{\log n}{2n}$, then the equality~(\ref{Tib}) holds with $\ell= n-\lfloor \log n\rfloor$, which gives us 
\begin{align}\label{OneMinusR}
1\,-\,\mathbf{r}_{n-\lfloor \log n\rfloor}\,=\,
&\frac{ \lfloor \log n\rfloor \, -\, \eta\log \log n \,-\,r}{ n  }\,+\,\underbracket{\frac{\eta}{n}\log\bigg(\frac{\log n}{n(1-\mathbf{r}_{n-\lfloor \log n\rfloor})}   \bigg)}\,+\,\mathit{o}\Big(\frac{1}{n}\Big)\,.
\end{align}
Note that~(\ref{Flattened}) holds provided that the bracketed term is $\mathit{o}(1/n)$ for $n\gg 1$.  Since $1-\mathbf{r}_{n-\lfloor \log n\rfloor}\geq \frac{\log n}{2n}$, we can get an upper bound for $1\,-\,\mathbf{r}_{n-\lfloor \log n\rfloor}$  by substituting $\frac{\log n}{2n}$ in place of $1-\mathbf{r}_{n-\lfloor \log n\rfloor}$ on the right side of~(\ref{OneMinusR}): 
\begin{align*}
\frac{\log n}{2n}\,\leq \,1\,-\,\mathbf{r}_{n-\lfloor \log n\rfloor}\,\leq \,
&\frac{ \lfloor \log n\rfloor \, -\, \eta\log \log n \,-\,r}{ n  }\,+\,\frac{\eta}{n}\log 2\,+\,\mathit{o}\Big(\frac{1}{n}\Big)\,.
\end{align*}
Thus $1\,-\,\mathbf{r}_{n-\lfloor \log n\rfloor}$ is bounded from above and below by constant multiples of $\frac{\log n}{n}$ for $n\gg 1$.  It follows that the bracketed term in~(\ref{OneMinusR}) is $\mathit{O}(1/n)$, and hence we can conclude from~(\ref{OneMinusR}) that $1-\mathbf{r}_{n-\lfloor \log n\rfloor}= \frac{\log n}{n}\big(1+\mathit{o}(1)\big)$.  Plugging this asymptotics for $1-\mathbf{r}_{n-\lfloor \log n\rfloor}$ back into the right side of~(\ref{OneMinusR}), however, yields that the bracketed term in~(\ref{OneMinusR}) is $\mathit{o}\big(1/n)$, which proves~(\ref{Flattened}) under the assumption that $1-\mathbf{r}_{n-\lfloor \log n\rfloor}\geq \frac{\log n}{2n}$.

\vspace{.3cm}

\noindent \textbf{(h) Establishing the validity of~(\ref{Tib}) when $\ell=n-\lfloor \log n\rfloor$:} It remains to show that $1-\mathbf{r}_{n-\lfloor \log n\rfloor}\geq \frac{\log n}{2n}$ holds for large enough $n$. Let $\ell^* \equiv \ell^*(n,r)$ be the smallest $\ell^* \in \mathbb{N}$ such that
\begin{align}\label{Lower}
1-\mathbf{r}_{\ell^*}\,\leq \,\frac{3 \log n }{4n}\,.
\end{align}
Since $1-\mathbf{r}_{\ell^*-1}> \frac{3 \log n }{4n}$ and $\mathbf{r}_{\ell^*}-\mathbf{r}_{\ell^*-1}=\frac{1}{n}+\mathit{o}(\frac{1}{n})$ by (iv) in part (d), we have  the inequality $1-\mathbf{r}_{\ell^*}\geq \frac{ \log n }{2n}$ for large enough $n$. Thus the equality~(\ref{Tib}) will hold with $\ell=\ell^*$ when $n\gg 1$:
\begin{align}\label{Chub}
\frac{3 \log n }{4n}\,\geq\,1-\mathbf{r}_{\ell^*}\, =\,& \frac{ n-\lfloor \log n\rfloor -\ell^*    }{ \mathbf{n}_{n,r}  }\,+\, \frac{ \lfloor \log n\rfloor \, -\, \eta\log \log n \,-\,r}{ n  }\,+\,\frac{\eta}{n}\log\bigg(\frac{\log n}{n(1-\mathbf{r}_{\ell^*})}   \bigg)\,+\,\mathit{o}\Big(\frac{1}{n}\Big)\,.\nonumber
\intertext{Using the upper bound~(\ref{Lower}) for $1-\mathbf{r}_{\ell^*}$ in the logarithm yields}
\, \geq \,& \underbracket{\frac{ n-\lfloor \log n\rfloor -\ell^*    }{ \mathbf{n}_{n,r}  }}_{\text{must be}\,\, <\,0\,\,\text{for large $n$}}\,+\, \underbrace{\frac{ \lfloor \log n\rfloor \, -\, \eta\log \log n \,-\,r}{ n  }+\frac{\eta  }{ n }\log\Big(\frac{4}{3}\Big)+\mathit{o}\Big(\frac{1}{n}\Big)}_{ > \,\,\frac{3 \log n }{4n} \text{ for large $n$}   }\,.
\end{align}
Since $1\,-\,\mathbf{r}_{\ell^*}$ is bounded by $\frac{3 \log n }{4n}$ and the braced term on the right side of~(\ref{Chub}) is  greater than  $ \frac{3 \log n }{4n}$ for large $n$, the first term on the right side of~(\ref{Chub}) must be negative when $n\gg 1$, and  therefore $\ell^*> n-\lfloor \log n\rfloor$. It follows that $1-\mathbf{r}_{n-\lfloor \log n\rfloor}\geq \frac{\log n}{2n}$ for large $n$.
\end{proof}

\subsection{Proof of Lemma~\ref{LemCond}} \label{SecLemCond}

Since the random variables $\mathbb{E}\big[\widehat{W}_n^{\omega}\big( \widehat{\beta}_{n,r} \big)\,\big|\, \mathcal{F}_{n}^{\lfloor \log n\rfloor} \big]$  and $\widehat{W}_n^{\omega}\big( \widehat{\beta}_{n,r} \big)-\mathbb{E}\big[\widehat{W}_n^{\omega}\big( \widehat{\beta}_{n,r} \big)\,\big|\, \mathcal{F}_{n}^{\lfloor \log n\rfloor} \big]  $ are uncorrelated, the square of the $L^2$ distance between $\widehat{W}_n^{\omega}\big( \widehat{\beta}_{n,r} \big)$ and $\mathbb{E}\big[\widehat{W}_n^{\omega}\big( \widehat{\beta}_{n,r} \big)\,\big|\, \mathcal{F}_{n}^{\lfloor \log n\rfloor}\big]$ is equal to 
\begin{align*}
  \textup{Var}\Big(\widehat{W}_n^{\omega}\big( \widehat{\beta}_{n,r} \big)\Big)\,-\,\textup{Var}\Big( \mathbb{E}\Big[\widehat{W}_n^{\omega}\big( \widehat{\beta}_{n,r} \big)\,\Big|\, \mathcal{F}_{n}^{\lfloor \log n\rfloor}\Big]\Big) \,=\,&\widehat{M}_{n,r}^{n }(0)\,-\,\textup{Var}\left( \mathcal{Q}^{\lfloor \log n\rfloor}\big\{ X_h^{n,r} \big\}_{h\in E_{\lfloor \log n\rfloor }}\right)\\ 
=\, &\widehat{M}_{n,r}^{n }(0)\,-\, M^{\lfloor \log n\rfloor }\left(\widehat{M}_{n,r}^{n-\lfloor \log n\rfloor }(0)\right)\,,
  \end{align*}
where the random variables  $X_h^{n,r}$  are independent copies of  $\widehat{W}_{n-\lfloor \log n\rfloor  }^{\omega}\big( \widehat{\beta}_{n,r} \big)$.
The equalities above use~(\ref{RecEqVar}), Proposition~\ref{PropReduce}, and (i) of Remark~\ref{RemarkArrayVar}.
It follows that Lemma~\ref{LemCond} is a corollary of the following:

\begin{lemma}\label{LemmaMMaps} The difference between $\widehat{M}_{n,r}^{n }(0)$ and $M^{\lfloor \log n\rfloor }\big(\widehat{M}_{n,r}^{n-\lfloor \log n\rfloor }(0)\big)$ vanishes as $n\rightarrow \infty$.\end{lemma}

\begin{remark}\label{RemarkM}  Note that  $M^{\lfloor \log n\rfloor }\big(\widehat{M}_{n,r}^{n-\lfloor \log n\rfloor }(0)\big)$  converges to $R(r)$ as $n\rightarrow \infty$.  This follows from Lemma~\ref{LemVar} since $\widehat{M}_{n,r}^{n-\lfloor \log n\rfloor }(0) $, which is equal to the variance of $\widehat{W}_{n-\lfloor \log n\rfloor}\big( \widehat{\beta}_{n,r}    \big) $, has  the large-$n$ asymptotics~(\ref{Unflat}) by Lemma~\ref{LemMtilde}.
\end{remark}

In the proof of Lemma~\ref{LemmaMMaps}, we will use Lemma~\ref{LemMaybe} below, which is a  result from~\cite[Lemma 2.2(iv)]{Clark1}.  Notice that applying the chain rule to the $k$-fold composition of $M(x)=\frac{1}{b}\big[(1+x)^b-1   \big]$ yields
\begin{align}\label{DerivM}
 \frac{d}{dx}M^{k}(x)\,=\,\prod_{j=1}^k \Big( 1+M^{j-1}(x)  \Big)^{b-1}\,=\,(k+1)^2D_{k}\big(M^{k}(x) \big)  \,,   
 \end{align}
where the function  $D_k:[0,\infty)\rightarrow [0,\infty)$ is defined by
\begin{align}\label{DefDk}
D_{k}(y)\,=\, \frac{1}{(k+1)^2}\prod_{\ell=1}^{k} \Big(1+M^{-\ell}(y)\Big)^{b-1}\,. 
\end{align}
In the above,  $M^{-\ell}$ denotes the $\ell$-fold composition of the function inverse of the map $M$. The following lemma gives us uniform bounds for the sequence in $k\in \mathbb{N}_0$ of functions $D_{k}$.
\begin{lemma}\label{LemMaybe}  The sequence of functions $\{D_{k}\}_{k\in \mathbb{N}_0} $ converges  uniformly over any bounded subinterval of $[0,\infty)$ to a limit function $D$.  In particular, $F(L):=\sup_{k\in \mathbb{N}_0}\sup_{x\in [0,L]}D_{k}(x)$ is finite for any $L>0$.
\end{lemma}

\begin{proof}[Proof of Lemma~\ref{LemmaMMaps}] Define $A_{n,r}:=\widehat{M}_{n,r}^{n-\lfloor \log n\rfloor }(0)$.  By Remark~\ref{RemarkM}, $M^{\lfloor \log n\rfloor  }(A_{n,r})$ converges to $R(r)$ as $n\rightarrow \infty$. For any $\ell \in \{0,\ldots, \lfloor \log n\rfloor\}$, the definition of $A_{n,r}$ implies that
\begin{align}
&\widehat{M}_{n,r}^{\ell+n-\lfloor \log n\rfloor }(0)-M^{\ell }\left(\widehat{M}_{n,r}^{n-\lfloor \log n\rfloor }(0)\right)\nonumber \\
&\hspace{.3cm}\,=\,\widehat{M}_{n,r}^{\ell} (A_{n,r})\,-\,M^{\ell }(A_{n,r})\,,
\intertext{which we can rewrite through the following telescoping sum:}
  & \hspace{.3cm}\,=\,\sum_{k=1}^{\ell  }     \bigg( \widehat{M}_{n,r}^{ k }\Big(M^{\ell-k }(A_{n,r})\Big)\,-\,\widehat{M}^{k-1 }_{n,r}\Big(M^{\ell-k+1 }(A_{n,r})\Big)\bigg)\,. \nonumber  \\
 \intertext{By the mean value theorem, there are points $y_k $ in the intervals $\Big(M^{\ell-k+1 }(A_{n,r}), \widehat{M}_{n,r}\big(M^{\ell-k }(A_{n,r})\big)\Big)$  such that the equality below holds.  }
&\hspace{.3cm}\,= \,\sum_{k=1}^{\ell } \bigg(  \widehat{M}_{n,r}\Big(M^{\ell-k }(A_{n,r})\Big)\,-\,M\Big(M^{\ell-k }(A_{n,r}) \Big) \bigg)  \frac{d}{dx}\widehat{M}_{n,r}^{ k-1 }(x)\Big|_{x=y_{k} } \nonumber 
 \intertext{Since the derivative of $\widehat{M}_{n,r}^{k-1 } $ is increasing $\widehat{M}_{n,r}(x)\geq M(x)$ for $x\geq 0$, the above is bounded by}
& \hspace{.3cm}\,\leq \,\sum_{k=1}^{\ell } \underbrace{\bigg(  \widehat{M}_{n,r}\Big(M^{\ell-k }(A_{n,r})\Big)\,-\,M\Big(M^{\ell-k }(A_{n,r}) \Big) \bigg)}_{ \textbf{(I)}} \underbrace{\frac{d}{dx}\widehat{M}_{n,r}^{ k-1 }(x)\Big|_{x= \widehat{M}_{n,r}\big(M^{\ell-k }(A_{n,r})\big)  }}_{\textbf{(II)}  } \,. \label{Pebble}
\end{align}
We will return to~(\ref{Pebble}) after obtaining bounds for the terms $\textbf{(I)}$ and $\textbf{(II)}$.\vspace{.3cm}

\noindent \textbf{Bound for (I):} The difference between the functions $\widehat{M}_{n,r}$ and $M$ has the bound,
\begin{align}\label{ChangeM}
  \widehat{M}_{n,r}(x)\,-\,M(x) \,=\,\frac{1}{b}(1 +x  )^{b}\Big[\big(1+V_{n,r}\big)^{b-1}-1   \Big]\,<\,V_{n,r} (1 +x  )^{b}\,,
\end{align}
where the inequality holds for large enough $n$ since $V_{n,r}$ is vanishing.  Thus for large  $n$
\begin{align}
 \widehat{M}_{n,r}\Big(M^{\ell-k }(A_{n,r})\Big)\,-\,M\Big(M^{\ell-k }(A_{n,r}) \Big)\,\leq \,&  V_{n,r}\Big(1 + M^{\ell-k }(A_{n,r})   \Big)^{b}\,.\nonumber 
  \intertext{Since $\ell\leq \lfloor \log n\rfloor $ and $x\leq M(x)$ for all $x\geq 0$, the above is bounded by  }
 \,\leq \,&  V_{n,r}\Big(1 + M^{\lfloor \log n \rfloor }(A_{n,r})   \Big)^{b}\,.\nonumber
\intertext{Since $M^{\lfloor \log n \rfloor }(A_{n,r})  $ converges to $R(r)$ as $n\rightarrow \infty$, we have the following inequality for large enough $n$:} 
 \,<\,& 2V_{n,r}\big(1+R(r)   \big)^{b} \,.\label{Zip}
\end{align}
 \vspace{.1cm}

\noindent \textbf{Bound for (II):} By the chain rule, the derivative of $\widehat{M}_{n,r}^{ k }$ can be written in the form
\begin{align}
\frac{d}{dx}\widehat{M}_{n,r}^{ k }(x)\,=\,&\big(1+ V_{n,r} \big)^{k(b-1)}\prod_{j=0}^{k-1} \Big(1+\widehat{M}_{n,r}^{j}(x)\Big)^{b-1}\,.\nonumber 
\intertext{Since $V_{n,r}=\mathit{O}(1/n^2)$ and $k\leq \lfloor \log n\rfloor $, the term $\big(1+ V_{n,r} \big)^{k(b-1)}$ is smaller than $2$ for large $n$.  Moreover, writing  $\widehat{M}_{n,r}^{j}=\widehat{M}_{n,r}^{-(k-j)}\widehat{M}_{n,r}^{k} $ and changing the index to $l=k-j$ yields  }
 \,\leq \,&2\prod_{l=1}^{k} \left(1+\widehat{M}_{n,r}^{-l}\Big(\widehat{M}_{n,r}^{ k }(x)\Big)\right)^{b-1}\,. \nonumber
\intertext{Since $\widehat{M}_{n,r}(x)\geq M(x)$ for all $x\geq 0$, $\widehat{M}_{n,r}^{-1}(y)\leq M^{-1}(y)$ for all $y\geq 0$.  Thus the above is bounded by}
 \,\leq \,&2\prod_{l=1}^{k} \Big(1+M^{-l}\Big(\widehat{M}_{n,r}^{ k }(x)\Big)\Big)^{b-1} \,= \,2(k+1)^2 D_k\Big(\widehat{M}_{n,r}^{ k }(x)\Big)\,,\label{Tibit}  
\end{align}
where the  equality uses the definition~(\ref{DefDk}) of the  function $D_k:[0,\infty)\rightarrow [0,\infty)$.  An application of~(\ref{Tibit}) to the term $\mathbf{(b)}$ gives us
\begin{align}
\frac{d}{dx}\widehat{M}_{n,r}^{ k-1 }(x)\Big|_{x= \widehat{M}_{n,r}\big(M^{\ell-k }(A_{n,r})\big)  } \,\leq \,&2k^2 D_{k-1}\Big(\widehat{M}_{n,r}^{ k }\Big(M^{\ell-k }(A_{n,r})\Big)\Big)\nonumber \\
\,\leq \,&2k^2 D_{k-1}\Big(\widehat{M}_{n,r}^{  \ell}(A_{n,r})\Big) \,,\label{Zap}
\end{align}
where the second inequality again uses that  $\widehat{M}_{n,r}(x)\geq M(x)$ for all $x\geq 0$.

\vspace{.3cm}

\noindent \textbf{Returning to~(\ref{Pebble}):} Applying~(\ref{Zip}) and~(\ref{Zap}) to~(\ref{Pebble}) gives us the first inequality below for all $0\leq \ell \leq \lfloor \log n\rfloor$ when $n$ is large enough.
\begin{align}
\widehat{M}_{n,r}^{\ell} (A_{n,r})\,-\,M^{\ell }(A_{n,r})\,\leq \,& 4 V_{n,r}\big(1+R(r)   \big)^{b}\sum_{k=1}^{\ell }   k^2 D_{k-1}\Big(\widehat{M}_{n,r}^{  \ell}(A_{n,r})\Big) \nonumber
\intertext{Define $F(L):=\sup_{k\in \mathbb{N}_0}\sup_{x\in [0,L]}D_k(x)$.   Recall that $F$ is finite-valued by  Lemma~\ref{LemMaybe}. Since $\ell\leq \lfloor \log n \rfloor $ and $V_{n,r}\propto 1/n^2$ with large $n$, there is a $c>0$ such that for all $n\in \mathbb{N}$}
\,\leq \,& c\frac{\lfloor \log n\rfloor^3}{n^2} F\Big(\widehat{M}_{n,r}^{  \ell}(A_{n,r}) \Big)\,.\label{Dip}
\end{align}

 Let $\ell^*_{n,r}\in \mathbb{N}$ be the minimum of $\ell=\lfloor \log n\rfloor$ and the largest $\ell$ such that $\widehat{M}_{n,r}^{\ell} (A_{n,r})\leq 2M^{\ell }(A_{n,r})$. Applying~(\ref{Dip}) with $\ell=\ell^*_{n,r}  $ yields
\begin{align}
\widehat{M}_{n,r}^{\ell^*_{n,r}} (A_{n,r})\,-\,M^{\ell^*_{n,r}}(A_{n,r})\,\leq\,& c\frac{\lfloor \log n\rfloor^3}{n^2} F\Big(\widehat{M}_{n,r}^{ \ell^*_{n,r}}(A_{n,r}) \Big)\,.\nonumber
\intertext{ Since $\widehat{M}_{n,r}^{\ell^*_{n,r}} (A_{n,r})\leq 2M^{\ell^*_{n,r} }(A_{n,r})\leq 2M^{\lfloor \log n\rfloor }(A_{n,r})$ and $M^{\lfloor \log n\rfloor }(A_{n,r})$ converges to $R(r)$ as $n\rightarrow \infty$, $\widehat{M}_{n,r}^{\ell^*_{n,r}} (A_{n,r})$ is smaller than $3R(r)$ for $n\gg 1$. Moreover,  $F$ is nondecreasing, and so for large $n$} \,\leq \,&c\frac{\lfloor \log n\rfloor^3}{n^2} F\big(3R(r)\big) \,=\,\mathit{O}\bigg(\frac{\log^3 n}{n^2}  \bigg) \,.\label{REW}
\end{align}
We can apply~(\ref{REW}) to get the inequality below 
\begin{align}
\widehat{M}_{n,r}^{\ell^*_{n,r}+1} (A_{n,r})\,=\,\widehat{M}_{n,r}\Big(\widehat{M}_{n,r}^{\ell^*_{n,r}} (A_{n,r})\Big)\,\leq \,&\widehat{M}_{n,r}\bigg(M^{\ell^*_{n,r}} (A_{n,r})\,+\,\mathit{O}\bigg(\frac{\log^3 n}{n^2}  \bigg) \bigg)\,.\nonumber
\intertext{Since $M^{\ell^*_{n,r}} (A_{n,r})\leq M^{\lfloor \log n \rfloor} (A_{n,r}) $ and $M^{\lfloor \log n \rfloor} (A_{n,r}) $ converges to $R(r)$ as $n\rightarrow \infty$,  $M^{\ell^*_{n,r}} (A_{n,r})$ is smaller than $2R(r)$ for large enough $n$.  Moreover, the derivative of $\widehat{M}_{n,r}$ is uniformly bounded over bounded intervals, so we have that}
 \,=\,&\widehat{M}_{n,r}\Big(M^{\ell^*_{n,r}} (A_{n,r})\Big)\,+\,\mathit{O}\bigg( \frac{ \log^3 n}{n^2}  \bigg)\,. \nonumber
\intertext{Finally~(\ref{ChangeM}) implies that  replacing   $\widehat{M}_{n,r}$ by $M$ yields a negligible error:} 
\,=\,& M^{\ell^*_{n,r}+1} (A_{n,r})\,+\,\mathit{O}\bigg( \frac{ \log^3 n}{n^2}  \bigg)\,.
 \label{The}
\end{align}
However, since $M^{\ell^*_{n,r}+1} (A_{n,r})\geq A_{n,r}\geq \frac{\kappa^2}{\lfloor\log n \rfloor}$ for large $n$, the inequality~(\ref{The}) precludes the possibility that $\widehat{M}_{n,r}^{\ell^*_{n,r}+1} (A_{n,r})> 2M^{\ell^*_{n,r}+1} (A_{n,r})$ when $n$ is large.  It follows from the definition of $\ell^*_{n,r}$ that $\ell^*_{n,r}:=\lfloor \log n\rfloor$, and thus~(\ref{REW}) implies that the difference between $\widehat{M}_{n,r}^{\lfloor \log n\rfloor} (A_{n,r})$ and $M^{\lfloor \log n\rfloor}(A_{n,r})$ vanishes with large $n$.
\end{proof}

\subsection{Proof of Lemma~\ref{LemmaHM}}\label{SecLemmaHM}

\begin{proof} It suffices to show that the (uncentered) positive integer moments of $\widehat{W}_{n-\lfloor \log n\rfloor}\big(\widehat{\beta}_{n,r}\big)$ all converge to one as $n\rightarrow \infty$.  For $m\in \{2,3,\ldots\}$, $n\in \mathbb{N}$, $r\in \R$, and $k\in \mathbb{N}_0$ define
$$ \mu_{n,r}^{(m)}(k)\,:=\,\mathbb{E}\Big[ \big(\widehat{W}_{k}(\widehat{\beta}_{n,r})\big)^m \Big]    \hspace{1cm}  \text{and} \hspace{1cm} \nu_{n,r}^{(m)}\,:=\,\mathbb{E}\Bigg[ \bigg(\frac{e^{ \widehat{\beta}_{n,r}\omega  } }{ \mathbb{E}[ e^{ \widehat{\beta}_{n,r}\omega  } ]  } \bigg)^m \Bigg]  \,.   $$
Note that $\mu_{n,r}^{(m)}(0)=1$ since $\widehat{W}_{0}(\widehat{\beta}_{n,r})=1$ by definition, and $\mu_{n,r}^{(m)}(k),\nu_{n,r}^{(m)}\geq 1$ by Jensen's inequality. We obtain the following recursive equation in $k\in \mathbb{N}$ by evaluating the $m^{th}$ moment of both sides of the distributional equality~(\ref{PartHierSymmII}):
\begin{align}\label{FT}
\mu_{n,r}^{(m)}(k+1)\,=\, \frac{1}{b^{m-1}}  \big(\mu_{n,r}^{(m)}(k)\big)^{b}\big(\nu_{n,r}^{(m)}\big)^b\,+\, \mathbf{P}_m\Big( \big(\mu_{n,r}^{(\ell)}(k)\big)^{b}\big(\nu_{n,r}^{(\ell)}\big)^b\,;\,\, \ell\in \{2,\ldots, m-1\}  \Big)\,,
\end{align}
where $\mathbf{P}_m( y_{2},\ldots,  y_{m-1} )$ is a  polynomial with nonnegative coefficients that sum to $1-1/b^{m-1}$. In particular, $\mathbf{P}_m( y_{2},\ldots,  y_{m-1} )=1-1/b^{m-1}$ when evaluated at $( y_{2},\ldots,  y_{m-1} )=(1,\ldots, 1) $.  Moreover,  $ 1-1/b^{m-1}$ is a lower bound for $\mathbf{P}_m\Big( \big(\mu_{n,r}^{(\ell)}(k)\big)^{b}\big(\nu_{n,r}^{(\ell)}\big)^b\,;\,\, \ell\in \{2,\ldots, m-1\}  \Big)$ since $\mu_{n,r}^{(m)}(k),\nu_{n,r}^{(m)}\geq 1$.

We will use induction to prove that  $\max_{0\leq k\leq  n-\lfloor \log n\rfloor } \big|\mu_{n,r}^{(m)}(k)-1\big|$ vanishes as $n\rightarrow \infty$ for each $m\in \{2,3,\ldots\}$.  As a consequence of Lemma~\ref{LemMtilde}, $\mu_{n,r}^{(2)}\big(n-\lfloor \log n\rfloor\big)$ converges to one as $n\rightarrow \infty$. Since $ \big\{\mu_{n,r}^{(2)}(k)\big\}_{k\in \mathbb{N}_0} $ is an increasing sequence and $\mu_{n,r}^{(2)}(0)=1$, it follows that $ \max_{0\leq k\leq  n-\lfloor \log n\rfloor }\big| \mu_{n,r}^{(2)}(k)-1\big|$ vanishes as $n\rightarrow \infty$.  Suppose for the purpose of a strong induction argument that 
\begin{align}\label{IndAss}
\max_{0\leq k\leq  n-\lfloor \log n\rfloor }\big|\mu_{n,r}^{(\ell)}(k)-1\big| \hspace{.2cm}\stackrel{n\rightarrow \infty}{\longrightarrow }\hspace{.2cm} 0
\end{align} 
for each $\ell\in\{2,\ldots,m\}$.  Note that $\nu_{n,r}^{(\ell)}$ converges to one as $n\rightarrow \infty$ for each $\ell\in \mathbb{N}$ since $\widehat{\beta}_{n,r}$ vanishes with large $n$.  Fix some $\epsilon\in (0,1)$.  Since $\mathbf{P}_{m+1}$ is  continuous and $\mathbf{P}_{m+1}(1,\ldots, 1)=1-1/b^{m}$, we can choose  $n\in \mathbb{N}$ large enough such  that 
\begin{align}\label{IndCond}
\max_{0\leq k\leq  n-\lfloor \log n\rfloor }  \mathbf{P}_{m+1}\Big( \big(\mu_{n,r}^{(\ell)}(k )\big)^{b}\big(\nu_{n,r}^{(\ell)}\big)^b\,;\,\, \ell\in \{2,\ldots, m\}  \Big)\,\leq \,(1+\epsilon)\Big(1-\frac{1}{b^m}  \Big)\,.
\end{align}
Let $k^*_{n,r,\epsilon}$ be the minimum of $k=n-\lfloor \log n\rfloor $ and the smallest $k\in \mathbb{N}$
such that
\begin{align} \label{Smallest}
\big(\mu_{n,r}^{(m+1)}(k)\big)^{b-1}\big(\nu_{n,r}^{(m+1)}\big)^b\,> \,1+\epsilon \,.
\end{align}
By~(\ref{FT}) and the definition of $k^*_{n,r,\epsilon}$,  we have the   recursive inequality in $k\in \{0,\ldots, k^*_{n,r,\epsilon}-1\}$ below.
\begin{align}\label{Recur<}
\mu_{n,r}^{(m+1)}(k+1)\,\leq \, \frac{1+\epsilon}{b^{m}}  \mu_{n,r}^{(m+1)}(k)\,+\, \mathbf{P}_{m+1}\Big( \big(\mu_{n,r}^{(\ell)}(k)\big)^{b}\big(\nu_{n,r}^{(\ell)}\big)^b\,;\,\, \ell\in \{2,\ldots, m\}  \Big)
\end{align}
Applying~(\ref{Recur<}) $k^*_{n,r,\epsilon}$ times and using that $\mu_{n,r}^{(m+1)}(0) =1$ yields
\begin{align}
\max_{0\leq k\leq  k^*_{n,r,\epsilon} } &\mu_{n,r}^{(m+1)}(k) \nonumber \,\\ \leq \,& \Big( \frac{1+\epsilon}{b^m}  \Big)^{ k^*_{n,r,\epsilon} }\,+\, \sum_{k=0}^{k^*_{n,r,\epsilon}-1  } \Big( \frac{1+\epsilon}{b^m}  \Big)^{ k^*_{n,r,\epsilon}-1-k }\mathbf{P}_{m+1}\Big( \big(\mu_{n,r}^{(\ell)}(k)\big)^{b}\big(\nu_{n,r}^{(\ell)}\big)^b\,;\,\, \ell\in \{2,\ldots, m\}  \Big)\,. \nonumber 
\intertext{Since $\mathbf{P}_{m+1}\Big( \big(\mu_{n,r}^{(\ell)}(k)\big)^{b}\big(\nu_{n,r}^{(\ell)}\big)^b\,;\,\, \ell\in \{2,\ldots, m\}  \Big)$ is bounded from below by  $ 1-1/b^{m}$, geometric summation gives us the inequality}
\,\leq \,& \frac{   1}{1-  \frac{1+\epsilon}{b^m}   } \underbracket{\max_{0\leq k\leq  n-\lfloor \log n\rfloor }  \mathbf{P}_{m+1}\Big( \big(\mu_{n,r}^{(\ell)}(k)\big)^{b}\big(\nu_{n,r}^{(\ell)}\big)^b\,;\,\, \ell\in \{2,\ldots, m\}  \Big)}\,.\label{Liz}
\end{align}
The bracketed term converges to $1-1/b^{m}$ as $n\rightarrow \infty$ by the same reasoning as for~(\ref{IndCond}).  We will show that $k^*_{n,r,\epsilon}= n-\lfloor \log n\rfloor$ holds for large enough $n$ by showing that the condition~(\ref{Smallest}) cannot hold for $k\leq n-\lfloor \log n\rfloor$ when $n\gg 1$.  Notice that
\begin{align*}
\max_{0\leq k\leq  n-\lfloor \log n\rfloor } \Big(&\mu_{n,r}^{(m+1)}(k)\Big)^{b-1}\big(\nu_{n,r}^{(m+1)}\big)^b\\ &\,\leq \,\underbrace{\bigg(\frac{ 1 }{1-  \frac{1+\epsilon}{b^m}   } \max_{0\leq k\leq  n-\lfloor \log n\rfloor }  \mathbf{P}_{m+1}\Big( \big(\mu_{n,r}^{(\ell)}(k)\big)^{b}\big(\nu_{n,r}^{(\ell)}\big)^b\,;\,\, \ell\in \{2,\ldots, m\}  \Big)\bigg)^{b-1}   \big(\nu_{n,r}^{(m+1)}\big)^b}_{ \text{\large This expression converges to $\Big(\frac{ 1 -\frac{1}{b^m} }{1-  \frac{1+\epsilon}{b^m}   }\Big)^{b-1}$ as $n\rightarrow \infty$.}}\,.
\end{align*}
Moreover, since $m\geq 2$, the following inequality holds for small $\epsilon>0$:
$$  \bigg(\frac{ 1 -\frac{1}{b^m} }{1-  \frac{1+\epsilon}{b^m}   }\bigg)^{b-1}\,=\,\bigg(1+\frac{  \epsilon }{b^m-1-\epsilon   }\bigg)^{b-1}  \, < \, 1+\epsilon\,.$$  Thus $k^*_{n,r,\epsilon}$ does not satisfy~(\ref{Smallest}) when $n$ is large, and therefore $k^*_{n,r,\epsilon}=n-\lfloor \log n\rfloor$ for large $n$.  Going back to~(\ref{Liz}) with $k^*_{n,r,\epsilon}=n-\lfloor \log n\rfloor$, we get
\begin{align*}
\limsup_{n\rightarrow \infty}\max_{0\leq k\leq  n-\lfloor \log n\rfloor } \mu_{n,r}^{(m+1)}(k)\,\leq \,\frac{ 1 -\frac{1}{b^m} }{1-  \frac{1+\epsilon}{b^m}   }\,.
\end{align*}
Since  $\epsilon>0$ is arbitrarily and $\mu_{n,r}^{(m+1)}(k)\geq 1$, the sequence $\big\{\max_{0\leq k\leq  n-\lfloor \log n\rfloor } \big|\mu_{n,r}^{(m+1)}(k)-1\big|\big\}_{n\in \mathbb{N}}$  is vanishing. Therefore, by induction, $\max_{0\leq k\leq  n-\lfloor \log n\rfloor } \big|\mu_{n,r}^{(m)}(k)-1\big|$ converges to zero for each $m\in \{2,3,\ldots\}$, which completes the proof.
\end{proof}

\section{Miscellaneous  proofs from Sections~\ref{SecMiscProofs} \& \ref{SectionSharpRegProof}}\label{SecMiscII}

\subsection{Proof of Proposition~\ref{PropUnif}}\label{SecLemUnif}
  To prepare for the proof of Proposition~\ref{PropUnif}, we will define some additional notation related to the recursive formulas governing the positive integer moments of  random variables in a $\mathcal{Q}$-pyramidic array generated from an i.i.d.\ array of random variables and cite a bound (Lemma~\ref{LemBefore}) from~\cite{Clark1}.

 Let $\big\{ X_h^{(n)} \big\}_{h\in  E_{n} }$ be an i.i.d.\ array of centered random variables with finite $m^{th}$ absolute moment for some $m\in \{2,3,\ldots\}$ and $\{ X^{(*,n)}_a\}_{a\in E_{*}}$  be the $\mathcal{Q}$-pyramidic array generated from it.  For $k\in \{0,\ldots, n\}$ and $a\in E_k$, we will use the notation
 \begin{align}\label{SigmaGen}
 \sigma^{(m)}_{k,n}\,:=\,\mathbb{E}\Big[ \big(  X^{(k,n)}_a \big)^m  \Big]
 \end{align}
  and condense subscripts when $k=n$ as follows: $\sigma^{(m)}_{n,n}\equiv \sigma^{(m)}_{n}$.  For $m=2$ note that  $ \sigma^{(2)}_{k,n}$ is interchangeable with our previous notation  $\sigma^{2}_{k,n} $ from~(\ref{DefLittleSigma}).
  By (i) of Remark~\ref{RemarkArrayVar}, the recursive relation $\{ X^{(k-1,n)}_a\}_{a\in E_{k-1}}:=\mathcal{Q}\{ X^{(k,n)}_a\}_{a\in E_{k}}$  implies that  $M(\sigma^2_{k,n})=\sigma^2_{k-1,n}$ for the polynomial $M(x)=\frac{1}{b}[(1+x)^b-1]$.  More generally,  the multilinear form of the map $\mathcal{Q}$ implies that the vector of higher moments $\big(\sigma^{(3)}_{k,n}, \ldots  , \sigma^{(m)}_{k,n}\big)$  obeys a recursive equation with $\sigma^{(2)}_{k,n}$ as an additional input:
\begin{eqnarray}\label{GG}
 \Big(\sigma^{(3)}_{k-1,n}, \ldots  , \sigma^{(m)}_{k-1,n}\Big)\,=\,\vec{P}_m\Big( \sigma^{(2)}_{k,n},\ldots,   \sigma^{(m)}_{k,n}\Big) \,,
 \end{eqnarray}
where $\vec{P}_m:\R^{m-1}\rightarrow \R^{m-2}$ is a vector of polynomials $P_j:\R^{j-1}\rightarrow \R$:\footnote{The polynomials $P_j$ are the same as  those in  (I) of Theorem~\ref{ThmHM}.}
$$\vec{P}_m(y_2,\ldots, y_m)\,=\,\Big( P_3(y_2,y_3), P_4(y_2,y_3,y_4), \,\ldots, P_m(y_2,\ldots, y_m)     \Big) \,.   $$
In the above, the variables $y_j$ are indexed according to the number $j$ of the moment, $ \sigma^{(j)}_{k,n}$, that they correspond to.  The polynomials $P_j$ have nonnegative coefficients and are thus nondecreasing in each variable $y_i$ for $i\in \{2,\ldots, j\}  $ on the subdomain $[0,\infty)^{j-1}$; see Lemma~\ref{LemPoly} for some additional properties of these polynomials.

Let $\vec{H}_m:(0,\infty)\rightarrow \R^{m-2}$ be defined as below for the limiting moment functions $R^{(j)}:\R\rightarrow [0,\infty)$ from Theorem~\ref{ThmHM}:
$$ \vec{H}_m(x)\,:=\,\Big(R^{(3)}\big(R^{-1}(x)\big),\ldots, R^{(m)}\big(R^{-1}(x)\big)\Big)\,,      $$
where $x > 0$ and $R^{-1}$ is the inverse of the variance function $R\equiv R^{(2)}$.  In other terms, $\vec{H}_m$ determines the vector of limiting higher moments with $3\leq j\leq m$ from the variance $x$.  In Definition~\ref{DefFH}, we use the functions $\vec{P}_m$ and $M$  to  construct functions $\vec{H}^{(k)}_{m}(x; y)$ from $(0,\infty)\times \R^{m-2}$ to  $\R^{m-2}$ that converge pointwise with large $k\in \mathbb{N}$  to $\vec{H}_m(x)$ when $y$ has small enough norm by~\cite[Lemma 3.2]{Clark1}.  For the purpose of proving Proposition~\ref{PropUnif}, the relevant properties of the functions $\vec{H}^{(k)}_{m}$ are the identities in Remarks~\ref{RemarkSig}~\&~\ref{RemarkR} below and the bound on their derivatives in Lemma~\ref{LemBefore}.  As before, $M^{-k}$ denotes the $k$-fold composition of the function inverse of $M$.
\begin{definition}\label{DefFH} For $m\in \{3,4,\ldots\}$, let $\vec{P}_m:\R^{m-1}\rightarrow \R^{m-2}$ be the vector of polynomials determined by~(\ref{GG}).  Given $x>0$ and $k\in \mathbb{N}$,  define   $F^{(x)}_{k}:\R^{m-2}\rightarrow \R^{m-2}$ such that for  $y=(y_3,\ldots, y_m)\in \R^{m-2}$
$$ F^{(x)}_{k}(y_3,\ldots, y_m)\,:=\,  \vec{P}_m\big(M^{-k}(x),\, y_3,\ldots, y_m\big)  \,.     $$
 Define $\vec{H}^{(k)}_{m}: (0,\infty)\times \R^{m-2}\rightarrow  \R^{m-2}$ through the $k$-fold composition of the maps $F^{(x)}_{j}$ given by
\begin{align*}
\vec{H}^{(k)}_{m}(x; y)\,:=\, F^{(x)}_{1}\circ F^{(x)}_{2}\circ \cdots \circ F^{(x)}_{k}(y)\,.
\end{align*}
We  denote the $(m-2)$-by-$(m-2)$ matrix of first-order derivatives of $\vec{H}^{(k)}_{m}(x; y)$ with respect to the variables $y_j$ for $j\in \{3,\ldots, m\}$  by $\mathbf{D}\vec{H}^{(k)}_{m}(x; y)$.

\end{definition}

\begin{remark}\label{RemarkSig} Let $n\in \mathbb{N}$ and $k\in \{0,\ldots, n\}$.  Since $\sigma^2_{k,n}:=M^{n-k}(\sigma^2_n)$,  the recursive relation~(\ref{GG}) implies the identity
$$\vec{H}^{(n-k)}_{m}\Big(\sigma^2_{k,n}; \,\sigma^{(3)}_n,\ldots, \sigma^{(m)}_n \Big)\,=\, \Big(\sigma^{(3)}_{k,n},\ldots, \sigma^{(m)}_{k,n}\Big)  \,. $$
\end{remark}

\begin{remark}\label{RemarkR}
 Note that $R(r-k)=M^{n-k}\big(R(r-n)\big)$ by part (I) of Lemma~\ref{LemVar}.  Hence part (I) of Theorem~\ref{ThmHM} implies that
$$\vec{H}^{(n-k)}_{m}\Big(R(r-k); \,R^{(3)}(r-n),\ldots, R^{(m)}(r-n) \Big)\,=\, \Big(R^{(3)}(r-k),\ldots, R^{(m)}(r-k)\Big)\,.   $$
\end{remark}

We will use the following simple vector notation. 
\begin{notation}\label{NotationVec} For $d\in \mathbb{N}$ let $y=(y_1,\ldots, y_d)$ and  $y'=(y_1',\ldots, y_d')$ be elements of $\R^d$ and  $A$ be a $d\times d$ real-valued matrix. 
\begin{enumerate}[(i)]
\item We write $y\leq y'$ if the inequality holds component-wise, i.e., $y_j\leq y_j'$  for all $j\in \{1,\ldots, d\}$.

\item $\|y\|_{\infty}$ denotes the max norm of $y$, i.e., $\|y\|_{\infty}=\max_{1\leq j\leq d}|y_j|$.

\item $\|A\|$ is the operator norm with respect to the  max norm on $\R^d$, i.e., $\|A\|=\max_{1\leq i\leq d}\sum_{j=1}^d|A_{i,j}| $.
\end{enumerate}
\end{notation}
\begin{remark}In the sense of (i) in Notation~\ref{NotationVec}, we will refer to a function $f:\R^d\rightarrow \R^{d} $ as being \textit{nondecreasing} on a subdomain $D\subset \R^d$ if $f(y_1)\leq f(y_2)$ holds for  all $y_1, y_2\in D$ with $y_1\leq y_2$. 
\end{remark}

\begin{remark}\label{RemarkInc}  Since the polynomials $P_j$ have nonnegative coefficients, $\vec{P}_m$ is nondecreasing on $[0,\infty)^{m-1}$.  Since $M$  is increasing,  it follows from the construction in Definition~\ref{DefFH} that $\vec{H}^{(k)}_{m}(x; y)$ is also nondecreasing on $[0,\infty)^{m-1}$.
\end{remark}

The  lemma below from~\cite[Eqn.\ 3.8]{Clark1} implies that the function $\vec{H}^{(k)}_{m}(x;y)$ is essentially independent of $y\in \R^{m-2} $ when $k\gg 1$ and $(x;y)\in (0,\infty)\times \R^{m-2}$ is restricted to a small region around the origin. 

\begin{lemma}\label{LemBefore} For any $m\in \{3,4,\ldots\}$, there is an $\epsilon\equiv \epsilon(m)>0$ such that for all $k\in \mathbb{N}$:
$$  \sup_{\substack{ x\leq \epsilon  \\ \|y\|_{\infty}\leq \epsilon  } }\big\|\big(\mathbf{D}\vec{H}^{(k)}_{m}\big)(x;y)\big\| \leq \Big(\frac{b+1}{2b}  \Big)^{k} \,. $$

\end{lemma}

\begin{proof}[Proof of Proposition~\ref{PropUnif}]Part (i):  Pick any $r_{\downarrow},r_{\uparrow}\in \R$ with $  r_{\downarrow}<r<r_{\uparrow}$.  Since $\sigma_{ n}^{2}=\textup{Var}\big(X_h^{(n)}\big)$ has the large $n$ asymptotics~(\ref{VARAsym}) and $R(s)$ has the asymptotics in (II) of Lemma~\ref{LemVar} as $s\rightarrow -\infty$, we have the following inequality for all $n$ larger than some $\widetilde{n}>0$
\begin{align}
    R(r_{\downarrow}-n)  < \, \sigma_{ n}^{2}\,< R(r_{\uparrow}-n)\,.
\end{align}
Thus for any $n>\widetilde{n}$ and  $k\in \{0,\ldots, n\}$ the relations below hold:
\begin{align}\label{Mic}
R(r_{\downarrow}-k)  \,=\,  M^{n-k}\big( R(r_{\downarrow}-n)\big)  < \, \sigma_{k, n}^{2}\,< M^{n-k}\big( R(r_{\uparrow}-n)\big)\,=\,R(r_{\uparrow}-k)  \,,
\end{align}
where we have used  (I) of Lemma~\ref{LemVar}, the definition $\sigma_{k, n}^{2}:=M^{n-k}\big(  \sigma_{ n}^{2} \big)$,  and that $M$ is increasing. Since $R(s)\sim \frac{\kappa^2}{-s}$ for $-s\gg 1$ and $R$ takes values in $(0,\infty)$, the terms $R(r_{\uparrow}-k) $ and $R(r_{\downarrow}-k)$ are respectively bounded from above and below by positive multiples, $c_{\uparrow}$ and $c_{\downarrow}$, of $\frac{1}{k+1}$.  Thus we have $\frac{c_{\downarrow}}{k+1}< \sigma_{k, n}^{2} <\frac{c_{\uparrow}}{k+1}$ for all $n> \widetilde{n}$ and $k\in \{0,\ldots, n\}$. Since there are only finitely many $k,n\in \mathbb{N}_0$ with $0\leq k\leq n\leq \widetilde{n} $, the inequalities $\frac{c_{\downarrow}}{k+1}< \sigma_{k, n}^{2} <\frac{c_{\uparrow}}{k+1}$ can be extended to all $k, n$ by choosing the constants $c_{\uparrow},c_{\downarrow}$ to be larger/smaller if needed. \vspace{.25cm}

\noindent Part (ii): Let $\epsilon\in (0,1/2)$ be small enough to satisfy the conclusion of Lemma~\ref{LemBefore} with $m=4$, and fix any $r^{\uparrow}\in (r,\infty)$.  Let $\widetilde{n}\in \mathbb{N}$ be  large enough such that statements (a)-(c) below hold for all  $n\in \mathbb{N}$ with $n > \widetilde{n}$.
\begin{enumerate}[(a)]
\item $\sigma^{(2)}_{k,n}\leq R(r^{\uparrow}-k)   $ for all $k\in \{0,\ldots, n\}$,

\item $\max_{m\in \{3,4\}}\big|\sigma^{(m)}_{n}\big|<\epsilon$, and

\item $\max_{m\in \{2,3,4\}} R^{(m)}(r^{\uparrow}-n)  <\epsilon$.

\end{enumerate}
To see that $\widetilde{n}$ exists, notice the following: statement (a) holds for large $n$ by the reasoning leading to~(\ref{Mic});
statement (b) holds for large enough $n$ as a consequence of  our  minimal regularity assumption that the fourth moments $\sigma^{(4)}_{n}$  vanish as $n\rightarrow \infty$; statement (c) holds for large enough $n$ since $R^{(m)}(s)$ vanishes as $s\searrow-\infty$ for each $m\in \{2,3,4\}$ by (II) of Theorem~\ref{ThmHM}; .

Since there are only finitely many terms $\sigma^{(4)}_{k,n}$  with $n \leq  \widetilde{n}$, we can focus on the case that $n > \widetilde{n}$.  For $n > \widetilde{n}$, let $k^*$ be the smallest element of $\{0,\ldots, n\}$ such that $R^{(4)}(r^{\uparrow}-k^*)<\epsilon$.  Note that $k^*$   much exist as a consequence of (c).  Since $\sigma^{(4)}_{k,n}$ converges to $R^{(4)}(r-k)$ as $n\rightarrow \infty$ for each  $k\in \mathbb{N}$ by (III) of Lemma~\ref{LemmaMom}, the following is finite:
\begin{align}
\max_{0\leq k \leq  k^* }\sup_{n\geq k}\sigma^{(4)}_{k,n}\,<\,\infty \,.
\end{align}
Thus it suffices for us to assume that $k >  k^*$ in the remainder of the proof.

  Let $n > \widetilde{n}$ and  $k\in \{k^*,\ldots, \lfloor n/2 \rfloor \}$.  The  equality below is the $m=4$ case of the identity in Remark~\ref{RemarkSig}. 
\begin{align}\label{qz}
 \big(\sigma^{(3)}_{k,n}, \sigma^{(4)}_{k,n} \big)\,= \,\vec{H}^{(n-k)}_{m}\left(\sigma^{(2)}_{k,n}; \sigma^{(3)}_{n}, \sigma^{(4)}_{n} \right)  \,\leq \,\vec{H}^{(n-k)}_{m}\Big(R(r^{\uparrow}-k); \big|\sigma^{(3)}_{n}\big|, \sigma^{(4)}_{n} \Big)
\end{align}
The inequality above holds by statement (a) and Remark~\ref{RemarkInc}.  Since statements (b) and (c) imply that $R(r^{\uparrow}-k)<\epsilon$, $\big\| \big(\big|\sigma^{(3)}_{n}\big|, \sigma^{(4)}_{n} \big) \big\|_{\infty}<\epsilon$, and $\big\|\big( R^{(3)}(r^{\uparrow}-n)  , R^{(4)}(r^{\uparrow}-n)  \big)  \big\|_{\infty}<\epsilon$, we can apply Lemma~\ref{LemBefore} to get the first inequality below.
\begin{align}
&\Big\| \vec{H}^{(n-k)}_{m}\Big(R(r^{\uparrow}-k); \big|\sigma^{(3)}_{n}\big|, \sigma^{(4)}_{n} \Big)- \underbracket{\vec{H}^{(n-k)}_{m}\Big(R(r^{\uparrow}-k); R^{(3)}(r^{\uparrow}-n)  , R^{(4)}(r^{\uparrow}-n)  \Big)}  \Big\|_{\infty}\nonumber  \\  &\,\leq\, \Big(\frac{b+1}{2b}  \Big)^{n-k}\Big\|\Big( \big|\sigma^{(3)}_{n}\big|, \sigma^{(4)}_{n} \Big)\,-\,\Big(R^{(3)}(r^{\uparrow}-n)  , R^{(4)}(r^{\uparrow}-n)  \Big)   \Big\|_{\infty}\,\leq\, 2\epsilon \Big(\frac{b+1}{2b}  \Big)^{n-k}\label{LH}
\end{align}
By Remark~\ref{RemarkR}, the bracketed term is equal to $\Big( R^{(3)}(r^{\uparrow}-k)  , R^{(4)}(r^{\uparrow}-k)  \Big)$, and   thus the inequality~(\ref{LH}) implies that  
\begin{align}
   \vec{H}^{(n-k)}_{m}\Big(R(r^{\uparrow}-k); \big|\sigma^{(3)}_{n}\big|, \sigma^{(4)}_{n} \Big)
      \, \leq  \,  \Big( R^{(3)}(r^{\uparrow}-k)  , R^{(4)}(r^{\uparrow}-k)  \Big)\,+\,2\epsilon \Big(\frac{b+1}{2b}  \Big)^{n-k} (1,1) \,, \label{Cati}
\end{align}
where $(1,1)$ refers to the vector in $\R^2$.  
Combining the vector inequalities~(\ref{qz}) and~(\ref{Cati})  yields the following  for the second components of the vectors:
$$ \sigma^{(4)}_{k,n}\,\leq \,R^{(4)}(r^{\uparrow}-k)\,+\,2\epsilon \Big(\frac{b+1}{2b}  \Big)^{n-k}\,< \,R^{(4)}(r^{\uparrow}-k)\,+\,\Big(\frac{b+1}{2b}  \Big)^{ \frac{n}{2}}\,.  $$
For the second inequality, we  have used that $\epsilon < 1/2$ and that $k\leq \lfloor n/2\rfloor$.  Since $R^{(4)}(s)$ is $\mathit{O}\big(\frac{1}{s^2}\big)$ for $-s\gg 1$ by part (II) of Theorem~\ref{ThmHM}, $R^{(4)}(r^{\uparrow}-k)$ is bounded by a constant multiple of $\frac{1}{(k+1)^2}$ for all $k\in \mathbb{N}_0$.    
Also, $\big(\frac{b+1}{2b}  \big)^{ n/2}$, which decays exponentially in $n$,  is bounded from above by a multiple of $\frac{1}{(k+1)^2}$ for all $k\leq n$.  Thus we have the desired inequality when  $n >\widetilde{n}$ and $k\in \{k^*,\ldots, \lfloor n/2\rfloor\}$, which completes the proof.
\end{proof}

\subsection{Proof of  Lemma~\ref{LemFourTerms}}\label{LemmaFourTerms}

Let $\{x_a\}_{a\in E_n}$ be an array of centered random variables with finite fourth moments, and define $Y_{\ell}:=\mathcal{L}^{\ell-1}\mathcal{E}\mathcal{L}^{n-\ell}\{x_a\}_{a\in E_n}$ for $\ell\in \{1,\ldots, n\}$.  Recall that  Lemma~\ref{LemFourTerms} states that  $\mathbb{E}\big[ \big( \sum_{\ell=1}^n Y_{\ell}   \big)^4  \big]$ is bounded by a constant multiple of $n\sum_{\ell=1}^n \mathbb{E}\big[Y_{\ell}^4\big] $. 
\begin{notation}
For distinct $a_1,a_2\in E_n$, let $\gamma(a_1,a_2)$ denote the smallest value of $k\in \{1,\ldots, n\}$ such that there exist distinct $\mathbf{b}_1,\mathbf{b}_2\in E_k$ with $a_1\in \mathbf{b}_1$ and  $a_2\in \mathbf{b}_2$.  When $a_1=a_2$, we define $\gamma(a_1,a_2)=\infty$.
\end{notation}
\begin{remark}\label{RemarkSubset} Let $S\subset E_n$ and $\mathbf{b}\in E_k$ for $1\leq k< n$.  If $\mathbf{b}\cap S \neq \emptyset$ and $\gamma(a_1,a_2)>k$ for all $a_1,a_2\in S$, then $S\subset \mathbf{b}$.
\end{remark}

\begin{remark}   Let  $S_1,S_2\subset E_n$. If $\gamma(a_1,a_2)$ is independent of  $a_1\in S_1$ and $a_2\in S_2$, then we define $\gamma(S_1,S_2):=\gamma(a_1,a_2)$ for  $a_1\in S_1$ and $a_2\in S_2$.
\end{remark}

\begin{proof}[Proof of Lemma~\ref{LemFourTerms}] Let $\sigma^2$ denote the variance of the  variables $x_a$,  $a\in E_n$. By foiling, we get 
\begin{align}\label{ManyTerms}
 \mathbb{E}\Bigg[ \bigg( \sum_{\ell=1}^n Y_{\ell}   \bigg)^4  \Bigg]\,=\,& \sum_{1\leq \ell_1, \ell_2,  \ell_3, \ell_4\leq n} \mathbb{E}\big[Y_{\ell_1}Y_{\ell_2}Y_{\ell_3}Y_{\ell_4}   \big]\nonumber \\  \,=\,& \sum_{\ell=1}^n \mathbb{E}\big[ Y_{\ell}^4     \big]\,+\,4 \sum_{\ell=1}^n \sum_{\ell <l\leq n } \mathbb{E}\big[\underbracket{Y_{\ell}^3Y_{l}} \big]\,+\,6\sum_{\ell=1}^n \sum_{\ell <l\leq n} \mathbb{E}\big[  \underbracket{Y_{\ell}^2Y_{l}^2}    \big]\nonumber \\ 
&\,+\,12\underbrace{\sum_{\ell=1}^n\sum_{\ell<l_1<l_2\leq n} \mathbb{E}\big[Y_{\ell}^2Y_{l_1}Y_{l_2} \big]}_{ \text{bounded by a multiple of $\big(\sigma^2\big)^4$}  }\,+\, 4\underbrace{\sum_{ \ell=1}^{n}\sum_{\ell< l_1,l_2, l_3\leq n  } \mathbb{E}\big[Y_{\ell}Y_{l_1}Y_{l_2}Y_{l_3} \big]}_{=\,0} \,.
\end{align}
By applying Young's inequality, $|xy|\leq \frac{1}{p}|x|^p+\frac{1}{q}|x|^{q}$, to the  bracketed   products above with  $(p,q)=(\frac{4}{3},4)$ and $(p,q)=(2,2)$, respectively, we can bound the second and third terms on the right side of~(\ref{ManyTerms}) by multiples of $n\sum_{\ell=1}^n \mathbb{E}\big[ Y_{\ell}^4     \big]$.  In the analysis below, we will show that $\mathbb{E}\big[Y_{\ell}Y_{l_1}Y_{l_2}Y_{l_3} \big]=0$ when $\ell < l_1,  l_2, l_3 $ and  thus that the last term on the right side of~(\ref{ManyTerms}) is zero.  We will also show that $\mathbb{E}\big[Y_{\ell}^2Y_{l_1}Y_{l_2} \big]$ is bounded by a constant multiple of $ (\sigma^2)^4 b^{-(l_1+l_2)}$ for all $\ell, l_1,l_2\in \mathbb{N}$ with $\ell<l_1<l_2 $, which implies that there are $C,C'>0$ such that the inequalities below hold for all $n\in \mathbb{N}$.
\begin{align*}
\sum_{ \ell=1}^{n}\sum_{\ell<l_1<l_2\leq n} \mathbb{E}\big[Y_{\ell}^2Y_{l_1}Y_{l_2} \big]\,\leq \, C\sum_{ \ell=1}^{n}\sum_{\ell<l_1<l_2\leq n}  \frac{(\sigma^2)^4}{b^{l_1+l_2}} \,\leq \,C^{\prime} (\sigma^2)^4 \,\leq &\,4C^{\prime}\big(M(\sigma^2)-\sigma^2   \big)^2\\ \,= \,& 4C^{\prime}\mathbb{E}\big[ Y_{1}^2  \big]^2 \,\leq \, 4C^{\prime}\mathbb{E}\big[ Y_{1}^4  \big]
\end{align*}
The third inequality holds since $M(x):=\frac{1}{b}\big[(1+x)^b-1\big]\geq x+\frac{b-1}{2}x^2$ for $x\geq 0$ and $b\geq 2$.  The equality holds by Remark~\ref{RemarkArrayVar} since $Y_{1}:= \mathcal{E}\mathcal{L}^{n-1}\{ x_a \}_{a\in E_n } $, and the last inequality is Jensen's.
It follows that the fourth term on the right side of~(\ref{ManyTerms}) is easily bounded by a constant multiple of $n \sum_{\ell=1}^n \mathbb{E}\big[Y_{\ell}^4    \big]$.

 For $1\leq \ell\leq n$ and $\mathbf{a} \in E_{\ell}$, define $ x_{\mathbf{a}}^{(\ell)} :=  \mathcal{L}^{n-\ell}\{ x_a \}_{a\in \mathbf{a}\cap E_n }=\frac{1}{b^{n-\ell}}\sum_{ a\in \mathbf{a}\cap E_{n} } x_a$. The random variable $Y_{\ell}$ can be written in the forms
\begin{align} \label{FormOfY}
Y_{\ell}\,=\,\frac{1}{b^{\ell}}\sum_{\mathbf{a}\in E_{\ell-1}  }\sum_{i=1}^b\Bigg(\prod_{j=1}^b\big(1+x_{\mathbf{a}\times (i,j)}^{(\ell)}   \big) \,-\,1    \Bigg)\,=\,&\frac{1}{b^{\ell}}\sum_{\mathbf{a}\in E_{\ell-1}  }\sum_{i=1}^b\sum_{\substack{A\subset \{1,\ldots, b\} \\  |A|\geq 2   }} \prod_{j\in A} x_{\mathbf{a}\times (i,j)}^{(\ell)} \nonumber  \\ \,=\,&\frac{1}{b^{\ell}}\sum_{\mathbf{a}\in E_{\ell-1}  }\sum_{\substack{ i \in \{1,\ldots, b\} \\ A\subset \{1,\ldots, b\} \\  |A|\geq 2   }} \prod_{j\in A}\Bigg( \frac{1}{b^{n-\ell}  }\sum_{a\in \mathbf{a}\times (i,j)\cap E_n   } x_{a} \Bigg) \,, \nonumber 
\intertext{where we have consolidated the summation over $i$ and $A$ into a single $\sum$.  For $i, \mathbf{a},A$ as above, let  $G_{i,A}^{n, \mathbf{a}}$ denote the set of functions $\phi:A\rightarrow \bigcup_{j\in A}  \mathbf{a}{\times} (i,j)\cap E_n   $ such that $\phi(j)\in \mathbf{a}{\times} (i,j)\cap E_n $ for each $j\in A$.\footnotemark \,\,In other terms, each $\phi\in G_{i,A}^{n, \mathbf{a}}$ is determined by choosing an element of $\mathbf{a} {\times} (i,j)\cap E_n$ for each $j\in A$, and thus $\big|G_{i,A}^{n, \mathbf{a}}\big|=b^{2|A|(n-\ell)}$. Expanding the  product above yields}
 \,=\,&\frac{1}{b^{\ell}}\sum_{\mathbf{a}\in E_{\ell-1}  }\sum_{\substack{ i \in \{1,\ldots, b\} \\ A\subset \{1,\ldots, b\} \\  |A|\geq 2   }}\frac{1}{b^{|A|(n-\ell)}  } \sum_{\phi\in G_{i,A}^{n, \mathbf{a}}}  \prod_{j\in A} x_{\phi(j)}   \,. 
\end{align}
\footnotetext{Thus $G_{i,A}^{n, \mathbf{a}}$ has a canonical one-to-one correspondence with the Cartesian product  $\prod_{j\in A}  \mathbf{a}{\times} (i,j)\cap E_n$.} 
From~(\ref{FormOfY}) we see that  $Y_{\ell}$ is a degree-$b$ multilinear polynomial in the variables $\{ x_a \}_{a\in E_n }$  consisting of a linear combination of  monomials $\prod_{a\in B }  x_a$  for  subsets $B $ of $E_{n}$ satisfying
\begin{enumerate}[(I)]
\item $|B |\geq 2$ and
\item  $\gamma(a_1,a_2)=\ell$ for any distinct $a_1,a_2\in B$.
\end{enumerate}
For numbers $k_{\epsilon} \in \{1,\ldots, n\}$ indexed by $\epsilon\in\{1,2,3,4\}$, let $B_{\epsilon}$ be a subset of $E_{n}$ satisfying (I)-(II)  for  $\ell=k_{\epsilon}$.  The product of the monomials $\prod_{a\in B_{\epsilon} }  x_{a}$  can be written as
\begin{align}\label{ProdTerms}
 \prod_{a_1\in B_{1} }  x_{a_1} \prod_{a_2\in B_{2} }  x_{a_2}\prod_{a_3\in B_{3} }  x_{a_3}\prod_{a_4\in B_{4} }  x_{a_4} \,=\,\prod_{a\in B_{1}\cup B_{2}\cup B_{3}\cup B_{4} }x_a^{\lambda(a)}\,,
\end{align}
where the exponent $\lambda(a)\in\{1,2,3,4\}$ is defined by  $ \lambda(a):=\big|\big\{ j\in\{1,2,3,4\}\,\big|\, a\in B_{j}\big\}\big|$.  The expectation of~(\ref{ProdTerms}) is zero if $\lambda(a)=1$ for some $a\in \cup_{\epsilon} B_{\epsilon} $. The first case below implies $\mathbb{E}\big[ Y_{\ell}Y_{l_1}Y_{l_2}Y_{l_3}]=0$ when $\ell<l_1, l_2, l_3$. \vspace{.2cm}

\noindent \textbf{Case $\mathbf{k_1<k_2,\, k_3,\, k_4}$:} To reach a contradiction, suppose that $k_1<k_2,\, k_3,\, k_4$ and $\lambda(a)\geq 2$ for all $a\in \cup B_{\epsilon}$.  Since  $B_{1} $ satisfies properties (I)-(II) with $\ell=k_1$, there must be distinct $a_1,a_2\in B_{1} $ and  distinct $ \mathbf{b}_1,\mathbf{b}_2\in E_{k_1}$ such that $a_1\in \mathbf{b}_1$ and $a_2\in\mathbf{b}_2$.  By our assumption that $k_1<k_{\epsilon}$ for $\epsilon\in \{2,3,4\}$, property (II) for $B_{\epsilon}$  implies that  for each $\epsilon\in \{2,3,4\}$ we have $\mathbf{b}_1\cap B_{\epsilon} =\emptyset$ or $\mathbf{b}_2\cap B_{\epsilon} =\emptyset$ (since otherwise there exist distinct $c_1,c_2\in B_{\epsilon}$ with $\gamma(c_1,c_2)<k_{\epsilon}$). In particular, $ \mathbf{b}_{1}$ or $ \mathbf{b}_{2}$ is disjoint from the sets $B_{\epsilon}$ for at least two values of $\epsilon\in \{2, 3, 4\} $.  Without loss of  generality, we can assume that $ \mathbf{b}_1\cap (B_3\cup B_4)=\emptyset$ and consequently that $a_1\notin B_3\cup B_4$.  Since  $a_1\notin B_{3}$ and  $a_1\notin B_{4}$,  we must have $a_1\in B_{2}$ to ensure  that $\lambda(a_1)\geq 2$. Thus $\mathbf{b}_1\cap B_2\neq \emptyset$.  Note that  $B_{2}\subset\mathbf{b}_{1}$, by Remark~\ref{RemarkSubset},  because $\mathbf{b}_1\cap B_2\neq \emptyset $, and $B_2$ satisfies property (II) with $\ell=k_2$ and $k_2>k_1$.  By properties (I)-(II) for $B_2$,  there exists $b\in B_2$ with $b\neq a_1$  and $\gamma(a_1,b)=k_2$. Since $a_1\in B_1$ and $\gamma(a_1,b)=k_2>k_1$, it follows from property (II) for $B_1$ that $b\notin B_1$.  Also  $b\notin B_3\cup B_4$ since $b\in B_2\subset \mathbf{b}_{1}$ and  $\mathbf{b}_1\cap (B_3\cup B_4) =\emptyset$.  To summarize, $b\in B_1$, but $b\notin B_\epsilon$ for all $\epsilon\in \{2,3,4\}$.  Therefore,  $\lambda(b)=1$, which is a contradiction.

\begin{remark}\label{RemarkBasicIdea} To summarize the above contradiction proof, both  $\mathbf{b}_1$ and $\mathbf{b}_2$ need to have $B_{\epsilon}\subset \mathbf{b}_1$ for at least two values of $\epsilon\in \{2,3,4\}$ to avoid having $b\in \cup_{\epsilon} B_{\epsilon} $ with $\lambda(b)=1$, however, this is inconsistent with $\mathbf{b}_1, \mathbf{b}_2\in E_{k_1}$ being distinct and thus disjoint when viewed as subsets of $E_n$. 
\end{remark}

\noindent \textbf{Case $\mathbf{k_1=k_2<k_3<k_4}$:} Let $B_{\epsilon}$ for $\epsilon\in\{1,2,3,4\}$ satisfy properties (I)-(II) above respectively for $\ell=k_{\epsilon}$ with  $k_1=k_2<k_3<k_4$. There are two special types---see ($\textup{I}^{\prime}$)-($\textup{II}^{\prime}$) below---of configurations  of the sets $B_{\epsilon}$ such that $\lambda(a)\geq 2$  for all $a\in \cup_{\epsilon}B_{\epsilon}$. For both types, $|B_1|=|B_2|$ and $|B_3|=|B_4|=2$.



\begin{itemize}

\item[($\textup{I}^{\prime}$)]  There exists $\mathbf{a}\in E_{k_1-1}$  and  distinct $\mathbf{b}_1,\mathbf{b}_2\in \mathbf{a}\cap E_{k_1}  $ such that $B_3\subset \mathbf{b}_1$ and $B_4\subset \mathbf{b}_2$.  The sets in the collection $\mathscr{P}:=\big\{B_{\epsilon}\cap B_{\delta}\,\big|\, \epsilon\in \{1,2\},\, \delta\in \{3,4\} \big\}$ are pairwise disjoint, have cardinality one, and their union is equal to $B_3\cup B_4=B_1\Delta B_2$.  In particular, $B_3\cap B_4=\emptyset$ and $|B_1\cap B_2|=|B_1|-2$.

\item[($\textup{II}^{\prime}$)]  There exists  $\mathbf{a}\in  E_{k_3-1}  $ and   $\mathbf{b}\in \mathbf{a}\cap E_{k_3}$ such that $B_3\subset \mathbf{a} $ and $B_4\subset \mathbf{b}$.  The sets $B_1\backslash B_2$, $B_2\backslash B_1$, $B_3\backslash B_4$, $B_4\backslash B_3$ have cardinality one and $B_1\Delta B_2=B_3\Delta B_4$.  In particular, $|B_3\cap B_4|=1$ and $|B_1\cap B_2|=|B_1|-1$.

\end{itemize} 
The types ($\textup{I}^{\prime}$) and ($\textup{II}^{\prime}$) correspond to the cases of $\gamma(B_3,B_4)=k_1$ and $\gamma(B_3,B_4)>k_1$, respectively.  The possibility $\gamma(B_3,B_4)<k_1$ can be excluded because it results in multiple $b\in \cup_{\epsilon}B_{\epsilon}$ with $\lambda(b)=1$ by simpler reasoning than  in the case $k_1<k_2,k_3,k_4$ discussed above. 

 To understand the type-($\textup{I}^{\prime}$) configuration, notice that the intersections $ B_{\epsilon}\cap B_{\delta}$ for  $\epsilon\in \{1,2\}$ and $\delta\in \{3,4\}$ contain at most one element since $B_\epsilon$ and $B_{\delta}$  satisfy property (II)  for $\ell=k_1=k_2$ and $\ell>k_1$, respectively.  Thus $B_1$ and $B_2$ can each contribute at most one to each of the sums $\sum_{a\in B_3}\lambda(a) $ and  $\sum_{a\in B_4}\lambda(a)$.  Since $B_3\cap B_4=\emptyset$ (because $B_3\subset \mathbf{b}_1$ and    $B_4\subset \mathbf{b}_2$ for distinct $\mathbf{b}_1,\mathbf{b}_2\in E_{k_1}$), it is only possible that  $\lambda(a)\geq 2$ for all $a\in B_3\cup B_4$ if $|B_3|=|B_4|=2$ and  the collection  $\mathscr{P}:=\big\{B_{\epsilon}\cap B_{\delta}\,\big|\,\epsilon\in \{1,2\},\,\delta\in \{3,4\}\big\} $ is a partition of $B_3\cup B_4$ comprised of single-element sets. Similarly, $B_1\Delta B_2:=(B_1\backslash B_2) \cup (B_2\backslash B_1)$ must be a subset of $B_3\cup B_4$ to avoid having an  $a\in B_1\cup B_2$ with $\lambda(a)=1$. Since the sets in $\mathscr{P}$ are disjoint and have union equal to $B_3\cup B_4$, it follows that $B_1\Delta B_2=B_3\cup B_4$.  Finally, $|B_1|=|B_2|$ and $|B_1\cap B_2|=|B_1|-2$ since sets in $\mathscr{P}$ have cardinality one.  

To derive the type-($\textup{II}^{\prime}$) configuration, suppose that there is a single $\mathbf{b^{\prime}}\in E_{k_1}$ such that $B_3,B_4\subset \mathbf{b^{\prime}}$.  Since $B_1$ and $B_2$ satisfy property (II) with $\ell=k_1=k_2$, the sets $B_1\cap \mathbf{b^{\prime}} $ and  $B_2\cap \mathbf{b^{\prime}} $ contain at most one element.  It follows that $B_1$ and $B_2$ can each contribute at most one to the sum $\sum_{a\in B_3\Delta B_4} \lambda(a) $.  Since  $B_3$ and $B_4$ satisfy property (II) respectively for $\ell=k_3$ and $\ell=k_4$ with $k_3<k_4$, the set $B_3\cap B_4$ has at most one element.  Under these constraints, it is only possible that $\lambda(a)\geq 2$ for all $a\in B_3\cup B_4$ if $|B_3|=|B_4|=2$, $|B_3\cap B_4|=1$, and the sets $(B_3\Delta B_4)\cap B_{1} $ and $(B_3\Delta B_4)\cap B_{2} $ have cardinality one and  are disjoint. Since $B_3,B_4\subset \mathbf{b^{\prime}}\in E_{k_1}$ and $B_1,B_2$ satisfy property (II) with $\ell=k_1=k_2$, the sets $B_3$, $B_4$ jointly contribute at most one to each of the sums $\sum_{a\in B_1\backslash B_2 }\lambda(a)$ and $\sum_{a\in B_2\backslash B_1 }\lambda(a)$. In order for $\lambda(a)\geq 2$ for all  $a\in B_1\Delta B_2$, it must be that $|B_2\backslash B_1|=1$ and $|B_1\backslash B_2|=1$.  Hence $|B_1\cap B_2|=|B_1|-1=|B_2|-1$.  Since $B_3$ and $B_4$ satisfy property (II) respectively for $\ell=k_3$ and $\ell=k_4>k_3$ with $B_3\cap B_4\neq \emptyset$, there exists $\mathbf{a}\in E_{k_3-1}$ and $\mathbf{b}\in E_{k_3}$ such that $B_3\subset \mathbf{a}$ and   $B_4\subset \mathbf{b}\subset \mathbf{a}$. \vspace{-.2cm}\\
\indent Next we bound the expectation of  $Y_{k_1} Y_{k_2}  Y_{k_3} Y_{k_4} $ when $k_1=k_2<k_3<k_4$.  Using the formula~(\ref{FormOfY}), we can write
\begin{align}\label{YProduct}
\mathbb{E}\big[  Y_{k_1} Y_{k_2}  Y_{k_3} Y_{k_4}\big]\,=\,& \mathbb{E}\Bigg[ \prod_{\epsilon=1}^4\Bigg(\frac{1}{b^{k_\epsilon}}\sum_{\mathbf{a}_{\epsilon}\in E_{k_\epsilon-1}  }\sum_{\substack{ i_\epsilon\in\{1,\ldots,b   \}  \\  A_\epsilon\subset \{1,\ldots, b\} \\  |A_\epsilon|\geq  2   }}\frac{1}{b^{|A_\epsilon|(n-k_\epsilon)}  } \sum_{\phi_\epsilon\in G_{i_{\mathsmaller{\mathsmaller{\epsilon}}},A_\epsilon}^{n,\mathbf{a}_\epsilon}}  \prod_{j_{\epsilon}\in A_\epsilon} x_{\phi_\epsilon(j_\epsilon)}   \Bigg) \Bigg] \\
 \,=\,& \big(\text{Contribution from type-($\textup{I}^{\prime}$) terms}\big)\,+\, \big(\text{Contribution from type-($\textup{II}^{\prime}$) terms}\big)\,,\nonumber \end{align}
where the second equality holds by foiling the product over $\epsilon \in \{1,2,3,4\}$ by our observations above. The type-($\textup{I}^{\prime}$) and type-($\textup{II}^{\prime}$) contributions to~(\ref{YProduct}) both yield multiples of $b^{-(k_3+k_4)}$.  The cases are similar, so we will discuss only the type-($\textup{I}^{\prime}$) case.

 When the product over $\epsilon \in \{1,2,3,4\}$  inside the expectation in~(\ref{YProduct}) is foiled, only the terms with $\mathbf{a}_1=\mathbf{a}_2$, $i_1=i_2$, $A_1=A_2$ can be of type-($\textup{I}^{\prime}$) or type-($\textup{II}^{\prime}$) and thus nonzero.  In the type-($\textup{I}^{\prime}$) case, there are distinct $j,J\in A_1$ such that $\mathbf{a}_3\in (\mathbf{a}_1{\times} (i_1,j))\cap E_{k_3}$ and $\mathbf{a}_4\in (\mathbf{a}_1{\times} (i_1,J))\cap E_{k_4}$, where $\mathbf{a}_1{\times} (i_1,j)$ and $\mathbf{a}_1{\times} (i_1,J)$ have the roles of $\mathbf{b}_1$ and $\mathbf{b}_2$, respectively, in the statement of ($\textup{I}^{\prime}$). The type-($\textup{I}^{\prime}$) contribution has the form
\begin{align}\label{MultiSum}
&\underbrace{\sum_{ \mathbf{a}_1\in E_{k_1-1}}}_{(\textup{i})} \underbrace{\sum_{ \substack{ i_1\in \{1,\ldots, b\} \\ A_1\subset \{1,\ldots, b\}\\ |A_1|\geq 2   }}}_{(\textup{ii})}\underbrace{\sum_{\substack{  j,J\in A_1 \\ j\neq J  } }}_{(\textup{iii})} \underbrace{\sum_{ \substack{  \mathbf{a}_3\in (\mathbf{a}_1\times (i_1,j))\cap E_{k_3-1}  \\ \mathbf{a}_4\in (\mathbf{a}_1\times (i_1,J)) \cap E_{k_4-1}   } }   }_{(\textup{iv})}\underbrace{\sum_{\substack{i_3,i_4\in \{1,\ldots, b\}\\ A_3,A_4\subset \{1,\ldots, b\} \\    |A_3|=|A_4|= 2   }} }_{(\textup{v})}  \underbrace{\sum_{\substack{\phi_1,\phi_2\in G_{i_1,A_1}^{n,\mathbf{a}_1}  \\ \phi_3\in G_{i_3,A_3}^{n,\mathbf{a}_3}  \\  \phi_4\in G_{i_4,A_4}^{n,\mathbf{a}_4}  }}  \frac{\eta(\phi_1,\phi_2,\phi_3,\phi_4)(\sigma^2)^4}{ \prod_{\epsilon=1}^{4}b^{k_{\epsilon}+|A_{\epsilon}|(n-k_{\epsilon})   }   } }_{(\textup{vi})}    \,,
\end{align}
 where we interpret $|A_2|:=|A_1|$ inside the product $\prod_{\epsilon=1}^{4}$, and   $\eta(\phi_1,\phi_2,\phi_3,\phi_4) \in \{0,1\} $ is defined as
\begin{align*} \eta(\phi_1,\phi_2,\phi_3,\phi_4)\,:=\, 1_{\big\{ \phi_1=\phi_2\text{ on } A\backslash\{j,J\}   \text{ and } \phi_1 \neq \phi_2\text{ on } \{j,J\} \big\} } 1_{\big\{\phi_{1}(j),\phi_2(j)\in \textup{Rng}(\phi_3) \big\} }1_{\big\{\phi_{1}(J),\phi_2(J)\in \textup{Rng}(\phi_4) \big\} } \,.  
\end{align*}
Note that the sets  $\textup{Rng}(\phi_3)=\phi_3(A_3)$ and $\textup{Rng}(\phi_4)=\phi_4(A_4)$ in the definition of $\eta(\phi_1,\phi_2,\phi_3,\phi_4)$ both contain exactly two elements. There are respectively $|G_{i_3,A_3}^{n,\mathbf{a}_3}|=b^{2|A_3|(n-k_3)}=b^{4(n-k_3)} $ and  $|G_{i_4,A_4}^{n,\mathbf{a}_4}|=b^{2|A_4|(n-k_4)}=b^{4(n-k_4)} $ choices for the functions $\phi_3$ and $\phi_4$.  When  $ \phi_3\in G_{i_3,A_3}^{n,\mathbf{a}_3} $ and $ \phi_4\in G_{i_4,A_4}^{n,\mathbf{a}_4} $ are given, there are $4b^{2(|A_1|-2)(n-k_1)  }$ combinatorial possibilities for the  pair of functions $\phi_1,\phi_2\in G_{i_1,A_1}^{n,\mathbf{a}_1}$ such that $\eta(\phi_1,\phi_2,\phi_3,\phi_4)=1$,  where the factor of $4$ comes from the assignment choices for $\phi_1,\phi_2$ on the subdomain $\{j,J\}$.   For the purpose of evaluating~(\ref{MultiSum}), it will be convenient  to reformulate the sums (\textup{ii})-(\textup{iii}) as
$$ \underbrace{\sum_{ \substack{ i_1\in \{1,\ldots, b\} \\ A_1\subset \{1,\ldots, b\}\\ |A_1|\geq 2   }}}_{(\textup{ii})}\underbrace{\sum_{\substack{  j,J\in A_1 \\ j\neq J  } }}_{(\textup{iii})} \, \,\equiv\, \, \underbrace{\sum_{ \substack{ i_1\in \{1,\ldots, b\} \\ j,J\in \{1,\ldots, b\} \\ j\neq J  }}}_{(\textup{ii}^{\prime})}\underbrace{\sum_{\substack{ A_1\subset \{1,\ldots, b\} \\\{j,J\}\subset A_1  } }}_{(\textup{iii}^{\prime})}   $$  The summation~(\ref{MultiSum}) is equal to
\begin{align} 
\underbrace{b^{2(k_1-1)}}_{ (\textup{i})  } \underbrace{b^2(b-1)}_{(\textup{ii}^{\prime} )  }  \underbrace{b^{2(k_3-1-k_1)}b^{2(k_4-1-k_1)}}_{(\textup{iv})  }\underbrace{\Big(\frac{b^2(b-1)}{2}\Big)^2}_{(\textup{v})  } \underbrace{\sum_{\substack{ A_1\subset \{1,\ldots, b\} \\\{j,J\}\subset A_1  } }\frac{  4b^{2(|A_1|-2)(n-k_1)  } b^{ 4(n-k_3)  } b^{ 4(n-k_4)  } (\sigma^2)^4 }{b^{2k_1+k_3+k_4+2|A_1|(n-k_{1})+2(n-k_3)+2(n-k_4)   }    }}_{ (\textup{iii}^{\prime})\,\, \& \,\, (\textup{vi})}\,,
\nonumber 
\end{align}
where the sum  is independent of a particular choice of $j,J\in \{1,\ldots,b\}$ with $j\neq J$. Moreover, the sum has $2^{b-2}$ terms and the summand is indepedent of $|A_1|$ because of the  cancellation of $b^{2|A_1|(n-k_1)  }$ between the numerator and the denominator.   The product above  is equal to $2^{b-2}(b-1)^3(\sigma^2)^4/b^{k_3+k_4}$. \end{proof}

\subsection{Proof of Lemma~\ref{LemVarApprox}}\label{SecSharpVC}

In this section, we will prove the following lemma, which uses more restrictive assumptions on the asymptotics for $x\equiv x^{n,r}$ in~(\ref{xAssump}) to gain more explicit control of the error in the convergence of $M^{n}(x)$ to $R(r)$ as $n\rightarrow \infty$  in Lemma~\ref{LemVar}.  Recall from Remark~\ref{RemarkArrayVar} that if  the random variables in an  i.i.d.\ array $\big\{ X_{h}^{(n)} \big\}_{h\in  E_{n} }$ are centered with variance $x$, then the  random variables in the array  $  \big\{X_{a}^{(k,n)}\big\}_{a\in E_k}:=\mathcal{Q}^{n-k}\big\{ X_{h}^{(n)} \big\}_{h\in  E_{n} }  $ have variance  $M^{n-k}(x)$. It follows that Lemma~\ref{LemmaVarApprox} below is equivalent to  Lemma~\ref{LemVarApprox}. 
\begin{lemma}\label{LemmaVarApprox} Fix $\mathbf{v}>0$, $\alpha\in (0,1)$, and a bounded interval $\mathcal{I}\subset \R$. There exists $C_{\mathcal{I},\mathbf{v},\alpha} >0$ such that for any  $x >0$, $n\in \mathbb{N}$, and  $r\in \mathcal{I}$ satisfing the inequality \begin{align}\label{SigmaBound}\left|x-\kappa^2\Big(\frac{1}{n}+\frac{\eta \log n  }{n^2}+\frac{ r  }{ n^2} \Big) \right|\leq \frac{\mathbf{v}}{ n^{2+\alpha}  }\,,\end{align}   the following  inequality holds:\begin{align}\label{UJ}\max_{0\leq k\leq n} \left|M^{n-k}(x) \,-\,   R(r-k) \right|\,\leq \,\frac{C_{\mathcal{I},\mathbf{v},\alpha}}{ n^{\alpha}  } \,.\end{align}\end{lemma}
 The proof of Lemma~\ref{LemmaVarApprox} will rely on an application of Lemma~\ref{LemMaybe}.
\begin{proof} Let $\mathbf{v}>0$, $\alpha\in (0,1)$, and $\mathcal{I}$ be a bounded interval in $\R$.  As a preliminary, note that the asymptotic form for $R(s)$ as $s\rightarrow -\infty$ in (II) of Lemma~\ref{LemVar} implies that there exists a $C_{\mathcal{I},\alpha}>0$ such that for all $r\in \mathcal{I}$ and $n\in \mathbb{N}$
\begin{align}\label{RAsympt}
\Big|\kappa^2\Big(\frac{1}{n}+\frac{\eta \log n }{ n^2  }+\frac{r}{n^2}  \Big)\,-\, R(r-n)    \Big| \,\leq \,\frac{C_{\mathcal{I},\alpha}  }{ n^{2+\alpha} } \,. 
\end{align}

Let $x>0$, $n\in \mathbb{N}$, and  $r\in \mathcal{I}$ be any values  satisfying the condition~(\ref{SigmaBound}), and let $k\in \{0,\ldots, n\}$. By (I) of Lemma~\ref{LemVar}, we can rewrite the difference between $M^{n-k}(x)  $ and $ R(r-k)$   as
\begin{align}\label{D}
M^{n-k}(x)\,-\, R(r-k) \,= \,M^{n-k}(x) \,-\,M^{n-k}\big(R(r-n)\big)  \,. 
\end{align}
Since the derivative of $M^{n-k}$ is increasing, the absolute value of~(\ref{D}) is bounded by
\begin{align}
\big|M^{n-k}(x)\,-\,R(r-k)\big|\,\leq \,&\big| x\,-\,  R(r-n) \big|\frac{d}{dy}M^{n-k}(y)\Big|_{y=\max\big(x,\, R(r-n)\big)  }\,.\nonumber 
 \intertext{The chain rule formula~(\ref{DerivM}) yields  }
 \,= \,&(n-k+1)^2\big|x\,-\, R(r-n) \big|  D_{n-k}\left(M^{n-k}(y)\right)\Big|_{y=\max\big(x,\, R(r-n)\big)  } \,. \nonumber 
\intertext{By applying~(\ref{SigmaBound}) and~(\ref{RAsympt}) along with the triangle inequality, we get }
 \,= \,&\frac{4(C_{\mathcal{I},\alpha}+\mathbf{v}) }{ n^{\alpha} }  D_{n-k}\left(M^{n-k}(y)\right)\Big|_{y=\max\big(x,\, R(r-n) \big)  } \,,\nonumber 
 \intertext{where the factor of $4$ covers $\frac{(n-k+1)^2   }{ n^2 }\leq \frac{(n+1)^2}{n^2}\leq 4$.  By Lemma~\ref{LemMaybe}, $F(L):=\sup_{ k\in \mathbb{N}_0} \sup_{ y\in [0, L]}D_{k}(y)  $ is finite for any $L>0$.  Since $ M^{n-k}\big(R(r-n) \big)=R(r-k)$,    the above is bounded by    }
 \,\leq \,&\frac{4(C_{\mathcal{I},\alpha}+\mathbf{v}) }{ n^{\alpha} } F(L)\Big|_{L=\max\big( R(r-k), \,M^{n-k}(x) \big)  }
 \,.\label{LittleDif}
\end{align}

Let $k^*\equiv k^*(x,n,r)$ be the smallest $k\in \{0,\ldots,n\}$ such that  
\begin{align}\label{Yb}
\big|M^{n-k}(x)\,-\,R(r-k)\big| \, \leq \,R(r)\,+\, C_{\mathcal{I},\alpha}\,+\,\mathbf{v}\,,
\end{align}
which exists because~(\ref{Yb}) is satisfied with $k=n$ by~(\ref{SigmaBound}) and~(\ref{RAsympt}).  Note that~(\ref{Yb}) implies that $M^{n-k^*}(x)\leq 2R(r)+ C_{\mathcal{I},\alpha}+\mathbf{v}  $ since $R$ is increasing.  Thus with~(\ref{LittleDif}),  for any $x>0$, $n\in \mathbb{N}$, $r\in \mathcal{I}$ satisfying~(\ref{SigmaBound}), we have that
\begin{align}\label{LastDiff}
\big|M^{n-k^*} (x)\,-\,R(r-k^*)\big|
 \,\leq \,\frac{ 4(C_{\mathcal{I},\alpha}+\mathbf{v}) }{ n^{\alpha} } F\big(2R(r)+ C_{\mathcal{I},\alpha}+\mathbf{v}\big)
 \,.
\end{align}
We will show that $k^*=0$ whenever $n\geq N_{\mathcal{I},\mathbf{v},\alpha}$, where $N_{\mathcal{I},\mathbf{v},\alpha}>0$ is defined by
$$ N_{\mathcal{I},\mathbf{v},\alpha}\,:=\, \sup_{r\in \mathcal{I}}\bigg(\frac{ 4(C_{\mathcal{I},\alpha}+\mathbf{v}) }{ R(r) } F\big(2R(r)+ C_{\mathcal{I},\alpha}+\mathbf{v}\big)\frac{d}{dy}M(y)\Big|_{ y=2 R(r) + C_{\mathcal{I},\alpha}+\mathbf{v} }\bigg)^{\frac{1}{\alpha}} \,. $$
 Suppose to reach a contradiction that $n\geq N_{\mathcal{I},\mathbf{v},\alpha}$ and $k^*\equiv k^*(x,n ,r) >0$ for some  $x>0$, $n\in \mathbb{N}$, $r\in \mathcal{I}$  such that~(\ref{SigmaBound}) holds.  Using similar reasoning as in~(\ref{LittleDif}), the difference between $M^{n-k^*+1} (x)$ and $R(r-k^*+1)$  is bounded by 
\begin{align}
\big|M^{n-k^*+1}(x)\,-\,R(r-k^*+1)\big|
 \,\leq  \, &\big|M^{n-k^*}(x)\,-\,R(r-k^*)\big|\frac{d}{dy}M(y)\Big|_{ y=\max\big( M^{n-k^*}(x), \, R(r-k^*)\big)  }\nonumber \\
  \,\leq  \, &\frac{4(C_{\mathcal{I},\alpha}+\mathbf{v}) }{ n^{\alpha} } F\big(2R(r)+ C_{\mathcal{I},\alpha}+\mathbf{v}\big)\frac{d}{dy}M(y)\Big|_{ y=2 R(r)+ C_{\mathcal{I},\alpha}+\mathbf{v}  } \nonumber
 \,,  
\end{align}
where we have applied~(\ref{LastDiff}) in the second inequality. Since $n\geq N_{\mathcal{I},\mathbf{v},\alpha}$, the above is smaller than $R(r)$. Thus $k:=k^*-1  $ satisfies  $\big|M^{n-k} (x)\,-\,R(r-k)\big| \leq R(r) $, which contradicts that $k^*$ is the smallest element of $\{0,\ldots,n\}$ satisfying~(\ref{Yb}).  Therefore,  $k^*=0$ when $n\geq N_{\mathcal{I},\mathbf{v},\alpha}$.

Since $\big|M^{n-k} (x)\,-\,R(r-k)\big| \leq R(r)+C_{\mathcal{I},\alpha}+\mathbf{v}   $  holds for all $k\in \{0,\ldots, n\}$  when  $x>0$, $n\in \mathbb{N}$, $r\in \mathcal{I}$ satisfy~(\ref{SigmaBound}) and    $n\geq N_{\mathcal{I},\mathbf{v},\alpha}$, under these conditions on  $x$, $n$, $r$  the inequality~(\ref{LittleDif}) yields 
\begin{align*}
\max_{k\in \{0,\ldots, n\}}\big|M^{n-k}(x)\,-\,R(r-k)\big|\,\leq \,\underbrace{\Big(4(C_{\mathcal{I},\alpha}+\mathbf{v}) \sup_{r\in \mathcal{I}}F\big( 2R(r)+C_{\mathcal{I},\alpha}+\mathbf{v}  \big)\Big)}_{=:C'_{\mathcal{I},\mathbf{v},\alpha}}\frac{1  }{ n^{\alpha} }\,.
\end{align*}
Thus we have the inequality that we sought under the restriction $n \geq N_{\mathcal{I},\mathbf{v},\alpha}$. The remaining case when $n $ is smaller than $ N_{\mathcal{I},\mathbf{v},\alpha}$ is trivial.
\end{proof}

\subsection{Proof  of Proposition~\ref{PropMomApprox}}
\label{SecSharpHMC}

For $m\in \{2,3,\ldots\}$, let the polynomial $P_{m}:\R^{m-1}\rightarrow \R$ be defined as in Section~\ref{SecLemUnif}. The following lemma is from~\cite[Proposition 3.1]{Clark1}.
\begin{lemma}\label{LemPoly}
The multivariate polynomial $P_m:\R^{m-1}\rightarrow \R$ satisfies the properties below.
\begin{enumerate}[(i)]

\item   $P_m(y_2,\ldots, y_{m})$  has nonnegative coefficients, no constant term, and its only linear term is $\frac{1}{b^{m-2}} y_{m}$.  In other words, there exist polynomials $U_m:\R^{m-1}\rightarrow \R$ and $V_m:\R^{m-2}\rightarrow \R$ 
with nonnegative coefficients such that
\begin{align}\label{PMRecur}
 P_m(y_2, \ldots, y_{m})\,=\,\frac{1}{b^{m-2}}y_{m}\,+\,y_{m}U_m(y_2,\ldots, y_{m})\,+\,V_m(y_2,\ldots, y_{m-1}) \, ,   
 \end{align}
where the polynomials $y_{m}U_m(y_2,\ldots, y_{m})$ and $V_m(y_2,\ldots,y_{m-1})$  have no constant or linear terms.

\item The polynomial  $V_m(y_2,\ldots,y_{m-1})$ is a linear combination of monomials $y_{j_1}\cdots y_{j_\ell}$ 
 with 
$$j_1+\cdots +j_\ell  \geq  \begin{cases}m & \quad  \text{$m$ even,}  \\    m+1& \quad  \text{$m$ odd.}   \end{cases}$$
The polynomial $y_{m}U_m(y_2,\ldots, y_{m})$ is a linear combination of monomials with $j_1+\cdots +j_\ell\geq m+2$.    

\end{enumerate}

\end{lemma}

The next lemma follows easily from  (II) of Theorem~\ref{ThmHM}.
\begin{lemma}\label{LemRM} For any $\frak{p}\in \mathbb{N}$ and bounded interval  $\mathcal{I}\subset \R$, there is a positive number $C_{\mathcal{I},\frak{p}}$ such that for all $r\in \mathcal{I}$ and $n\in \mathbb{N}$ 
$$  R^{(2\frak{p})}(r-n)  \,\leq \, \frac{C_{\mathcal{I},\frak{p}}}{ n^{\frak{p}}}  \,. $$
\end{lemma}

We will use the notation  $\sigma^{(m)}_{k,n}:=\mathbb{E}\big[ \big(X_a^{(k,n)}\big)^m   \big]  $ and $\sigma^{(m)}_{n,n}\equiv  \sigma^{(m)}_{n}$ from~(\ref{SigmaGen}) throughout the following proof. The $m^{th}$ absolute moment of variables in the generating array $\big\{X_h^{(n)}\big\}_{h\in E_n}$ will be denoted by $\widebar{\sigma}^{(m)}_{n}$.
\begin{proof}[Proof of  Proposition~\ref{PropMomApprox}] Fix $\mathbf{v},\varkappa\geq 1$, $\alpha\in (0,1)$, and a bounded interval $\mathcal{I}\subset \R$.\footnote{Without losing any generality we can assume $\mathbf{v},\varkappa\geq 1$ rather than  $\mathbf{v},\varkappa > 0$.}  We will use  induction in $m\in\{2, 3,\ldots \}$ to show that there is a $c_m\equiv c_m(\mathcal{I},\mathbf{v}, \alpha, \varkappa)>0$ such that  for any $r\in \mathcal{I}$, $n\in \mathbb{N}$, and  i.i.d.\ array of centered random variables  $\{ X^{(n)}_h\}_{h\in E_n}$ satisfying 
\begin{enumerate}[(I)]
\item $\Big| \sigma_{n}^{2}-\kappa^2\big(\frac{1}{n}+\frac{\eta\log n}{n^2}+\frac{r}{n^2}\big)  \Big| <\frac{\mathbf{v}  }{ n^{2+\alpha}}  $ and

\item $\widebar{\sigma}^{(m)}_{n}<\frac{\varkappa}{n^{m/2}}$,

\end{enumerate}
the following inequality holds for all  $k\in \{0,\ldots, n\} $
\begin{align}\label{Rho}
\big|\sigma_{k,n}^{(m)}\big|\, \leq  \,  \frac{c_m}{(k+1)^{\frac{m}{2}}}\,.  
\end{align}

Notice that the existence of $c_2$ follows from  Lemma~\ref{LemRM} with $\frak{p}=1$ and Lemma~\ref{LemVarApprox}. Assume for the purpose of a strong induction argument that there exist constants $c_m\equiv c_m(\mathcal{I},\mathbf{v},\alpha,\varkappa)>0$ satisfying the statement above  for each $m\in \{2,\ldots, \mathbf{m}-1\}$ for some $\mathbf{m}\in \{3,4,\ldots\}$.  Let $r\in \mathcal{I}$, $n\in \mathbb{N}$, and $\{ X^{(n)}_h\}_{h\in E_n}$ be an i.i.d.\ array of centered random variables satisfying (I)-(II) for $m=\mathbf{m}$. Note that for any $m\in \{2,\ldots,\mathbf{m}-1\}$  Jensen's inequality and condition (II) give us the first two inequalities below:
$$\widebar{\sigma}^{(m)}_{n}\,\leq \, \big(\widebar{\sigma}^{(\mathbf{m})}_{n}\big)^{\frac{m}{\mathbf{m}}} \, <\Big(\frac{\varkappa}{n^{\frac{\mathbf{m}}{2}}}\Big)^{\frac{m}{\mathbf{m}}}\,\leq \,\frac{\varkappa}{n^{\frac{m}{2}}} \,. $$
The third inequality holds since $\varkappa\geq 1$.
Thus $\{ X^{(n)}_h\}_{h\in E_n}$ satisfies condition (II)  for each $m\in \{2,\ldots,\mathbf{m}-1\}$, and therefore~(\ref{Rho}) holds for all  $m\in \{2,\ldots,\mathbf{m}-1\}$ by our induction assumption.  Define $c:=\max_{2\leq m \leq \mathbf{m}-1} c_m $. 

 The last component of the recursive relation~(\ref{GG}) implies that
\begin{align}
 \big|\sigma_{k-1,n}^{(\mathbf{m})}\big|\, =  \,&   \Big|P_{\mathbf{m}}\Big( \sigma_{k,n}^{(2)} , \ldots, \sigma_{k,n}^{(\mathbf{m})} \Big)\Big|\,.
 \intertext{For $U_\mathbf{m}$ and $V_\mathbf{m}$ defined as in part (i) of Lemma~\ref{LemPoly}, the triangle inequality gives us } 
  \leq \,&\frac{1}{b^{\mathbf{m}-2}}\big|\sigma_{k,n}^{(\mathbf{m})}\big|\,+\,
 \big|\sigma_{k,n}^{(\mathbf{m})} \big|\,\Big|U_\mathbf{m}\Big(  \sigma_{k,n}^{(2)} , \ldots, \sigma_{k,n}^{(\mathbf{m})}  \Big) \Big| \,+\,\Big|V_\mathbf{m}\Big( \sigma_{k,n}^{(2)} , \ldots, \sigma_{k,n}^{(\mathbf{m}-1)} \Big)\Big|  \,.\label{Hugh}
\end{align}

Since~(\ref{Rho}) holds  for all $m\in \{2,\ldots, \mathbf{m}-1\}$, the term $\Big|V_{\mathbf{m}}\Big( \sigma_{k,n}^{(2)} , \ldots, \sigma_{k,n}^{(\mathbf{m}-1)} \Big) \Big| $  has the bound
   \begin{align}\label{zum}
  \Big|V_{\mathbf{m}}\Big( \sigma_{k,n}^{(2)} , \ldots, \sigma_{k,n}^{(\mathbf{m}-1)} \Big)\Big| \,\leq \,V_{\mathbf{m}}\Big(    c (k+1)^{-1}   , \ldots,    c (k+1)^{-\frac{\mathbf{m}-1}{2}}  \Big)   \,\leq \,\frac{c'}{(k+1)^{ \frac{\mathbf{m}}{2}}   } \,,
  \end{align}
where $c'\equiv c'(\mathcal{I},\mathbf{v},\varkappa, \alpha,\mathbf{m})$ is defined by $c'=\sup_{\ell\in \mathbb{N}_0}\, (\ell+1)^{\frac{\mathbf{m}}{2}}  V_{\mathbf{m}}\Big(     c (\ell+1)^{-1}   , \ldots,    c (\ell+1)^{-\frac{\mathbf{m}-1}{2}} \Big)  $, and we have used that $V_{\mathbf{m}}$ has nonnegative coefficients.   The supremum above is finite as a  consequence of part (ii) of Lemma~\ref{LemPoly}.  
  
Again invoking that~(\ref{Rho}) holds for all $m\in \{2,\ldots, \mathbf{m}-1\}$,  the factor $\big|U_\mathbf{m}\big(  \sigma_{k,n}^{(2)} , \ldots, \sigma_{k,n}^{(\mathbf{m})}  \big)\big|$ in~(\ref{Hugh}) has the bound
\begin{align}\label{UINEQ}
\Big|U_\mathbf{m}\Big(  \sigma_{k,n}^{(2)} , \ldots, \sigma_{k,n}^{(\mathbf{m})}  \Big)\Big|\,\leq \,& U_\mathbf{m}\Big(    c (k+1)^{-1}  , \ldots,  c  (k+1)^{-\frac{\mathbf{m}-1}{2}} \,   , \big|\sigma_{k,n}^{(\mathbf{m})}\big| \Big) \,.
\end{align}
The above also uses that the coefficients of the polynomial $U_\mathbf{m}$ are nonnegative. Since the polynomial $U_{\mathbf{m}}$ has no constant term by (i) of Lemma~\ref{LemPoly},  there is a $\mathbf{c}\equiv\mathbf{c}(\mathcal{I},\mathbf{v}, \varkappa,\alpha, \mathbf{m})>0$ such that for all $k\in \mathbb{N}_0$ and $y\in [0,1]$
\begin{align}\label{UBound}
 U_\mathbf{m}\Big(     c (k+1)^{-1}  , \ldots,   c (k+1)^{-\frac{\mathbf{m}-1}{2}}   , \, y \Big) \,\leq \,&\mathbf{c}y \,+\, \frac{\mathbf{c}}{k+1}\,.
\end{align}
Define $N_{\varkappa, \mathbf{c} }:=\textup{max}(\varkappa,8\mathbf{c})$. Note that when $n\geq N_{\varkappa, \mathbf{c} }$ the inequalities below are satisfied for $k=n$ as a consequence of assumption (II) with $m=\mathbf{m}$:
\begin{align}\label{DefEll}
\big|\sigma_{k,n}^{(\mathbf{m})}\big|\,\leq \, \min\Big( 1, \frac{1}{8\mathbf{c}}   \Big) \hspace{1cm}\text{and}\hspace{1cm} \frac{\mathbf{c}}{k+1}\,\leq \,\frac{1}{8}    \,.
\end{align}
 For $n\geq N_{\varkappa, \mathbf{c} }$ define
 $k^*_n$  as the smallest  $ k\in\{0,\ldots, n\}$ satisfying~(\ref{DefEll}).  Note that for all $k\in\{k^*_n,\ldots, n \}$
 \begin{align}\label{Condition}
  \frac{1}{ b^{\mathbf{m}-2} } + \Big|U_\mathbf{m}\Big(  \sigma_{k,n}^{(2)} , \ldots, \sigma_{k,n}^{(\mathbf{m})}  \Big)\Big|\,\leq \, \frac{3}{4}   
\end{align}  
    by~(\ref{UINEQ})-(\ref{DefEll}) and since $\mathbf{m}\geq 3$ and $b\geq 2$.

 Assume $n\geq N_{\varkappa, \mathbf{c} }$. By the  bounds~(\ref{Hugh}),~(\ref{zum}), and~(\ref{Condition}), we have the inequality below for all $k\in\{k^*_n,\ldots, n \}$.
\begin{align}\label{BowWow}
 \big| \sigma_{k-1,n}^{(\mathbf{m})}\big| \, \leq  \, \frac{3}{4}  \big|\sigma_{k,n}^{(\mathbf{m})}\big|\,+\,\frac{c'}{(k+1)^{\frac{\mathbf{m}}{2}} }
\end{align}
   Using~(\ref{BowWow}) recursively,  it follows that for any $k\in \{k^*_n-1,\ldots, n\}$
\begin{align}
 \big|\sigma_{k,n}^{(\mathbf{m})}\big| \, \leq  & \, \Big(\frac{3}{4}\Big)^{n-k}\big|\sigma_{n}^{(\mathbf{m})}\big|\,+\, \sum_{j=1}^{n-k}\Big(\frac{3}{4}\Big)^{j-1}\frac{c'}{(k+j+1)^{\frac{\mathbf{m}}{2}} }\,.\nonumber 
\intertext{The term $\big|\sigma_{n}^{(\mathbf{m})}\big|$ is bounded by the $m^{th}$ absolute moment, $\widebar{\sigma}^{(\mathbf{m})}_{n}$, and is thus  smaller than $ \frac{\varkappa }{n^{\mathbf{m}/2}} $ by assumption (II). By using $ (k+1)^{-\mathbf{m}/2}   $ to bound the terms $(k+j+1)^{-\mathbf{m}/2}     $ in the above sum, we are left with a geometric sum that we can bound by } 
   \,\leq & \, \Big(\frac{3}{4}\Big)^{n-k}\frac{ \varkappa }{  n^{\frac{\mathbf{m}}{2}} } \,+\,\frac{4c'}{(k+1)^{\frac{\mathbf{m}}{2}} } \,\leq \,\frac{c''}{(k+1)^{\frac{\mathbf{m}}{2}} }\,,\label{HERE}
\end{align}
where $c'':=\varkappa 2^{\mathbf{m}/2}+4c'$, and we have used the crude bound $k+1\leq 2n$. It follows from~(\ref{HERE}) that  $k^*_n$ is bounded from above by $ \widehat{k}\equiv\widehat{k}(\mathcal{I},\mathbf{v},\varkappa, \alpha,\mathbf{m})$ defined by
$$  \widehat{k}\,:=\,\textup{max}\Big( \big(c''\big)^{\frac{2}{\mathbf{m}}}  , \big(8\mathbf{c}c''\big)^{\frac{2}{\mathbf{m}}}  , 8\mathbf{c}  \Big)      \,.     $$

If $n\geq \max\big(\widehat{k},N_{\varkappa, \mathbf{c} }\big)$, then~(\ref{HERE})  has the form of  our desired inequality~(\ref{Rho}) for $m=\mathbf{m}$ and all  $k\in \{\widehat{k},\ldots, n\}$.   Since there are only finitely many remaining $k\in \{0,\ldots, \widehat{k}-1\}$, we can use the recursive relation $\sigma_{k-1,n}^{(\mathbf{m})} =    P_{\mathbf{m}}\Big( \sigma_{k,n}^{(2)} , \ldots, \sigma_{k,n}^{(\mathbf{m})} \Big)$ and our induction assumption to bound the remaining terms by a constant depending only on $\mathcal{I}$, $\mathbf{v}$, $\varkappa$, $\alpha$, and $\mathbf{m}$.  Finally, we can pick our constant  large enough to extend the inequality  to the finitely many $n\in \mathbb{N}$ with $n<\max\big(\widehat{k},N_{\varkappa, \mathbf{c} }\big)$.  By induction this completes the proof. 
\end{proof}

\begin{appendix}

\section{Inverse temperature scaling}\label{AppendBetaScale}
We will outline the calculation verifying that the variance scaling~(\ref{VarAsym}) determines the inverse temperature scaling $\beta_{n,r}$ in~(\ref{BetaForm}).  In other terms,  $V(\beta_{n,r})=V_{n,r}+\mathit{o}(1/n^2)$ as $n\rightarrow \infty$ for
$$  V_{n,r}\,:=\, \frac{\kappa_b^2}{n}\,+\,\frac{\kappa_b^2\eta_b\log n}{n^2}\,+\, \frac{\kappa_b^2 r}{n^2}\hspace{1cm}\text{and}\hspace{1cm}V(\beta)\,:=\,\textup{Var}\bigg( \frac{e^{\beta\omega}}{\mathbb{E}[ e^{\beta\omega}] }  \bigg)\,.
$$ 
Recall that  $\tau:=\mathbb{E}[\omega^3]$ and $\tau':=\mathbb{E}[\omega^4]-3$. Since $\mathbb{E}[ e^{\beta\omega}]=1+\frac{1}{2}\beta^2 + \frac{\tau}{6}\beta^3+\frac{\tau'+3}{24}\beta^4+\mathit{O}(\beta^5)$    for $0<\beta\ll 1$, a computation shows that 
\begin{align}
V(\beta)\,:=\,  \frac{\mathbb{E}[e^{2\beta\omega}]- \mathbb{E}[e^{2\beta\omega}]^2 }{\mathbb{E}[ e^{\beta\omega}]^2 } \,=\, \beta^2 + \tau\beta^3+\bigg(\frac{1}{2}+\frac{7\tau'}{12}\bigg)\beta^4+\mathit{O}\big(\beta^5\big)  \,.\label{VAsymp}
\end{align}   
Another computation using the expansion~(\ref{VAsymp}) shows that for small $\beta>0$
\begin{align}\label{kn}
\beta\,=\,\sqrt{V(\beta)}-\frac{1}{2}\tau V(\beta)+\bigg(\frac{5 \tau^2 }{8}-\frac{1}{4}-\frac{7\tau'}{24}\bigg) \big(V(\beta)   \big)^{\frac{3}{2}}+\mathit{O}\big(   \beta^4\big)\,.
\end{align}
Substituting  $ V_{n,r}+ \mathit{o}\big(\frac{1}{n^2}  \big) $ in for  $V(\beta)$ on the right side of~(\ref{kn}) yields
\begin{align}
\sqrt{V_{n,r} }-\frac{\tau}{2} V_{n,r} &+\bigg(\frac{5 \tau^2 }{8}-\frac{1}{4}-\frac{7\tau'}{24}\bigg) V_{n,r}^{\frac{3}{2}}+\mathit{O}\Big(  \frac{1}{n^{2}} \Big) \nonumber   \\
\,=\,&
\frac{\kappa_b}{\sqrt{n}}  -\frac{\tau}{2}  \frac{\kappa_b^2}{n} +\frac{\kappa_b\eta_b\log n}{2n^{\frac{3}{2}}}
+ \frac{\kappa_b r+\kappa_b^3\big(\frac{5 \tau^2}{4} -\frac{1}{2}-\frac{7\tau'}{12}\big) }{2n^{\frac{3}{2}}}  + \mathit{o}\Big(\frac{1}{n^{\frac{3}{2}}}\Big)\,,\label{kl}
\end{align}
which is the asymptotic form for $\beta_{n,r} $ in~(\ref{BetaForm}).  Alternatively, if we substitute the sharper asymptotic form $ V_{n,r}+ \mathit{O}\big(\frac{1 }{n^{2+\alpha}}  \big) $ in for  $V(\beta)$ on the right side of~(\ref{kn}), then $\mathit{o}\big(\frac{1}{n^{3/2}}\big) $ can be replaced by $\mathit{O}\big(\frac{1}{n^{3/2+\alpha}}\big) $ on the right side of~(\ref{kl}).

\vspace{.1cm}

For the site-disorder model, the inverse temperature scaling~(\ref{BetaForm2}) results in the variance scaling~(\ref{VarScaling2}) since by~(\ref{VAsymp}) we have
\begin{align*}
  V\big(\widehat{\beta}_{n,r}\big)\,=\,\widehat{\beta}_{n,r}^2+\tau  \widehat{\beta}_{n,r}^3+\mathit{O}\big( \widehat{\beta}_{n,r}^4 \big) \,=\,&\bigg( \frac{  \widehat{\kappa}_b }{ n }+\frac{  \widehat{\kappa}_b\eta_b \log n  }{ n^2 }+\frac{ \widehat{\kappa}_b r-\widehat{\kappa}_b^2\frac{\tau}{2} }{ n^2 }+\mathit{o}\Big( \frac{1}{n^2} \Big)   \bigg)^2+\tau\Big( \frac{  \widehat{\kappa}_b }{ n }  \Big)^3+  \mathit{O}\Big( \frac{1}{n^4}  \Big)\\
  \,=\,&  \widehat{\kappa}_b^2\bigg(  \frac{ 1 }{ n^2 }+\frac{  2\eta_b \log n  }{ n^3 } + \frac{  2r }{ n^3 }\bigg)+\mathit{o}\Big( \frac{1}{n^3}  \Big)\,.
  \end{align*}

\section{Variance function consistency check}\label{AppendixVarCheck}
There is instructional value in implementing a consistency check between properties (I) and (II) in the statement of Lemma~\ref{LemVar}, i.e., between the claim that $M\big(R(r)\big)=R(r+1)$ and the $-r\gg 1$ asymptotics
\begin{align}\label{RAsymp}
R(r)\,=\,\frac{\kappa^2  }{-r}\,+\,\frac{\kappa^2 \eta\log(-r)  }{r^2}\,+\,\mathit{O}\bigg( \frac{\log^2(-r)}{r^3}  \bigg)\,,
\end{align}
where $\kappa^2:=\frac{2}{b-1}$ and $\eta:= \frac{b+1}{3(b-1)}  $.  Fix some $r$ with $-r\gg 1$ and define $V_n=R(r-n)$ for $n\in \mathbb{N}_0$.  We begin by writing $R(r)$ as a telescoping sum
\begin{align}\label{VTelescope}
R(r)\,=\,&\sum_{k=1}^\infty \big(  V_{k+1} \, -\, V_k  \big)\,=\,\sum_{k=1}^\infty \big( M( V_k  ) \, -\, V_k  \big)\,.
\intertext{Since $V_n$ vanishes as $n\rightarrow \infty$ and the map $M(x)=\frac{1}{b}\big[(1+x)^b\,-\,1  \big]$ has the $0<x\ll 1$ asymptotics $M(x)= x+\frac{b-1}{2}x^2 +\frac{(b-1)(b-2)}{6}x^3+\mathit{O}(x^4) $, the equality~(\ref{VTelescope})  can be written as }
\,=\,&\underbrace{\frac{b-1}{2}\sum_{k=1}^\infty V_{k}^2}_{\text{(\textbf{a})}}\, +\, \underbrace{\frac{(b-1)(b-2)}{6}\sum_{k=1}^\infty V_{k}^3}_{\text{(\textbf{b})}} \, +\,\underbrace{\sum_{k=1}^\infty \mathit{O}\big( V_{k}^4   \big) }_{\text{(\textbf{c})}} \,.\nonumber
\end{align}
We will analyze the expressions (\textbf{a}), (\textbf{b}), and (\textbf{c}) to verify that the right side of~(\ref{VTelescope}) has the asymptotics~(\ref{RAsymp}). The expression (\textbf{c}) is $\mathit{O}(1/r^3)$ since the terms $V_k$ are bounded by a constant multiple of $(k-r)^{-1}$ as a consequence of~(\ref{RAsymp}).

Applying~(\ref{RAsymp}) to $V_k$ in the  expression (\textbf{a}) yields
 \begin{align*}
\text{(\textbf{a})}\,=\,&\frac{b-1}{2}  \sum_{k=1}^\infty \bigg( \frac{\kappa^2}{k-r}\,+\,\frac{\eta\kappa^2 \log (k-r)}{(k-r)^2}\,+\,\mathit{O}\bigg(\frac{\log^2(k-r)}{(k-r)^3}\bigg)  \bigg)^2  \,. \intertext{Foiling the square and using that $\kappa^{-2}=(b-1)/2$, we can write} \,=\,& \kappa^2\sum_{k=1}^\infty \frac{1}{(k-r)^2}\,+\, 2\eta\kappa^2 \sum_{k=1}^\infty\frac{\log (k-r)}{(k-r)^3}\,  +\, \sum_{k=1}^\infty\mathit{O}\bigg(\frac{\log^2(k-r)}{(k-r)^4}\bigg)  \\ \,=\,& \frac{\kappa^2}{-r}\underbracket{\,-\,\frac{1}{(b-1)r^2}}\,+\, \eta\kappa^2 \frac{\log (-r)}{r^2}\,+\,\underbracket{\frac{\eta\kappa^2}{2r^2}}\,+\,\mathit{O}\bigg( \frac{\log^2(-r)}{ r^3} \bigg)\,, 
\end{align*}
where we have used a trapezoidal Riemann approximation to get
\begin{align*}
\sum_{k=1}^\infty \frac{1}{(k-r)^2}\,=\,&-\frac{1}{2r^2}\,+\, \frac{1}{2}\sum_{k=1}^\infty\Big( \frac{1}{(k-r)^2}+ \frac{1}{(k-1-r)^2}  \Big) \,+\,\mathit{O}\Big(\frac{1}{r^3}\Big)    \\ \,=\,&-\frac{1}{2r^2}\,+\, \frac{1}{-r}\int_{0}^\infty \frac{1}{(1+x)^2}dx \,+\,\mathit{O}\Big(\frac{1}{r^3}\Big) 
\\ \,=\,&-\frac{1}{2r^2}\,+\, \frac{1}{-r} \,+\,\mathit{O}\Big(\frac{1}{r^3}\Big) \,,
\end{align*}
and right-hand Riemann approximations to get
\begin{align*}
\sum_{k=1}^\infty \frac{\log (k-r)}{(k-r)^3}\,=\,&\frac{\log (-r)}{ (-r)^3}\sum_{k=1}^\infty \frac{1}{(1+\frac{k}{-r})^3}\,+\, \frac{1}{(-r)^3}\sum_{k=1}^\infty \frac{\log (1+\frac{k}{-r})}{(1+\frac{k}{-r})^3}\\
\,=\,&\frac{\log (-r)}{ r^2}\int_0^\infty \frac{1}{(1+x)^3}dx \,+\, \frac{1}{r^2}\int_0^\infty \frac{\log(1+ x)}{(1+x)^3}dx \,+\, \mathit{O}\bigg(\frac{\log (-r) }{ r^3}\bigg)
\\
\,=\,&\frac{1}{2}\frac{\log (-r)}{ r^2}\,+\, \frac{1}{4r^2}\,+\, \mathit{O}\bigg(\frac{\log (-r) }{ r^3}\bigg)\,.
\end{align*}

Again applying~(\ref{RAsymp}) to $V_k$, foiling, and using that $\kappa^{-2} =\frac{b-1}{2} $,  the expression (\textbf{b}) is equal to
\begin{align*}
\text{(\textbf{b})}\,=\,&\frac{(b-1)(b-2)}{6}  \sum_{k=1}^\infty \bigg( \frac{\kappa^2}{k-r}\,+\,\frac{\eta\kappa^2 \log (k-r)}{(k-r)^2}\,+\,\mathit{O}\bigg(\frac{\log^2 (k-r)}{(k-r)^3}\bigg)  \bigg)^3 \\ \,=\,&\frac{b-2}{3}  \sum_{k=1}^\infty \bigg( \frac{\kappa^4}{(k-r)^3}\,+\,\mathit{O}\bigg(\frac{\log (k-r)}{(k-r)^4}\bigg) \bigg)  \\  
\,=\, & \underbracket{\frac{b-2}{3}\frac{\kappa^4  }{2r^2 }}\,+\,\mathit{O}\bigg(\frac{\log(-r)}{r^3}\bigg)\,,
\end{align*}
where we have used the Riemann approximation
\begin{align*}
\sum_{k=1}^\infty  \frac{1}{(k-r)^3} \,=\, \frac{1}{r^2}\int_0^\infty \frac{1}{(1+x)^3}dx\,+\,\mathit{O}\Big(\frac{1}{r^3}\Big)\,=\,\frac{1}{2r^2}\,+\,\mathit{O}\Big(\frac{1}{r^3}\Big)\,.
\end{align*}
Summing up (\textbf{a}), (\textbf{b}), and (\textbf{c}) gives the desired asymptotics~(\ref{RAsymp}) as a result of the cancellation  $-\frac{1}{(b-1)r^2}+ \frac{\eta\kappa^2}{2r^2}+\frac{b-2}{3}\frac{\kappa^4  }{2r^2 }=0$ between the bracketed terms above.

\section{The zero bias approach to Stein's method}\label{AppendixGoldstein}

We will discuss the zero bias variation on Stein's method introduced in~\cite{Goldstein2}, which provides an easy proof of Lemma~\ref{LemNorm} (restated in Lemma~\ref{LemGold}).

\subsection{Zero bias  transformation}

Let $X$ be a centered random variable with variance $\sigma^2$. The zero bias transformation, $X^*$, of $X$    is the distribution satisfying
$$ \mathbb{E}\big[ f'(X^*)  \big]\,=\,  \frac{1}{\sigma^2} \mathbb{E}\big[X   f(X)  \big]    $$
for all absolutely continuous functions $f$ on $\R$. The right side above can be written as   
$$\frac{1}{\sigma^2} \mathbb{E}\big[X   f(X)  \big]\,=\, \mathbb{E}\Bigg[\frac{X^2}{\sigma^2} \frac{\int_0^{X}f'(r)dr}{X}  \Bigg] \,.  $$
Thus if $X$ has distribution measure $\mu$, then $X^*$ is constructed by choosing a number $x$ using the measure $\nu(dx)=\frac{x^2}{\sigma^2}\mu(dx)$ and then picking a number uniformly at  random from the interval between $0$ and $x$.  The normal distribution is the unique fixed point for the zero bias transformation:
\begin{lemma} Let $X$ be a centered random variable with variance $\sigma^2$.  Then $X\stackrel{d}{=}X^*$ iff $X\sim \mathcal{N}(0,\sigma^2)$.
\end{lemma}

\begin{lemma}\label{LemZeroMom} Let $X$ be a centered random variable with variance $\sigma^2$ and finite absolute moment $\varsigma_n:=\mathbf{E}\big[ |X|^n  \big]$ for some $n\geq 3$. The absolute moment $\varsigma_{n-2}^*$ of $X^* $  is finite and equal to $ \varsigma_{n-2}^*= \frac{ \varsigma_{n}}{ \sigma^2 (n-1)}$.       \end{lemma}
\begin{proof}This follows easily from the definition of $X^*$ since
\begin{align*}\varsigma_{n-2}^*\,=\,\mathbb{E}\big[ |X^*|^{n-2}  \big]\,=\,\mathbb{E}\Bigg[\frac{X}{\sigma^2} \int_0^{X} |r|^{n-2} dr   \Bigg]\,=\,\frac{\mathbb{E}[ |X|^{n}  ]   }{\sigma^2 (n-1) }\,=\,\frac{\varsigma_{n}   }{\sigma^2 (n-1) }\,.\end{align*}\end{proof}

The  lemma below gives a key distributional identity for the zero bias transformation of a finite sum of independent random variables; see, for instance, Lemma 2.2 of~\cite{Goldstein1} for the proof.
\begin{lemma} Let $X_1, \ldots, X_n$ be independent centered random variables with $\textup{Var}(X_k)=\sigma_k^2$.  Let $\textbf{i}$ be a variable taking values in $\{1, 2,  \ldots , n\}$ with probability $ \mathcal{P}\big[\textbf{i}=k  \big]=\frac{\sigma_k^2}{\sigma_1^2+\cdots +\sigma_n^2   }  $. The distribution of $(X_1 +\cdots + X_n)^*$ has the form
$$ (X_1 +\cdots + X_n)^*\,\stackrel{d}{=}\, X_1 +\cdots + X_n\,+\,\big(X^*_\textbf{i}-X_\textbf{i})\,,    $$
where $\textbf{i}$ is independent of the random variables $X_k$ and $X^*_k$. In other terms, the $k^{th}$ variable $X_k$ in the sum is replaced by $X_k^*$ with probability $\frac{\sigma_k^2}{\sigma_1^2+\cdots +\sigma_n^2   } $.
\end{lemma}

\subsection{Relation to Stein's method}

Recall that  $ \rho_1(X,Y    ) :=\sup_{ h\in \textup{Lip}_1  } \mathbb{E}\big[h(X)-h(Y)    \big]      $ for two random variables $X$ and $Y$ with finite first absolute moments.  Also, recall that the auxiliary function $f$  for a  given 
 $h\in \textup{Lip}_1$ in Stein's method  satisfies the differential equation
\begin{align*} 
 f'(x)\,-\, \frac{x}{\sigma^2} f(x)\,=\,h(x)\,-\,\int_{\R} h(r)\frac{ e^{-\frac{r^2}{2\sigma^2}  }}{\sqrt{2\pi \sigma^2}  }dr\,
 \end{align*}
and that the first- and second-order derivatives have the bounds $\sup_x |f'(x)| \leq 1$ and $\sup_x |f''(x)| \leq 2$.  In particular $f'$ is absolutely continuous with Lipschitz constant $\leq 2$.   If $X$ is a centered random variable with variance $\sigma^2$ and $\mathcal{X}\sim \mathcal{N}\big(0,\sigma^2\big) $, then by definition of $X^*$ we have
\begin{align*}
\mathbb{E}\big[h(X)\,-\,h(\mathcal{X})\big]\,=\,\mathbb{E}\Big[f'(X)\,-\,\frac{X}{\sigma^2}f(X)\Big]\,=\,\mathbb{E}\big[f'(X)\,-\,f'(X^*)\big]\,.
\end{align*}
Thus, by supremizing over $h\in \textup{Lip}_1$ above, we have the bound $ \rho(X,\mathcal{X}    )\,\leq \, 2\rho\big(X, X^*  \big)   $
since $|f''|\leq 2 $.  Therefore, the Wasserstein-$1$ norm between $X$ and the normal random variable $\mathcal{X}$ is smaller than two times the Wasserstein-$1$ norm between  $X$  and its zero bias transformation.

\begin{lemma}\label{LemGold}
Let $X_1$,\dots,$X_{n}$ be i.i.d.\ variables with mean $0$ and variance $\sigma^2$.  For $Y_n  := \frac{X_1+\cdots +X_n}{\sqrt{n}} $, we have the inequality
\begin{align*}
\rho_1\big( Y_n, Y_n^* \big)\, \leq \,\frac{1}{\sqrt{n}}\rho\big(X_1,X_1^*  \big)   
\end{align*}
for $Y_n  = \frac{X_1+\cdots +X_n}{\sqrt{n}} $.  Moreover, if 
$\mathbb{E}\big[ |X_1|^3  \big]<\infty   $ and $\mathcal{Y}\sim \mathcal{N}\big(0,\sigma^2\big)$, then 
\begin{align*}
\rho_1\big( Y_n, \mathcal{Y}\big)\, \leq \,\frac{3}{\sqrt{n}}\frac{ \mathbb{E}\big[ |X_1|^3  \big]  }{ \sigma^2 }\,.  
\end{align*}

\end{lemma}

\begin{proof} Let the pairs $(X_k,X_k^*)$ be i.i.d.\ couplings of the variables $X_k$ and $X_k^*$ such that  $$\rho_1(X_k, X_k^*)\,=\,\mathbb{E}\big[|X_k-X_k^*|\big]\,. $$
Then $\rho_1( Y_n, Y_n^* )$ is bounded as follows:
\begin{align*}
\rho_1( Y_n, Y_n^* )\,=\,\sup_{ \|h\|_{Lip}\leq 1  } \mathbb{E}\big[h(Y_n)-h(Y_n^*)    \big]  \,\leq \,\mathbb{E}\big[|Y_n - Y_n^*|    \big]  \,=\,\frac{1}{\sqrt{n}}\mathbb{E}\big[|X_\mathbf{i} - X_\mathbf{i}^*|    \big] \,=\,\frac{1}{\sqrt{n}}\mathbb{E}\big[|X_1 - X_1^*|    \big]\,,
\end{align*}
and the last term is equal to $\frac{1}{\sqrt{n}}\rho_1(X_1 , X_1^*)$ by assumption.  Next we simply observe that 
$$\rho_1(X_1 , X_1^*)\,= \,\mathbb{E}\big[|X_1 - X_1^*|    \big]\,\leq \,\mathbb{E}\big[|X_1|\big] \,+\, \mathbb{E}\big[| X_1^* |    \big]\,\leq \,\mathbb{E}\big[|X_1|\big] \,+\, \frac{1}{2\sigma^2}\mathbb{E}\big[| X_1 |^3    \big]\,\leq \,\frac{3}{2\sigma^2}\mathbb{E}\big[| X_1 |^3    \big]  \,,  $$
where the second inequality is by Lemma~\ref{LemZeroMom}.  The result then holds because $\rho_1( Y_n, \mathcal{Y} )\leq 2\rho_1( Y_n, Y_n^* )$.\end{proof}

\begin{proof}[Proof of Lemma~\ref{CorNorm}]
Lemma~\ref{Lemma1to2} gives us the inequality
\begin{align}
\rho_2\big(\widebar{X}_n, \mathcal{X}\big)\,\leq   &\,2^{\frac{2}{3}} \left(\rho_1\big(\widebar{X}_n, \mathcal{X}\big)   \right)^{\frac{1}{3}}\Big(\mathbb{E}\big[  \widebar{X}_n^4\big]^{\frac{1}{6}}\, +\, \mathbb{E}\big[\mathcal{X}^{4}\big]^{\frac{1}{6}} \Big)\,. \nonumber 
  \intertext{Since $\mathbb{E}\big[\mathcal{X}^{2}\big]=\sigma^2$ and $\mathcal{X}\sim \mathcal{N}(0,\sigma^2)$, we have $\mathbb{E}\big[\mathcal{X}^{4}\big]       =3\sigma^4 \leq 3 \mathbb{E}\big[\widebar{X}_n^4\big]        $.     Thus with Lemma~\ref{LemNorm},   }
  \, \leq  &\, 2^{\frac{2}{3}} \bigg(\frac{3}{\sqrt{n}}\frac{ \mathbb{E}\big[ |X_1|^3  \big]  }{ \sigma^2 }\bigg)^{\frac{1}{3}}\big(1+3^{\frac{1}{6}}\big) \mathbb{E}\big[\widebar{X}_n^4\big]^{\frac{1}{6}} \, \leq  \,6 \Bigg(\frac{1}{\sqrt{n}}\frac{ \mathbb{E}\big[ X_1^4 \big]^{\frac{3}{4}}  }{ \sigma^2 }\Bigg)^{\frac{1}{3}}\mathbb{E}\big[X_1^4\big]^{\frac{1}{6}}\,\leq \, 6  n^{-\frac{1}{6}} \frac{  \mathbb{E}\big[X_1^4\big]^{\frac{5}{12}}  }{ \sigma^{\frac{2}{3}}   }\,.\nonumber
\end{align}
The second inequality above uses that $\mathbb{E}\big[\widebar{X}_n^4\big]= 3\sigma^4\big(1-\frac{1}{n}\big)+\frac{1}{n}\mathbb{E}\big[X_1^4\big] $ is smaller than $3\mathbb{E}\big[X_1^4\big]$ and $2^{\frac{2}{3}}3^{\frac{1}{3}}\big(1+3^{\frac{1}{6}}\big)<6 $.\end{proof}

\end{appendix}

\end{document}